\documentclass[hebrew,english,12pt,a4paper,oneside,onecolumn]{book}

\usepackage{silence}
\usepackage[T1]{fontenc}
\usepackage[utf8x]{inputenc}

\makeatletter\let\l@hebrew\l@nohyphenation\makeatother
\usepackage[main=english,hebrew]{babel}

\usepackage{amsthm}
\usepackage[framemethod=default]{mdframed}

\makeatletter
\newcommand{\addnumberlesstotoc}[2]{
  \addcontentsline{toc}{#1}{\protect\numberline{}#2}
}
\makeatother

\usepackage{datetime}
\ddmmyyyydate

\usepackage{setspace}
\usepackage{fancyhdr}
\fancypagestyle{plain}{\pagestyle{fancy}}
\usepackage[labelformat=simple]{subcaption}

\usepackage{amsmath}
\usepackage{diagbox}

\usepackage{mdframed}
\usepackage{listings}
\usepackage{xcolor}
\usepackage[splitrule]{footmisc}
\usepackage{titlesec}
\usepackage{float}
\usepackage{graphicx}
\usepackage{bigints}
\graphicspath{{images/}}
\usepackage{mathpazo}
\usepackage{hyperref}
\usepackage{bookmark}
\titleformat{\chapter}[hang]{\bf\huge}{\thechapter}{2pc}{}
\usepackage{multirow}
\DeclareMathOperator{\spn}{span}
 
\DeclareMathOperator*{\argmin}{argmin}
\usepackage{tikz}

\usepackage{epigraph}
\usepackage{bbm}
\usepackage{relsize}
\usepackage{mathtools}

\usetikzlibrary{arrows, automata,
                quotes,
                positioning
                }
\usetikzlibrary{patterns}

\DeclareMathOperator{\rank}{rank}
\DeclareMathOperator{\diag}{diag}

\newmdtheoremenv{theorem}{Theorem}[section]
\newtheoremstyle{definition}
{9pt}           
{9pt}           
{}              
{}           
{\bfseries}     
{\ }            
{0.5em}             
{}              

\theoremstyle{definition}
\newtheorem{definition}{Definition}[section]

\theoremstyle{plain}
\newmdtheoremenv{corollary}{Corollary}[section]
\newmdtheoremenv{lemma}{Lemma}[section]

\newmdtheoremenv{proposition}{Proposition}[section]

\theoremstyle{remark}
\newtheorem{remark}{Remark}[section]
\usepackage{algorithm2e}
\definecolor{pycomment}{rgb}{0.0, 0.5, 0.0}
\definecolor{pykeyword}{rgb}{0.7, 0.0, 0.5}
\definecolor{pystring}{rgb}{0.6, 0.1, 0.1}

\lstdefinestyle{mypython}{
  language=Python,
  basicstyle=\ttfamily\scriptsize,
  keywordstyle=\color{pykeyword}\bfseries,
  commentstyle=\color{pycomment}\slshape,
  stringstyle=\color{pystring},
  showstringspaces=false,
  breaklines=true,
  tabsize=2,
  numbers=left,
  numberstyle=\tiny\color{gray},
  numbersep=8pt,
  frame=lines,
  framesep=3mm
}

\lstdefinestyle{mypythoninline}{
  language=Python,
  basicstyle=\ttfamily\small,
  keywordstyle=\color{pykeyword}\bfseries,
  commentstyle=\color{pycomment}\slshape,
  stringstyle=\color{pystring},
  showstringspaces=false
}

\newcommand*{\pyth}{\lstinline[style=mypythoninline]}
\newcommand*{\pythoninline}{\lstinline[style=mypythoninline]}

\lstnewenvironment{minted}[2][]{%
  \lstset{style=mypython}%
}{}
\title{BGU Thesis Template}

\begin{document}

\pdfbookmark{Cover}{Cover}
\frontmatter
\begin{titlepage}
    \begin{center}
        \vspace*{1cm}
        
        \includegraphics[width=0.1\textwidth]{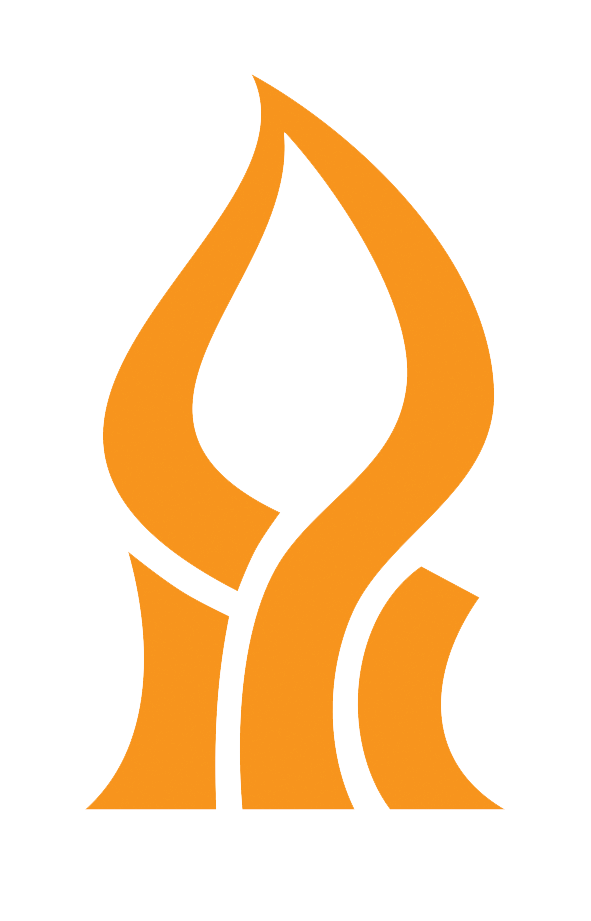}\\
        Ben-Gurion University of the Negev\\
        The Faculty of Natural Sciences\\
        The Department of Computer Science
        
        \vspace{2cm}
        
        {\Large \textbf{Differential Games for Compositional Handling of Competing Control Tasks}}
        
        \vspace{1.5cm}
        
        \textbf{Joshua Shay Kricheli}
        
        \vspace{1cm}
        
        Thesis submitted in partial fulfillment of the requirements\\for the Master of Sciences degree
        
        \vspace{1cm}
        
        Under the supervision of \\\textbf{Prof. Gera Weiss}, The Department of Computer Science \\
        and \textbf{Dr. Shai Arogeti}, The Department of Mechanical Engineering
        
        \vfill
        
        \textbf{November 2022}
    \end{center}
\end{titlepage}
\begin{titlepage}
    \begin{center}
        
        \includegraphics[width=0.1\textwidth]{logos/bgu.png}\\
        Ben-Gurion University of the Negev\\
        The Faculty of Natural Sciences\\
        The Department of Computer Science
        
        \vspace{1.3cm}
        
        {\Large \textbf{Differential Games for Compositional Handling of Competing Control Tasks}}
        
        \vspace{1cm}
        
        \textbf{Joshua Shay Kricheli}
        
        \vspace{1cm}
        
        Thesis submitted in partial fulfillment of the requirements\\for the Master of Sciences degree
        
        \vspace{1cm}
        
        Under the supervision of \\\textbf{Prof. Gera Weiss}, The Department of Computer Science \\
        and \textbf{Dr. Shai Arogeti}, The Department of Mechanical Engineering
        
        \vspace{1cm}
        \tikz[remember picture,overlay]
    \path (current page.south west) ++ (10.2,6.9) node {\includegraphics[width=0.15\columnwidth]{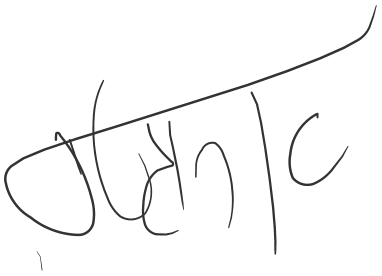}
    };
    \tikz[remember picture,overlay]
    \path (current page.south west) ++ (9,8.3) node {\includegraphics[width=0.15\columnwidth]{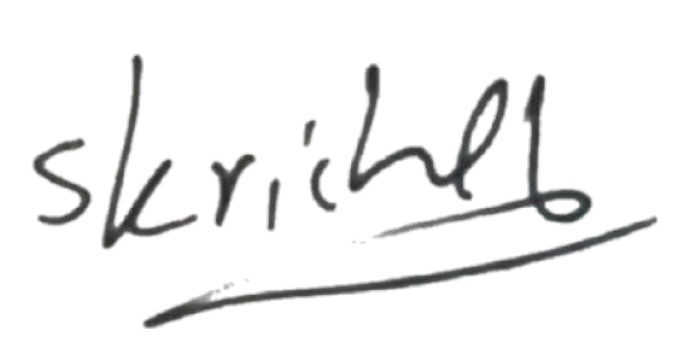}
    };
    \tikz[remember picture,overlay]
    \path (current page.south west) ++ (9.3,7.5) node {\includegraphics[width=0.08\columnwidth]{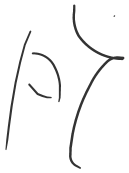}
    };
        \begin{flushleft}
        Signature of student: \( \rule{3cm}{0.15mm} \) \hfill Date: \( \underline{ \,\,\,\,\,\,\,\,08.11.2022\,\,\,\,\,\,\,\,\,\,   } \)\\
        Signature of supervisor: \( \rule{3cm}{0.15mm} \) \hfill Date: \( \underline{ \,\,\,\,\,\,\,\,08.11.2022\,\,\,\,\,\,\,\,\,\,   } \)\\
        Signature of supervisor: \( \rule{3cm}{0.15mm} \) \hfill Date: \( \underline{ \,\,\,\,\,\,\,\,08.11.2022\,\,\,\,\,\,\,\,\,\,   } \)\\
        Signature of chairperson of the\\committee for graduate studies: \( \rule{3cm}{0.15mm} \) \hfill Date: \( \underline{ \,\,\,\,\,\,\,\,\,\,\,\,\,\,\,\,\,\,\,\,\,\,\,\,\,\,\,\,\,\,\,\,\,\,\,\,\,\,\,\,\,   } \)
        \end{flushleft}
        \vfill
        \textbf{November 2022}
    \end{center}
\end{titlepage}
\clearpage

\pdfbookmark{Abstract}{Abstract}
\thispagestyle{plain}

\begin{center}
    \Large
    Differential Games for Compositional Handling of Competing Control Tasks

    \vspace{0.2cm}
    \large
    Joshua Shay Kricheli
       
    \vspace{0.2cm}
    \large
    Master of Sciences Thesis 
      
    \vspace{0.2cm}
    \large
    Ben-Gurion University of the Negev 
    
    \vspace{0.2cm}
    \large
    2022
    
    \vspace{0.2cm}
    \Large
    \textbf{Abstract}
\end{center}

We introduce a novel Divide and Conquer control design methodology to implement the use of differential games in single-agent, multi-objective dynamical systems. The described approach starts with associating each control objective with a virtual input and then considering a non cooperative, finite or infinite, differential game between a corresponding set of representative players, each trying to attain the best strategy for his designated goal, all the while knowing the rest of the players' chosen optimal policy. The approach is flexible in that it associates a virtual cost function to each player, effectively supplying each its weighting parameters for his objective, the entire system state, and all other virtual inputs. By guarantying a Nash Equilibrium for this game, we synthesize a composite controller that establishes a stable balance between the objectives and allows the control engineer an approachable manner to re-tune the parameters along the design cycle. We provide a mathematical derivation of the design pattern, both for continuous and discrete-time systems, and employ it in single-agent large-scale frameworks, where multiple elaborate control tasks can often conflict with each other dynamically, thus rendering the weighting of the overall system beforehand highly challenging. To demonstrate the use of the suggested approach, we develop an open-source Python package, in which we enact a novel algorithm we developed for solving Algebraic Riccati Equations that arise in the infinite horizon differential game derivation. We use two case studies to formulate appropriate differential games and secure a solution; an inverted pendulum on a moving cart and a non-linear hierarchically-controlled quadrotor. We compare the resulting performance with the Linear Quadratic Regulator optimal control technique across multiple transient and steady-state control metrics and show preferable results.
\clearpage

\pdfbookmark{Acknowledgements}{Acknowledgements}
\chapter*{Acknowledgements}
This paper was written after several years of copious efforts and it is the culmination of the work our team has put together. I was in the Computer Science department at Ben Gurion University for the past 4.5 years, which was a remarkable journey. I hope this work concludes this journey respectfully.

I was at the start of the fourth year of my B.Sc. degree in Mechanical Engineering at Ben Gurion University, near the end of 2017, when I decided I wish to proceed and strengthen my knowledge in the subjects I majored in - Robotics, Mechatronics and mainly, what this work focuses on - Modern Control Theory. As I found interest in its formal mathematical, algorithmic and computational aspects, I reached out to the Computer Science department at the university. Upon doing so, I met Prof. Gera Weiss, who was the Head of the Teaching Committee of advanced studies. 

Gera was extremely welcoming and enthusiastic and greeted me with nothing but kindness and graciousness. Moreover, Gera shared with me that his Ph.D. thesis was also exactly in Control Theory, which immediately initiated a cooperative effort we have established since then. 

We founded a team of multiple expert researchers, from Mechanical Engineering, Software Engineering, Information Systems Engineering, and Computer Science, each with his expertise, but all with a shared enthusiasm to improve contemporary computational intelligent models. Gera led (and still leading when this is written) the team, along with Dr. Shai Arogeti, who is the head of the Bachelor's track and Master's program in Robotics, Mechatronics, and Control at the Mechanical Engineering department. 

We have since then embarked on an adventurous journey together, published several articles as a team, submitted and won several prizes, and made interpersonal friendships with one another, so much so that when we all meet it does not even fell as a work meeting but a '\textit{parliament of friends}' as Shai Arogeti once said.

I wish to thank all my colleagues, friends, and fellow researchers in the Software Engineering, Information System Engineering, Electrical Engineering, Mechanical Engineering, Mathematics and Physics departments, mostly at Ben Gurion University but in some other institutions as well. Most notably I would like to thank \textbf{Dr. Aviran Sadon} who was there at troublesome nights to help me with stubborn coding bugs. Without his help, both our conference paper and this thesis would not have been.

I wish to extend the utmost gratitude and respect to my two advisors: 

\begin{itemize}
    \item \textbf{Prof. Gera Weiss}, who in our countless meetings, phone talks, and correspondences has become a close friend of mine, to a degree that we feel comfortable with each other like we had known each other from high school. He has aided me in so many ways and has taught me so many things, that I am a better researcher and a professional thanks to him.
    \item \textbf{Dr. Shai Arogeti}, whom I know and work with for more than 5 years through my Mechanical Engineering studies, has always been an extraordinary person that upheld the most excellent research standards, all the while maintaining a fun collaborative environment to work in. Shai has aided me in this work in so many ways that I am ever thankful to him.
\end{itemize}

And in the end it all boils down to my dearest friends and then... \\
To my family, \textbf{Dudu}, \textbf{Elizabeth}, \textbf{Mom} and \textbf{Dad}. This is all for you.
\vspace{7pt}
\hrule
\vspace{5pt}
This research was supported in part by the Helmsley Charitable Trust through the Agricultural, Biological and Cognitive (ABC) Robotics Initiative and by the Marcus Endowment Fund both at Ben-Gurion University of the Negev, Israel. This research was also supported by The Israeli Smart Transportation Research Center (ISTRC) by The Technion and Bar-Ilan Universities, Israel.

\tikz[remember picture,overlay]
    \path (current page.south west) ++ (7,5) node {\includegraphics[width=0.4\columnwidth]{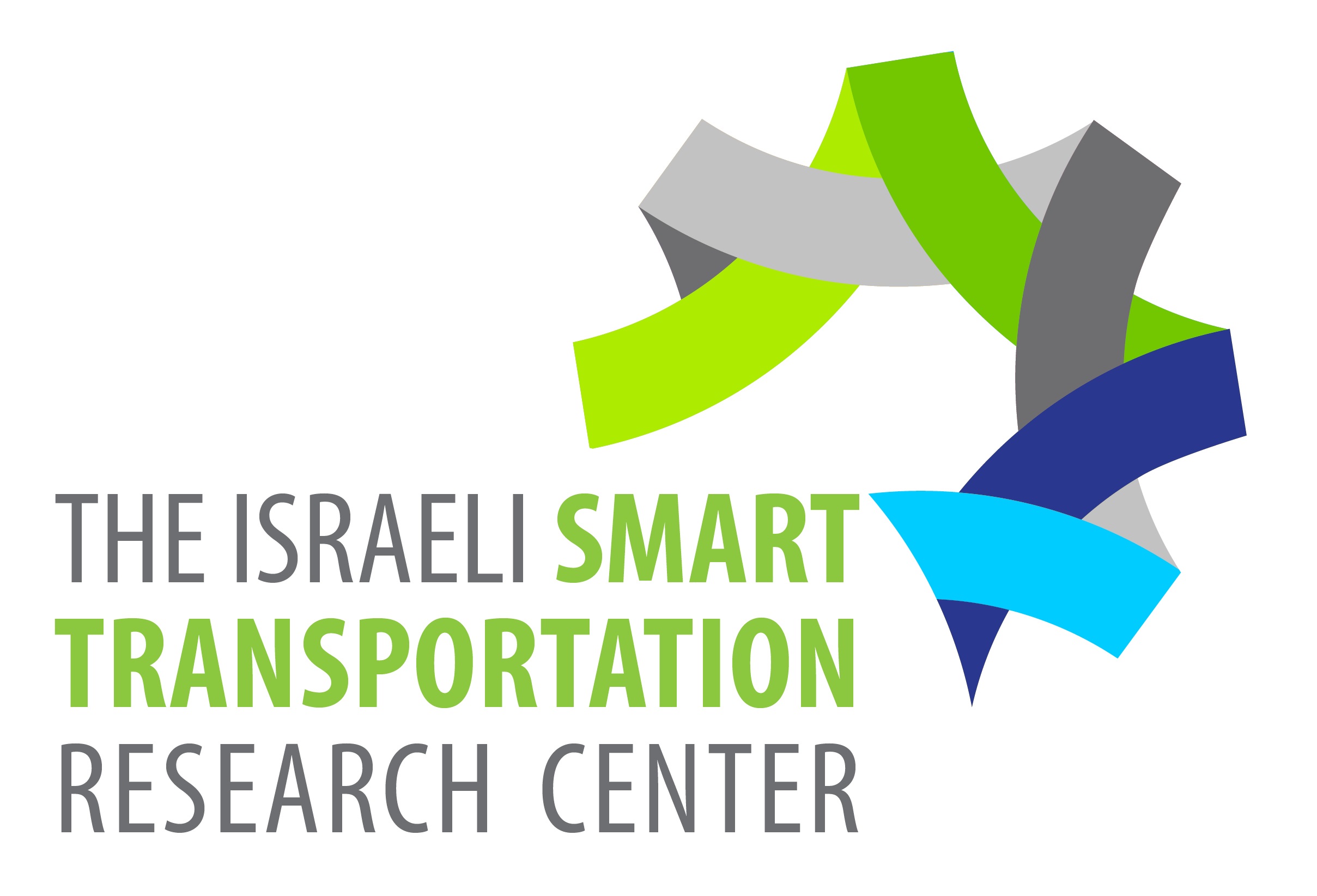}};
    
\tikz[remember picture,overlay]
    \path (current page.south west) ++ (14,5) node {\includegraphics[width=0.45\columnwidth]{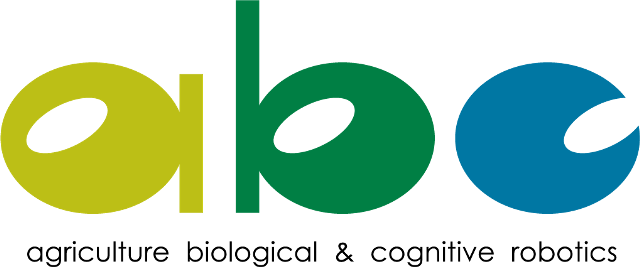}};
\clearpage

\pdfbookmark{\contentsname}{\contentsname}
\tableofcontents
\clearpage

\makeatletter
\let\everyparbck\everypar
\let\everypar\heb@o@everypar
\makeatother

\phantomsection
\addnumberlesstotoc{chapter}{\listfigurename}
\listoffigures
\cleardoublepage

\phantomsection
\addnumberlesstotoc{chapter}{\listtablename}
\listoftables

\mainmatter

\chapter{Introduction}
\label{chap:intro}
\epigraph{\textit{A captain should endeavor with every act to divide the forces of the enemy.}}{\textbf{Niccolò Machiavelli}}

As engineers, scientists and academics are constantly racing to improve the human understanding, modeling and handling of physical systems, the methods that are involved in doing so are constantly evolving. Tremendous energy is allocated to this effort, as we try to better understand how things work and how we can make them do good by us. Principles from various scientific fields are wielded to aid in this effort, and as time goes by, the tides are shifted towards those that provide better results. Some problems are better solved by different methods, and one of the main challenges of modern engineering is to try and accommodate solutions to a given problem, such that are befitting the situation at hand and the complications that arise with it.

In this we focus on combining aspects from multiple fields to provide a framework for handling use cases in \textbf{Control Theory}, namely control tasks in specific situations where there are multiple dynamically changing goals we wish to achieve simultaneously. We propose simple, yet effective, tools for designing controllers that are capable of handling a reality where objectives change dynamically depending on the evolving conditions of the system and its environment, e.g., the road and the behavior of the surrounding cars.

\section{Novel Differential Game Approach for \\Single-Agent Multi-Objective Systems}

In \textbf{Economics}, forming competition among internal business units within a company is considered an effective tool for balancing competing objectives in complex organizations~\cite{ECO} and maintaining centralized power by the entity enacting the strategy.
The maxim commonly used to describe this concept, '\textit{Divide and Conquer}' (abbreviated '\textit{D}\&\textit{C}') is frequently invoked not just in economics, but also in political and legal theory~\cite{posner2010divide}, among other social sciences and in computer science as well~\cite{CandC}. In a broad sense, Divide and Conquer is the act of:
\begin{quote}
"Gaining and maintaining overall power by breaking up larger concentrations of power into pieces that individually have less power than the one initially implementing the strategy"~\cite{Divide_History}.
\end{quote}
 In social structures, this term often invokes a darker and more sinister tone; as it employed as a means to control the public. In his book, '\textit{Political Writings}',  Immanuel Kant\footnote{\textbf{Immanuel Kant}, 1724 – 1804, German philosopher, and one of the central Enlightenment thinkers.} demonstrates support for D\&C, noting:
\begin{quote}
    "The problem of setting up a state can be solved even by a nation of devils, so long as they possess an appropriate constitution which pits opposing factions against each other with a system of checks and balances"~\cite{Political_Writings}.
\end{quote}

Niccolò Machiavelli\footnote{\textbf{Niccolò di Bernardo dei Machiavelli}, 1469 – 1527, Notorious Italian diplomat, philosopher and historian. He is best known for his political book '\textit{The Prince}', where he mainly claims politics is always bound to be played with deception, treachery, and crime.} suggests an application to warfare, advising in his book '\textit{The Art of War}': 

\begin{quote}
    ``A captain should endeavor with every act to divide the forces of the enemy''~\cite{machiavelli2009art}.
\end{quote}

 In economics, the notion is used to refer to a strategy to motivate the most potential out of factions in a competitive market.

Implementing these concepts, either theoretically or physically, is often by means of \textbf{Game Theory}, a branch of mathematics that deals with modeling strategic interactions between multiple players in a game~\cite{game_th}. The field defines and models several types of games, and the concept of D\&C can relate to both \textbf{non-cooperative}~\cite{zahedi2020seeking} and \textbf{adversarial}~\cite{RL_in_AG} games. The main distinguishing feature of non-cooperative games from \textbf{cooperative} games, is the absence of a protocol that allows for cooperative behavior. Thus players cannot form coalitions and must compete independently. This, however, does not mean that a valid solution to such a game cannot be such that all players fully achieve (or fail) their respective goals, as opposed to the adversarial case, in which players are necessarily opponents of each other, such that each one's gain is some other's loss. Consequently, adversarial games are also termed \textbf{zero-sum games}, as the sum of all the gains (assumed as some positive quantity) and losses (assumed as negative) across all players is zero.~\cite{zero_sum} Non-adversarial games (either cooperative or non-cooperative) are termed \textbf{non-zero-sum games}.

Some board games we can think about such as Backgammon, Checkers and Chess and so on are  \textbf{discrete}, i.e., each move is treated as a single atomic operation chosen from a finite set, as opposed to a \textbf{continuous} game, as in a basketball and football. Another property of games is their length. The aforementioned games last a \textbf{finite} amount of time, while other games such as the world economy or climate change last forever (or an \textbf{infinite} amount of time).

A special class of games we will be dealing in this work is \textbf{differential games}, which are closely related to optimal control theory (chapter 10 of~\cite{optimal_control_book}), as they describe continuous, infinite games where the dynamics of the players' behavior is described by \textbf{differential equations}~\cite{A_First_Course_in_Differential_Equations}.

One problem in using differential games in real-time controller design is that it is hard to handle a multitude of dynamically-changing objectives. Imagine, for example, the many considerations that a basketball player must weigh in real time. Current man-made controllers are very far from achieving such performance.

In this thesis, we propose to handle complex controllers by separating control tasks to components that represent different objectives, thus forming an internal competition between them as they act as players in a game, by use of optimizing appropriate cost functions, to ultimately achieve a sufficiently balanced solution to the problem at hand. Specifically, we apply the aforementioned techniques from economics in control theory.

Our and others' experience shows that defining cost functions by arithmetic operations such as a weighted sum with manually-designed parameters solves certain problems but does not constitute a holistic, comprehensive, solution to the problem of control in multi-purpose scenarios~\cite{lqr_ea, localized_weighted_sum, preference_based_weighted_sum}. The main reason for this is that sometimes engineers cannot easily model the intricate relations between the different objectives and thus resort to elaborate estimation and prediction techniques~\cite{weight_changing_model, heuristic_cost_function_weights}.
Specifically, it can be shown mathematically that if one considers control as a game between several players with each having several goals, the result obtained as an optimal solution of the game is not the result obtained for any possible weighting of the goal functions~\cite{differential_game_approach_formation}. 

In relation to control theory, we thus propose a novel approach, one which its core principle is more common in more recent Artificial Intelligence and Machine Learning works.

\textbf{Artificial Intelligence} (AI) is an umbrella term, used to refer to many methods and approaches that aim to automate the concept of acting \textit{rationally}\footnote{A \textbf{rational} entity is one that always aims to perform optimal actions based on given premises and information.}. Instead of a computer program designed to perform specific tasks explicitly, by imperative commands, AI is used in situations where there is a measure of uncertainty in the system model to consider, the possible changes that can occur or any other possible aspect that can demonstrate unforeseeable behaviour.~\cite{AI_1, AI_2, AI_3} AI is highly interdisciplinary and thus has many sub-branches, with one of them being Machine Learning. \textbf{Machine Learning} (ML) is a field which is at the center of attention for the past few decades, as it provides superior performance across a multitude of use cases and applications in many scientific fields~\cite{ML_1, ML_2, ML_3}. One specific type of learning is \textbf{Adversarial Learning}, where the learning task is modeled as a carefully designed competition between competing elements, such that their mutual gradual improvement corresponds to learning the \textit{statistics} of the data, and thus gaining the ability to make future predictions \cite{adv_1, adv_2}. One class of adversarial learning models specifically worth of mentioning due to its use of zero-sum games is the \textbf{Generative Adversarial Network (GAN)}, originally used for unsupervised learning. Since their introduction by Ian Goodfellow\footnote{\textbf{Ian J. Goodfellow} ,born 1985 or 1986, a researcher and director of machine learning in \textit{Apple Inc.}. He was previously employed at \textit{Google Brain} and has made several contributions to the field of deep learning, such as the formulation aforementioned GAN.}, et. al. at 2014~\cite{GAN}, GANs have become a staple of generative modeling. The GAN is composed of two elements that are defined as opposing rivals: a Generator, tasked with generating new 'made up' data from the learned distribution, and a Discriminator that is tasked with distinguishing between real and the 'fake' generated examples. These two compete in a zero-sum game that is tailored to produce the desired distribution of the input data.

Zero-sum formulations indeed provide a clear, concise manner to model the mutual dynamics of multiplayer games as they satisfy the property that each result of a zero-sum game is \textbf{Pareto-optimal}\footnote{Named after \textbf{Vilfredo Federico Damaso Pareto} (1848–1923), an Italian civil engineer and economist, a situation is called \textbf{Pareto-optimal} or \textbf{Pareto-efficient} if there does not exist a change that could lead to greater gain for some player without some other player suffering loss.}, and so under reasonable assumptions, a solution can be guaranteed~\cite{pareto_zero_sum}. For this reason zero-sum games are extremely useful in situations where it is reasonable to assume the zero-sum property, like in the case of a GAN. That being said, as demonstrated by Robert Wright in his book '\textit{Nonzero: The Logic of Human Destiny}':

\begin{quote}
    "All aspects of society are generally non-zero-sum as they becomes more complex, specialized, and interdependent"~\cite{nonzero}.
\end{quote}

 Zero-sum games generally fail to fully represent the conflicts faced in the everyday world. Problems in day-to-day basis usually do not have a straightforward, unique result, and could have none or multiple results. Thus non-zero-sum games better represent situations in the world we live in. This is the reason why the work on non-zero sum games, and specifically \textbf{Nash Equilibria}\footnote{\textbf{Nash Equilibrium} is named after \textbf{John Forbes Nash Jr.} (1928 – 2015), an American mathematician and Nobel Prize laureate in Economics, who made multiple contributions to Game Theory and other mathematical fields. Plainly put, it is a situation in a non-cooperative game where each player is assumed to know the equilibrium strategies other players, and no player has anything to gain by only changing their own strategy~\cite{Elementary_game_theory}.}, is the preferable tool to use in many applications.

We thus suggest, in a nutshell, the use of non-cooperative, non-zero-sum game strategy to be employed in single-agent multi-objective control problems. More precisely, given a control task to be applied on a system, we propose to:
\begin{enumerate}
    \item First, decompose the control input of the target system to multiple abstract \textit{virtual inputs};
    \item Next, assign a \textit{virtual objective} to each virtual input;
    \item Then, consider a differential game where each player controls a virtual input and aims at minimizing a \textit{cost function} that models its own personal objective;
    \item We ultimately propose to compute a \textit{Nash Equilibrium} for this game and use it to control the system along its progression.
\end{enumerate}
This methodology is defined and developed in details in a paper we published at the \textit{2021 29th Mediterranean Conference on Control and Automation (MED)} named \textit{Composition of Dynamic Control Objectives Based on Differential Games}~\cite{med_paper}.
 
Our method is most effective for multi-objective control problems where there is a decomposition of the input space such that each virtual input influences one of the goals significantly and its affects on the other goals is small. This situation is common in multi-objective controllers because independent goals like steering, speed control, etc., are usually affected by independent actuators (or sets thereof) such as the steering system or the gas pedal as demonstrated in~\cite{multi_actuator}.

Our approach allows control engineers to define dynamically changing games with time varying input decomposition and objective functions, as demonstrated by a proof-of-concept tool that we have developed. The game abstraction, allows a design of complex controllers without having to explicitly balance all goals and objectives. The freedom to change the game dynamically allows reaction to changing environments.  The tool we demonstrate our results with is a Python library that we have made available online~\cite{package}. It provides engineers with a rich and accessible specification language for controller design. We demonstrate how this tool performs in action by elaborating on the design process of a control strategy for a quadrotor. The control system that we develop, for this demonstration, handles time-varying partially competing objectives using our methodology.

Chapter~\ref{chap:d&c} provides a full discussion. 

\section{Novel Method for Solving AREs}
\label{sect:ARE_intro}

A class of equations arising frequently in optimal control problems is the Algebraic Riccati Equations (AREs). These equations appear in many optimal control problems, especially in cases where the objective function is quadratic. In this study we will first demonstrate formally how they are derived in cases where the optimality problem to be solved is considered along an infinite time horizon. These equations consist of a class of non-linear, quadratic matrix algebraic equations for an unknown matrix, which is required for setting the optimal controller. The AREs have been studied extensively from both mathematical and applicational points of view~\cite{riccati_review}.

In addition to the introduction of the aforementioned method, this work also focused on formulating a novel approach to solve these.  We have seen that solving this class of equations, while mainly done by numerical algorithms, can manifest solvability issues due to its large amount of possible solutions and high dependence on initial guesses~\cite{yet}. 

Moreover, due to the solution of this set being there to solve an underlying control problem, a specific solution is required that satisfies certain characteristics and is also always unique. This fact coupled with the complexity issues mentioned can cause extreme difficulties in finding that exact single desired solution, since, as we show, the number of solutions increases exponentially when the system state variables and a number of objectives (in the case of finding the Nash Equilibrium as described earlier) grow, and can even be infinite and even more so uncountable. 

Moreover, in order to obtain closed-loop strategies in a Nash differential game with an infinite horizon, as we will show, one needs to solve a system of \textbf{coupled} AREs. Unlike the regular LQR case, it is well established that there is no general set of conditions to guarantee uniqueness or even existence to this set~\cite{new_algorithm_coupled_algebraic_Riccati_equations}.

Thus we will suggest converting these equations to a corresponding set of DREs (Differential Riccati Equations), which will (if exists, and under certain conditions) yield the unique desired solution, which is always the one that fits the control-related required properties. 

Chapter~\ref{chap:riccati} provides a full discussion. 

        

\section{Main Contributions}
This thesis makes the following core contributions:

\begin{enumerate}
    \item Derivation of the theoretical basis and formal mathematical formulation for controllers that apply for single-agent, multi-objective dynamic systems, by formulating and solving non-cooperative, non-zero-sum differential games and obtaining their respective Nash Equilibria, for continuous-time control systems;
    \item Extending the aforementioned theoretical basis and formal mathematical formulation of the technique of single-agent multi-objective Nash Equilibria, for direct-design discrete-time control systems;
    \item Development of an open-source Python package named \textit{PyDiffGame}, implementing the proposed method, both for the continuous and discrete-time case;
    \item Derivation of a novel method for solving matrix algebraic Riccati equations (AREs) by converting them to differential Riccati equations (DREs) and solving them repetitively until convergence;
    \item Implementing the method of solving AREs by reduction to DREs in the Python package \textit{PyDiffGame}.
\end{enumerate}

\chapter{Preliminaries}
\label{chap:preliminaries}
\epigraph{\textit{Dogmas take endless forms, and when you can persuade different people to hold opposing dogmas, the manipulation of conflict and control through "divide and rule" becomes easy. It is happening today in the same way - more so, in fact - as it has throughout human history.}}{\textbf{David Icke}}

In this section, we will present several advanced principles in control theory required to showcase our work. An introductory section to system modeling is laid out at Appendix~\ref{app:sys_modeling} an to control systems at~\ref{app:controlled_systems}.

\section{Continuous Optimal Control}
\label{section:optimal_control}
One means of controlling composite systems, is by the notion of \textbf{optimality}. Optimal control deals with controllers that are derived by optimizing a set of performance criteria, defined with accordance to some preferred behaviour we would like the system to exhibit. In this section we provide a detailed introduction to the field, which is the basis to our proposed method. The matter is covered extensively in the book with the same title~\cite{optimal_control_book}. 

One way to derive an optimal controller is by means of defining a \textbf{cost function} that assigns a value to some undesired performance the system exhibits, and then finding a controller that minimizes that function. The cost function can be set to consider the deviation of the state vector from its desired value and the effort required in order to obtain that preferable performance. In that context of optimality, we use the term 'preferable' with accordance to some pre-determined criterion we call an \textbf{objective}. All of these concepts will now be formally introduced, based on the system modeling notations described in Appendix~\ref{app:sys_modeling}. 

\begin{definition}
 We associate a cost function for the linear dynamics of the continuous system $\mathcal{S}$. In the continuous case this is referred to as a \textbf{value integral} $J$ \footnote{The definition of $\mathcal{X}$, $\mathcal{U}$ and $\mathcal{I}$, along with any following non-explicitly stated notations, are given in Appendix~\ref{app:sys_modeling}.} $ \colon \mathcal{X} \times \mathcal{U} \times\mathcal{I} \rightarrow \mathbb{R}^+$ and is of the following general form:
\begin{equation}
\begin{aligned}
        J\left(\mathbf{x}(t), \mathbf{u}(t)\footnotemark, t\right) \coloneqq 
       \int_{t}^{T_f} r\left(\mathbf{x}(\tau), \mathbf{u}(\tau)\right)\mathrm{d}\tau + r_f\left(\mathbf{x}(T_f)\right),
\end{aligned}
\label{eqn:general_cost}
\end{equation}

\footnotetext{The cost function $J$ is a \textbf{functional}, as in a function that receives functions as inputs. In this context, $\mathbf{u}(\cdot)$ is a function which we wish to find in order to minimize the value for $J$. Thus $\mathbf{u}$ is in fact not a function of the variable $t$ only, but of all the time points in the interval $[t,T_f]$. With that, we still denote $\mathbf{u}(t)$ for simplicity.}
where:
\begin{itemize}
    \item $r \colon \mathcal{X} \times \mathcal{U} \rightarrow \mathbb{R}^+$ is a temporal cost function defined for all $\tau \in \mathcal{I}$. We will denote it as $r(t)$ for conciseness. For a given point in time $t \in \mathcal{I}$, $r(t)$ returns the total evaluated cost with accordance to the virtual inputs and state vector at the point in time;
    \item $r_f \colon \mathcal{X} \rightarrow \mathbb{R}^+$ is a function returning the cost of the final state vector;
    \item Depending on the scenario, we would consider either a finite horizon ($T_f < \infty$), for which we will define a non-zero value for $r_f$ in Equation~\eqref{eqn:general_cost}, or an infinite horizon ($T_f \rightarrow \infty$), for which we will set $r_f \equiv 0$.
\end{itemize}
We will denote $J(t)$ for conciseness.
\end{definition}

\subsection{Continuous Linear Quadratic Regulator (LQR)}
The main predecessor to our proposed D\&C method is the widely-known Linear-Quadratic Regulator (LQR), as it handles single-agent multiple objective cases, but with a single weighing of the state variables and control effort, and does not easily permit the decomposition into virtual inputs as we will describe. 

The LQR is defined by using the following approach:

Given a continuous LTI system $\mathcal{S}$ adhering the model at~\eqref{eqn:basic_sys}\footnote{Appendix~\ref{app:LTI} introduces LTI systems in detail.}, let us define the following:

\begin{itemize}
    \item \begin{definition}
A control input $\mathbf{u}(t) \in \mathbb{R}^m$ is said to be \textbf{admissible} if it is continuous over $\mathcal{I}$, stabilizes
the system $\mathcal{S}$, and makes the value of $J(t)$ at~\eqref{eqn:general_cost} finite. This definition induces the definition of the set $\mathcal{U}$;
\end{definition}
\item \begin{definition}
Let us assume there exists a pre-determined objective $O$ we would like to enact upon the system $\mathcal{S}$. This \textbf{objective} is defined by a pair of the form:
\begin{equation}
    O \coloneqq \left(Q, R\right),
\end{equation}
where:
\begin{itemize}
    \item $Q \in \mathbb{R}^{n \times n}, Q \geq 0$, i.e., $Q$ is positive-semi-definite (PSD);
    \item $R \in \mathbb{R}^{m \times m}$ and $R > 0$, i.e., $R$ is positive-definite (PD).
\end{itemize}
$Q$ and $R$ are weight matrices that dictate how each state and input variable precisely affects the corresponding cost function.
\end{definition}
\item \begin{definition}
We deal with cost function assignments that are quadratic with respect to the state and input vectors, i.e., $r(t)$ in~\eqref{eqn:general_cost} takes on the following form:

\begin{equation}
   r\left(\mathbf{x}(t), \mathbf{u}(t)\right) \equiv \mathbf{x}^T(t)Q\mathbf{x}(t) + \mathbf{u}^T(t)R\mathbf{u}(t).
   \label{eqn:LQR_cost}
\end{equation}

In this scenario, the control model is called a \textbf{Linear Quadratic Regulator} (LQR), since the model is linear and the cost function is quadratic with respect to the state.
\end{definition}
\end{itemize}

\begin{definition}
Given a continuous LTI system $\mathcal{S}$ with an LQR cost function $J\left( \mathbf{x}(t),\mathbf{u}(t), t\right)$, we define the \textbf{continuous open-loop LQR optimal control problem}, formally stated as:
\begin{equation}
    \begin{aligned}
         \min_{\substack{\mathbf{u}(t) \in \mathcal{U}}} \quad & J\left(\mathbf{x}(t), \mathbf{u}(t), t_0\right) \\
         = \min_{\substack{\mathbf{u}(t) \in \mathcal{U}}} \quad & \int_{\tau \in \mathcal{I}} r\left(\mathbf{x}(\tau), \mathbf{u}(\tau)\right)\mathrm{d}\tau + r_f\left(\mathbf{x}(T_f)\right) \\
         = \min_{\substack{\mathbf{u}(t) \in \mathcal{U}}} \quad & \bigintssss_{\tau \in \mathcal{I}} \left[ \mathbf{x}^T(\tau)Q\mathbf{x}(\tau) + \mathbf{u}^T(\tau)R\mathbf{u}(\tau) \right]\mathrm{d}\tau + r_f\left(\mathbf{x}(T_f)\right)\\
         \textrm{subject to} \quad & \mathbf{x}(t_0) = \mathbf{x_0};\\
         \quad & \dot{\mathbf{x}}(t) = A\mathbf{x}(t) + B \mathbf{u}(t); \quad & \forall t \in \mathcal{I}.
    \end{aligned}
    \label{eqn:optimal_control_problem_formulation}
\end{equation}
\end{definition}
Given a continuous LTI system $\mathcal{S}$, let us consider the open-loop LQR optimal control problem given in~\eqref{eqn:optimal_control_problem_formulation} along with its cost function $J\left( \mathbf{x}(t),\mathbf{u}(t), t\right)$, let us define the following:
\begin{itemize}
    \item \begin{definition}
Let $\mathbf{u^*}(t) \in \mathcal{U}$ be a \textbf{optimal input} of $\mathcal{S}$ that solves the optimal control problem, defined as:

\begin{equation}
     \mathbf{u^*}(t) \coloneqq \argmin_{\mathbf{u}(t) \in \mathcal{U}} J\left( \mathbf{x}(t),\mathbf{u}(t), t\right).
\end{equation}
\label{def:optimal_input}
\end{definition}
\begin{remark}
Notice Definition~\eqref{def:optimal_input} correlates to optimality in the open-loop sense, as we refer to a minimizing \textit{input}, which does not necessarily has to be a feedback controller;
\end{remark}
\item \begin{definition}
Let $\mathbf{x^*}(t) \in \mathcal{X}$ be the \textbf{optimal state trajectory} of $\mathcal{S}$ with respect to the optimal control problem, defined as the solution to the LTI model given at~\eqref{eqn:basic_sys} with the aforementioned optimal input, i.e., it satisfies:
\begin{equation}
        \dot{\mathbf{x}}^*(t) =  A\mathbf{x^*}(t) + B \mathbf{u^*}(t),
    \label{eqn:optimal_trajectory}
\end{equation}
and by that extension, using the solution for the state solution at~\eqref{eqn:LTI_solution}:
\begin{equation}
        \mathbf {x^*}(t) \coloneqq e^{A(t-t_0)}{\mathbf  {x_0}}+\int _{{t_{0}}}^{t}e^{A(t-\tau)}{B}{\mathbf  {u^*}}(\tau )\mathrm{d}\tau.
\end{equation}
\end{definition}
\end{itemize}

\begin{theorem}[\textbf{Continuous LQR Quadratic Optimal Cost}, Table 3.3-1, Page 145, Chapter 3 of~\cite{optimal_control_book}]
Given a continuous LTI system $\mathcal{S}$, let us consider the open-loop LQR optimal control problem given in~\eqref{eqn:optimal_control_problem_formulation} along with its cost function $J(t)$ and induced optimal state trajectory $\mathbf{x^*}(t)$. Then for any $T_f$, there exists $P \colon \mathcal{I} \rightarrow \mathbb{R}^{n \times n}$ such that the optimal cost is quadratic with respect to the optimal state, i.e.:
\begin{equation}
    J^*(t) = \mathbf{x^*}^T(t) P(t) \mathbf{x^*}(t)
\end{equation}
for all $t \in \mathcal{I}$.
\label{the:lqr_quad_cost}
\end{theorem}
\begin{remark}
To solve the open-loop optimal control problem, we will use an approach based on the Hamiltonian function and Pontryagin’s Minimum Principle. Then by applying the obtained solution to Bellman's Dynamic Programming method, we will show how the open loop optimal input is actually a closed-loop controller, exactly as described in chapter 10 of~\cite{optimal_control_book}. 
\end{remark}
Given a continuous LTI system $\mathcal{S}$, let us consider the open-loop LQR optimal control problem given in~\eqref{eqn:optimal_control_problem_formulation} along with its cost function $J\left( \mathbf{x}(t),\mathbf{u}(t), t\right)$ and induced optimal state trajectory $\mathbf{x^*}(t)$, let us define the following:
\begin{itemize}
    \item \begin{definition}
 The \textbf{partial derivatives of the cost function}\footnote{Notice $\partial_t J(t)$ and $\mathbf{\partial_x J}(t)$ also take values from $\mathcal{X} \times \mathcal{U} \times\mathcal{I}$, but we write them as a function of only $t$ for conciseness. Moreover, notice $\partial_t J$ is the derivative of a scalar with respect to another scalar - thus it is also a scalar, but $\mathbf{\partial_x J}$ is the derivative of a scalar with respect to a vector - which is a vector. As a matter of fact, we have: $\partial_t J \colon \mathcal{X} \times \mathcal{U} \times\mathcal{I} \rightarrow \mathbb{R}$ and $\mathbf{\partial_x J} \colon \mathcal{X} \times \mathcal{U} \times\mathcal{I} \rightarrow \mathbb{R}^n$.} with respect to time $t$ and state $\mathbf{x}(t)$, respectively, are:
\begin{equation}
\begin{aligned}
    \partial_t J(t) \coloneqq & \frac{\partial J\left( \mathbf{x}(t), \mathbf{u}(t), t \right)}{\partial t};\\
     \mathbf{\partial_x J}(t) \coloneqq & \frac{\partial J\left( \mathbf{x}(t), \mathbf{u}(t), t \right)}{\partial \mathbf{x}(t)} \\
     = & \left[ \frac{\partial J\left( \mathbf{x}(t), \mathbf{u}(t), t \right)}{\partial x_1(t)}, \frac{\partial J\left( \mathbf{x}(t), \mathbf{u}(t), t \right)}{\partial x_2(t)}, \dots, \frac{\partial J\left( \mathbf{x}(t), \mathbf{u}(t), t \right)}{\partial x_n(t)} \right]^T;
\end{aligned}
\label{eqn:partial_def}
\end{equation}
\end{definition}
\item \begin{definition}
 Let us use the gradient operator on the cost function $\nabla J\left( \mathbf{x}(t),\mathbf{u}(t), t\right)$ to define the \textbf{derivative with respect to either time $t$ or state $\mathbf{x}(t)$}.
\end{definition}
\item \begin{definition}
 Let us define the \textbf{optimal cost value} $J^*(t)$ and \textbf{optimal cost value gradient} $\nabla J^*(t)$, both with regards to state trajectory\footnote{Even though the optimal state trajectory is defined by the optimal input, these definitions are valid, assuming there exists an optimal input that induces the described state trajectory.} as:
\begin{equation}
\begin{aligned}
     J^*(t) \coloneqq & J\big( \mathbf{x}(t),\mathbf{u}(t), t\big) \ \big|_{\mathbf{x}(t) \equiv \mathbf{x^*}(t)};\\
     \nabla J^*(t) \coloneqq & \nabla J\big( \mathbf{x}(t),\mathbf{u}(t), t\big) \ \big|_{\mathbf{x}(t) \equiv \mathbf{x^*}(t)} \ ,
\end{aligned}
\end{equation}
and more specifically, optimal cost value derivatives with respect to time and state vector:
\begin{equation}
\begin{aligned}
    \mathbf{\partial_x J^*}(t) \coloneqq \mathbf{\partial_x J}(t) \ \big|_{\mathbf{x}(t) \equiv \mathbf{x^*}(t)};\\
    \partial_t J^*(t) \coloneqq \partial_t J(t) \ \big|_{\mathbf{x}(t) \equiv \mathbf{x^*}(t)}.
    \end{aligned}
\end{equation}
\end{definition}
\end{itemize}

\begin{proposition}
Given a continuous LTI system $\mathcal{S}$, let us consider the open-loop LQR optimal control problem given in~\eqref{eqn:optimal_control_problem_formulation} along with its cost function $J\left( \mathbf{x}(t),\mathbf{u}(t), t\right)$, induced optimal state trajectory $\mathbf{x^*}(t)$ and corresponding state partial derivative $\mathbf{\partial_x J}(t)$. The optimal input that solves the problem is then given by the following expression:
\begin{equation}
    \mathbf{u^*}(t) = - \frac{1}{2} R^{-1}B^T \ \mathbf{\partial_x J}^{*}(t).
\end{equation}
\end{proposition}

\begin{proof}
Let us differentiate both sides of the definition of the value integral at~\eqref{eqn:general_cost} ,with respect to time $t$, using Leibniz’s formula for differentiation under the integral sign\footnote{
\begin{definition}
In calculus, the \textbf{Leibniz integral rule for differentiation under the integral sign}~\cite{Intermediate_Calculus}, named after \textbf{Gottfried Wilhelm Leibniz} (1646–1716), a German mathematician, philosopher, scientist, and diplomat, states that for an integral of the form $\int _{a(x)}^{b(x)}f(x,t)\,dt$, the derivative is expressible as:
$$\displaystyle {\frac {d}{dx}}\left[\int _{a(x)}^{b(x)}f(x,t)\,dt\right]=f{\big (}x,b(x){\big )}\cdot {\frac {d}{dx}}b(x)-f{\big (}x,a(x){\big )}\cdot {\frac {d}{dx}}a(x)+\int _{a(x)}^{b(x)}{\frac {\partial }{\partial x}}f(x,t)\,dt$$
\end{definition}
}:

\begin{equation}
    \partial_t J(t) =
       - r\big(\mathbf{x}(t), \mathbf{u}(t)\big).
       \label{eqn:J_dot}
\end{equation}

Using the chain rule for vectors~\cite{matrix_cookbook} the LHS can be written as:
\begin{equation}
\begin{aligned}
    \partial_t J(t) & = \frac{\partial J\big( \mathbf{x}(t), \mathbf{u}(t), t \big)}{\partial t} = \left< \frac{\partial J\big( \mathbf{x}(t), \mathbf{u}(t), t \big)}{\partial \mathbf{x}(t)},  \frac{\mathrm{d} \mathbf{x}(t)}{\mathrm{d}t}\right>\footnotemark \\
    & = \left< \mathbf{\partial_x J}(t),  \frac{\mathrm{d} \mathbf{x}(t)}{\mathrm{d}t}\right> = \mathbf{\partial_x J}^T(t) \ \frac{\mathrm{d} \mathbf{x}(t)}{\mathrm{d}t} = \mathbf{\partial_x J}^T(t) \ \dot{\mathbf{x}}(t).
\end{aligned}
    \label{eqn:vectors_chain_rule}
\end{equation}
\footnotetext{The notation $\left< \cdot \right>$ signifies the inner-product operation.}

 Plugging this into~\eqref{eqn:J_dot} and moving sides we have:

\begin{equation}
    \mathbf{\partial_x J}^T(t) \ \dot{\mathbf{x}}(t)
       + r\big(\mathbf{x}(t), \mathbf{u}(t)\big) = 0.
\end{equation}

Plugging in the expression for $\dot{\mathbf{x}}(t)$ from~\eqref{eqn:basic_sys}, we have:

\begin{equation}
    \mathbf{\partial_x J}^T(t) \ \left(A\mathbf{x}(t) + B\mathbf{u}(t) \right)
       + r\big(\mathbf{x}(t), \mathbf{u}(t)\big) = 0.
       \label{eqn:pre_hamilton}
\end{equation}

\begin{definition}
We denote the LHS of~\eqref{eqn:pre_hamilton} as the \textbf{Control Hamiltonian} $\mathcal{H}$ of $\mathcal{S}$:

\begin{equation}
\begin{aligned}
     \mathcal{H} & \left( \mathbf{x}(t), \nabla J(t), \mathbf{u}(t) \right)  \coloneqq \partial_t J(t) + r\left(\mathbf{x}(t), \mathbf{u}(t)\right) \\
      = & \mathbf{\partial_x J}^T(t) \ \left(A\mathbf{x}(t) + B\mathbf{u}(t) \right)
       + r\left(\mathbf{x}(t), \mathbf{u}(t)\right),
       \label{eqn:hamiltonian}
\end{aligned}
\end{equation}

which is an ancillary function used to solve the optimal control problem.\footnote{Inspired by, but distinct from, the Hamiltonian of classical mechanics, the Hamiltonian of optimal control theory was developed by \textbf{Lev Pontryagin} (1908 – 1988), a Soviet mathematician, as part of his mentioned principle. Pontryagin proved that a necessary condition for solving the optimal control problem is that the control should be chosen so as to minimize the Control Hamiltonian.}
\end{definition}

\begin{theorem}[\textbf{Cont. Pontryagin’s Minimum Principle (PMP)}~\cite{pontryagin}]
\label{the:PMP}
See Equations 3.3-6, Page 136, Chapter 3 of~\cite{optimal_control_book}.
Given a continuous LTI system $\mathcal{S}$, let us consider the open-loop LQR optimal control problem given in~\eqref{eqn:optimal_control_problem_formulation} along with its cost function $J\left(t\right)$, its gradient $\nabla J(t)$ and optimal value $\nabla J^*(t)$ corresponding to the induced optimal state trajectory $\mathbf{x^*}(t)$. A necessary and sufficient
condition to obtain the optimal control policy that solves the optimal control problem is to obtain the control policy that minimizes the Hamiltonian. Formally, this may be stated as:

\begin{equation}
     \mathbf{u^*}(t) = \argmin_{\mathbf{u}(t) \in \mathcal{U}} \mathcal{H} \left( \mathbf{x^*}(t), \nabla J^*(t), \mathbf{u}(t) \right),
     \label{eqn:argmin_Hamiltonian}
\end{equation}

which is equivalent to solving the following stationarity condition:

\begin{equation}
     \frac{\partial}{\partial \mathbf{u}(t)} \mathcal{H} \left( \mathbf{x^*}(t), \nabla J^*(t), \mathbf{u}(t) \right) \big|_{\mathbf{u}(t) \equiv \mathbf{u^*}(t)} = 0.
     \label{eqn:stationarity}
\end{equation}
\end{theorem}

Given an LTI system $\mathcal{S}$ with an LQR cost function $J(t)$, optimal cost value gradient in regards to state trajectory $\nabla J^*(t)$ and optimal state trajectory $\mathbf{x^*}(t)$, we now apply PMP. Differentiating the Hamiltonian from~\eqref{eqn:hamiltonian} with respect to the input $\mathbf{u}(t)$ yields:

\begin{equation}
\begin{aligned}
    & \frac{\partial}{\partial \mathbf{u}(t)} \mathcal{H} \left( \mathbf{x^*}(t), \nabla J^*(t), \mathbf{u}(t) \right) =\\
    & \frac{\partial}{\partial \mathbf{u}(t)} \left[ \mathbf{\partial_x J}^{*^T}(t)\Big(A\mathbf{x^*}(t) + B\mathbf{u}(t) \Big)
       + r\big(\mathbf{x^*}(t), \mathbf{u}(t)\big) \right] = \\
       & \frac{\partial}{\partial \mathbf{u}(t)} \left[ \mathbf{\partial_x J}^{*^T}(t)\Big(A\mathbf{x^*}(t) + B\mathbf{u}(t) \Big)
       + \mathbf{x^*}^T(t)Q\mathbf{x^*}(t) + \mathbf{u}^T(t)R\mathbf{u}(t) \right].
\end{aligned}
\label{eqn:diff_H_1}
\end{equation}

By the definition of the cost function at~\eqref{eqn:infinte_value_integral}, we can see that its derivative with respect to the state vector $\mathbf{x}(t)$ is not a function of the input $\mathbf{u}(t)$. Thus we have that the cross-derivative cancels out, i.e.:

\begin{equation}
    \frac{\partial}{\partial \mathbf{u}(t)}\mathbf{\partial_x J}^{*^T}(t) = 0,
    \label{eqn:null_cross_derivative}
\end{equation}
and so the first term in~\eqref{eqn:diff_H_1} yields\footnote{Given two matrices $X \in \mathbb{R}^{n \times m}, Y \in \mathbb{R}^{m \times k}$, we have~\cite{matrix_cookbook}:
$$\frac{\partial XY}{\partial Y} = X^T$$}:

\begin{equation}
\begin{aligned}
    \frac{\partial}{\partial \mathbf{u}(t)} \left[ \mathbf{\partial_x J}^{*^T}(t)\left(A\mathbf{x^*}(t) + B\mathbf{u}(t) \right) \right]  =
    \left( \mathbf{\partial_x J}^{*^T}(t) B \right)^T = B^T \ \mathbf{\partial_x J}^{*}(t).
\end{aligned}
\label{eqn:diff_H_2}
\end{equation}

The second term in~\eqref{eqn:diff_H_1} is not a function of $\mathbf{u}(t)$ so its derivative with respect to it also cancels out, and the third term yields\footnote{Given a vector $\mathbf{x} \in \mathbb{R}^{n}$ and a matrix $Y \in \mathbb{R}^{n \times n}$, we have~\cite{matrix_cookbook}:
$$\frac{\partial \mathbf{x}^T Y \mathbf{x}}{\mathbf{x}} = (Y + Y^T) \mathbf{x}.
$$}

\begin{equation}
\begin{aligned}
    \frac{\partial \mathbf{u}^T(t)R\mathbf{u}(t)}{\partial \mathbf{u}(t)} = \big(R + R^T)\mathbf{u}(t) =  2R\mathbf{u}(t).
\end{aligned}
\label{eqn:diff_H_3}
\end{equation}

Plugging~\eqref{eqn:diff_H_2} and~\eqref{eqn:diff_H_3} into~\eqref{eqn:diff_H_1} and equating to zero (as in~\eqref{eqn:stationarity}) we have:

\begin{equation}
    \left[B^T \ \mathbf{\partial_x J}^{*}(t) + 2R\mathbf{u}(t)\right]\Big|_{\mathbf{u}(t) \equiv \mathbf{u^*}(t)} = 0.
\end{equation}

Solving this for $\mathbf{u}(t)$ yields the following optimal value $\mathbf{u^*}(t)$:

\begin{equation}
    \mathbf{u^*}(t) = - \frac{1}{2} R^{-1}B^T \ \mathbf{\partial_x J}^{*}(t).
    \label{eqn:optimal_u}
\end{equation}
Notice $R >0$ and thus it has no zero eigenvalues, and so $R^{-1}$ exists.
\end{proof}

Given a continuous LTI system $\mathcal{S}$, let us consider the open-loop LQR optimal control problem given in~\eqref{eqn:optimal_control_problem_formulation} along with its cost function $J\left( \mathbf{x}(t),\mathbf{u}(t), t\right)$, its gradient $\nabla J(t)$ and optimal value $\nabla J^*(t)$ corresponding to the induced optimal state trajectory $\mathbf{x^*}(t)$. Let us then define:

\begin{itemize}
\item
\begin{definition}
    \item Let us use the notation for the Control Hamiltonian $\mathcal{H}$ from~\eqref{eqn:hamiltonian} and plug it in~\eqref{eqn:pre_hamilton} to define the \textbf{Bellman Equation}\footnote{The \textbf{Bellman Equation}, named after \textbf{Richard E. Bellman} (1920 –  1984), an American applied mathematician, holds true in any case and moreover is a necessary condition for optimality associated with the \textbf{dynamic programming} optimization technique. In simple terms, this method considers the value of a problem at a certain point in time in terms of the payoff from some initial choices and the value of the remaining problem that results from those initial choices. This breaks a dynamic optimization problem into a sequence of simpler sub-problems, as Bellman's “principle of optimality" entails~\cite{bellman}.}:
\begin{equation}
    \mathcal{H} \Big( \mathbf{x}(t), \nabla J(t), \mathbf{u}(t) \Big)  = 0.
    \label{eqn:bellman}
\end{equation}
Bellman equation is a \textbf{partial differential equation (PDE)} for the value function $J(t)$, along with the initial condition $J\big( \mathbf{x_0},\mathbf{u}(t_0), t_0 \big) = 0$.
\end{definition}
\item 
\begin{definition}
We now combine the LTI optimal control value obtained by PMP and the Bellman equation to define a new PDE, this time for the optimal cost value $J^*(t)$. Plugging the expression for the optimal value of $\mathbf{u^*}(t)$ from~\eqref{eqn:optimal_u} and also plugging in the optimal (yet unknown) value $\mathbf{\partial_x J}^{*}(t)$ - both to the Bellman equation at~\eqref{eqn:bellman} yields:

\begin{equation}
\begin{aligned}
    \mathbf{\partial_x J}^{*^T}(t) & \Big[A\mathbf{x^*}(t) -\frac{1}{2}BR^{-1}B^T \ \mathbf{\partial_x J}^{*}(t)  \Big]
       + \mathbf{x^*}^T(t)Q\mathbf{x^*}(t)\\
       & + \frac{1}{4}\Big( R^{-1}B^T \ \mathbf{\partial_x J}^{*}(t) \Big)^TR R^{-1}B^T \ \mathbf{\partial_x J}^{*}(t) = 0,
\end{aligned}
\end{equation}

which is called the \textbf{Hamilton-Jacobi-Bellman} (HJB) Equation\footnote{A scalar version of the HJB equation is examined in Appendix~\ref{app:scalar_HJB}, providing a elementary reasoning for claiming the optimal LQR cost being quadratic.}, this time a PDE for the optimal value $J^*(t)$. Rearranging, we have:
\begin{equation}
\begin{aligned}
    \mathbf{\partial_x J}^{*^T}(t) &\Big[A\mathbf{x^*}(t) -\frac{1}{4}BR^{-1}B^T \ \mathbf{\partial_x J}^{*}(t)  \Big]
     +\mathbf{x^*}^T(t)Q\mathbf{x^*}(t)= 0.
       \label{eqn:HJB}
\end{aligned}
\end{equation}
\end{definition}
\end{itemize}

To solve the optimal control problem, we first solve the HJB Equation~\eqref{eqn:HJB}
for the optimal value $J^*(t)$, then the optimal control is given as  $\mathbf{u^*}\big(J^*(t)\big)$ in terms of the HJB solution by~\eqref{eqn:optimal_u}. 




\subsection{Continuous Time Infinite Horizon LTI LQR}
\label{subsection:cont_inf_lqr}

Let us assume $T_f \rightarrow \infty$. Let us define the following value for the LQR cost function:

\begin{equation}
\begin{aligned}
    J\big( \mathbf{x}(t), \mathbf{u}(t), t \big)  \equiv & 
       \int_{t}^{\infty} r\big(\mathbf{x}(\tau), \mathbf{u}(\tau)\big)\mathrm{d}\tau  \\
       = &  \int_{t}^{\infty} \Big[ \mathbf{x}^T(\tau)Q\mathbf{x}(\tau) + \mathbf{u}^T(\tau)R\mathbf{u}(\tau) \Big] \mathrm{d}\tau
\end{aligned}
       \label{eqn:infinte_value_integral}
\end{equation}
\begin{proposition}
Given a continuous LTI system $\mathcal{S}$, let us consider the open-loop LQR optimal control problem given in~\eqref{eqn:optimal_control_problem_formulation} along with its induced optimal state trajectory $\mathbf{x^*}(t)$. If $T_f \rightarrow \infty$ then:
\begin{enumerate}
    \item The optimal solution $\mathbf{u^*}(t)$ to the problem is actually a closed-loop feedback controller, which is linear with regards to optimal state, with a constant coefficient, i.e. it is of the form:
\begin{equation}
    \mathbf{u^*}(t) = - K^*\mathbf{x^*}(t),
\end{equation}
with a constant matrix $K^* \in \mathbb{R}^{n \times m}$.
\item The optimal controller is of the following form:
\begin{equation}
    K^* \coloneqq R^{-1}B^T P
\end{equation}
when $P = P^T \in \mathbb{R}^{n \times n}$ is a solution to the Continuous Algebraic Riccati Equation (CARE):
\begin{equation}
    PA + A^TP -PBR^{-1}B^T P + Q = 0
\end{equation}
\end{enumerate}

\end{proposition}
\begin{proof}

Due to Theorem~\ref{the:lqr_quad_cost}, let us consider a proposed solution that is quadratic with respect to the state. Let us propose the following expression for the optimal cost:

\begin{equation}
    J^*(t) \equiv \mathbf{x^*}^T(t) P \mathbf{x^*}(t)
    \label{eqn:infinite_horizon_optimal_cost}
\end{equation}

for some yet-to-be-determined time-independant matrix $P \in \mathbb{R}^{n \times n}$ which is assumed to be real and symmetric, i.e. $P = P^T$. Thus, we can now directly evaluate $\mathbf{\partial_x J}^{*}(t)$ using matrix calculus:

\begin{equation}
    \mathbf{\partial_x J}^{*}(t) = \frac{\partial \mathbf{x^*}^T(t) P \mathbf{x^*}(t) }{\partial \mathbf{x}(t)}  = (P+P^T)\mathbf{x^*}(t) = 2 P \mathbf{x^*}(t).
    \label{eqn:dJ_dx}
\end{equation}

With this we can now directly evaluate the optimal input by plugging this result to~\eqref{eqn:optimal_u}:

\begin{equation}
    \mathbf{u^*}(t) = - \frac{1}{2} R^{-1}B^T \cdot 2 P \mathbf{x^*}(t) = - R^{-1}B^T P \mathbf{x^*}(t),
    \label{eqn:optimal_proposed_value}
\end{equation}
which proves 1. Denoting the optimal controller as 
\begin{equation}
    K^* \coloneqq R^{-1}B^T P
    \label{eqn:optimal_K_infinite}
\end{equation}

we have:
\begin{equation}
    \mathbf{u^*}(t) = - K^* \mathbf{x^*}(t),
\end{equation}
meaning the optimal controller in this case is linear. Plugging~\eqref{eqn:dJ_dx} in the HJB equation at~\eqref{eqn:HJB} we have:

\begin{equation}
\begin{aligned}
    \Big( 2 P \mathbf{x^*}(t) \Big)^T\Big[A\mathbf{x^*}(t) -\frac{1}{4}BR^{-1}B^T \cdot 2 P \mathbf{x^*}(t)  \Big] +\mathbf{x^*}^T(t)Q\mathbf{x^*}(t)= 0.
\end{aligned}
\end{equation}


Rearranging:

\begin{equation}
\begin{aligned}
    \mathbf{x^*}^T (t) \Big(2PA -PBR^{-1}B^T P + Q\Big)\mathbf{x^*} (t)
       = 0.
       \label{eqn:pre_riccati}
\end{aligned}
\end{equation}

Since $\mathbf{x^*}^T (t)PA\mathbf{x^*} (t)$ is a scalar, it is equal to its transpose, and so:

\begin{equation}
\begin{aligned}
    2\mathbf{x^*}^T (t)PA\mathbf{x} (t) & = \mathbf{x^*}^T (t)PA\mathbf{x^*} (t) + \mathbf{x^*}^T (t)PA\mathbf{x^*} (t) \\
    &= \mathbf{x^*}^T (t)PA\mathbf{x^*} (t) + \mathbf{x^*}^T (t)A^TP^T\mathbf{x} (t) \\& = \mathbf{x^*}^T (t)PA\mathbf{x^*} (t) + \mathbf{x^*}^T (t)A^TP\mathbf{x^*} (t)
    \\& = \mathbf{x^*}^T (t)\Big( PA + A^TP\Big)\mathbf{x^*} (t)
    \label{eqn:2xPAx}
\end{aligned}
\end{equation}

\begin{definition}
Since it must hold for all optimal state trajectory $\mathbf{x^*}(t)$, Equation~\eqref{eqn:pre_riccati} can be rewritten as:

\begin{equation}
    PA + A^TP -PBR^{-1}B^T P + Q = 0,
    \label{eqn:riccati}
\end{equation}
Equation~\eqref{eqn:riccati}, is referred to as the \textbf{Continuous Algebraic Riccati Equation} (CARE), which is a quadratic matrix equation.
\end{definition}

This concludes the proof of 2. 

\end{proof}
\begin{remark}
As mentioned in the introduction at~\ref{sect:ARE_intro}, this work relies extensively on the properties of the Riccati Equation, and so Section~\ref{app:solving_the_care}, of the next Chapter~\ref{chap:related_work}, is dedicated to showcase a few of its known characteristics, some of which are taken from a book dedicated to the study of it~\cite{riccati_Equation}. 
\end{remark}

\subsection{Continuous Finite Horizon LTI LQR}
\label{subsect:finite_horizon_con_LQR}

In the case where $T_f < \infty$, we define the finite horizon cost function as choosing the values for $r(t)$ and $r_f$ in~\eqref{eqn:general_cost} such that:

\begin{equation}
\begin{aligned}
    J\left( \mathbf{x}(t), \mathbf{u}(t), t \right)  \equiv & 
       \int_{t}^{T_f} \left[ \mathbf{x}^T(\tau)Q\mathbf{x}(\tau) + \mathbf{u}^T(\tau)R\mathbf{u}(\tau) \right] \mathrm{d}\tau + \mathbf{x}^T(T_f)Q_f\mathbf{x}(T_f).
\end{aligned}
       \label{eqn:finite_value_integral}
\end{equation}
\begin{proposition}
Given a continuous LTI system $\mathcal{S}$ with an LQR cost function, let us consider the open-loop optimal control problem given in~\eqref{eqn:optimal_control_problem_formulation} along with its induced optimal state trajectory $\mathbf{x^*}(t)$. If $T_f < \infty$ then:
\begin{enumerate}
    \item The optimal solution $\mathbf{u^*}(t)$ to the problem is actually a closed-loop feedback controller, which is linear with regards to optimal state, with a time-varying coefficient, i.e. it is of the form:
\begin{equation}
    \mathbf{u^*}(t) = - K^*(t)\mathbf{x^*}(t),
\end{equation}
with a time-varying matrix $K^* \colon \mathcal{I} \rightarrow \mathbb{R}^{n \times m}$.
\item The optimal controller is of the following form:
\begin{equation}
    K^*(t) \coloneqq R^{-1}B^T P(t)
\end{equation}
when $P = P^T \colon \mathcal{I} \rightarrow \mathbb{R}^{n \times n}$ is a solution to the Continuous Differential Riccati Equation (CDRE):
\begin{equation}
    \dot{P}(t) + P(t)A + A^TP(t)+  Q - P(t) B  R^{-1}B^T P(t) = 0.
\end{equation}
\end{enumerate}

\end{proposition}
\begin{proof}

In this case, again due to Theorem~\ref{the:lqr_quad_cost}, we consider a similar proposed solution as in~\eqref{eqn:infinite_horizon_optimal_cost}, but this time with a time-dependant positive-definite matrix $P(t) \in \mathbb{R}^{n \times n}$:

\begin{equation}
    J^*(t) \equiv \mathbf{x^*}^T(t) P(t) \mathbf{x^*}(t).
    \label{eqn:finite_horizon_optimal_cost}
\end{equation}

Differentiating both sides of~\eqref{eqn:finite_horizon_optimal_cost} with respect to time while using the product rule for the RHS we have:

\begin{equation}
\begin{aligned}
    \partial_t J^{*}(t) & = \dot{\mathbf{x}}^{*^T}(t)P(t) \mathbf{x^*}(t) + \mathbf{x^*}^T(t)\dot{P}(t) \mathbf{x^*}(t) + \mathbf{x^*}^T(t)P(t)\dot{\mathbf{x}}^*(t).
    \label{eqn:J_time_derivative}
\end{aligned}
\end{equation}

Plugging this to the expression for the Control Hamiltonian at~\eqref{eqn:hamiltonian} we have:

\begin{equation}
\begin{aligned}
    & \mathcal{H} \left( \mathbf{x^*}(t), \nabla J^*(t), \mathbf{u}(t) \right) = \dot{\mathbf{x}}^{*^T}(t)P(t) \mathbf{x^*}(t)  \\ & + \mathbf{x^*}^T(t)\left[\left(\dot{P}(t) + Q\right)\mathbf{x^*}(t) + P(t)\dot{\mathbf{x}}^*(t) \right]  + \mathbf{u}^T(t)R\mathbf{u}(t).
    \label{eqn:P_t_1}
\end{aligned}
\end{equation}

Plugging the expression for $\dot{\mathbf{x}}^{*}(t)$ to~\eqref{eqn:P_t_1} and rearranging we have:

\begin{equation}
\begin{aligned}
    & \mathcal{H} \left( \mathbf{x^*}(t), \nabla J^*(t), \mathbf{u}(t) \right) = \\ & \mathbf{x^*}^T(t)\left[\left(\dot{P}(t) + P(t)A + A^TP(t)+  Q\right)\mathbf{x^*}(t) + P(t) B \mathbf{u}(t) \right] \\
    & + \mathbf{u}^T(t)\left[ R\mathbf{u}(t) + B^T P(t) \mathbf{x^*}(t) \right].
    \label{eqn:P_t_2}
\end{aligned}
\end{equation}

Applying PMP from~\eqref{eqn:stationarity} to~\eqref{eqn:P_t_2} by differentiating with respect to $\mathbf{u}(t)$ and equating to zero, we have the following stationarity condition:

\begin{equation}
\begin{aligned}
    \left[ 2B^TP\mathbf{x^*}^T(t) + 2R\mathbf{u}(t)\right]\Big|_{\mathbf{u}(t) \equiv \mathbf{u^*}(t)} = 0
\end{aligned}
\end{equation}

which yields the following optimal policy:

\begin{equation}
    \mathbf{u}^*(t) = - R^{-1}B^T P(t) \mathbf{x^*}(t).
\end{equation}

which indicates that in this case, the optimal controller is also linear,  but is time-dependant this time around:

\begin{equation}
    K^*(t) \coloneqq R^{-1}B^T P(t),
\end{equation}

and so 1 is proven. 
\begin{definition}
Plugging this to~\eqref{eqn:P_t_2} and rearranging, we get the finite horizon version of the HJB Equation:

\begin{equation}
\begin{aligned}
    & \mathbf{x^*}^T(t)\left[\dot{P}(t) + P(t)A + A^TP(t)+  Q - P(t) B  R^{-1}B^T P(t)  \right]\mathbf{x^*}(t)  = 0,
    \label{eqn:P_t_3}
\end{aligned}
\end{equation}

which yields the \textbf{Continuous Differential Riccati Equation} (CDRE):

\begin{equation}
\begin{aligned}
    \dot{P}(t) + P(t)A + A^TP(t)+  Q - P(t) B  R^{-1}B^T P(t) = 0,
    \label{eqn:CDRE}
\end{aligned}
\end{equation}

which, along with the terminal condition 

\begin{equation}
\begin{aligned}
    P(T_f) \equiv Q_f,
\end{aligned}
\end{equation}
represents a first-order quadratic differential equation which can be solved backwards in time to obtain the value for the temporal function $P(t)$.
\end{definition}

This concludes the proof of 2.
\end{proof}

\begin{remark}
As the CDRE is a more general form of the CARE, it has its own share of theory regarding its solution. Section~\ref{app:solving_the_CDRE} provides an overview of the matter.
\end{remark}

\section{Discrete Optimal Control}

Now we broaden the scope of optimal control to handle discrete systems, which are described in Appendix~\ref{sect:lti_discrete_appendix}. As in continuous-time systems, discrete-time systems can also be controlled by means of optimal control, by defining a corresponding cost function and finding a controller that minimizes it.

Given a discrete LTI system $\mathcal{S}_D$ adhering the model at~\eqref{eqn:discrete_LTI_model_2}\footnote{Appendix~\ref{sect:lti_discrete_appendix} introduces discrete LTI systems in detail.}, let us define the following:
\begin{itemize}
    \item \begin{definition}
As $\mathcal{S}_D$ is sampled at discrete sampling points, the performance cost is evaluated by a sum instead of an integral. Thus the \textbf{$k$'th measurement of the discrete cost function} $J_D \colon \mathcal{X_D} \times \mathcal{U}_D \times \bigcup_{i=0}^{L-1}\{i\} \rightarrow \mathbb{R}^+$ takes on the following general form:

\begin{equation}
\begin{aligned}
        J_D\left(\mathbf{x}[k], \mathbf{u}[k], k\right) \coloneqq 
       \mathlarger{\sum}_{t=k}^{L-1} r_D\left(\mathbf{x}[t], \mathbf{u}[t]\right) + r_{f_D}\left(\mathbf{x}[L]\right)
       \label{eqn:discrete_general_cost}
\end{aligned}
\end{equation}

where $r_D$ and $r_{f_D}$ are the discrete-time equivalents of $r$ and $r_f$ from~\eqref{eqn:general_cost} and each corresponds to the discrete aggregated cost of the control effort and deviation from the desired state - as in the continuous case, but this time calculated between each sampling interval discretely, and from the $k$'th measurements onwards until the $L-1$ measurement which corresponds to $T_f$, the final operation time of the system.
\end{definition}
\begin{remark}
In the context described in Appendix~\ref{sect:lti_discrete_appendix}, the discrete sampling index $k$ in Equation~\eqref{eqn:discrete_general_cost} is the discrete-time equivalent of the continuous time variable $t$ in Equation~\eqref{eqn:general_cost}.
\end{remark}
\item 
\begin{definition}
Using the same definition for an objective as in the continuous case, given an objective $O = (Q,R)$, we deal with cost function assignments that are quadratic with respect to the state and input vectors, i.e., $r_D[k]$ takes on the following form:

\begin{equation}
   r_D\Big(\mathbf{x}[k], \mathbf{u}[k]\Big) \equiv \mathbf{x}^T[k]Q\mathbf{x}[k] + \mathbf{u}^T[k]R\mathbf{u}[k].
\end{equation}

for all $k \in \bigcup_{i=0}^{L-1}\{i\}$. In this scenario, the control model is called a \textbf{Discrete Linear Quadratic Regulator} (DLQR), similar to what defined in the continuous case.
\end{definition}

\end{itemize}

\begin{definition}

Given a discrete LTI system $\mathcal{S}_D$ and DLQR cost function measurements $\left(J_D[k]\right)_{k=0}^{L-1}$, we define \textbf{the discrete open-loop optimal control problem}, formally stated as:

\begin{equation}
    \begin{aligned}
         \min_{\substack{\mathbf{u}[k] \in \mathcal{U}_D}} \quad & J_D\left(\mathbf{x}[k], \mathbf{u}[k], 0\right) \\
         = \min_{\substack{\mathbf{u}[k] \in \mathcal{U}_D}} \quad & \mathlarger{\sum}_{t=0}^{L-1} r_D\left(\mathbf{x}[t], \mathbf{u}[t]\right) + r_{f_D}\left(\mathbf{x}[L]\right)\\
         = \min_{\substack{\mathbf{u}[k] \in \mathcal{U}_D}} \quad & \mathlarger{\sum}_{t=0}^{L-1} \left[\mathbf{x}[t]^TQ\mathbf{x}[t] +  \mathbf{u}[t]^TR\mathbf{u}[t]\right] + r_{f_D}\left(\mathbf{x}[L]\right) \\
         \textrm{subject to} \quad & \mathbf{x}[0] = \mathbf{x_0};\\
         \quad & \mathbf{x}[k+1] = \tilde{A}\mathbf{x}[k] + \tilde{B} \mathbf{u}[k]; \quad & \forall k \in \bigcup_{i=0}^{L-1}\{i\}.
    \end{aligned}
    \label{eqn:discrete_optimal_control_problem_formulation}
\end{equation}
\end{definition}
Given a discrete LTI system $\mathcal{S}_D$, let us consider the open-loop optimal control problem given in~\eqref{eqn:discrete_optimal_control_problem_formulation} along with its cost function measurements $\left(J_D[k]\right)_{k=0}^{L-1}$, let us define:

\begin{itemize}
    \item \begin{definition}
 Let $\mathbf{u^*}[k] \in \mathcal{U}_D$ be the \textbf{$k$'th measurement of the optimal input} of $\mathcal{S}_D$ with respect to the $k$'th measurement of the given cost function $J_D[k]$, defined as:
\begin{equation}
     \mathbf{u^*}[k] \coloneqq \argmin_{\mathbf{u}[k] \in \mathcal{U}_D} J_D\big( \mathbf{x}[k],\mathbf{u}[k], k\big).
\end{equation}
\end{definition}
\item \begin{definition}
Given optimal input measurements $\left(\mathbf{u^*}[k]\right)_{k=0}^{L-1}$, let the series $\left(\mathbf{x^*}[k]\right)_{k=0}^{L-1}$ be the \textbf{optimal state trajectory measurements} of $\mathcal{S}_D$, defined as the solution to the discrete LTI model given at~\eqref{eqn:discrete_LTI_model_2} with the optimal input measurements, i.e., it satisfies:
\begin{equation}
\begin{aligned}
     \mathbf{x^*}[k+1] = \tilde{A}\mathbf{x^*}[k] + \tilde{B} \mathbf{u^*}[k],
\end{aligned}
\end{equation}
for all $k \in \bigcup_{i=0}^{L-2}\{i\}$.
\end{definition}
\end{itemize}

\begin{theorem}[\textbf{Discrete LQR Quadratic Optimal Cost}, Table 2.2-1, Page 33, Chapter 2 of~\cite{optimal_control_book}]
Given a discrete LTI system $\mathcal{S}_D$, let us consider the open-loop optimal control problem given in~\eqref{eqn:discrete_optimal_control_problem_formulation} along with its cost function and induced optimal state trajectory measurements $\left(J_D[k]\right)_{k=0}^{L-1}$ and $\left(\mathbf{x^*}[k]\right)_{k=0}^{L-1}$. Then for any $L$, there exists $P_D \colon \bigcup_{i=0}^{L-1}\{i\} \rightarrow \mathbb{R}^{n \times n}$ such that the optimal cost is quadratic with respect to the optimal state, i.e.:
\begin{equation}
    J^*[k] = \mathbf{x^*}^T[k] P_D[k] \mathbf{x^*}[k]
\end{equation}
for all $k \in \bigcup_{i=0}^{L-1}\{i\}$.
\label{the:discrete_lqr_quad_cost}
\end{theorem}
\begin{definition}
Let us denote $\Delta J_D$ to be the  \textbf{$k$'th finite difference of the discrete cost function}:
\begin{equation}
\begin{aligned}
    \Delta J_D\left(\mathbf{x}[k], \mathbf{u}[k], k\right) & \coloneqq J_D[k+1] - J_D[k]\\
    & = J_D\left(\mathbf{x}[k+1], \mathbf{u}[k+1], k+1\right) - J_D\left(\mathbf{x}[k], \mathbf{u}[k], k\right)
\end{aligned}
\end{equation}
\end{definition}
\begin{definition}
Let us denote the \textbf{$k$'th measurement of the optimal value in regards to state trajectory of the cost function $J_D[k]$} and its \textbf{finite difference} as:
\begin{equation}
\begin{aligned}
     J_D^*[k] \coloneqq & J_D\left(\mathbf{x}[k], \mathbf{u}[k], k\right) \ \big|_{\mathbf{x}[k] \equiv \mathbf{x^*}[k]};\\
    \Delta J_D^*[k] \coloneqq & \Delta J_D[k] \ \big|_{\mathbf{x}[k] \equiv \mathbf{x^*}[k]} = J_D^*[k+1] - J_D^*[k]\ .
\end{aligned}
\end{equation}
\end{definition}
\begin{definition}

Note that by its definition, for all $k \in \bigcup_{i=0}^{L-2}\{i\}$, we can write Equation~\eqref{eqn:discrete_general_cost} in the form:

\begin{equation}
\begin{aligned}
        J_D\left(\mathbf{x}[k], \mathbf{u}[k], k\right) & =
        r_D\left(\mathbf{x}[k], \mathbf{u}[k]\right) +
        \mathlarger{\sum}_{t=k+1}^{L-1} r_D\left(\mathbf{x}[t], \mathbf{u}[t]\right) +
        r_{D_f}\left(\mathbf{x}[L]\right)\\
        & = r_D\left(\mathbf{x}[k], \mathbf{u}[k]\right) +
        J_D\left(\mathbf{x}[k+1], \mathbf{u}[k+1], k+1\right).
        \label{eqn:discrete_bellman_1}
\end{aligned}
\end{equation}

Moving sides and multiplying by $-1$, we have:

\begin{equation}
\begin{aligned}
        J_D\left(\mathbf{x}[k+1], \mathbf{u}[k+1], k+1\right) - J_D\left(\mathbf{x}[k], \mathbf{u}[k], k\right)& = - r_D\left(\mathbf{x}[k], \mathbf{u}[k]\right).
        \label{eqn:discrete_bellman_2}
\end{aligned}
\end{equation}

Using the notation for the $k$'th finite difference of the discrete cost function in Equation~\eqref{eqn:discrete_bellman_2} we have the discrete-time equivalent of Equation~\eqref{eqn:J_dot}:

\begin{equation}
\begin{aligned}
        \Delta J_D[k] = - r_D\left(\mathbf{x}[k], \mathbf{u}[k]\right).
\end{aligned}
\end{equation}

Using this, we symmetrically define the \textbf{$k$'th measurement of the discrete time control Hamiltonian} $\mathcal{H}_D$, similar to Equation~\eqref{eqn:hamiltonian}:

\begin{equation}
    \mathcal{H}_D \left( \mathbf{x}[k], \Delta J_D[k], \mathbf{u}[k] \right)  \coloneqq \Delta J_D[k] + r_D\left(\mathbf{x}[k], \mathbf{u}[k]\right).
    \label{eqn:discrete_hamiltonian}
\end{equation}
\end{definition}
\begin{theorem}[\textbf{Discrete Pontryagin’s Minimum Principle (DPMP)} (Equation 2.2-6, Page 33, Chapter 2 of ~\cite{optimal_control_book})]
Given a discrete LTI system $\mathcal{S}_D$, let us consider the open-loop optimal control problem given in~\eqref{eqn:discrete_optimal_control_problem_formulation} along with its cost function measurements $\left(J_D[k]\right)_{k=0}^{L-1}$, their gradients \\
$\left(\Delta J_D[k]\right)_{k=0}^{L-1}$ and optimal values $\left(\Delta J_D^*[k]\right)_{k=0}^{L-1}$ induced by optimal state trajectory measurements $\left(\mathbf{x^*}[k]\right)_{k=0}^{L-1}$. A necessary and sufficient
condition to obtain the optimal control policy that solves the discrete optimal control problem is to obtain the discrete control policy that minimizes the discrete Hamiltonian. Formally, this may be stated as:

\begin{equation}
     \mathbf{u^*}[k] = \argmin_{\mathbf{u}[k] \in \mathcal{U}_D} \mathcal{H}_D \left( \mathbf{x^*}[k], \Delta J_D^*[k], \mathbf{u}[k] \right)
     \label{eqn:argmin_Hamiltonian_discrete}
\end{equation}

which is equivalent to solving the following stationarity condition:

\begin{equation}
     \frac{\partial}{\partial \mathbf{u}[k]} \mathcal{H}_D \left( \mathbf{x}^*[k],  \Delta J^*_D[k], \mathbf{u}[k] \right) \big|_{\mathbf{u}[k] \equiv \mathbf{u^*}[k]} = 0.
     \label{eqn:discrete_stationarity}
\end{equation}
\end{theorem}

\begin{definition}
Given a discrete LTI system $\mathcal{S}_D$, let us consider the open-loop optimal control problem given in~\eqref{eqn:discrete_optimal_control_problem_formulation} along with its cost function measurements $\left(J_D[k]\right)_{k=0}^{L-1}$ and their gradients $\left(\Delta J_D[k]\right)_{k=0}^{L-1}$. Let us use the notation for the Discrete Control Hamiltonian $\mathcal{H}_D$ from~\eqref{eqn:hamiltonian} and plug it in~\eqref{eqn:pre_hamilton} to define the \textbf{Discrete Bellman Equation}:

\begin{equation}
    \mathcal{H}_D \left( \mathbf{x}[k], \nabla J[k], \mathbf{u}[k] \right)  = 0.
    \label{eqn:discrete_bellman}
\end{equation}
\end{definition}
Let us now consider the infinite and finite horizon separately.

\subsection{Discrete Infinite Horizon LTI LQR}

In the infinite horizon case we consider the following form\footnote{Full derivation of the discretization procedure of the infinite horizon cost function is detailed at Appendix\ref{app:performance_discretization}.} for the discrete cost function:

\begin{equation}
\begin{aligned}
        J_D\left(\mathbf{x}[k], \mathbf{u}[k], k\right) \equiv
      \mathlarger{\sum}_{t=k}^{\infty} \left[\mathbf{x}[t]^TQ\mathbf{x}[t] +  \mathbf{u}[t]^TR\mathbf{u}[t]\right]
\end{aligned}
\end{equation}

\begin{proposition}
Given a discrete LTI system $\mathcal{S}_D$, let us consider the optimal control problem given in~\eqref{eqn:discrete_optimal_control_problem_formulation} along with the measurements of its induced optimal state trajectory $\left(\mathbf{x^*}[k]\right)_{k=0}^{L-1}$. If $L \rightarrow \infty$ then:
\begin{enumerate}
    \item The measurements of the optimal solution $\left(\mathbf{u^*}[k]\right)_{k=0}^{L-1}$ to the problem are actually that of a closed-loop feedback controller, which is linear with regards to optimal state, with a constant coefficient, i.e. it is of the form:
\begin{equation}
    \mathbf{u^*}[k] = - K_D^*\mathbf{x^*}[k],
\end{equation}
with a constant matrix $K_D^* \in \mathbb{R}^{n \times m}$, for all $k \in \bigcup_{i=0}^{L-1}\{i\}$.
\item The optimal controller $K_D^*$ is of the following form:
\begin{equation}
    K_D^* \coloneqq R^{-1}B^T P_D
\end{equation}
when $P_D = P_D^T \in \mathbb{R}^{n \times n}$ is a solution to the Discrete Algebraic Riccati Equation (DARE):
\begin{equation}
    P_D =  \tilde{A}^T P_D \tilde{A}  - \tilde{A}^T P_D \tilde{B}\left(\tilde{B}^TP_D\tilde{B} + R \right)^{-1}\tilde{B}^TP_D\tilde{A} + Q.
\end{equation}
\end{enumerate}
\end{proposition}
\begin{proof}

Due to Theorem~\ref{the:discrete_lqr_quad_cost}, we now consider the $k$'th measurement of the optimal cost function to be quadratic with respect to the $k$'th measurement of the optimal state:

\begin{equation}
    J_D^*[k] \equiv \mathbf{x^*}^T[k] P_D \mathbf{x^*}[k].
    \label{eqn:discrete_infinite_horizon_optimal_cost}
\end{equation}

for some constant positive-definite matrix $P_D \in \mathbb{R}^{n \times n}$. Plugging all optimal the expressions to the Hamiltonian at~\eqref{eqn:discrete_hamiltonian} and:

\begin{equation}
\begin{aligned}
      \mathcal{H}_D & \left( \mathbf{x}^*[k],  \Delta J^*_D, \mathbf{u}[k] \right)    \\ & =J^*_D\left(\mathbf{x}[k+1], \mathbf{u}[k+1]\right)
    - J^*_D\left(\mathbf{x}[k], \mathbf{u}[k], k\right) + r_D\left(\mathbf{x}^*[k], \mathbf{u}[k]\right)\\
    &= \mathbf{x^*}^T[k+1] P_D \mathbf{x^*}[k+1]
    - \mathbf{x^*}^T[k] P_D \mathbf{x^*}[k] + \mathbf{x}^{*^T}[k]Q\mathbf{x}^*[k] +  \mathbf{u}[k]^TR\mathbf{u}[k]\\
    &= \left(\tilde{A}\mathbf{x}^*[k] + \tilde{B} \mathbf{u}[k]\right)^T P_D \left(\tilde{A}\mathbf{x}^*[k] + \tilde{B} \mathbf{u}[k]\right)
     + \mathbf{x}^{*^T}[k]\left(Q-P_D\right)\mathbf{x}^*[k] +  \mathbf{u}[k]^TR\mathbf{u}[k]\\
     & = \mathbf{x}^{*^T}[k] \left( \tilde{A}^T P_D \tilde{A} + Q - P_D\right)\mathbf{x}^{*}[k] + \mathbf{x}^{*^T}[k]\tilde{A}^TP_D\tilde{B}\mathbf{u}[k] +  \mathbf{u}^T[k]\tilde{B}^TP_D\tilde{A}\mathbf{x}^*[k]\\
     & + \mathbf{u}[k]^T\left( \tilde{B}^TP_D\tilde{B} + R\right)\mathbf{u}[k].
     \label{eqn:discrete_infinite_Hamiltonian}
\end{aligned}
\end{equation}

Applying DPMP to Equation~\eqref{eqn:discrete_infinite_Hamiltonian} we have:

\begin{equation}
\begin{aligned}
      \frac{\partial}{\partial \mathbf{u}[k]} \mathcal{H}_D & \Big( \mathbf{x}^*[k],  \Delta J^*_D, \mathbf{u}[k] \Big)\Big|_{\mathbf{u}[k] \equiv \mathbf{u^*}[k]}\\ 
     & = \Bigg[ \left(\mathbf{x}^{*^T}[k]\tilde{A}^TP_D\tilde{B}\right)^T +  \tilde{B}^TP_D\tilde{A}\mathbf{x}^*[k]
     \\
     & + \Big[ \tilde{B}^TP_D\tilde{B} + R + \big( \tilde{B}^TP_D\tilde{B}  + R\big)^T\Big]\mathbf{u}[k]\Bigg]\Bigg|_{\mathbf{u}[k] \equiv \mathbf{u^*}[k]}\\
     & = \left[2\tilde{B}^TP_D\tilde{A}\mathbf{x}^*[k] + 2 \left(\tilde{B}^TP_D\tilde{B} + R \right)\mathbf{u}[k]\right]\Big|_{\mathbf{u}[k] \equiv \mathbf{u^*}[k]} = 0.
\end{aligned}
\end{equation}

Which yields the following discrete optimal input:

\begin{equation}
\begin{aligned}
      \mathbf{u^*}[k] = - \left(\tilde{B}^TP_D\tilde{B} + R \right)^{-1}\tilde{B}^TP_D\tilde{A}\mathbf{x}^*[k].
\end{aligned}
\end{equation}

Once again denoting the optimal discrete controller as:

\begin{equation}
    K_D^* \coloneqq \left(\tilde{B}^TP_D\tilde{B} + R \right)^{-1}\tilde{B}^TP_D\tilde{A},
\end{equation}

we have the the optimal control law in the discrete case is also linear. 

\begin{definition}
Plugging the optimal input to the expression for the Hamiltonian at~\eqref{eqn:discrete_infinite_Hamiltonian} and then to the discrete Bellman equation at~\eqref{eqn:discrete_bellman}, we have the \textbf{discrete HJB (DHJB)} Equation:

\begin{equation}
\begin{aligned}
     &\mathbf{x}^{*^T}[k] \Bigg[ \tilde{A}^T P_D \tilde{A} - P_D - \tilde{A}^T P_D \tilde{B}\left(\tilde{B}^TP_D\tilde{B} + R \right)^{-1}\tilde{B}^TP_D\tilde{A}
      + Q \Bigg]\mathbf{x}^{*}[k]= 0.
     \label{eqn:discrete_HJB}
\end{aligned}
\end{equation}

Which corresponds to the \textbf{Discrete time Algebraic Riccati Equation} (DARE):

\begin{equation}
\begin{aligned}
     P_D =  \tilde{A}^T P_D \tilde{A}  - \tilde{A}^T P_D \tilde{B}\left(\tilde{B}^TP_D\tilde{B} + R \right)^{-1}\tilde{B}^TP_D\tilde{A} + Q
     \label{eqn:DARE}
\end{aligned}
\end{equation}

from which $P_D$ can be obtained.
\end{definition}
\end{proof}

\subsection{Discrete Finite Horizon LTI LQR}

In the finite horizon case we consider the following form\footnote{Full derivation of the discretization procedure of the finite horizon cost function is detailed at Appendix\ref{app:performance_discretization}.} for the discrete cost function:

\begin{equation}
\begin{aligned}
        J_D\Big(\mathbf{x}[k], \mathbf{u}[k], k\Big) \equiv \delta
      \mathlarger{\sum}_{t=k}^{{L-1}} \left[\mathbf{x}[t]^TQ\mathbf{x}[t] +  \mathbf{u}[t]^TR\mathbf{u}[t]\right] + \mathbf{x}[L]^TQ\mathbf{x}[L].
\end{aligned}
\end{equation}
\begin{proposition}
Given a discrete LTI system $\mathcal{S}_D$, let us consider the optimal control problem given in~\eqref{eqn:discrete_optimal_control_problem_formulation} along with the measurements of its induced optimal state trajectory $\left(\mathbf{x^*}[k]\right)_{k=0}^{L-1}$. If $L < \infty$ then:
\begin{enumerate}
    \item The measurements of the optimal solution $\left(\mathbf{u^*}[k]\right)_{k=0}^{L-1}$ to the problem are actually that of a closed-loop feedback controller, which is linear with regards to optimal state, with a time-varying coefficient, i.e. it is of the form:
\begin{equation}
    \mathbf{u^*}[k] = - K_D^*[k]\mathbf{x^*}[k],
\end{equation}
with a time-varying matrix $K_D^* \colon \bigcup_{i=0}^{L-1}\{i\} \rightarrow \mathbb{R}^{n \times m}$, for all $k \in \bigcup_{i=0}^{L-1}\{i\}$.
\item The measurements of the optimal controller $\left(K_D^*[k]\right)_{k=0}^{L-1}$ are of the following form:
\begin{equation}
    K_D^*[k] \coloneqq R^{-1}B^T P_D[k]
\end{equation}
for all $k \in \bigcup_{i=0}^{L-1}\{i\}$, when $P_D = P_D^T \colon \bigcup_{i=0}^{L-1}\{i\} \rightarrow \mathbb{R}^{n \times n}$ is a solution to the Discrete Dynamic Riccati Equation (DDRE):
\begin{equation}
    P_D[k] = \tilde{A}^T A_{cl}[k+1]\tilde{A} + \delta Q,
\end{equation}
where:
\begin{equation}
    A_{cl}[k] = P_D[k] \left[I_{n} -  \tilde{B}\left(\tilde{B}^TP_D[k]\tilde{B} + \delta R \right)^{-1}\tilde{B}^TP_D[k]\right]
\end{equation}
\end{enumerate}
\end{proposition}
\begin{proof}

Now we also propose the optimal cost to be quadratic with respect to the state due to Theorem~\ref{the:discrete_lqr_quad_cost}:

\begin{equation}
    J_D^*[k] \equiv \mathbf{x^*}^T[k] P_D[k] \mathbf{x^*}[k].
    \label{eqn:discrete_finite_horizon_optimal_cost}
\end{equation}

but this time with the positive-definite matrix $P_D: \mathbb{N} \rightarrow \mathbb{R}^{n \times n}$ being dynamic with respect to the measurement index, and supposing the cost of the $k$'th measurement is calculated using the value of the function $P_D$ at the next ($k+1$)'th measurement, i.e. $P_D[k+1]$. Thus we again expect to solve backwards in measurements to find the values for $P_D[k]$ and $J_D[k]$. Plugging~\eqref{eqn:discrete_finite_horizon_optimal_cost} to the expression for the Hamiltonian at~\eqref{eqn:discrete_hamiltonian} we have:

\begin{equation}
\begin{aligned}
      \mathcal{H}_D & \left( \mathbf{x}^*[k],  \Delta J^*_D[k], \mathbf{u}[k] \right)    \\ & =J^*_D\left(\mathbf{x}[k+1], \mathbf{u}[k+1]\right)
    - J^*_D\left(\mathbf{x}[k], \mathbf{u}[k], k\right) + r_D\left(\mathbf{x}^*[k], \mathbf{u}[k]\right)\\
    &= \mathbf{x^*}^T[k+1] P_D[k+1] \mathbf{x^*}[k+1]
    - \mathbf{x^*}^T[k] P_D[k] \mathbf{x^*}[k] + \delta\mathbf{x}^{*^T}[k]Q\mathbf{x}^*[k] \\
    &=+  \delta\mathbf{u}[k]^TR\mathbf{u}[k]\\
    &= \left(\tilde{A}\mathbf{x}^*[k] + \tilde{B} \mathbf{u}[k]\right)^T P_D[k+1] \left(\tilde{A}\mathbf{x}^*[k] + \tilde{B} \mathbf{u}[k]\right) + \mathbf{x}^{*^T}[k]\left(\delta Q-P_D[k]\right)\mathbf{x}^*[k]
     \\
     & +  \mathbf{u}[k]^T\delta R\mathbf{u}[k]\\
     & = \mathbf{x}^{*^T}[k] \left( \tilde{A}^T P_D[k+1] \tilde{A} + \delta Q - P_D[k]\right)\mathbf{x}^{*}[k] + \mathbf{x}^{*^T}[k]\tilde{A}^TP_D[k+1]\tilde{B}\mathbf{u}[k] \\
     & +  \mathbf{u}^T[k]\tilde{B}^TP_D[k+1]\tilde{A}\mathbf{x}^*[k]+ \mathbf{u}[k]^T\left( \tilde{B}^TP_D[k+1]\tilde{B} + \delta R\right)\mathbf{u}[k].
     \label{eqn:discrete_finite_Hamiltonian}
\end{aligned}
\end{equation}

Applying DPMP to Equation~\eqref{eqn:discrete_finite_Hamiltonian} we have:

\begin{equation}
\begin{aligned}
      \frac{\partial}{\partial \mathbf{u}[k]} \mathcal{H}_D & \left( \mathbf{x^*}[k],  \Delta J^*_D[k], \mathbf{u}[k] \right)\Big|_{\mathbf{u}[k] \equiv \mathbf{u^*}[k]}\\ 
     & =  \Bigg[\left(\mathbf{x^*}^{T}[k]\tilde{A}^TP_D[k]\tilde{B}\right)^T +  \tilde{B}^TP_D[k]\tilde{A}\mathbf{x^*}[k]
     \\
     & + \left[ \tilde{B}^TP_D[k]\tilde{B} + \delta R + \left( \tilde{B}^TP_D[k]\tilde{B}  + \delta R\right)^T\right]\mathbf{u}[k]\Bigg]\Bigg|_{\mathbf{u}[k] \equiv \mathbf{u^*}[k]}\\
     & = \left[2\tilde{B}^TP_D[k]\tilde{A}\mathbf{x^*}[k] + 2 \left(\tilde{B}^TP_D[k]\tilde{B} + \delta R \right)\mathbf{u}[k]\right]\Big|_{\mathbf{u}[k] \equiv \mathbf{u^*}[k]} = 0.
\end{aligned}
\end{equation}

Which yields the following discrete optimal input:

\begin{equation}
\begin{aligned}
      \mathbf{u^*}[k] = - \left(\tilde{B}^TP_D[k]\tilde{B} + \delta R \right)^{-1}\tilde{B}^TP_D[k]\tilde{A}\mathbf{x^*}[k].
\end{aligned}
\end{equation}

Once again denoting the optimal discrete controller as:

\begin{equation}
    K_D^*[k] \coloneqq \left(\tilde{B}^TP_D[k]\tilde{B} + \delta R \right)^{-1}\tilde{B}^TP_D[k]\tilde{A},
\end{equation}

we have the the optimal control law in the discrete case is also linear. Plugging the optimal input to the expression for the Hamiltonian at~\eqref{eqn:discrete_finite_Hamiltonian} and then to the discrete Bellman equation at~\eqref{eqn:discrete_bellman}, we have the DHJB equation  for the finite horizon case:

\begin{equation}
\begin{aligned}
     \mathbf{x}^{*^T}[k] \Big[& \tilde{A}^T P_D[k+1] \tilde{A} - P_D[k] \\
     &- \tilde{A}^T P_D[k+1] \tilde{B}\left(\tilde{B}^TP_D[k+1]\tilde{B} + \delta R \right)^{-1}\tilde{B}^TP_D[k+1]\tilde{A}\\
      & + \delta Q \Big]\mathbf{x}^{*}[k]= 0.
     \label{eqn:finite_discrete_HJB}
\end{aligned}
\end{equation}

Which corresponds to the \textbf{Discrete-time Dynamic Riccati Equation} (DDRE):

\begin{equation}
\begin{aligned}
     P_D[k] = & \tilde{A}^T P_D[k+1] \left[I_{n} -  \tilde{B}\left(\tilde{B}^TP_D[k+1]\tilde{B} + \delta R \right)^{-1}\tilde{B}^TP_D[k+1]\right]\tilde{A} + \delta Q,
     \label{eqn:DDRE}
\end{aligned}
\end{equation}

or to its more concise form:

\begin{equation}
\begin{aligned}
     P_D[k] = & \tilde{A}^T A_{cl}[k+1]\tilde{A} + \delta Q,
\end{aligned}
\end{equation}

while denoting:
    
\begin{equation}
    A_{cl}[k+1] \coloneqq P_D[k+1] \left[I_{n} -  \tilde{B}\left(\tilde{B}^TP_D[k+1]\tilde{B} + \delta R \right)^{-1}\tilde{B}^TP_D[k+1]\right],
\end{equation}

from which $P_D[k]$ can be obtained by solving iteratively backwards in time with the terminal condition\footnote{While setting this terminal condition might seem odd, as the matrix $Q$ is the weighting matrix for the state deviation, and has no immediate connection to take the place for $P$ at the terminal time, We will later see that setting the terminal condition to by any semi-positive $n \ times n$ matrix will provide desirable results, and for this reason this methodology is typically used.}: $P[L] \equiv Q$.
\end{proof}

\chapter{Related Work}
\label{chap:related_work}
\epigraph{\textit{The public is not to see where power lies, how it shapes policy, and for what ends. Rather, people are to hate and fear one another.}}{\textbf{Noam Chomsky}}

Many have dealt with `\textit{divide and conquer}' methods in different science and engineering disciplines. In Computer Science, for example, it is common, when focusing on the design of software components, to consider implementing them by means of the result of non-cooperative, competitive and even adversarial behaviours of their corresponding  components~\cite{clark2016software}. Divide and Conquer actually is an accepted name for a whole class of algorithms that enact the aforementioned methodology to solve problems.

\section{Divide and Conquer Algorithms}
 A D\&C algorithm is termed so as it breaks a problem down into smaller sub-problems again and again, until they become simple enough to be solved separately, then the solutions are combined to account for the original problem. 
 Directly quoted from the introduction to Chapter 5 of~\cite{alma9926618719804361}:
 
 \begin{quote}
     One of the most powerful techniques for solving problems is to break them down into smaller, more easily solved pieces. Smaller problems are less overwhelming, and they permit us to focus on details that are lost when we are studying the whole thing. A recursive algorithm starts to become apparent whenever we can break the problem into smaller instances of the same type of problem. Two important algorithm design paradigms are based on breaking problems down into smaller problems:
     
     \begin{enumerate}
         \item Dynamic programming, which typically removes one element from the problem, solves the smaller problem, and then adds back the element to the solution of this smaller problem in the proper way.
         \item Divide and conquer instead splits the problem into (say) halves, solves each half, then stitches the pieces back together to form a full solution.
     \end{enumerate}
     
       Thus, to use divide and conquer as an algorithm design technique, we must divide the problem into two smaller subproblems, solve each of them recursively, and then meld the two partial solutions into one solution to the full problem. Whenever the merging takes less time than solving the two subproblems, we get an efficient algorithm. 
       
       Mergesort~\cite{mergesort}, is the classic example of a divide-and-conquer algorithm. It takes only linear time to merge two sorted lists of $n/2$ elements, each of which was obtained in $O(nlogn)$ time. Divide and conquer is a design technique with many important algorithms to its credit, including mergesort, the Fast Fourier Transform (FFT)~\cite{FFT}, and Strassen’s matrix multiplication algorithm~\cite{STRASSEN1969}.
 \end{quote}

 This paragraph demonstrates one motivation for D\&C algorithms in the form of a \textbf{time speedup}~\cite{HennessyPatterson12}, which is of course reasonable and very much desirable. What we will claim, is that our method is not necessarily \textbf{faster} than performing direct LQR or MPC, or other methods, but that it provides a much \textbf{simpler way to distribute weights across a large amount of objectives}, especially when the number of state variables and objectives grows, and that can result in better transient and steady-state performance.



\section{Single-Agent Multi-Objective Methods}

The methodology described in this paper brings this kind of separation of concerns to control engineering. We are not the first to propose ways of modular construction of composite controllers. Specifically for control systems, the following methodologies were described:

\begin{itemize}
    \item Behavioral Control~\cite{mataricbehavior,brooks2robust};
    \item Null Space Behaviour (NSB)~\cite{ antonelli2010flocking, arad2014coordinated};
    \item Model Predictive Control (MPC)~\cite{ camacho2013model, lee2011model};
    \item  Multi-Objective Optimization (MOO)~\cite{MOO_dial_a_ride, MOO_energy, MOO_fuel, MOO_green_vehicle, MOO_HMSC}
\end{itemize}
 These are state-of-the-art methods to achieve goals similar to what we are aiming to solve. We will now provide a brief overview of some of these methods.
 
 \subsection{Behavioral Control}
 \label{subsect:BP}
 The \textbf{Behavioral Control} technique, first proposed by Brooks is a popular control method to handle conflicts between objectives. This includes the layered approach and the motor schema approach~\cite{arkin1989motor}. In the layered approach components that handle objectives compete on execution time where a lower priority components can only be executed after the completion of higher priority components. The motor schema approach is based on cooperation: the output is computed as a weighted sum of the tasks~\cite{Dynamic_Task_Priority_Planning}.
 
 \subsection{Null Space Behaviour (NSB)}
 One such cooperative approach, proposed by Antonelli and Chiaverini~\cite{antonelli2009experiments} is called \textbf{Null-Space-based Behavior}. NSB gives tools to design controllers when there is a clear (possibly dynamic) hierarchy of the goals. In such case, the method proposes to pursue less important goals only withing the "null space" of the more important ones, i.e., only by changing the control signal in a way that does not affect the performance of the controller that pursues the important goals. This has proven very useful in many applications, but it is very conservative. It is not hard to imagine situations where strict hierarchy is not optimal, e.g., when we can gain much by degrading the performance of the important controller a little in exchange to a considerable improvement of the performance of less important ones. Another issue with NSB is that it composes the control signals at a specific time without looking at other times.

\subsection{Model  Predictive  Control  (MPC)}
\label{subsect:MPC}
\textbf{Model Predictive Control} is an effective means of dealing with large multi-variable constrained control problems. The main idea is to choose the control action by repeatedly solving online an optimal control problem, aiming to minimize a performance criterion over a future horizon, possibly subject to constraints on the manipulated inputs and outputs. Future behavior is computed according to a model of the plant~\cite{camacho2013model}. The main weaknesses of the approach are: 
\begin{enumerate}
    \item It is hard to guarantee closed-loop stability;
    \item It is hard to handle model uncertainty;
    \item It requires heavy online computations~\cite{camacho2013model, eren2017model}.
\end{enumerate}
While these are significant limitations, many case studies have shown that the approach gives viable solutions for many practical problems, especially when the system's dynamics is relatively slow~\cite{ vazquez2014model}. Although MPC in its basic configuration allows some representation of multiple objectives, through weighting and through the representation of some objectives as constraints, in the last decade significant work has been done on expanding MPC's scope for explicit multi-objective control, under the title of Multi-Objective MPC (MOMPC)~\cite{ peitz2017multiobjective, zhang2019multi}.

\subsection{Multi-Objective Optimization (MOO)}

Many complex control tasks can be treated as Multi-Objective Optimization (MOO) problems, since the decisions required by a controller often affect various parameters such as safety, speed, fuel consumption and more~\cite{MOO_dial_a_ride, MOO_fuel, MOO_energy, MOO_green_vehicle, MOO_HMSC}. While MOO issues have been explored since the nineteenth century following Pareto's groundbreaking work~\cite{pareto}, tools are still lacking for control engineers to build controllers that take into account the need for constant balancing of different targets as needed for solving, e.g., complex road situations~\cite{control_limits, inconsistencies}.

\section{Contemporary Methods for Solving the CARE}
\label{app:solving_the_care}

 The CARE at~\eqref{eqn:riccati} has been extensively studied for its methods of solutions and their properties. This~\cite{CARE_comparison} article provides, as it is titles a \textit{Numerical Comparison of Different Solvers for Large-Scale, Continuous-Time Algebraic Riccati Equations and LQR Problems}. While these methods rely on numerical analysis, we will take on a different approach and suggest a novel theoretical method of solving the equation. We will of course employ numerical solvers to implement these methods at real-time, but the main contribution in this work is from a theoretical point of view. 
 
 Let us ten first discuss the possible solutions the CARE can produce. It is well-known that it can possess a variety of solutions. First of all, it
may have no solution at all. If it does have one, there can be both real and complex
solutions, some of them being Hermitian or symmetric. Finally, there can be even an infinite and moreover uncountable amount of solutions. It is also known that due to the underlying physical problem, only non-negative solutions are applicable solutions.~\cite{riccati_review}

\subsection{Number of Solutions to the CARE}
\label{sect:single_CARE_solution_num}
 As we laid out, the CARE is a quadratic matrix equation for which the solution is a constant matrix $P$. When deriving the equation, we considered $P$ to be symmetric with respect to the number of variables in the state vector, i.e., $P \in \mathbb{R}^{n \times n}$. This means $P$ can be written in the following general form:
 \begin{equation}
     P = \begin{bmatrix}
     p_{1,1} & p_{1,2} & \cdots & p_{1,n-1} & p_{1,n}\\
     p_{1,2} & p_{2,2} & \cdots & p_{2,n-1} & p_{2,n}\\
     \vdots & \vdots & & \vdots &\vdots\\
     p_{1,n-1} & p_{2,n-1} & \cdots & p_{n-1,n-1} & p_{n-1,n}\\
     p_{1,n} & p_{2,n} & \cdots & p_{n-1,n} & p_{n,n}
      \end{bmatrix}.
 \end{equation}
 
 where $\forall 1 \leq r < s \leq n \ ; \ p_{r,s} \in \mathbb{R}$. Let us consider the matrix equation as an equivalent set of scalar equations. 
 
 \begin{proposition}
 \label{prop:balanced_CARE}
The CARE at~\eqref{eqn:riccati} admits a balanced\footnote{
\begin{definition}
Consider a set of $e \in \mathbb{N}$ polynomial equations with $u \in \mathbb{N}$ unknowns. Then:
 \begin{itemize}
     \item If $e < u$, then we call the set \textbf{underdetermined};
     \item If $e > u$, then we call the set \textbf{overdetermined};
     \item If $e = u$, then we call the set \textbf{balanced}.
 \end{itemize}
 \end{definition}} set of equations.
 \end{proposition}
  \begin{proof}
 Notice the total amount of unique scalar elements in $P$ will be less than $n^2$, and more precisely, exactly $T_n$, the triangular number\footnote{
 \begin{definition}
 Let $k \in \mathbb{N}$. The sum of all numbers from 1 to $k$, commonly known as the \textbf{triangular number of order $k$}~\cite{Hayes2006GausssDO}, is denoted $T_k$ and given by:
 \begin{equation}
     T_k = \sum_{i=1}^k i = \frac{k(k+1)}{2}
     \label{eqn:T_n}
 \end{equation}
 \end{definition}
 } of order $n$. With that, we have that the matrix $P$ has all of its unique scalar elements in its upper triangle. Since the LHS of the CARE at~\eqref{eqn:riccati} is also symmetric, this fact will cause the equations of the lower triangle elements to be linearly-dependant upon the equations of the upper triangle elements. Thus we can discard either set of equations and end up with a set of $T_n$  unique quadratic equations for $T_n$ unique unknowns. 
 \end{proof}
 
 \begin{corollary}
 \label{cor:CARE_exp}
 Let us consider the CARE at~\eqref{eqn:riccati}, and assume it is well-behaved\footnote{\begin{definition}
 Consider a set of equations. If it has a solution, then we call the set \textbf{consistent}, and \textbf{inconsistent} otherwise.
 \end{definition}
 \begin{definition}
 Consider a consistent set of equations. If it admits a finite number of solutions, then we call the set \textbf{zero-dimensional}. Otherwise, in the case it admits a infinite number of solutions, then we call the set \textbf{positive-dimensional}.
 \end{definition}
 \begin{definition}
 Consider a consistent and balanced set of equations. If the set is also zero-dimensional then we call the set \textbf{well-behaved}.
 \end{definition}}. Then the maximal number of solutions to the CARE is exponential with respect to $n$.
 \end{corollary}
 \begin{proof}
 The CARE at~\eqref{eqn:riccati} is well-behaved and balanced with $(T_i)_{i=1}^{T_n}$ polynomial equations. Moreover, the CARE is a set of second degree\footnote{\begin{definition}
 Consider a polynomial equation
 \begin{equation}
     T(x)\equiv \sum_{i=1}^k a_i x^i = 0,
 \end{equation}
 with $x$ as its unknown, i.e. a function of the form $T \colon \mathbb{R} \rightarrow \mathbb{R}$. The \textbf{degree} $d(T)$ of the equation is the maximal power of $x$ in the equation.
 \end{definition}} polynomials, i.e., it satisfies $d(T_i) \equiv 2$ for all $1 \leq i \leq T_n$, since each polynomial is quadratic. 
 \begin{lemma}
 [\textbf{Bézout's Theorem}~\cite{bezout}]
 \label{the:bez}
 Consider a well-behaved series of $e \in \mathbb{N}$ polynomial equations $(T_i)_{i=1}^{e}$, with corresponding degrees $(d(T_i))_{i=1}^{e}$. Then the number of solutions $\mathcal{N}$ to the set satisfies:
 \begin{equation}
     \mathcal{N} = O\left(\prod_{i=1}^{e}d(T_i)\right)
 \end{equation}
 \end{lemma}
 Thus, by Lemma~\ref{the:bez} we have that the its number of possible solutions to the CARE satisfies:
 \begin{equation}
         \mathcal{N} = O\left(\prod_{i=1}^{T_n}d(T_i)\right) =  O\left(\prod_{i=1}^{T_n}2\right)= O(2^{T_n}) = O(2^{n^2}).
     \end{equation}
 \end{proof}

\subsection{General Solution Construction to the CARE}

We will now describe the general method for constructing symmetric solutions to the CARE, as demonstrated in section 13.1 of~\cite{robust_optimal}.

\begin{definition}
Let us consider the CARE at~\eqref{eqn:riccati}. Let the \textbf{CARE Hamiltonian} $H \in \mathbb{R}^{2n \times 2n}$ be defined as:
\begin{equation}
    H \coloneqq \begin{bmatrix}
    A & -BR^{-1}B^T\\
    -Q & -A^T
    \end{bmatrix}
\end{equation}
\end{definition}
\begin{proposition}
Given the CARE at~\eqref{eqn:riccati} and its CARE Hamiltonian $H$, we have that $H$ is Hamiltonian~\ref{subsect:Hamiltonian_Matrices}.
\end{proposition}
\begin{proof}
Plugging $H$ to the LHS of Equation~\eqref{eqn:hamiltonian_matrix_condition} to obtain:

\begin{equation}
\begin{aligned}
    W H^T W = & \begin{bmatrix}
    O_{n \times n} & -I_{n}\\
    I_{n} &O_{n \times n}
    \end{bmatrix} \begin{bmatrix}
    A^T & -Q\\
    -BR^{-1}B^T & -A
    \end{bmatrix} \begin{bmatrix}
    O_{n \times n} & -I_{n}\\
    I_{n} &O_{n \times n}
    \end{bmatrix} \\
    = & \begin{bmatrix}
    BR^{-1}B^T & A\\
    A^T & -Q
    \end{bmatrix} \begin{bmatrix}
    O_{n \times n} & -I_{n}\\
    I_{n} &O_{n \times n}
    \end{bmatrix} \\
    = & \begin{bmatrix}
    A & BR^{-1}B^T\\
    -Q & -A^T
    \end{bmatrix} = H
    \end{aligned}
\end{equation}
\end{proof}

\begin{proposition}[Theorem 13.2 from page 321 of~\cite{robust_optimal}]
\label{prop:care_con}
Let us consider the CARE at~\eqref{eqn:riccati} and its CARE Hamiltonian $H$. Then the CARE is consistent, i.e. there exists a solution $X \in \mathbb{R}^{n \times n}$ if and only if the following conditions are satisfied:
\begin{enumerate}
    \item There exists an $H$-invariant subspace $\mathcal{W} \subseteq \mathbb{R}^{2n}$;
    \item There exist two real square matrices $X_1,X_2 \in \mathbb{R}^{n \times n}$ such that:
\begin{equation}
    \mathcal{W} = \operatorname{Im} \left(\begin{bmatrix}
    X_1 \\
    X_2
    \end{bmatrix}\right);
    \label{eqn:W_image}
    \end{equation}
    \item $X_1$ in nonsingular and thus $X_1^{-1}$ exists.
    \item The solution $X$ is of the form $X \coloneqq X_2X_1^{-1}$
\end{enumerate}
\end{proposition}
\begin{proof}
\begin{enumerate}
    \item $\rightarrow$ \\
    Let us assume conditions $1-4$ are satisfied. Due to the relation in Equation~\eqref{eqn:W_image} we trivially have: 
\begin{equation}
    \begin{bmatrix}
    X_1 \\
    X_2
    \end{bmatrix} \in \mathcal{W}.
\end{equation}
Moreover, since $\mathcal{W}$ is $H$-invariant, by \eqref{eqn:inv_sub_cond} we also have: 
\begin{equation}
    H \begin{bmatrix}
    X_1 \\
    X_2
    \end{bmatrix} \in \mathcal{W}.
\end{equation}

Thus by Equation~\eqref{eqn:W_image}, 
there exists a matrix $\Lambda \in \mathbb{R}^{n \times n}$ such that:

\begin{equation}
    H \begin{bmatrix}
    X_1 \\
    X_2
    \end{bmatrix} = \begin{bmatrix}
    X_1 \\
    X_2
    \end{bmatrix} \Lambda.
\end{equation}

Now we multiply by $X_1^{-1}$ from the right and extract $X_1$ from the RHS to obtain:

\begin{equation}
    H \begin{bmatrix}
    I_{n} \\
    X_2 X_1^{-1}
    \end{bmatrix} = \begin{bmatrix}
    I_{n} \\
    X_2 X_1^{-1}
    \end{bmatrix} X_1 \Lambda X_1^{-1}.
\end{equation}

Let us denote $X \coloneqq X_2 X_1^{-1}$ and then multiply by $\begin{bmatrix}-X & I_{n} \end{bmatrix}$ from the right to obtain:

\begin{equation}
    \begin{bmatrix}-X & I_{n} \end{bmatrix} H \begin{bmatrix}
    I_{n} \\
    X
    \end{bmatrix} = \begin{bmatrix}-X & I_{n} \end{bmatrix} \begin{bmatrix}
    I_{n} \\
    X
    \end{bmatrix} X_1 \Lambda X_1^{-1}.
\end{equation}
Notice the RHS is now:

\begin{equation}
    \begin{bmatrix}-X & I_{n} \end{bmatrix} \begin{bmatrix}
    I_{n} \\
    X
    \end{bmatrix} X_1 \Lambda X_1^{-1} = \left(-X + X\right) X_1 \Lambda X_1^{-1} = 0_{n \times n},
\end{equation}
Thus we have:

\begin{equation}
    \begin{bmatrix}-X & I_{n} \end{bmatrix} \begin{bmatrix}
    A & -BR^{-1}B^T\\
    -Q & -A^T
    \end{bmatrix} \begin{bmatrix}
    I_{n} \\
    X
    \end{bmatrix} = 0_{n \times n},
\end{equation}

which corresponds to:

\begin{equation}
    \begin{bmatrix}
    -XA - Q & XBR^{-1}B^T - A^T
    \end{bmatrix} \begin{bmatrix}
    I_{n} \\
    X
    \end{bmatrix} = 0_{n \times n}.
\end{equation}

Performing matrix multiplication:

\begin{equation}
    -XA - Q + XBR^{-1}B^TX - A^TX
     = 0_{n \times n},
\end{equation}

and finally, multiplying both sides by $-1$ and rearranging we have:

\begin{equation}
     XA + A^TX - XBR^{-1}B^TX + Q 
     = 0_{n \times n}.
\end{equation}

We see that we got an equivalent expression to the CARE from \eqref{eqn:riccati}, and thus $X$ solves the CARE.
    \item $\leftarrow$ \\
    Let us assume there exists a solution $X \in \mathbb{R}^{n \times n}$ that solves the CARE. 
    Thus the closed-loop dynamics is:
    \begin{equation}
        A_{cl} \equiv A - BR^{-1}B^TX.
    \end{equation}
    Plugging this to the CARE we have:
    \begin{equation}
        XA_{cl} = - Q - A^T X .
    \end{equation}
    Notice the last two equations can be written in block matrix form:
    \begin{equation}
    \begin{aligned}
        \begin{bmatrix}
         A - BR^{-1}B^TX\\
        - Q  - A^T X
        \end{bmatrix} = \begin{bmatrix}
        A_{cl}\\
        XA_{cl}
        \end{bmatrix}.
    \end{aligned}
    \end{equation}
    Taking $A_{cl}$ out from the RHS:
    \begin{equation}
    \begin{aligned}
        \begin{bmatrix}
        A - BR^{-1}B^TX\\
        - Q - A^T X
        \end{bmatrix} = \begin{bmatrix}
        I_n\\
        X
        \end{bmatrix}A_{cl},
    \end{aligned}
    \end{equation}
    Notice the LHS can be written as:
    \begin{equation}
    \begin{aligned}
        \begin{bmatrix}
        A & - BR^{-1}B^T\\
        - Q & - A^T
        \end{bmatrix} 
        \begin{bmatrix}
        I_n \\
        X
        \end{bmatrix}= 
        \begin{bmatrix}
        I_n\\
        X
        \end{bmatrix}A_{cl},
    \end{aligned}
    \end{equation}
    Notice the leftmost matrix is exactly the CARE Hamiltonian $H$, and thus we have:
    \begin{equation}
    \begin{aligned}
        H \begin{bmatrix}
        I_n \\
        X
        \end{bmatrix}= \begin{bmatrix}
        I_n\\
        X
        \end{bmatrix}A_{cl}.
    \end{aligned}
    \label{eqn:before_jordan}
    \end{equation}
    Let:
    \begin{equation}
        J_{A_{cl}} \coloneqq X_1^{-1} A_{cl} X_1,
        \label{eqn:a_cl_jordan}
    \end{equation}
     be an $A_{cl}$-Jordan matrix and let $\sigma ( A_{cl} )$ be the corresponding spectrum of $A_{cl}$. Thus by the definition of the Jordan normal form, we have:
     \begin{equation}
        J_{A_{cl}} = \diag \left( ( \lambda )_{\lambda \in \sigma ( A_{cl} ) }\right),
    \end{equation} 
    with $X_1$ composed the generalised eigenvectors of $A_{cl}$. Moreover, from Theorem~\ref{the:real_jordan}, we can arrange for $X_1$ to be real. So let us choose $J_{A_{cl}}$ such that $X_1 \in \mathbb{R}^{n \times n}$.  Multiplying~\eqref{eqn:before_jordan} by $X_1X_1^{-1} = I$ (and thus not changing anything) in the middle of the RHS we have:
    \begin{equation}
    \begin{aligned}
        H \begin{bmatrix}
        I_n \\
        X
        \end{bmatrix}= \begin{bmatrix}
        X_1\\
        XX_1
        \end{bmatrix}X_1^{-1}A_{cl}.
    \end{aligned}
    \end{equation}
    Let us define:
    \begin{equation}
        X_2 \coloneqq X X_1.
        \label{eqn:x_x1_x2}
    \end{equation}
    Notice we have $X_2 \in \mathbb{R}^{n \times n}$. Multiplying both sides by $X_1$ from the right we have:
    \begin{equation}
    \begin{aligned}
        H \begin{bmatrix}
        X_1 \\
        X_2
        \end{bmatrix}= \begin{bmatrix}
        X_1\\
        X_2
        \end{bmatrix}X_1^{-1}A_{cl} X_1.
    \end{aligned}
    \end{equation}
    Notice the rightmost expression is exactly $J_{A_{cl}}$ and thus we have:
    \begin{equation}
    \begin{aligned}
        H \begin{bmatrix}
        X_1 \\
        X_2
        \end{bmatrix}= \begin{bmatrix}
        X_1\\
        X_2
        \end{bmatrix}J_{A_{cl}}.
    \end{aligned}
    \end{equation}
    Let us define $\mathcal{W} \coloneqq \operatorname{Im} \left(\begin{bmatrix}
    X_1 \\
    X_2
    \end{bmatrix}\right)$ and thus we have:
    \begin{itemize}
        \item Since $\begin{bmatrix}
        X_1 & X_2
        \end{bmatrix}^T J_{A_{cl}} \in \mathcal{W}$, then we have that applying $H$ to $\begin{bmatrix}
        X_1 & X_2
        \end{bmatrix}^T \in \mathcal{W}$ keeps it in $\mathcal{W}$, so $\mathcal{W} \subseteq \mathbb{R}^{2n}$ is an $H$-invariant subspace by definition;
        \item $\mathcal{W}$ is the image of the two real square matrices $X_1, X_2$;
        \item From its definition by the Jordan normal form in~\eqref{eqn:a_cl_jordan}, $X_1$ is invertible;
        \item Multiplying by $X_1^{-1}$ from both sides of~\eqref{eqn:x_x1_x2} we have the solution is of the form: $X = X_2X_1^{-1} \in \mathbb{R}^{n \times n}$.
    \end{itemize}
\end{enumerate}
\end{proof}

\begin{remark}
There are a few things to notice regarding Proposition~\ref{prop:care_con}:

\begin{itemize}
    \item By this method we see that in order to obtain solutions to the Riccati Equation, it is necessary to be able to construct bases for those invariant subspaces of $H$, a task which is not always so easily done.
    \item One way of constructing those invariant subspaces is to use eigenvectors and generalized eigenvectors of $H$, as described in an example in page 322 of~\cite{robust_optimal}, but this strategy lacks numerical stability and is mainly avoided for this reason.~\cite{indefinite_linear_algebra}
    \item This construction guarantees the solution $X$ is square, but does not force it to necessarily be symmetric or even real. 
    \item That being said, by the derivation of the CARE itself as we carried out, assigned with real coefficient matrices $A, B, R, Q$, we already assumed the conditions for these to happen (as elaborated in chapter 13 of~\cite{robust_optimal}) are satisfied. 
    \item According to the last point mentioned, the aforementioned construction, in this case and under the specified conditions, will provide a solution $X$ that is indeed real and symmetric, but we are searching for a solution that stabilizes a controlled system, so it must satisfy an additional set of conditions.
\end{itemize}
\end{remark}
 \subsection{Stabilizing Solution to the CARE}
 \label{sect:stabalizing_care}

\begin{definition}
Let us consider the CARE at~\eqref{eqn:riccati} and assume it is consistent. Let $X \in \mathbb{R}^{n \times n}$ be a solution to the CARE. If the controlled dynamics $A_{cl} \equiv A - BR^{-1}B^TX$ is stable, then $X$ is  \textbf{stabilizing}.
\end{definition}
\begin{remark}
Now let us reconsider Proposition~\ref{prop:care_con}:

\begin{itemize} 
    \item  Alas, by definition, the aforementioned general method provides no guarantee that any obtained solution will be stabilizing.
    \item The CARE is guaranteed to be consistent, but still it can be positive-dimensional with infinite solutions, and even if its zero-dimensional with a finite amount of solutions, this amount can still be exponential.
    \item As discussed in section 13 of~\cite{robust_optimal}, there may some solutions that are not necessarily real, not necessarily Hermitian, not necessarily non-negative, and not necessarily stabilizing the resulting controlled system.
\end{itemize}
\end{remark}

Now let us describe the conditions for the existence and uniqueness of the stabilizing solution. To do so, we present the following:

\begin{corollary}
Let us consider the CARE at~\eqref{eqn:riccati}, with its CARE Hamiltonian $H$ and its spectrum $\sigma (H )$. Each solution to the CARE $X \coloneqq X_2 X_1 ^{-1}$ corresponds to a choice of $n$ eigenvalues $\{ \lambda_i \}_{i=1}^n$ of $\sigma (H )$ such that $X_1^{-1}$ exists.
\label{cor:sol_eig_asso}
\end{corollary}
\begin{proof}
This is a direct result from Proposition~\ref{prop:care_con}.
\end{proof}
 \begin{corollary}
 Let us assume there exists a solution to the CARE at~\eqref{eqn:riccati} $X \coloneqq X_2 X_1 ^{-1} \in \mathbb{R}^{n \times n}$. Let us consider the corresponding closed loop dynamics $A_{cl} \equiv A - BR^{-1}B^TX$ and its corresponding spectrum $\sigma ( A_{cl})$. Then the eigenvalues in the spectrum $\sigma ( A_{cl})$ are associated with the column vectors of $X_1$.
 \label{cor:a_cl_eig_asso}
 \end{corollary}
 \begin{proof}
 As in the proof for Proposition~\ref{prop:care_con}, with $X_1$ let us consider the real $A_{cl}$-Jordan matrix $J_{A_{cl}} \coloneqq X_1^{-1} A_{cl} X_1$. Since $J_{A_{cl}}$ is an $A_{cl}$-Jordan matrix, we have that the eigenvalues of $A_{cl}$ are associated with the generalised eigenvectors of its transformation matrix. But the transformation matrix of $J_{A_{cl}}$ is exactly $X_1$ and thus the claim follows.
 \end{proof}
\begin{proposition}
Consider the CARE at~\eqref{eqn:riccati}. Then there exist $C \in \mathbb{R}^{n \times n}$ such that $Q = C^2$. Let us denote $\sqrt{Q} \coloneqq C$.
\end{proposition}
\begin{proof}
Let us consider $Q$ from the derivation of the CARE at~\eqref{eqn:riccati} and its spectrum $\sigma (Q)$. Since $Q \leq 0$, it does not have any negative eigenvalues. Since $Q$ is real symmetric, then by the spectral theorem there exists an orthogonal (and thus invertible) matrix $Q_1 \in \mathbb{R}^{n \times n}$  such that:
\begin{equation}
    Q = Q_1^{-1} \diag \left(( \lambda)_{\lambda \in \sigma (Q)}\right) Q_1.
\end{equation}
Thus, we can apply the function $f(x) = \sqrt{x}$ on $Q$ to obtain:
\begin{equation}
    \sqrt{Q} = Q_1^{-1} \diag \left(( \sqrt{\lambda})_{\lambda \in \sigma (Q)}\right) Q_1.
\end{equation}
\end{proof}
\begin{proposition}
\label{prop:existence_stability_zero_spectrum}
Let us consider the CARE at~\eqref{eqn:riccati} and its Hamiltonian $H$. Then a stabilizing solution $X \in \mathbb{R}^{n \times n}$ to the CARE exists if and only if:
\begin{enumerate}
    \item The pair $(A,B)$ is stabilizable;
    \item The spectrum of $H$ has no eigenvalues with zero real part, as in:
    \begin{equation}
        \forall \lambda \in \sigma (H) \ \colon \ \operatorname{Re}(\lambda) \neq 0 
    \end{equation}
\end{enumerate}
\end{proposition}
\begin{proof}
\begin{enumerate}
    \item $\rightarrow$ :\\
    Let us assume a stabilizing solution to the CARE exists, denoted $X \in \mathbb{R}^{n \times n}$. Then from Proposition~\ref{prop:care_con} we have  $X \coloneqq X_2X_1^{-1}$ with $X_1, X_2 \in \mathbb{R}^{n \times n}$.
    Since $X$ is stabilizing, then by definition, all eigenvalues of the closed loop dynamics $A_{cl} \equiv A - BR^{-1}B^TX$ have negative real parts. Thus we have:
    \begin{equation}
        \forall \lambda \in \sigma ( A_{cl} ) \ \colon \ \operatorname{Re}(\lambda) < 0.
    \end{equation}
    From Corollary~\ref{cor:a_cl_eig_asso} we have that the spectrum $\sigma ( A_{cl} )$ is associated with the column vectors of $X_1$. Let us now consider the CARE Hamiltonian $H$. From Corollary~\ref{cor:sol_eig_asso} we know $X_1$ is associated with a choice of $n$ eigenvalues of $H$. So we have that there exist $n$ eigenvalues of $H$ with strictly negative real parts. From Proposition~\ref{prop:H_symmetry} we have that the rest of the $n$ eigenvalues of $H$ have strictly positive real parts, and thus no eigenvalue of $H$ has a zero real part. Moreover, since $X$ is stabilizing, then the corresponding optimal control law given in~\eqref{eqn:optimal_K_infinite} stabilizes the system, so the pair $(A,B)$ is stabilizable.
    \item Let us assume the pair $(A,B)$ is stabilizable and the spectrum of $H$ has no zero real parts. Since the spectrum of $H$ has no zero real parts then from their symmetry, there are $n$ eigenvalues on the left-hand complex plane and $n$ eigenvalues on the right-hand complex plane. The $n$ ones on the left are the stable ones which are associated with $X_1$ which is then associated with the spectrum $\sigma ( A_{cl} )$. Since the pair $(A,B)$ is stabilizable, the closed loop dynamics $A_{cl} \equiv A - BR^{-1}B^TX$ is stable, for some stabilizing solution to the CARE $X \in \mathbb{R}^{n \times n}$.
\end{enumerate}
\end{proof}

Let us present the following Theorem:

\begin{theorem}[Theorem 4 from~\cite{riccati_review}]
\label{the:unique_stab_detectable}
Let us consider the CARE at~\eqref{eqn:riccati}. The stabilizing solution is the only positive semi-definite solution to the CARE if
and only if the pair $(A, \sqrt{Q})$ is detectable.
\end{theorem}

Now we can conclude the following:

\begin{corollary}[\textbf{CARE Existence and Uniqueness Conditions}, Corollary 13.8 from~\cite{robust_optimal}; Theorem 5 from~\cite{riccati_review}]
\label{CARE_Existence_Uniqueness}
Let us consider the CARE at~\eqref{eqn:riccati}, and assume the pair $(A, B)$ is stabilizable and the pair $(A, \sqrt{Q})$ is detectable.
Then the CARE at~\eqref{eqn:riccati} admits a unique positive semi-definite solution $X \in \mathbb{R}^{n \times n}$ which is also stabilizing.
\end{corollary}
\begin{proof}
Both sides are a direct result of Proposition~\ref{prop:existence_stability_zero_spectrum} and Theorem~\ref{the:unique_stab_detectable}.
\end{proof}

Corollary~\ref{CARE_Existence_Uniqueness} provides sufficient conditions of stabilizability and detectability under which we are promised a unique solution to the CARE, which are the standard prerequisites when dealing with the LQR problem. That being said, there still could be an exponential, if not infinite, if not, god forbid, uncountable amount of solutions to the equation, and the existence and uniqueness does not promise us we will be able to solve for that solution that easily, especially for a large number of state variables.

\subsubsection{Construction of the Unique Stabilizing Solution}

\begin{itemize}

    \item To construct the unique stabilizing solution, a class of algorithms exist known as Schur methods, which performs well on systems where the state vector is of length $n<100$.~\cite{schur}, but do less so for larger systems.
    \item The Schur method is the one used in MATLAB software's \pythoninline{care} and more recent \pythoninline{icare} methods.
    \item In plain words, the methods assumes Corollary~\ref{CARE_Existence_Uniqueness} and with that it finds the solution by performing a variation of the eigenvector decomposition approach, by using a set of self-proclaimed Schur vectors.~\cite{schur}
    \item The main giveaway is that this method performs at a complexity of $O(n^3)$ (Section IV, subsection C of~\cite{schur}).
\end{itemize}
\begin{remark}

The theory described so far corresponds to the continuous equation, but see Section III of~\cite{schur} for a corresponding similar solution for the discrete-time case.
\end{remark}



\chapter{Divide and Conquer for Control Tasks}
\label{chap:d&c}
\epigraph{\textit{Divide and rule}}{\textbf{Gaius Julius Caesar}}

 Sometimes, we are willing to sacrifice offline computing in order to achieve better online results, especially when dealing with complex systems that are composed of multiple modular sub-systems that can all affect each other in circumstances that are not so clear to easily point out from an overhead view. This is the reason why setting the weighs for an LQR controller of the constants for a PID~\cite{PID} controller is mainly a task of tiresome trial and error. Our methodology then, attempts to alleviate this procedure and provide an engineer the ability to deconstruct the system as he sees fit and the allocate the hyper-parameters accordingly. 
 
The method that we are presenting in this paper offers, in a sense, a midway between BP~\ref{subsect:BP} and MPC~\ref{subsect:MPC}. Instead of performing complex computations on-the-fly, as done with MPC, we propose to run lighter computations for solving simple differential games. The solution that we get, based on these differential games, does not depend on priorities nor it requires a specific technique for combining the control signals. 

As an alternative to NSB, our approach allows flexible compositions of the controllers over time. We rely on a competition between so called '\textit{players}' in the differential game we set in motion to automatically provide a balance between the objectives. 

Our approach is a generalization of LQR control for cases where there are multiple objectives to consider. Specifically, we propose to model each of the given objective as player in the game and the solve the multi-objective system using a game theory tool called Open-Loop Nash Equilibrium. In this context, one may ask how our approach compares with the alternative of designing one LQR controller for an objective that is a weighted sum of the objectives of the players that we consider. The answer is that our approach combines the objectives more dynamically while manual priority assignment often requires deep understanding of the inner details of the system. 

Experience shows that, in nature and in man-made systems, solving conflicts by allowing internal competition gives an agile solution that effectively balances competing objectives, especially when the number of objectives is large. This is reminiscent of a method used in economics where complex organizations avoid inefficiencies by dividing themselves to `\textit{profit units}'.

From a theoretical perspective, our method relies on the theory of differential games, that can either be finite or infinite. The approach to differential games in this work relies extensively on the book '\textit{Dynamic
Noncooperative
Game Theory}'~\cite{Basar}. Chapter 3 of Part I introduces Noncooperative Finite N-Person Nonzero-Sum games, with sections 3.4, 3.5 covering Nash Equilibria. Chapter 5 of Part II introduces the General Formulation of Infinite Dynamic Games, both for discrete and continuous time, and Chapter 6 deals with Nash and Saddle-Point Equilibria of Infinite Dynamic Games.

\section{Continuous LTI System Decomposition}

In this section, we will present a formal description of the approach we propose to implement. Let $\mathcal{S}$ be a dynamical continuous-time LTI system adhering the model in equation~\eqref{eqn:basic_sys} and let $\mathbf{x}(t)$ be the state vector of $\mathcal{S}$, defined for all $t \in \mathcal{I}$.\footnote{$\mathcal{I}$, the system interval, along with any following non-explicitly stated notations, are given in Appendix~\ref{app:sys_modeling}.}. We will propose a specific expression for the assigned input to the system, that captures the decomposition of the input signals into designated virtual inputs, corresponding to some pre-defined objectives.

\begin{definition}
Given a continuous LTI system $\mathcal{S}$, and $N \in \mathbb{N}$, we define its equivalent \textbf{continuous LTI virtual decomposed system}, denoted $\mathcal{S}_N$ as the LTI system that adheres the following specific model:

\begin{equation}
    \dot{\mathbf{x}}(t) = A\mathbf{x}(t) + \sum _{i=1}^N B_i \mathbf{v_i}(t)
    \label{eqn:sys_with_B1___Bn}
\end{equation}
\end{definition}
where:
\begin{definition}
    Let $\big( \mathbf{v_i}(t)  \big)_{i=1}^N$ be the \textbf{virtual input vectors}, with each being a vector function of length $m_i \in \mathbb{N}$ and $\mathbf{v_i} \colon \mathcal{I} \rightarrow \mathcal{V}_i$, where the vector subspace $\mathcal{V}_i \subseteq \mathbb{R}^{m_i}$ is the set of all possible admissable virtual inputs for the $i$'th objective. Let us denote $\mathcal{V} = \bigtimes_{i=1}^N\mathcal{V}_i$ so that $\mathcal{V} \subseteq \mathbb{R}^{\sum_{i=1}^Nm_i}$ is a subset of all the sorted possible admissable virtual input vectors of total length $\sum_{i=1}^Nm_i$;
\end{definition}
    \begin{definition}
    \label{def:virtual_coef}
    Let $\big(B_i\big)_{i=1}^N$ be the \textbf{virtual input coefficients} matrices where each $B_i \in \mathbb{R}^{n \times m_i}$.
\end{definition}

\begin{remark}
$\big( \mathbf{v_i}(t)  \big)_{i=1}^N$ and $\big(B_i\big)_{i=1}^N$ will be defined by closed-form expressions later on, given the matrix $B$ of $\mathcal{S}$ and a series of input design parameters we assume the system engineer submits as input, as we will explain. The point to be made here is that either way, since the original system model of $\mathcal{S}$ is as given at~\eqref{eqn:basic_sys}, then $\mathcal{S}_N$ must remain true to that same model, i.e. when designing $\mathcal{S}_N$, $\big( \mathbf{v_i}(t) \big)_{i=1}^N$ and $\big(B_i\big)_{i=1}^N$ must be such that:

\begin{equation}
    B \mathbf{u}(t) = \sum _{i=1}^N B_i \mathbf{v_i}(t),
\end{equation}

in any constellation the designer deems plausible. Notice that in the general case, this allows multiple ways of performing the aforementioned decomposition, as in given $\big( \mathbf{v_i}(t) \big)_{i=1}^N$ and $\big(B_i\big)_{i=1}^N$, the input $\mathbf{u}(t)$ is not determined uniquely.
\end{remark}

\begin{remark}
Once designed, the system $\mathcal{S}_N$ can be visualized using the following block diagram:

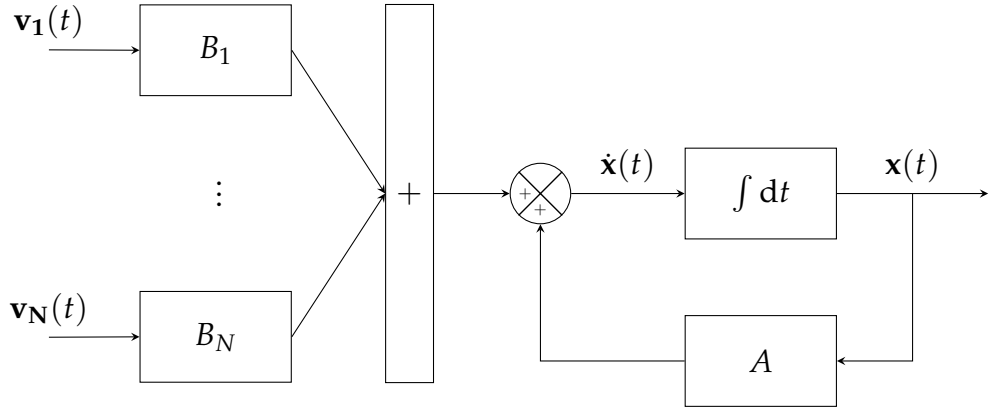
\begin{figure}[H]
\centering
\begin{tikzpicture}
\node[draw,
    circle,
    minimum size=0.8cm
] (sum) at (0,0){};
 
\draw (sum.north east) -- (sum.south west)
    (sum.north west) -- (sum.south east);
 
\draw (sum.north east) -- (sum.south west)
(sum.north west) -- (sum.south east);
 
\node[left=-1pt] at (sum.center){\tiny $+$};
\node[below] at (sum.center){\tiny $+$};
 
\node [draw,
    minimum width=2cm,
    minimum height=1.2cm,
    right=1.5cm of sum
]  (integrator) {$\int \mathrm{d}t$};

\node [draw,
    minimum width=0.2cm, 
    minimum height=5cm,
    left=1cm of sum
] (summer) {$+$};

\node [draw,
    minimum width=2cm, 
    minimum height=1.2cm,
    above left =1cm and 3cm of sum
] (B_1) {$B_1$};

\node[text width=3cm] at (-2.8,0.1) 
    {$\vdots$};
    
\node [draw,
    minimum width=2cm, 
    minimum height=1.2cm,
    below left =1cm and 3cm of sum
] (B_N) {$B_N$};
 
\node [draw,
    minimum width=2cm, 
    minimum height=1.2cm, 
    below=1cm of integrator
]  (A) {$A$};
 
\draw[-stealth] (sum.east) -- (integrator.west)
    node[midway,above]{$\dot{\mathbf{x}}(t)$};
 
\draw[-stealth] (integrator.east) -- ++ (2,0)
    node[midway,above](x){$\mathbf{x}(t)$};
 
\draw[-stealth] (x.south) |- (A.east);
 
\draw[-stealth] (A.west) -| (sum.south) 
    node[near end,left]{};
 
\draw[-stealth] ++(-6.5,1.9)
    node[above]{$\mathbf{v_1}(t)$} -- (B_1.west);

 \draw[-stealth] (B_1.east) -- (summer.west)
     node[midway,above]{};

\draw[-stealth] ++(-6.5,-1.9)
    node[above]{$\mathbf{v_N}(t)$} -- (B_N.west);

\draw[-stealth] (B_N.east) -- (summer.west)
    node[midway,above]{};
    
\draw[-stealth] (summer.east) -- (sum.west)
    node[midway,above]{};
 
\end{tikzpicture}
\caption{Block diagram of the decomposed system $\mathcal{S}_N$ described in equation~\eqref{eqn:sys_with_B1___Bn}.} \label{fig:S_N_diagram}
\end{figure}
\end{remark}

\begin{remark}
Let us note that the aforementioned transformation is only executed \textit{virtually}, i.e., not by means of a physical nor mathematical transformation. $\mathcal{S}_N$ differs with $\mathcal{S}$ only by means of its hypothetical control inputs and their corresponding coefficients. The exact division tactic will be elaborated later on. 
\end{remark}

\section{Input Division Strategy}
\label{sect:input_division_tactic}

We now present a formal method to apply our described transformation.

\begin{definition}
Let $\mathcal{S}_N$ be a continuous LTI decomposed system with virtual inputs $\big( \mathbf{v_i}(t) \big)_{i=1}^N$ and corresponding lengths $(m_i)_{i=1}^N$. Let us define the \textbf{augmented virtual input vector} as a function $\mathbf{v} \colon \mathcal{I} \rightarrow \mathcal{V}$ defined by:

\begin{equation}
    \mathbf{v}(t) \coloneqq 
    \begin{bmatrix} \mathbf{v_1}(t)\\
    \vdots \\
    \mathbf{v_N}(t)  
    \end{bmatrix}.
    \label{eqn:v_t_def}
\end{equation}
\end{definition}
Let $\mathcal{S}$ be a continuous LTI system with input vector $\mathbf{u}(t) \in \mathcal{U} \subseteq \mathbb{R}^m$.
Let $\mathcal{S}_N$ be a continuous LTI decomposed system of $\mathcal{S}$, with virtual inputs $\big( \mathbf{v_i}(t) \big)_{i=1}^N$ and corresponding lengths $(m_i)_{i=1}^N$. We define the following:
\begin{itemize}
    \item \begin{definition}
    Let us define the \textbf{input division matrices} $(M_i)_{i=1}^N$, where $M_i \in \mathbb{R}^{m_i \times m}$, which will be used to define the division of the input vector into virtual inputs, i.e.:
\begin{equation}
    \forall 1 \leq i \leq N \ \colon \ \mathbf{v_i}(t) = M_i \mathbf{u}(t)
    \label{eqn:u_to_v_i}
\end{equation}
\end{definition}
\item \begin{definition}
 Plugging~\eqref{eqn:u_to_v_i} to~\eqref{eqn:v_t_def} we have:
\begin{equation}
    \mathbf{v}(t) =
    \begin{bmatrix} \mathbf{v_1}(t)\\
    \vdots \\
    \mathbf{v_N}(t)  
    \end{bmatrix} = \begin{bmatrix} M_1 \mathbf{u}(t)\\
    \vdots \\
    M_N \mathbf{u}(t)  
    \end{bmatrix} = \begin{bmatrix} M_1 \\
    \vdots \\
    M_N  
    \end{bmatrix}\mathbf{u}(t).
\end{equation}
Denoting:
\begin{equation}
    M \coloneqq \begin{bmatrix}
          M_1 \\
          \vdots \\
          M_N 
            \end{bmatrix}
\end{equation}
such that $M \in \mathbb{R}^{\sum_{i=1}^N m_i \times m}$, we have:
\begin{equation}
    \mathbf{v}(t) = M \mathbf{u}(t).
    \label{eqn:u_to_v}
\end{equation}

In this context, $M$ is referred to as the \textbf{augmented division matrix} of $\mathcal{S}_N$ and is set to represent the linear transformation that decomposes $\mathcal{S}$, which is a system of one input vector, into $\mathcal{S}_N$, a system of $N$ distinct input vectors.
\end{definition}
\begin{remark}
The relationship at~\ref{eqn:u_to_v} is one that does not strictly require the matrix $M$ to be invertible or even square, i.e., generally we assume $\sum_{i=1}^Nm_i \neq m$. The relationship between $\mathbf{u}(t)$ and $\mathbf{v}(t)$ provides a general prescription to transform the original input of the original system $\mathcal{S}$ into the virtual inputs to the decomposed system $\mathcal{S}_N$, and is up to the control engineer to define per case.
\label{remark:u_to_v}
\end{remark}
\end{itemize}

\begin{definition}
Let $\mathcal{S}_N$ be a decomposed system of $\mathcal{S}$ with virtual inputs $\big( \mathbf{v_i}(t) \big)_{i=1}^N$ and corresponding lengths $(m_i)_{i=1}^N$ and let $M \in \mathbb{R}^{\sum_{i=1}^N m_i \times m}$ be the augmented division matrix of $\mathcal{S}_N$. If
\begin{itemize}
    \item $\sum_{i=1}^Nm_i \equiv m$;
    \item $M \in \mathbb{R}^{m \times m}$ is invertible;
\end{itemize}
 then we refer to the $\mathcal{S}_N$ as \textbf{inversely designable}.
\end{definition}

\begin{proposition}
Let $\mathcal{S}$ be a dynamical system. Let $\mathcal{S}_N$ be a decomposed system of $\mathcal{S}$. If $\mathcal{S}_N$ is inversely designable then when transforming  $\mathcal{S}$ into $\mathcal{S}_N$, $\mathbf{u}(t)$ is uniquely determined.
\end{proposition}
\begin{proof}
Since $\mathcal{S}_N$ is inversely designable, we have that $M^{-1} \in \mathbb{R}^{m \times m}$ exists. Left-multiplying both sides of~\eqref{eqn:u_to_v} by $M^{-1}$, we have in that case:

\begin{equation}
    \mathbf{u}(t) = M^{-1}\mathbf{v}(t).
\end{equation}
 Knowing the lengths $(m_i)_{i=1}^N$, there exists a unique series of matrices $(\tilde{M}_i)_{i=1}^N$ such that each $\tilde{M}_i$ satisfies $\tilde{M}_i \in \mathbb{R}^{m \times m_i}$ and is a decomposition of $M^{-1}$ of the form:

\begin{equation}
    M^{-1} = \begin{bmatrix}
    \tilde{M}_1 & \cdots & \tilde{M}_N.
    \end{bmatrix}
\end{equation}

With that, let us consider the following expressions for the virtual input coefficients $\big( B_i \big)_{i=1}^N$: 

\begin{equation}
    \forall 1 \leq i \leq N \ \colon\  B_i \coloneqq B \tilde{M}_i.
\end{equation}
Notice we have $B_i \in \mathbb{R}^{n \times m_i}$. Plugging to the model at~\eqref{eqn:basic_sys} we have:
\begin{equation}
\begin{aligned}
    \dot{\mathbf{x}}(t) = & A \mathbf{x}(t) + B \mathbf{u}(t) \Big|_{\mathbf{u}(t) \equiv M^{-1}\mathbf{v}(t)}\\
    = & A \mathbf{x}(t) + B M^{-1}\mathbf{v}(t) \\
    = & A \mathbf{x}(t) + B \begin{bmatrix}
    \tilde{M}_1 & \cdots & \tilde{M}_N
    \end{bmatrix} \mathbf{v}(t)\\
    = & A \mathbf{x}(t) +  \begin{bmatrix}
    B\tilde{M}_1 & \cdots & B\tilde{M}_N
    \end{bmatrix} \mathbf{v}(t)\\
    = & A \mathbf{x}(t) + \begin{bmatrix}
    B_1 & \cdots & B_N
    \end{bmatrix}\begin{bmatrix} \mathbf{v_1}(t)\\
    \vdots \\
    \mathbf{v_N}(t)  
    \end{bmatrix}\\
    = & A \mathbf{x}(t) + \sum_{i=1}^N B_i \mathbf{v_i}(t)
\end{aligned}
\end{equation}
\end{proof}

\section{Continuous Multi-Objective Optimal Control}
Now that we speak in terms of the decomposed system, we can carry on to the next step, which is to define several criteria that quantify the performance of the system. As described at the LQR optimal control section at~\ref{section:optimal_control}, we do this by defining a set of appropriate cost functions. This section is based on the reasoning described for the regular LQR and extends the methods described for multiplayer games.

\begin{definition}
Let us assume there exists a series of pre-determined virtual objectives $(O_i)_{i=1}^N$ 
we would like to enact upon the system $\mathcal{S}_N$. Each \textbf{virtual objective} is defined by a triplet of the form:
\begin{equation}
    O_i \coloneqq \Big(Q_i, (R_{ij})_{j=1}^N, B_i\Big)
\end{equation}
where $Q_i \in \mathbb{R}^{n \times n}, Q_i \geq 0$, i.e., $Q_i$ is positive-semi-definite (PSD), $R_{ij} \in \mathbb{R}^{m_j \times m_j}$ and $R_{ij} > 0$, i.e., $R_{ij}$ is positive-definite (PD) and the matrix $B_i$ is the corresponding virtual input coefficient matrix. $Q_i$ and $(R_{ij})_{j=1}^N$ are weight matrices that dictate how each state and input variable precisely affects the corresponding cost function. While this relation can be defined in many ways, we choose the standard way of specifying weights in a quadratic function with respect to the variables.
\end{definition}

\begin{definition}

We now propose to define a \textbf{virtual cost function} $J_i$ for each control objective $O_i$, similar to what we defined in Equation~\eqref{eqn:general_cost}. We would like the input of our cost functions to be a series $(\mathbf{v_j}(t))_{j=1}^N$ of virtual inputs, along with the the state vector $\mathbf{x}(t)$ - as both of these elements affect the value the function will return. So we define two sets of functions $(r_i)_{i=1}^N$ and $(r_{f_i})_{i=1}^N$, and the set of cost functions $(J_i)_{i=1}^N$, where each $J_i \colon \mathcal{X} \times \mathcal{V} \times \mathcal{I} \rightarrow \mathbb{R}^+$ assumes the following form:

\begin{equation}
\begin{aligned}
        J_i\Big(\mathbf{x}(t), (\mathbf{v_j}(t))_{j=1}^N, t\Big) \coloneqq 
       \int_{t}^{T_f} r_i\Big(\mathbf{x}(\tau), (\mathbf{v_j}(\tau))_{j=1}^N\Big)\mathrm{d}\tau + r_{f_i}\Big(\mathbf{x}(T_f)\Big)
\end{aligned}
\label{objective_function}
\end{equation}
\end{definition}
Once again, we will denote $J_i(t)$ for conciseness.
\subsection{Optimality by Open-Loop Nash Equilibrium}
\begin{definition}
For a given decomposed system $\mathcal{S}_N$ with a state vector $\mathbf{x}(t)$, a series of virtual policies $ (\mathbf{v^*_j}(t))_{j=1}^N$ is said to constitute an \textbf{Open-Loop Nash Equilibrium}\footnote{Note that obtaining a Nash Equilibrium does not necessarily minimize the sum of the costs but represent a certain balance between the described objectives.} if for all $1 \leq i \leq N$, we have:

\begin{equation}
    J_i\bigg(\mathbf{x}(t), (\mathbf{v^*_j}(t))_{j=1}^N, t\bigg) \leq J_i\bigg(\mathbf{x}(t), (\mathbf{v^*_j}(t))_{\substack{j=1 \\ j \neq i}}^N, \mathbf{v_i}(t), t\bigg)
    \label{nash_equilibrium}
\end{equation}

for any other admissable virtual input $\mathbf{v_i}(t) \in \mathcal{V}_i$ and for all $t \in \mathcal{I}$. In words, a series of virtual inputs $ (\mathbf{v^*_j}(t))_{j=1}^N$ constitutes an Open-Loop Nash equilibrium if it is not possible to decrease any cost function $J_i(t)$ only by changing its corresponding virtual input $\mathbf{v_i^*}(t)$ to some other input $\mathbf{v_i}(t)$. Equality in~\eqref{nash_equilibrium} will be in the case where $\mathbf{v_i}(t) \equiv \mathbf{v_i^*}(t)$.
\end{definition}
\begin{remark}
Relating to what was described earlier, our choice of using the concept of Open-Loop Nash Equilibrium is because we propose to choose the decomposed system's virtual \textbf{inputs} by solving a game between objectives that may have common and competing interests. It so happens that the optimal inputs will be closed-loop control policies.
\end{remark}

Notice, similar to~\eqref{eqn:optimal_control_problem_formulation} this definition induces the following optimal problem:

\begin{definition}
Given an decomposed LTI system $\mathcal{S}_N$ adhering the model at~\eqref{eqn:sys_with_B1___Bn} and with GLQR cost functions $\left(J_i(t)\right)_{i=1}^N$, we define the \textbf{open-loop Nash Equilibrium optimal control problem}, formally stated as:
\begin{equation}
    \begin{aligned}
    \min_{\substack{\mathbf{v_i}(t) \in \mathcal{V}_i}} \quad & J_i\Big(\mathbf{x}(t), (\mathbf{v_j}(t))_{j=1}^N, t_0\Big) \\
         = \min_{\substack{\mathbf{v_i}(t) \in \mathcal{V}_i}} \quad & \int_{\tau \in \mathcal{I}} r_i\Big(\mathbf{x}(\tau), (\mathbf{v_j}(\tau))_{j=1}^N\Big)\mathrm{d}\tau + r_{f_i}\Big(\mathbf{x}(T_f)\Big) \\
    = \min_{\substack{\mathbf{v_i}(t) \in \mathcal{V}_i}} \quad & \bigintsss_{\tau \in \mathcal{I}} \Bigg[ \mathbf{x}^T(\tau)Q_i\mathbf{x}(\tau) + \sum_{j=1}^N \mathbf{v_j}^T(\tau)R_{ij}\mathbf{v_j}(\tau) \Bigg]\mathrm{d}\tau + r_f\Big(\mathbf{x}(T_f)\Big)\\
         \textrm{subject to} \quad & \mathbf{x}(t_0) = \mathbf{x_0};\\
         \quad & \dot{\mathbf{x}}(t) = A\mathbf{x}(t) + \sum _{j=1}^N B_j \mathbf{v_j}(t); \quad & \forall t \in \mathcal{I};\\
         & \forall 1 \leq i \leq N.
    \end{aligned}
    \label{eqn:optimal_nash_control_problem_formulation}
\end{equation}
\end{definition}

\begin{definition}
Given a decomposed LTI system $\mathcal{S}_N$, let us define a \textbf{game optimal state trajectory} of $\mathcal{S}_N$, as a state trajectory $\mathbf{x^*}(t) \in \mathcal{X}$ which is the solution to the model of $\mathcal{S}_N$ at~\eqref{eqn:sys_with_B1___Bn}, assigned with the Nash Equilibrium policies $ (\mathbf{v^*_j}(t))_{j=1}^N$, i.e., it satisfies:
\begin{equation}
    \dot{\mathbf{x}}^*(t) = A\mathbf{x^*}(t) + \sum _{i=1}^N B_i \mathbf{v^*_i}(t),
\end{equation}
and by that extension, using the solution for the state solution at~\eqref{eqn:LTI_solution}:
\begin{equation}
        \mathbf {x^*}(t) \coloneqq e^{A(t-t_0)}{\mathbf  {x_0}}+\int _{{t_{0}}}^{t}e^{A(t-\tau)}\sum _{i=1}^N B_i \mathbf{v^*_i}(\tau)\mathrm{d}\tau.
\end{equation}
\end{definition}
\begin{definition}
Let us use the gradient operator on each cost function $\nabla J_i(t)$ to indicate \textbf{derivative with respect to time} $t$ \textbf{or state vector} $\mathbf{x}(t)$.
\end{definition}
\begin{definition}
With $\mathbf{x^*}(t)$, and $\nabla J_i(t)$, let us denote the \textbf{optimal values in regards to state trajectory} for every cost function $J_i(t), \ 1 \leq i \leq N$, and its gradient as:
\begin{equation}
\begin{aligned}
     J_i^*(t) \coloneqq & J_i\big( \mathbf{x}(t),\mathbf{u}(t), t\big) \ \big|_{\mathbf{x}(t) \equiv \mathbf{x^*}(t)};\\
    \nabla J_i^*(t) \coloneqq & \nabla J_i \ \big|_{\mathbf{x}(t) \equiv \mathbf{x^*}(t)} \ .
\end{aligned}
\end{equation}
\end{definition}
  \begin{definition}
  Similar to equation \eqref{eqn:partial_def}, let us define the\textbf{ partial derivatives with respect to state vector and time} of each cost function as:

\begin{equation}
\begin{aligned}
    \partial_t J_i(t) \coloneqq & \frac{\partial J_i\big( \mathbf{x}(t), \mathbf{u}(t), t \big)}{\partial t};\\
     \mathbf{\partial_x J_i}(t) \coloneqq & \frac{\partial J_i\big( \mathbf{x}(t), \mathbf{u}(t), t \big)}{\partial \mathbf{x}(t)} \\
     = & \left[ \frac{\partial J_i\big( \mathbf{x}(t), \mathbf{u}(t), t \big)}{\partial x_1(t)}, \frac{\partial J_i\big( \mathbf{x}(t), \mathbf{u}(t), t \big)}{\partial x_2(t)}, \dots, \frac{\partial J_i\big( \mathbf{x}(t), \mathbf{u}(t), t \big)}{\partial x_n(t)} \right]^T.
\end{aligned}
\end{equation}
  \end{definition}
\begin{definition}
With that let us define the \textbf{optimal partial derivatives} as:
\begin{equation}
\begin{aligned}
    \mathbf{\partial_x J_i}^*(t) \coloneqq \mathbf{\partial_x J_i} \ \big|_{\mathbf{x}(t) \equiv \mathbf{x^*}(t)};\\
    \partial_t J_i^*(t) \coloneqq \partial_t J_i \ \big|_{\mathbf{x}(t) \equiv \mathbf{x^*}(t)}.
    \end{aligned}
\end{equation}
\end{definition}
\begin{definition}
Given an decomposed LTI system $\mathcal{S}_N$ adhering the model at~\eqref{eqn:sys_with_B1___Bn}, we again deal with quadratic cost function, i.e., the functions $(r_i)_{i=1}^N$ take on the following forms:

\begin{equation}
r_i\Big(\mathbf{x}(t), (\mathbf{v_j}(t))_{j=1}^N\Big) \equiv \mathbf{x}^T(t)Q_i\mathbf{x}(t) + \sum_{j=1}^N \mathbf{v_j}^T(t)R_{ij}\mathbf{v_j}(t).
\end{equation}

Correspondingly, the control model is called a \textbf{Game Linear Quadratic Regulator} (GLQR), with:

\begin{itemize}
    \item $Q_i \in \mathbb{R}^{n \times n}$ is a positive semi-definite matrix;
    \item $\forall 1 \leq j \leq N \ \colon \ R_{ij} \in \mathbb{R}^{m_j \times m_j}$ is a positive definite matrix.
\end{itemize}
\end{definition}
\begin{remark}
Notice we assume a one-to-one correspondence between $B_i$, $Q_i$, and $(R_{ij})_{j=1}^N$. This relation reflects the aforementioned assignment of the virtual actuators to their corresponding virtual objectives. Our experience, as demonstrated later in this work, shows that it is often natural to choose a value for $B_i$  for each objective cost function $J_i(t)$ such that the virtual input $\mathbf{v_i}(t)$ affects mostly $J_i(t)$. We do not insist that other cost functions $J_k(t)$ (where $k \neq i$) are not affected by $\mathbf{v_i}(t)$, only that such effects are small relative to its effect on $J_i(t)$.
\end{remark}

\begin{proposition}
Given a decomposed LTI system $\mathcal{S}_N$ with GLQR cost functions $(J_i)_{i=1}^N$, let us consider the Open-Loop Nash Equilibrium optimal control problem given in~\eqref{eqn:optimal_nash_control_problem_formulation}. If a solution to the problem exists, it is obtained by the series $(\mathbf{v^*_i}(t))_{i=1}^N$ where:
\begin{equation}
    \mathbf{v^*_i}(t) = - \frac{1}{2} R_{ii}^{-1}B_i^T \ \mathbf{\partial_x J_i}^{*}(t),
\end{equation}
and $\left(\mathbf{\partial_x J^*_i}(t)\right)_{i=1}^N$ are the corresponding partial derivatives of the GLQR cost functions with respect to the state, assigned with the induced game optimal state $\mathbf{x^*}(t)$.
\end{proposition}
\begin{proof}

\begin{definition}
In practical terms, similar to what was defined in equation~\eqref{eqn:hamiltonian}, let us first define a \textbf{Virtual Control Hamiltonian} for each virtual input $\mathbf{v_i}(t), \ 1 \leq i \leq N$ by differentiating both sides of~\eqref{objective_function} with respect to time using Leibniz's formula for differentiation under the integral sign and then moving sides to give:

\begin{equation}
\begin{aligned}
     \mathcal{H}_i & \Big( \mathbf{x}(t), \nabla J_i(t), (\mathbf{v_j}(t))_{j=1}^N \Big)  \coloneqq \partial_t J_i(t) + r_i\Big(\mathbf{x}(t), (\mathbf{v_j}(t))_{j=1}^N\Big) \\
      = & \mathbf{\partial_x J_i}^T(t) \ \dot{\mathbf{x}}(t)
       + \mathbf{x}(t)^TQ_i\mathbf{x}(t) + \sum_{j=1}^N \mathbf{v_j}^T(t)R_{ij}\mathbf{v_j}(t)\\
       = & \mathbf{\partial_x J_i}^T(t) \ \left(A\mathbf{x}(t) + \sum _{j=1}^N B_j \mathbf{v_j}(t)\right)
       + \mathbf{x}(t)^TQ_i\mathbf{x}(t) + \sum_{j=1}^N \mathbf{v_j}^T(t)R_{ij}\mathbf{v_j}(t).
       \label{eqn:hamiltonian_i}
\end{aligned}
\end{equation}

Notice this holds for both the finite and infinite horizon case. 
\end{definition}

\begin{definition}
Similar to what we described in equation~\eqref{eqn:bellman}, we describe the following set of \textbf{Game Bellman Equations} as:

\begin{equation}
    \mathcal{H}_i \Big( \mathbf{x}(t), \nabla J_i(t), (\mathbf{v_j}(t))_{j=1}^N \Big)  = 0
    \label{eqn:bellman_N}
\end{equation}
\end{definition}
Since we assume there exists a solution to the problem, we can employ PMP to the $i$'th input individually to describe the following implementation of PMP:
\begin{definition}
Applying the stationarity condition at equation~\eqref{eqn:stationarity} by differentiating each Hamiltonian $\mathcal{H}_i$ by the corresponding $i$'th virtual input $\mathbf{v_i}(t)$, equating to zero and solving for $\mathbf{v_i}(t)$, we define the \textbf{Game PMP (GPMP)}. In that we obtain the set of optimal virtual inputs $\left(\mathbf{v^*_i}(t)\right)_{i=1}^N$. Thus the stationarity condition for multiplayer games corresponds to the following set of $N$ coupled equations\footnote{Notice this definition is just for the sake of simplicity, and does not constitute a separate theorem or even an extension of Theorem~\ref{the:PMP}, as it is just applying the same principle but to every individual virtual cost function and induced virtual Control Hamiltonian.}:
 
 \begin{equation}
     \frac{\partial}{\partial \mathbf{v_i}(t)} \mathcal{H}_i \Big( \mathbf{x^*}(t), \nabla J^*_i(t), (\mathbf{v_j}(t))_{j=1}^N \Big)\Big|_{\mathbf{v_i}(t) \equiv \mathbf{v_i^*}(t)} = 0.
     \label{eqn:stationarity_N}
\end{equation}
\end{definition}

Applying GPMP by differentiating the $i$'th Hamiltonian at~\eqref{eqn:hamiltonian_i} with respect to the $i$'th virtual input $\mathbf{v_i}(t)$, similar to what was done at~\eqref{eqn:diff_H_1}, we have:

\begin{equation}
\begin{aligned}
     \frac{\partial}{\partial \mathbf{v_i}(t)} &\mathcal{H}_i \Big( \mathbf{x^*}(t), \nabla J^*_i(t), (\mathbf{v_j}(t))_{j=1}^N \Big) =\\
     \frac{\partial}{\partial \mathbf{v_i}(t)} \Bigg[& \mathbf{\partial_x J_i}^{*^T}(t) \ \left(A\mathbf{x^*}(t) + \sum _{j=1}^N B_j \mathbf{v_j}(t)\right)
       \\
       &+ \mathbf{x^*}(t)^TQ_i\mathbf{x^*}(t) + \sum_{j=1}^N \mathbf{v_j}^T(t)R_{ij}\mathbf{v_j}(t) \Bigg].
\end{aligned}
\end{equation}

Similar to what we had in equation~\eqref{eqn:null_cross_derivative}, we have that the cross-derivative $\frac{\partial }{\partial \mathbf{v_i}(t)} \mathbf{\partial_x J_i(t)}$ cancels out, and so:

\begin{equation}
\begin{aligned}
    & \frac{\partial}{\partial \mathbf{v_i}(t)} \mathcal{H}_i \Big( \mathbf{x^*}(t), \nabla J^*_i(t), (\mathbf{v_j}(t))_{j=1}^N \Big)\Big|_{\mathbf{v_i}(t) \equiv \mathbf{v_i^*}(t)} =\\
    & \left[\left(\mathbf{\partial_x J_i}^{*^T}(t) B_i \right)^T
        + 2R_{ii}\mathbf{v_i}(t)\right]\Bigg|_{\mathbf{v_i}(t) \equiv \mathbf{v_i^*}(t)} =\\ & \left[B_i^T \partial_x J^*_i(t) + 2R_{ii}\mathbf{v_i}(t)\right]\Bigg|_{\mathbf{v_i}(t) \equiv \mathbf{v_i^*}(t)} = 0.
\end{aligned}
\end{equation}

Solving this for $\mathbf{v_i}(t)$ yields the following optimal value $\mathbf{v^*_i}(t)$:

\begin{equation}
    \mathbf{v^*_i}(t) = - \frac{1}{2} R_{ii}^{-1}B_i^T \ \mathbf{\partial_x J_i}^{*}(t).
    \label{eqn:optimal_v_i}
\end{equation}
\end{proof}

\begin{definition}

Plugging the expression for the optimal inputs at equation~\eqref{eqn:optimal_v_i} to the Game Bellman equations at~\eqref{eqn:bellman_N}, along with the optimal state and cost functions derivative, we have the following set of \textbf{Game Hamilton-Jacobi-Bellman} (GHJB) equations:

\begin{equation}
    \begin{aligned}
     \mathbf{\partial_x J_i}^{*^T}(t) & \left[A\mathbf{x^*}(t) - \frac{1}{2} \sum _{j=1}^N B_j  R_{jj}^{-1}B_j^T \ \mathbf{\partial_x J_j}^{*}(t)\right]
       + \mathbf{x^*}(t)^TQ_i\mathbf{x^*}(t) \\
       & + \frac{1}{4}\sum_{j=1}^N \mathbf{\partial_x J_j}^{*^T}(t)B_jR_{jj}^{-1} R_{ij} R_{jj}^{-1}B_j^T \ \mathbf{\partial_x J_j}^{*}(t) = 0
       \label{eqn:HJB_i}
    \end{aligned}
\end{equation}

Notice that all the derivation so far  holds for both the finite and infinite horizon case.
\end{definition}

Now let us consider each case separately.\footnote{It is well established at Chapter 3 of Part I and Chapter 6 of Part II at~\cite{Basar} that the optimal cost to the open-loop Nash Equilibrium optimal control problem with quadratic cost functions is also quadratic, as was in the LQR case. We will thus use this result as a given.
}

\subsection{Continuous Infinite Horizon LTI Game}
Let us consider an infinite time interval ($T_f \rightarrow \infty$). To this end, let us propose expressions for the aforementioned functions that will result in the following infinite-horizon \textbf{Game Linear Quadratic Regulators} (GLQR) of the form:

\begin{equation}\label{inf_objective_function}
\begin{aligned}
        J_i\Big(\mathbf{x}(t), (\mathbf{v_j}(t))_{j=1}^N, t\Big) = 
       \int_{t}^{\infty} \Big[\mathbf{x}(\tau)^TQ_i\mathbf{x}(\tau) + \sum_{j=1}^N \mathbf{v_j}^T(\tau)R_{ij}\mathbf{v_j}(\tau)\Big]\mathrm{d}\tau
\end{aligned}
\end{equation}
\begin{proposition}
Given a decomposed LTI system $\mathcal{S}_N$ with GLQR cost functions $(J_i)_{i=1}^N$, let us consider the Open-Loop Nash Equilibrium optimal control problem given in~\eqref{eqn:optimal_nash_control_problem_formulation}, in the case where $T_f \rightarrow \infty$. Let us assume there exist a series of open-loop optimal solutions $\left(\mathbf{v_i^*}(t)\right)_{i=1}^N$ to the problem. Then $\left(\mathbf{v_i^*}(t)\right)_{i=1}^N$ are all in fact closed-loop feedback controllers and are all linear with regards to the optimal state, with corresponding constant coefficients, i.e. they are of the form:
\begin{equation}
    \mathbf{v_i^*}(t) = - K_i^*\mathbf{x^*}(t),
\end{equation}
for some constants $\left(K_i^*\right)_{i=1}^N$ where $K_i^* \in \mathbb{R}^{n \times m}$.
\end{proposition}
\begin{proof}

 Similar to what was proposed at equation~\eqref{eqn:infinite_horizon_optimal_cost}, let us employ the same reasoning to suggest the following quadratic values with regards to the optimal state, for the optimal costs values $(J_i^*)_{i=1}^N$:

\begin{equation}
    J_i^*(t) \equiv \mathbf{x^*}^T(t) P_i \mathbf{x^*}(t),
    \label{eqn:infinite_horizon_optimal_cost_N}
\end{equation}

given a set of positive-definite matrices $(P_i)_{i=1}^N$ where $\in \mathbb{R}^{n \times n}$ that we will determine later on using the Game Bellman equations. Using vector calculus~\cite{matrix_cookbook} on~\eqref{eqn:infinite_horizon_optimal_cost_N} and using the fact that the matrices $(P_i)_{i=1}^N$ are positive-definite and thus symmetric, we have:
\begin{equation}
    \mathbf{\partial_x J_i}^{*}(t) = (P_i+P_i^T)\mathbf{x^*}(t) = 2P_i \mathbf{x^*}(t).
    \label{eqn:inf_horizon_optimal_cost_derivative}
\end{equation}

Plugging this to equation \eqref{eqn:optimal_v_i}, we have:

\begin{equation}
    \mathbf{v^*_i}(t) = -  R_{ii}^{-1}B_i^T P_i \mathbf{x^*}(t).
    \label{eqn:optimal_inf_solution}
\end{equation}

Denoting the set of optimal controllers:

\begin{equation}
    K^*_i \coloneqq R_{ii}^{-1}B_i^T P_i,
\end{equation}

we have:

\begin{equation}
    \mathbf{v^*_i}(t) = -  K_i^* \mathbf{x^*}(t).
\end{equation}

Thus in this case, as was for the regular LQR case, the optimal controller is linear with respect to the state. Plugging the optimal costs derivatives from~\eqref{eqn:inf_horizon_optimal_cost_derivative} to the GHJB equations at~\eqref{eqn:HJB_i} we have the set of infinite horizon GHJB equations:

\begin{equation}
    \begin{aligned}
     2\mathbf{x^*}^T(t) P_i  & \left[A\mathbf{x^*}(t) -  \sum _{j=1}^N B_j  R_{jj}^{-1}B_j^T \ P_j \mathbf{x^*}(t)\right]
       + \mathbf{x^*}(t)^TQ_i\mathbf{x^*}(t) \\
       & + \sum_{j=1}^N \mathbf{x^*}^T(t)P_j B_jR_{jj}^{-1} R_{ij} R_{jj}^{-1}B_j^T P_j \mathbf{x^*}(t) = 0
    \end{aligned}
\end{equation}


Rearranging we have:
\begin{equation}
    \begin{aligned}
    \mathbf{x^*}^T(t) \Bigg[ & 2 P_i  \ \left(A -  \sum _{j=1}^N B_j  R_{jj}^{-1}B_j^T P_j \right) + Q_i\\
       & + \sum_{j=1}^N P_j B_jR_{jj}^{-1} R_{ij} R_{jj}^{-1}B_j^T P_j \Bigg] \mathbf{x^*}(t)= 0
       \label{eqn:HJB_i_1}
    \end{aligned}
\end{equation}

Let us denote:

\begin{equation}
    \begin{aligned}
     A_{cl} \coloneqq A - \sum _{j=1}^N B_j R_{jj}^{-1} B_j^T P_j.
    \end{aligned}
\end{equation}

Plugging this into~\eqref{eqn:HJB_i_1} we have:

\begin{equation}
    \begin{aligned}
    \mathbf{x^*}^T(t) \Bigg[ & 2 P_i A_{cl} + Q_i + \sum_{j=1}^N P_j B_jR_{jj}^{-1} R_{ij} R_{jj}^{-1}B_j^T P_j \Bigg] \mathbf{x^*}(t)= 0
       \label{eqn:HJB_i_2}
    \end{aligned}
\end{equation}

Let us consider the first term $2 \mathbf{x^*}^T(t) P_i A_{cl} \mathbf{x^*}(t)$. By the shapes of all the matrices involved, we can notice the shape of the term is $(1 \times n) \times (n \times n) \times (n \times n) \times (n \times 1) = (1 \times 1)$, so it is a scalar. Thus it is equal to its transpose. By the same reasoning demonstrated in equation~\eqref{eqn:2xPAx} we have:

\begin{equation}
    2 \mathbf{x^*}^T(t) P_i A_{cl} \mathbf{x^*}(t) = \mathbf{x^*}^T(t) \left( P_i A_{cl} + A_{cl}^T P_i\right) \mathbf{x^*}(t)
\end{equation}

Plugging this into~\eqref{eqn:HJB_i_2} we have:

\begin{equation}
    \begin{aligned}
    \mathbf{x^*}^T(t) \Bigg[ & P_i A_{cl} + A_{cl}^T P_i + Q_i + \sum_{j=1}^N P_j B_jR_{jj}^{-1} R_{ij} R_{jj}^{-1}B_j^T P_j \Bigg] \mathbf{x^*}(t)= 0
    \end{aligned}
\end{equation}

This must hold for all optimal inputs $\mathbf{x}(t)$ and thus,  similar to what was shown in equation~\eqref{eqn:HJB}, this yields a set of \textbf{Game Continuous Algebraic Riccati Equations} (GCARE):

\begin{equation}
    \begin{aligned}
     P_iA_{cl} + A^T_{cl} P_i + Q_i  + \sum_{j=1}^N P_j B_j R_{jj}^{-1} R_{ij}R_{jj}^{-1} B_j^T P_j
       = 0.
    \end{aligned}
    \label{eqn:GCARE}
\end{equation}
\end{proof}
\subsection{Solutions to a System of GCAREs}

Note that when using our proposed method, by considering multiple objectives, we have introduced an additional layer of complexity with respect to a regular LQR; the GCAREs at equation~\eqref{eqn:GCARE} is actually a set of $N$ coupled algebraic matrix Riccati equations (or CAREs) for the matrices $(P_i)_{i=1}^N$, which have to be solved simultaneously. We would now like to establish claims regarding the properties of the solutions to the set of GCAREs.

Section 10.4 of Chapter 10 at~\cite{optimal_control_book} deals with multiplayer Nash differential games and the corresponding GCAREs that arise in them. As it well establishes, these coupled equations are very difficult to solve.

As an extension of the reasoning for the number of solutions to a single CARE provided in~\ref{sect:single_CARE_solution_num}, we would like to apply Bézout's Theorem~\cite{bezout} once again to establish the following natural (as the proof is trivial when considering the difference in sizes) extension of Proposition~\ref{prop:balanced_CARE} regarding the solutions to~\eqref{eqn:GCARE}:

\begin{theorem}
The set of GCAREs at~\eqref{eqn:GCARE} admits a balanced set of equations.
\label{the:gcare_bal}
\end{theorem}

Using this and also extending Corollary~\ref{cor:CARE_exp}, we have the following result:

\begin{theorem}
Let us consider the set of GCAREs at~\eqref{eqn:GCARE}, and assume it is well-behaved. Then the maximal number of solutions to the set of GCAREs is exponential with respect to $N \cdot n^2$.
\label{the:gcare_exp}
\end{theorem}

This last result demonstrates we have a balanced system with the same respective maximal number of solutions just as a usual CARE does, assuming it is well behaved. But can we know beforehand it will be well behaved? So what can we say about existence of solutions?

\subsubsection{GCAREs Solutions Existence and Uniqueness}
We may intuitively argue that the set of GCAREs will yield solutions in similar fashion, respective to scaling considerations, as a single CARE. As in, there may be no solution, but under some considerations we may be able to guarantee existence of a solution. Sections 11.4 of~\cite{riccati_Equation} and 6.5.3 of~\cite{Basar} both tell us otherwise. In the case of the coupled matrix Riccati equations, there is no general set of conditions that would
guarantee uniqueness or even existence of a solution. So we cannot state a game-complementary theorem such as~\ref{CARE_Existence_Uniqueness} to the GCAREs. Let us then present what we do know for sure:

\begin{theorem}[\textbf{GCAREs Nash Equilibrium} Proposition 6.8 in page 337 of \cite{Basar}]
Let us consider the set of GCAREs at~\eqref{eqn:GCARE}, and assume it is well-behaved. Let us assume there exists a series $(P^*_i)_{j=1}^N$ that solves the GCAREs. With that, denoting:
\begin{equation}
\begin{aligned}
    A^*_{{cl}_i} \coloneqq A - \sum_{\substack{j=1 \\ j \neq i}}^N B_j R_{jj}^{-1} B_j^T P^*_j,
\end{aligned}
\end{equation}
let us further assume the following conditions hold for all $1 \leq i \leq N$:
\begin{itemize}
    \item The pair $\left(A^*_{{cl}_i} \ , \  B_i\right)$ is stabilizable;
    \item The pair $\left(A^*_{{cl}_i} \ , \  \sqrt{Q_i + \sum_{\substack{j=1 \\ j \neq i}}^N P^*_j B_j R_{jj}^{-1}R_{ij}R_{jj}^{-1}B_j^TP^*_j} \right)$ is detectable.
\end{itemize}
Then the following holds:
\begin{enumerate}
    \item The optimal policies at~\eqref{eqn:optimal_inf_solution} assigned with $(P^*_i)_{j=1}^N$, i.e.:
\begin{equation}
    \mathbf{v^*_i}(t)\Big|_{P_i \equiv P^*_i} = -  R_{ii}^{-1}B_i^T P^*_i \mathbf{x^*}(t),
\end{equation}
provide a solution to the open-loop Nash Equilibrium~\eqref{nash_equilibrium} problem, thus in fact providing a closed-loop feedback;
\item The optimal costs in the Nash Equilibrium sense as $t \rightarrow \infty$ converge to the following forms:
\begin{equation}
    \lim_{t \rightarrow \infty} J_i^*(t) =  \lim_{t \rightarrow \infty} \mathbf{x^*}^T(t) P_i \mathbf{x^*}(t) = \mathbf{x_0}^T P_i \mathbf{x_0};
\end{equation}
\item The resulting closed-loop dynamics 
\begin{equation}
    \dot{\mathbf{x}}(t) = A^*_{{cl}}\mathbf{x}(t),
\end{equation}
where $A^*_{{cl}} \coloneqq A_{cl}\big|_{(P_j)_{j=1}^N \equiv (P^*_j)_{j=1}^N} = A - \sum_{j=1}^N B_j R_{jj}^{-1} B_j^T P^*_j$, is asymptotically stable.
\end{enumerate}
\label{the:GCAREs_no_cond}
\end{theorem}

This last discussion illustrates the claim presented at the introduction for the solution of coupled AREs, at~\ref{sect:ARE_intro}, as if there exists a series $(P^*_i)_{j=1}^N$ that solves the GCAREs, along with other conditions, then we are guaranteed the solution will yield a Nash equilibrium with a stable closed-loop dynamics.

\subsection{Continuous Finite Horizon LTI Game}
Similar to what was shown in the finite horizon case in Appendix \ref{subsect:finite_horizon_con_LQR}, a possible alternative formulation for the problem can be defined using the finite horizon case where the cost function is calculated up to a finite time $T_f < \infty$:

\begin{equation}\label{finite_objective_function}
\begin{aligned}
         J_i\Big(\mathbf{x}(t), (\mathbf{v_j}(t))_{j=1}^N, t\Big) = & \mathbf{\mathbf{x}}^T(T_f)Q_{f_i}\mathbf{\mathbf{x}}(T_f)  +\\  \int_{t}^{T_f} & \Big[\mathbf{x}(\tau)^TQ_i\mathbf{x}(\tau)  + \sum_{j=1}^N \mathbf{v}_j^T(\tau)R_{ij}\mathbf{v}_j(\tau)\Big]\mathrm{d}\tau
\end{aligned}
\end{equation}

with the same matrix sizes as described in equations~\eqref{inf_objective_function} and where $Q_{f_i} \in \mathbb{R}^{n \times n}$ are additional design parameters that account for the weights of the final state, defined for each $1 \leq i \leq N$, similar to the matrix $Q_f$ in equation \eqref{eqn:finite_value_integral}.

\begin{proposition}
Given a decomposed LTI system $\mathcal{S}_N$ with GLQR cost functions, let us consider the Nash Equilibrium optimal control problem given in~\eqref{nash_equilibrium} along with its induced optimal state trajectory $\mathbf{x^*}(t)$. Let us assume there exist a series of open-loop optimal solutions $\left(\mathbf{v_i^*}(t)\right)_{i=1}^N$ to the problem. Then $\left(\mathbf{v_i^*}(t)\right)_{i=1}^N$ are all in fact closed-loop feedback controllers and are all linear with regards to the optimal state, with corresponding time-varying coefficients, i.e. they are of the form:
\begin{equation}
    \mathbf{v_i^*}(t) = - K_i^*(t)\mathbf{x^*}(t),
\end{equation}
for some functions $\left(K_i^*(t)\right)_{i=1}^N$ where $K_i^* \ \colon \ \mathcal{I} \rightarrow \mathbb{R}^{n \times m}$.
\end{proposition}
\begin{proof}

We propose the following expressions for the optimal costs:

\begin{equation}
    J_i^* \equiv \mathbf{x^*}^T(t) P_i(t) \mathbf{x^*}(t),
    \label{eqn:finite_horizon_optimal_cost_N}
\end{equation}

where in this case each positive definite matrix $P_i : \mathcal{I} \rightarrow \mathbb{R}^{n \times n}$ is time-dependant. In this case we would have: $\mathbf{\partial_x J_i}^{*} = 2P_i(t) \mathbf{x^*}(t)$ and thus:

\begin{equation}
    \mathbf{v}_i^*(t) =-  R_{ii}^{-1} B_i^T P_i(t) \mathbf{x^*}(t).
    \label{eqn:finite_horizon_optimal_inputs}
\end{equation}

Denoting in this case the optimal controller as:

\begin{equation}
    K^*_i(t) =  R_{ii}^{-1} B_i^T P_i(t).
\end{equation}

we have:

\begin{equation}
    \mathbf{v}_i^*(t) =-  K^*_i(t) \mathbf{x^*}(t).
\end{equation}

So we have that the optimal controller in this case is also linear, but time-dependant this time. Appendix \ref{subsect:finite_horizon_con_LQR} describes the finite horizon case for the regular LQR (i.e. when $N=1$) by considering the temporal partial derivative of the optimal cost function. Equivalently, differentiating~\eqref{eqn:finite_horizon_optimal_cost_N} with respect to time while using the chain rule, as was done in~\eqref{eqn:J_time_derivative}, we have:

\begin{equation}
\begin{aligned}
    \partial_t J^{*}_i & = \dot{\mathbf{x}}^{*^T}(t)P_i(t) \mathbf{x^*}(t) + \mathbf{x^*}^T(t)\dot{P}_i(t) \mathbf{x^*}(t) + \mathbf{x^*}^T(t)P_i(t)\dot{\mathbf{x}}^*(t).
    \label{eqn:J_time_derivative_N}
\end{aligned}
\end{equation}

Plugging the optimal costs temporal derivatives from~\eqref{eqn:J_time_derivative_N}, along with the optimal inputs from~\eqref{eqn:finite_horizon_optimal_inputs} into the Game Bellman equations at~\eqref{eqn:bellman_N} we have:

\begin{equation}
    \begin{aligned}
     & \dot{\mathbf{x}}^{*^T}(t)P_i(t) \mathbf{x^*}(t) + \mathbf{x^*}^T(t)\dot{P}_i(t) \mathbf{x^*}(t) + \mathbf{x^*}^T(t)P_i(t)\dot{\mathbf{x}}^*(t) + \mathbf{x^*}^T(t)Q_i\mathbf{x^*}(t) \\
     & + \sum_{j=1}^N \left[R_{jj}^{-1} B_j^T P_j(t) \mathbf{x^*}(t)\right]^TR_{ij}R_{jj}^{-1} B_j^T P_j(t) \mathbf{x^*}(t) = 0
    \end{aligned}
\end{equation}

Plugging in the expression for $\mathbf{\dot{x}^*}(t)$ from the system model at~\eqref{eqn:sys_with_B1___Bn}, we have:

\begin{equation}
    \begin{aligned}
     & \left[A\mathbf{x^*}(t) + \sum _{j=1}^N B_j \mathbf{v^*_j}(t)\right]^{T}P_i(t) \mathbf{x^*}(t) + \mathbf{x^*}^T(t)\dot{P}_i(t) \mathbf{x^*}(t) 
     \\ & + \mathbf{x^*}^T(t)P_i(t)\left[A\mathbf{x^*}(t) + \sum _{j=1}^N B_j \mathbf{v^*_j}(t)\right] + \mathbf{x^*}^T(t)Q_i\mathbf{x^*}(t) \\
     & + \mathbf{x^*}^T(t)\sum_{j=1}^N P_j(t) B_j R_{jj}^{-1} R_{ij}R_{jj}^{-1} B_j^T P_j(t) \mathbf{x^*}(t) = 0
    \end{aligned}
\end{equation}

Plugging again the values for the optimal inputs, we have:

\begin{equation}
    \begin{aligned}
     & \left[\mathbf{x^*}^T(t)A^T - \sum _{j=1}^N \left(R_{jj}^{-1} B_j^T P_j(t) \mathbf{x^*}(t)\right)^TB^T_j \right]P_i(t) \mathbf{x^*}(t) + \mathbf{x^*}^T(t)\dot{P}_i(t) \mathbf{x^*}(t) 
     \\ & + \mathbf{x^*}^T(t)P_i(t)\left[A\mathbf{x^*}(t) - \sum _{j=1}^N B_j R_{jj}^{-1} B_j^T P_j(t) \mathbf{x^*}(t)\right] + \mathbf{x^*}^T(t)Q_i\mathbf{x^*}(t) \\
     & + \mathbf{x^*}^T(t)\sum_{j=1}^N P_j(t) B_j R_{jj}^{-1} R_{ij}R_{jj}^{-1} B_j^T P_j(t) \mathbf{x^*}(t) = 0
    \end{aligned}
    \label{eqn:finite_GHJB_1}
\end{equation}

Let us now denote:

\begin{equation}
    \begin{aligned}
     A_{cl}(t)  \coloneqq A - \sum _{j=1}^N B_j R_{jj}^{-1} B_j^T P_j(t).
    \end{aligned}
\end{equation}

Plugging $A_{cl}(t)$ into \eqref{eqn:finite_GHJB_1} and rearranging, we have:

\begin{equation}
    \begin{aligned}
     \mathbf{x^*}^T(t)\Big[ & A^T_{cl}(t)P_i(t) + \dot{P}_i(t) + P_i(t)A_{cl}(t) + Q_i 
     \\ & + \sum_{j=1}^N P_j(t) B_j R_{jj}^{-1} R_{ij}R_{jj}^{-1} B_j^T P_j(t) \Big] \mathbf{x^*}(t) = 0,
    \end{aligned}
\end{equation}

which yields a set of \textbf{Game Continuous Differential Riccati Equations} (GDARE):

\begin{equation}
\begin{aligned}
    \dot{P}_i(t) + A^T_{cl}(t)P_i(t) + P_i(t)A_{cl}(t) + Q_i + \sum_{j=1}^N P_j(t) B_j R_{jj}^{-1} R_{ij}R_{jj}^{-1} B_j^T P_j(t) = 0,
     \end{aligned}
     \label{eqn:GDARE}
\end{equation}
that along with appropriate terminal conditions of the form:
\begin{equation}
    \left(P_i(T_f)\right)_{i=1}^N \equiv \left(Q_{f_i}\right)_{i=1}^N
\end{equation}
are again solvable backwards in time to yield the time-varying matrix functions $\left(P_i(t\right)_{i=1}^N$.
\end{proof}

\subsection{Solutions to a System of GDAREs}

Let us now discuss the GDAREs. It is easy to determine that they provide the same amount of equations and variables as did the GCAREs, and as a result we can extend Theorems~\ref{the:gcare_bal} and~\ref{the:gcare_exp} naturally to the GDAREs. In this case though, these are systems of coupled matrix differential equations, and their properties are of interest to us in a general sense, as well as when applying our algorithm to transform AREs to DREs. 

We saw at Corollary~\ref{CARE_Existence_Uniqueness} that a simple LQR-induced CARE provides assurances for existence and uniqueness of the solution under standard assumptions, but the set of GCAREs does not provide the same assurances, and the best we can come with is what Theorem~\ref{the:GCAREs_no_cond} describes.

The method we suggest though, to transform AREs to DREs, becomes even more interesting, given the following Theorem regarding a given system of GDAREs:

\begin{theorem}[\textbf{GDAREs Nash Equilibrium} (Corollary 6.5 in page 323 of~\cite{Basar}]
Let us consider the set of GDAREs at~\eqref{eqn:GDARE}, and assume it is well-behaved. Let us assume there exists a series $\left(P^*_i(t)\right)_{j=1}^N$ that solves the GDAREs. Then the optimal policies at~\eqref{eqn:finite_horizon_optimal_inputs} assigned with $(P^*_i(t))_{j=1}^N$, i.e.:
\begin{equation}
    \mathbf{v^*_i}(t)\Big|_{P_i(t) \equiv P^*_i(t)} = -  R_{ii}^{-1}B_i^T P^*_i(t) \mathbf{x^*}(t),
\end{equation}
provide a solution to the open-loop Nash Equilibrium~\eqref{nash_equilibrium} problem, thus in fact providing a closed-loop feedback.
\label{the:GCDREs_cond}
\end{theorem}

Theorem~\ref{the:GCDREs_cond} then says, that if there is a solution to the GDAREs, it will render a closed-loop solution to the finite-horizon open-loop Nash Equilibrium problem at~\ref{eqn:optimal_nash_control_problem_formulation}. 

Coincidentally, the system of GDAREs is a system of differential equations, and matrix, coupled or whatever else it may be, along with corresponding terminal values, we are guaranteed a unique solution due to it being a Cauchy Initial Value problem~\cite{existence}. 

\begin{remark}
A subtle point to notice here is that we are promised a unique solution to the GDAREs, that will provide a corresponding Nash Equilibrium, as Theorem~\ref{the:GCDREs_cond} states, but that Nash Equilibrium may not be unique, and this may require additional assumptions. See Remark 6.16 at page 324 of~\cite{Basar} for more detail.
\end{remark}

\section{Discrete Multi-Objective Optimal Control}

The models presented thus far were designed for continuous-time systems, meant to be implemented as analogue components. That being said, due to the digital nature of the modern day computer, controllers of actual physical systems are implemented at discrete intervals at real-time. We would thus want to derive a discrete-time version of the model we described thus far, to provide a broader contribution and to handle cases in which this approach may be preferable. We will use direct design, as described in Appendix~\ref{sect:lti_discrete_appendix}.

\subsection{Direct Discrete Decomposed Controllers Design}

To perform direct design, let us consider a decomposed continuous-time LTI system $\mathcal{S}_N$ adhering equation~\eqref{eqn:sys_with_B1___Bn}, that we would now want to discretize, by extending the derivation of a regular LTI model discretization, as detailed in Appendix~\ref{app:lti_discretization}. Specifically, given we use the same Interval Discretization technique and then State Discretization by Sampling combined with Input Discretization by Sampling and Zero-Order-Hold, for each of the virtual objectives, to arrive at the discrete-time LTI equivalent system of the decomposed system $\mathcal{S}_N$.


\begin{definition}
Given a series of virtual inputs $(\mathbf{v_i}(t))_{i=1}^N$ and $L \in \mathbb{N}$, let $ \left(\mathbf{x}[k]\right)_{k=0}^{L-1}$ and $ \left(\left(\mathbf{v_i}[k]\right)_{i=0}^{N}\right)_{k=0}^{L-1}$ be $L$ corresponding state and virtual inputs samples measured from $\mathcal{S}_N$, where $\mathbf{x}[k] \in \mathcal{X}_D \subseteq \mathcal{X}$ and $\mathbf{v_i}[k] \in \mathcal{V}_{i_D} \subseteq \mathcal{V}_{i}$. With that, for all $0 \leq k < L -1$, the \textbf{discrete decomposed LTI model} of $\mathcal{S}_N$, that we shall refer to as $\mathcal{S}_{N_D}$:

\begin{equation}
    \mathbf{x}[k+1] = \tilde{A}\mathbf{x}[k] + \sum _{i=1}^N \tilde{B_i} \mathbf{v_i}[k]
    \label{eqn:discrete_decomposed_sys}
\end{equation}

with $\mathbf{x}[0]=\mathbf{x}_0$, and with:

\begin{equation}
    \begin{aligned}
\tilde{A} \coloneqq & \ e^{\delta A}\\
\tilde{B_i} \coloneqq & \int_{0}^\delta e^{t A} \mathrm{d}t \ B_i
\end{aligned}
\label{eqn:S_N_D_model}
\end{equation}
\end{definition}

\subsection{Discrete Decomposed Optimal Control}

The discretization of the virtual cost functions $\left( J_i(t) \right)_{i=1}^N$, defined at~\ref{objective_function} is no different than the discretization of a single LQR cost function.
\begin{definition}
Given a series of continuous virtual cost functions $\left( J_i(t) \right)_{i=1}^N$, let us define the corresponding \textbf{discrete virtual cost functions} $\left( J_{i_D} \right)_{i=1}^N$ for the system $\mathcal{S}_{N_D}$, where each $J_{i_D} \colon \mathcal{X}_D \times \mathcal{V_D} \times \bigcup_{i=0}^{L-1}\{i\} \rightarrow \mathbb{R}^+$ with $\mathcal{V_D} \coloneqq \bigtimes_{i=1}^N\mathcal{V}_{i_D}$ by the following form:

\begin{equation}
\begin{aligned}
       J_{i_D}\Big(\mathbf{x}[k], (\mathbf{v_j}[k])_{j=1}^N, k\Big) \coloneqq 
       \mathlarger{\sum}_{t=k}^{L-1} r_{i_D}\Big(\mathbf{x}[t], (\mathbf{v_j}[k])_{j=1}^N\Big) + r_{f_{i_D}}\Big(\mathbf{x}[L]\Big)
\end{aligned}
\label{discrete_objective_function}
\end{equation}
where $(r_{i_D})_{i=1}^N$ and $(r_{f_{i_D}})_{i=1}^N$ are the discrete-time equivalents of $(r_{i})_{i=1}^N$ and $(r_{f_{i}})_{i=1}^N$ from~\eqref{objective_function}.
\end{definition}
Once again, we will denote $J_{i_D}[k]$ for conciseness.
\begin{definition}
Given an discrete decomposed LTI system $\mathcal{S}_{N_D}$, that adheres the model at~\eqref{eqn:S_N_D_model}, we deal with cost function assignments that are quadratic with respect to the state and input vectors, i.e., $r_{i_D}[k]$ takes on the following form:

\begin{equation}
   r_{i_D}\Big(\mathbf{x}[k], (\mathbf{v_j}[k])_{j=1}^N\Big) \equiv \mathbf{x}^T[k]Q_i\mathbf{x}[k] + \sum_{j=1}^N \mathbf{v_j}^T[k]R_{ij}\mathbf{v_j}[k].
\end{equation}
with:

\begin{itemize}
    \item $Q_i \in \mathbb{R}^{n \times n}$ is a positive semi-definite matrix;
    \item $\forall 1 \leq j \leq N \ \colon \ R_{ij} \in \mathbb{R}^{m_j \times m_j}$ is a positive definite matrix.
\end{itemize}
In this scenario, the control model is called a \textbf{Discrete Game Linear Quadratic Regulator} (DGLQR), similar to what defined in the continuous case.
\end{definition}

\subsection{Optimality by Discrete Open-Loop Nash Equilibrium}
\begin{definition}
For a given discrete decomposed system $\mathcal{S}_{N_D}$ with state vector  measurements $\left(\mathbf{x}[k]\right)_{k=0}^{L-1}$,  a series of virtual policies measurements $\left((\mathbf{v^*_j}[k])_{k=0}^{L-1}\right)_{j=1}^{N}$ is said to constitute a \textbf{Discrete Open-Loop Nash Equilibrium} if for all $1 \leq i \leq N$, we have:

\begin{equation}
    J_{i_D}\bigg(\mathbf{x}[k], (\mathbf{v^*_j}[k])_{j=1}^N, k\bigg) \leq J_{i_D}\bigg(\mathbf{x}[k], (\mathbf{v^*_j}[k])_{\substack{j=1 \\ j \neq i}}^N, \mathbf{v_i}[k], k\bigg)
    \label{discrete_nash_equilibrium}
\end{equation}
for any other admissable virtual input measurement $\mathbf{v_i}[k] \in \mathcal{V}_{i_D}$ and for all $k \in \bigcup_{i=0}^{L-1}\{i\}$.
\end{definition}
 \begin{definition}

Given an discrete decomposed LTI system $\mathcal{S}_{N_D}$, that adheres the model at~\eqref{eqn:discrete_decomposed_sys} and with a series of DGLQR cost functions $\left(J_{i_D}[k]\right)_{i=1}^N$, we define the \textbf{discrete open-loop Nash Equilibrium optimal control problem}, formally stated as:

\begin{equation}
    \begin{aligned}
         \min_{\substack{\mathbf{v_i}[k] \in \mathcal{V}_{i_D}}} \quad & J_{i_D}\Big(\mathbf{x}[k], (\mathbf{v_j}[k])_{j=1}^N, 0\Big) \\
         = \min_{\substack{\mathbf{v_i}[k] \in \mathcal{V}_{i_D}}} \quad & \mathlarger{\sum}_{t=0}^{L-1} r_{i_D}\Big(\mathbf{x}[t], (\mathbf{v_j}[k])_{j=1}^N\Big) + r_{f_{D_i}}\Big(\mathbf{x}[L]\Big)\\
         = \min_{\substack{\mathbf{v_i}[k] \in \mathcal{V}_{i_D}}} \quad & \mathlarger{\sum}_{t=0}^{L-1} \left[\mathbf{x}^T[k]Q_i\mathbf{x}[k] + \sum_{j=1}^N \mathbf{v_j}^T[k]R_{ij}\mathbf{v_j}[k]\right] + r_{D_f}\Big(\mathbf{x}[L]\Big) \\
         \textrm{subject to} \quad & \mathbf{x}[0] = \mathbf{x_0};\\
         \quad & \mathbf{x}[k+1] = \tilde{A}\mathbf{x}[k] + \sum _{i=1}^N \tilde{B_i} \mathbf{v_i}[k]; \quad & \forall k \in \bigcup_{i=0}^{L-1}\{i\};\\
         & \forall 1 \leq i \leq N.
    \end{aligned}
    \label{eqn:discrete_game_optimal_control_problem_formulation}
\end{equation}
\end{definition}

\begin{definition}
Given a discrete decomposed LTI system $\mathcal{S}_{N_D}$, let us define a \textbf{game optimal state trajectory measurements} of $\mathcal{S}_{N_D}$, as a series $\left(\mathbf{x^*}[k]\right)_{k=0}^{L-1}$ which is the solution to the model of $\mathcal{S}_{N_D}$ at~\eqref{eqn:discrete_decomposed_sys}, assigned with the discrete optimal open-loop Nash Equilibrium policies measurements $\left((\mathbf{v^*_j}[k])_{k=0}^{L-1}\right)_{j=1}^{N}$, i.e., it satisfies:
\begin{equation}
    \mathbf{x^*}[k+1] = \tilde{A}\mathbf{x^*}[k] + \sum _{i=1}^N \tilde{B_i} \mathbf{v^*_i}[k],
\end{equation}
for all $k \in \bigcup_{i=0}^{L-2}\{i\}$.
\end{definition}

\begin{definition}
Let us denote $\Delta J_{i_D}$ to be the \textbf{$k$'th finite difference of the $i$'th discrete cost function}:
\begin{equation}
\begin{aligned}
    \Delta J_{i_D}[k] & \coloneqq J_{i_D}[k+1] - J_{i_D}[k]\\
    & = J_{i_D}\Big(\mathbf{x}[k+1], \mathbf{u}[k+1], k+1\Big) - J_{i_D}\Big(\mathbf{x}[k], \mathbf{u}[k], k\Big)
\end{aligned}
\end{equation}
\end{definition}

\begin{definition}

Note that by its definition, for all $k \in \bigcup_{i=0}^{L-2}\{i\}$, we can write Equation~\eqref{discrete_objective_function} in the form:

\begin{equation}
\begin{aligned}
        J_{i_D}\Big(\mathbf{x}[k], (\mathbf{v_j}[k])_{j=1}^N, k\Big) & =
        r_{i_D}\Big(\mathbf{x}[k], (\mathbf{v_j}[k])_{j=1}^N\Big) \\
        & +
        \mathlarger{\sum}_{t=k+1}^{L-1} r_{i_D}\Big(\mathbf{x}[t], (\mathbf{v_j}[k])_{j=1}^N\Big)  + r_{f_{i_D}}\Big(\mathbf{x}[L]\Big)\\
        & = r_{i_D}\Big(\mathbf{x}[k], (\mathbf{v_j}[k])_{j=1}^N\Big) \\
        & +
        J_{i_D}\Big(\mathbf{x}[k+1], (\mathbf{v_j}[k+1])_{j=1}^N, k+1\Big).
        \label{eqn:discrete_game_bellman_1}
\end{aligned}
\end{equation}

Moving sides and multiplying by $-1$, we have:

\begin{equation}
\begin{aligned}
        \Delta J_{i_D}[k] = - r_{i_D}\Big(\mathbf{x}[k], (\mathbf{v_j}[k])_{j=1}^N\Big).
\end{aligned}
\end{equation}

Using this, we symmetrically define the \textbf{discrete time virtual control Hamiltonian}:

\begin{equation}
    \mathcal{H}_{i_D} \Big( \mathbf{x}[k], \Delta J_{i_D}[k], (\mathbf{v_j}[k])_{j=1}^N \Big)  \coloneqq \Delta J_{i_D}[k] + r_{i_D}\Big(\mathbf{x}[k], (\mathbf{v_j}[k])_{j=1}^N\Big).
    \label{eqn:game_discrete_hamiltonian}
\end{equation}
\begin{definition}
Next, to define the \textbf{discrete-time Game Bellman Equation} for the $k$'th measurement:

\begin{equation}
    \mathcal{H}_{i_D} \Big( \mathbf{x}[k], \Delta J_{i_D}[k], (\mathbf{v_j}[k])_{j=1}^N \Big)  = 0.
    \label{eqn:discrete_game_bellman}
\end{equation}
\end{definition}

\begin{definition}
Let us denote the $k$'th measurement of \textbf{the optimal value in regards to state trajectory} of the $i$'th cost function $J_{i_D}[k]$ and its finite difference as:
\begin{equation}
\begin{aligned}
     J_{i_D}^*[k] \coloneqq & J_{i_D}\big( \mathbf{x}[k],(\mathbf{v_j}[k])_{j=1}^N \big) \ \big|_{\mathbf{x}[k] \equiv \mathbf{x^*}[k]};\\
    \Delta J_{i_D}^*[k] \coloneqq & \Delta J_{i_D}[k] \ \big|_{\mathbf{x}[k] \equiv \mathbf{x^*}[k]} = J_{i_D}^*[k+1] - J_{i_D}^*[k]\ .
\end{aligned}
\end{equation}
\end{definition}

\begin{definition}
Applying the stationarity condition at equation~\eqref{eqn:discrete_stationarity} by differentiating each $k$'th measurement of the discrete Hamiltonian $\mathcal{H}_{i_D}[k]$ by the $k$'th measurement of the corresponding $i$'th virtual input $\mathbf{v_i}[k]$, equating to zero and solving for $\mathbf{v_i}[k]$, we define the \textbf{Discrete Game PMP (DGPMP)}. In that we obtain a set of the $k$'th measurements of the optimal virtual inputs $\left(\mathbf{v^*_i}[k]\right)_{i=1}^N$. Thus the stationarity condition for multiplayer games corresponds to the following set of $N$ coupled equations:
 
 \begin{equation}
     \frac{\partial}{\partial \mathbf{v_i}[k]} \mathcal{H}_{i_D} \Big( \mathbf{x^*}[k], \Delta J^*_{i_D}[k], (\mathbf{v_j}[k])_{j=1}^N \Big)\Big|_{\mathbf{v_i}[k] \equiv \mathbf{v_i^*}[k]} = 0.
     \label{eqn:discrete_stationarity_N}
\end{equation}
\end{definition}


\end{definition}

\subsection{Discrete Infinite Horizon LTI Game}

In the infinite horizon case we consider the following form for the discrete cost function:

\begin{equation}
\begin{aligned}
        J_{i_D}\left(\mathbf{x}[k], (\mathbf{v_j}[k])_{j=1}^N, k\right) = 
      \mathlarger{\sum}_{t=k}^{\infty} \left[\mathbf{x}[t]^TQ_i\mathbf{x}[t] +  \sum_{j=1}^N \mathbf{v_j}^T[t]R_{ij}\mathbf{v_j}[t]\right]
\end{aligned}
\end{equation}

\begin{proposition}
Given a discrete decomposed LTI system $\mathcal{S}_{N_D}$ with DGLQR cost functions $(J_{i_D})_{i=1}^N$, let us consider the discrete open-loop Nash Equilibrium optimal control problem given in~\eqref{eqn:discrete_game_optimal_control_problem_formulation}, in the case where $T_f \rightarrow \infty$. Let us assume there exist a series of open-loop optimal solutions measurements $\left((\mathbf{v^*_j}[k])_{k=0}^{L-1}\right)_{j=1}^{N}$ to the problem. Then:
\begin{enumerate}
    \item The measurements of the optimal inputs in regards to the open-loop Nash Equilibrium  $\left((\mathbf{v^*_j}[k])_{k=0}^{L-1}\right)_{j=1}^{N}$ are all in fact closed-loop feedback controllers and are all linear with regards to the optimal state, with corresponding constant coefficients, i.e. they are of the form:
\begin{equation}
    \mathbf{v_i^*}[k] = - K_{i_D}^*\mathbf{x^*}[k],
\end{equation}
where each $K_{i_D}^* \in \mathbb{R}^{n \times m_i}$, for all $k \in \bigcup_{i=0}^{L-1}\{i\}$.
\item The optimal controllers $\left(K_{i_D}^*\right)_{i=1}^N$, each $K_{i_D}^* \in \mathbb{R}^{n \times m_i}$, satisfy the following set of coupled equations:

\begin{equation}
    K^*_{i_D}= \Big(R_{ii} + \tilde{B}^T_iP_{i_D}\tilde{B}_i \Big)^{-1}\tilde{B}^T_iP_{i_D}\Big(\tilde{A}- \sum _{\substack{j=1 \\ j \neq i}}^N \tilde{B}_j K^*_{j_D} \Big)
\end{equation}

and the matrices $\left(P_{i_D}\right)_{i=1}^N$, each $P_{i_D} \in \mathbb{R}^{n \times n}$ are obtained by solving the following coupled set of matrix Riccati equations simultaneously with the aforementioned relations for $\left(K_{i_D}^*\right)_{i=1}^N$:

\begin{equation}
    P_{i_D} = Q_i  + \sum_{j=1}^N K_{j_D}^{*^T}R_{ij}K_{j_D}^*
       +\Big(\tilde{A} - \sum _{j=1}^N \tilde{B}_j K_{j_D}^* \Big)^TP_{i_D}\Big(\tilde{A} - \sum _{j=1}^N \tilde{B}_j K_{j_D}^* \Big)
\end{equation}
\end{enumerate}
\label{prop:discrete_nash_solution}
\end{proposition}
\begin{proof}

we consider the optimal cost function to be quadratic with respect to the state:

\begin{equation}
    J_{i_D}^*[k] \equiv \mathbf{x^*}^T[k] P_{i_D} \mathbf{x^*}[k].
    \label{eqn:discrete_game_infinite_horizon_optimal_cost}
\end{equation}

for some constant positive-definite matrix $P_{i_D} \in \mathbb{R}^{n \times n}$. Plugging all optimal the expressions to the Hamiltonian at~\eqref{eqn:game_discrete_hamiltonian}:

\begin{equation}
\begin{aligned}
      \mathcal{H}_{i_D} & \left( \mathbf{x^*}[k], \Delta J^*_{i_D}[k], (\mathbf{v_j}[k])_{j=1}^N \right)  \\ & =J^*_{i_D}\left(\mathbf{x}[k+1], (\mathbf{v_j}[k+1])_{j=1}^N, k+1\right)
    - J^*_{i_D}\left(\mathbf{x}[k], (\mathbf{v_j}[k])_{j=1}^N, k\right) \\
    & + r_{i_D}\left(\mathbf{x}^*[k], (\mathbf{v_j}[k])_{j=1}^N\right)\\
    &= \mathbf{x^*}^T[k+1] P_{i_D} \mathbf{x^*}[k+1]
    - \mathbf{x^*}^T[k] P_{i_D} \mathbf{x^*}[k] + \mathbf{x}^{*^T}[k]Q_i\mathbf{x}^*[k] +  \sum_{j=1}^N \mathbf{v_j}^T[k]R_{ij}\mathbf{v_j}[k]\\
    &= \left(\tilde{A}\mathbf{x^*}[k] + \sum _{j=1}^N \tilde{B_j} \mathbf{v_j}[k]\right)^T P_{i_D} \left(\tilde{A}\mathbf{x^*}[k] + \sum _{j=1}^N \tilde{B_j} \mathbf{v_j}[k]\right)\\
    & + \mathbf{x}^{*^T}[k]\left(Q_i-P_{i_D}\right)\mathbf{x}^*[k] +  \sum_{j=1}^N \mathbf{v_j}^T[k]R_{ij}\mathbf{v_j}[k].
     \label{eqn:discrete_game_infinite_Hamiltonian}
\end{aligned}
\end{equation}

Applying DGPMP to Equation~\eqref{eqn:discrete_game_infinite_Hamiltonian} we have:

\begin{equation}
\begin{aligned}
      \frac{\partial}{\partial \mathbf{v_i}[k]} & \mathcal{H}_{i_D} \Big( \mathbf{x^*}[k], \Delta J^*_{i_D}[k], (\mathbf{v_j}[k])_{j=1}^N \Big)\Big|_{\mathbf{v_i}[k] \equiv \mathbf{v_i^*}[k]}\\ 
     & = 2\tilde{B}_i^TP_{i_D}\Big(\tilde{A}\mathbf{x^*}[k] + \tilde{B}_i \mathbf{v^*_i}[k] + \sum _{\substack{j=1 \\ j \neq i}}^N \tilde{B}_j \mathbf{v_j}[k]\Big) + 2R_{ii}\mathbf{v^*_i}[k] = 0.
\end{aligned}
\end{equation}

Notice we have the term $\mathbf{v^*_i}[k]$ both in the second element as well as in the sum. Moving sides and collecting terms we have:

\begin{equation}
    \Big(R_{ii} + \tilde{B}^T_iP_{i_D}\tilde{B}_i \Big) \mathbf{v^*_i}[k]= -\tilde{B}_i^T P_{i_D}\Big(\tilde{A}\mathbf{x^*}[k]+ \sum _{\substack{j=1 \\ j \neq i}}^N \tilde{B}_j \mathbf{v^*_j}[k] \Big)
\end{equation}

Assuming $R_{ii} + \tilde{B}^T_iP_i\tilde{B}_i$ is invertible, we arrive at the following relation:
\begin{equation}
    \mathbf{v^*_i}[k]= -\Big(R_{ii} + \tilde{B}^T_iP_{i_D}\tilde{B}_i \Big)^{-1}\tilde{B}^T_iP_{i_D}\Big(\tilde{A}\mathbf{x^*}[k]+ \sum _{\substack{j=1 \\ j \neq i}}^N \tilde{B}_j \mathbf{v_j}[k] \Big)
\end{equation}

With induction on $N$, performing the same procedure, plugging this expression simultaneously for all $1 \leq i \leq N$ in the rightmost sum and rearranging, claim 1 follows.

In order to obtain an expression for the matrices $\left(P_{i_D}\right)_{i=1}^N$, let us plug in $\mathbf{v^*_i}[k] = - K_{i_D}^*\mathbf{x}[k]$ to the discrete-time Game Bellman Equation at~\ref{eqn:discrete_game_bellman} and rearrange to obtain:

\begin{equation}
    \begin{aligned}
    \mathbf{x}[k]^T\Bigg[ & Q_i- P_{i_D} + \sum_{j=1}^N K_{j_D}^{*^T}R_{ij}K_{j_D}^*
       \\
       & \left(\tilde{A} - \sum _{j=1}^N \tilde{B}_j K_{j_D}^* \right)^TP_{i_D}\left(\tilde{A} - \sum _{j=1}^N \tilde{B}_j K_{j_D}^* \right)\Bigg]\mathbf{x}[k] = 0
    \end{aligned}
\end{equation}

This must hold for all $\mathbf{x}[k]$ and thus:

\begin{equation}
    P_{i_D} = Q_i  + \sum_{j=1}^N K_{j_D}^{*^T}R_{ij}K_{j_D}^*
       +\Big(\tilde{A} - \sum _{j=1}^N \tilde{B}_j K_{j_D}^* \Big)^TP_{i_D}\Big(\tilde{A} - \sum _{j=1}^N \tilde{B}_j K_{j_D}^* \Big)
       \label{eqn:nash_eq_1}
\end{equation}

Plugging in $\mathbf{v^*_i}[k] = - K_{i_D}^*\mathbf{x}[k]$ to the last equation obtained we get:

\begin{equation}
    - K_{i_D}^*\mathbf{x}[k]= -\Big(R_{ii} + \tilde{B}^T_iP_i\tilde{B}_i \Big)^{-1}\tilde{B}^T_iP_i\Big(\tilde{A}- \sum _{\substack{j=1 \\ j \neq i}}^N \tilde{B}_j K_{i_D}^* \Big)\mathbf{x}[k]
\end{equation}

This must hold for all $\mathbf{x}[k]$ and thus we get the following set of linear matrix equations for the optimal controllers $(K^*_{i_D})_{i=1}^N$:

\begin{equation}
    K^*_{i_D}= \Big(R_{ii} + \tilde{B}^T_iP_{i_D}\tilde{B}_i \Big)^{-1}\tilde{B}^T_iP_{i_D}\Big(\tilde{A}- \sum _{\substack{j=1 \\ j \neq i}}^N \tilde{B}_j K^*_{j_D} \Big)
    \label{eqn:nash_eq_2}
\end{equation}

which proves claim 2.
\end{proof}

One extremely interesting result is that Proposition~\ref{prop:discrete_nash_solution} that we have just proven is just a private case of a more general Theorem:
\begin{theorem}[\textbf{Discrete Infinite Horizon Nash Equilibrium Uniqueness} (Corollary 6.1 at page 280 of~\cite{Basar}]
    Given a discrete decomposed LTI system $\mathcal{S}_{N_D}$ with DGLQR cost functions $(J_{i_D})_{i=1}^N$, let us consider the discrete open-loop Nash Equilibrium optimal control problem given in~\eqref{eqn:discrete_game_optimal_control_problem_formulation}, in the case where $T_f \rightarrow \infty$. The control problem yields a unique Nash Equilibrium solution if and only if equations~\eqref{eqn:nash_eq_1} and~\eqref{eqn:nash_eq_2} yield a unique solution $(K^*_{i_D}, P_{i_D})_{i=1}^N$.
\end{theorem}

The discrete finite horizon game theory is left for future work, as will be described in~\ref{chap:conclusions}.

\section{Handling Nonlinear Systems}
The approach presented thus far applies to linear time-invariant systems with static objectives. We propose to use it as an ingredient in a wider framework that can also deal with nonlinear and possibly time-varying objectives. The main principle is to change the game dynamically when needed. This method is to be implemented on complex, non-linear systems, and we will demonstrate at the second case study of the quadrotor at~\ref{sect:quadrotor}.

\begin{definition}
Let us consider a general nonlinear control system $\mathcal{S}$ operating on a closed or semi-open interval $\mathcal{I}$ as described in~\eqref{eqn:dynamical_system}. Considering a time $t^* \in \mathcal{I}$ at which the system is mostly operating in the small vicinity of given equilibrium state, i.e., it satisfies $f(\mathbf{x}(t^*), \mathbf{u}(t^*))=0$, then the system linearization of $\mathcal{S}$ can be performed in the \textbf{standard linearization form}~\cite{linearization} by the following:
\begin{equation}
\begin{aligned}
A := \quad & \left. \frac{\partial f\left(\mathbf{x}(t), \mathbf{u}(t)\right)}{\partial \mathbf{x}(t)}
\right|_{\begin{matrix}
\mathbf{x}(t) \equiv \mathbf{x}(t^*) \\ \mathbf{u}(t) \equiv \mathbf{u}(t^*) \end{matrix}} \in \mathbb{R} ^{n \times n} \\
  B := \quad & \left. \frac{\partial f\left(\mathbf{x}(t), \mathbf{u}(t)\right)}{\partial \mathbf{u}(t)} \right|_{\begin{matrix}
\mathbf{x}(t) \equiv \mathbf{x}(t^*) \\ \mathbf{u}(t) \equiv \mathbf{u}(t^*) \end{matrix}}
  \in \mathbb{R} ^{n \times m}
\end{aligned}
\end{equation}
\end{definition}

\begin{definition}
Alternatively, one uses the \textbf{Linear Parameter Varying (LPV)} approach~\cite{LPV}: Assume a nonlinear model can be formulated as a linear system, where the model matrices $A$ and $B$ are functions of time-varying parameters, represented by the two vectors $\mathbf{z_1}(t)$ and $\mathbf{z_2}(t)$, respectively. Note that these time-varying parameters may include state variables (or functions of state variables) if their value is known (e.g., measurable) in real-time for the fixed time interval $\mathcal{I}$. Hence, the model is now represented as,

\begin{equation}\label{linearized}
\dot{\mathbf{x}}(t) = A(\mathbf{z_1}(t)) \mathbf{x}(t) + B(\mathbf{z_2}(t))\mathbf{u}(t).
\end{equation}

Depending on the time scale of $\mathbf{z_1}$ and $\mathbf{z_2}$, one may select $k \in \mathbb{N}$ time intervals $\left([t_i,t_{i+1}]\right)_{i=0}^{k-1}$ which are small enough, for which the time-varying system can be frozen at $t_i$ and assumed to act as an LTI system until $t_{i+1}$, with an acceptable model error (or uncertainty). The result is a piecewise LTI system
defined in each of the intervals $\left([t_i,t_{i+1}]\right)_{i=0}^{k-1}$.

For design simplicity, the time intervals can be taken as fixed equidistant intervals. Then, a Nash equilibrium is solved at the beginning of each interval, based on the instantaneous approximated linear model and the most relevant objectives. Note that this arraignment allows changing the game dynamically when ever needed, and the considered time intervals, where the game is updated, are case-dependent and not necessarily equidistant. 
\end{definition}

\section{Python Implementation}

Along with the formal foundations, we put in effort in implementing the  approach we suggested as a Python package. The implementation of the case studies that will follow in Chapter~\ref{chap:case_studies} uses a Python package we designed specifically for this purpose - called \textit{PyDiffGame}. The package allows the user to set up a multi-objective control problem and solve it for several scenarios. The main class of the package is given at Appendix~\ref{pyth:main}.

\subsection{Differential Game Formulation}
To solve a differential game using the package, the user must specify the required parameters to define the problem in a concise manner. After the problem is formulated with all the required parameters, a corresponding Python class\footnote{More on Object Oriented Programming in Python at~\cite{oop_python}.} is constructed called \pyth{PyDiffGame}. The class keeps all the system parameters of the differential game described in~\eqref{eqn:optimal_nash_control_problem_formulation} and allows the user to '\textit{solve the game}', i.e., solving the required equations to simulate the controller for the system based on the prescribed objectives and corresponding differential game. Then the user can simulate the dynamic behaviour of the described system under the control objectives detailed in the input and finally the temporal behaviour of the system can be plotted with respect to time by '\textit{solving the state space}'.

Thus far we have implemented the \pyth{PyDiffGame} base class that keeps all the required information for the formulation of a differential game according to the theory described. This class is defined to be \textbf{abstract}\footnote{An abstract class is a class that one cannot instantiate an object of, due to it not containing all the required methods or attributes to fully construct an object. One must implement a sub-class that provides the missing assets to the abstract class to then instantiate an object which is \textbf{inherited} from the abstract class.} as it is required to define if the system to deal with is modeled by a continuous or discrete model.

Consequentially, we have written two subclasses, one for the modeling of continuous systems, called \pyth{ContinuousPyDiffGame} and one for direct discrete control systems, called \pyth{DiscretePyDiffGame}, which include the design specifications described earlier per situation.

\subsection{Package Input}
A \pyth{PyDiffGame} class prescribed input is the following matrices that were described in the formulation of the Nash Differential Game at~\eqref{eqn:optimal_nash_control_problem_formulation}: 
\begin{itemize}
    \item \pyth{A}: \pyth{numpy} 2-d \pyth{array} of shape $(n, n)$ - the system dynamics matrix;
    
    \item \pyth{B}: \pyth{list} of \pyth{numpy} 2-d \pyth{array}s of length $N$, each matrix \pyth{B_i} of shape $(n, m_i)$ - virtual input coefficients matrices for each control objective;
    
    \item \pyth{Q}: \pyth{list} of \pyth{numpy} 2-d \pyth{array}s of length $N$, each matrix \pyth{Q_i} of shape $(n, n)$ - cost function state weights for each control objective;
    
    \item \pyth{R}: \pyth{list} of \pyth{numpy} 2-d \pyth{array}s of length $N$, each matrix \pyth{R_i} of shape $(m_i, m_i)$ - cost function input weights for each control objective\footnote{we assumed $R_{ij} \equiv 0_{m_j \times m_j}$ for all $i \neq j$ for simplicity.};
    
    \item \pyth{x_0}: \pyth{numpy} 1-d \pyth{array} of shape $n$, optional - initial state vector;
    
    \item \pyth{x_T}: \pyth{numpy} 1-d \pyth{array} of shape $n$, optional - final state vector, in case of signal tracking;
    
    \item \pyth{T_f}: positive \pyth{float}, optional, default = \pyth{2} - system dynamics horizon, should be given in the case of finite horizon;
    
    \item \pyth{P_f}: \pyth{list} of \pyth{numpy} 2-d \pyth{array}s of length $N$, each matrix \pyth{P_f_i} of shape $(n, n)$, optional, default = uncoupled solution of \pyth{scipy}'s \pyth{solve_are} - final condition for the Riccati equation matrix. Should be given in the case of finite horizon;
    
    \item \pyth{show_legend}: \pyth{bool}, optional, default = \pyth{True} - indicates whether to display a legend in the plots (\pyth{True}) or not (\pyth{False});
    
    \item \pyth{epsilon}: \pyth{float}, optional, default = \pyth{1 / (10 ** 7)} - the convergence threshold for numerical convergence;
    
    \item \pyth{L}: \pyth{int}, optional, default = \pyth{1000} - number of data points;
    
    \item \pyth{eta}: \pyth{int}, optional, default = \pyth{5} - number of last matrix norms to consider for convergence;
    
    \item \pyth{last_norms_number}: \pyth{int}, optional, default = \pyth{5} - the number of last matrix norms to consider for convergence;
    
    \item \pyth{debug}: \pyth{bool}, optional, default = \pyth{False} - indicates whether to display debug information or not;
\end{itemize}

\chapter{Novel AREs Solution Method}
\label{chap:riccati}
\epigraph{\textit{Divide your movements into easy-to-do sections. If you fail, divide again.}}{\textbf{Peter Nivio Zarlenga}}

The solution we propose for solving Algebraic Riccati Equations (AREs) that arise in infinite-horizon optimal control problems case takes advantage of the mechanism utilized to solve differential equations, that arise in finite-horizon optimal control problems. We will discuss the formalism under a continuous time system, and leave the work on discrete time for future work.
 
 In some cases, attempting to solve a set of algebraic Riccati equations can result in multiple solutions, as described in detail at~\ref{sect:single_CARE_solution_num}. Moreover, only one of those is stabilizing for the underlying control problem, as described in detail at~\ref{sect:stabalizing_care}.
 
 We propose to convert a set of CAREs, such as described in~\eqref{eqn:riccati}, in into their corresponding differential form, namely DREs, such as described in~\eqref{eqn:CDRE}. By doing so, we wish to consider cases where it is reasonable to invoke an assumption that the solution of differential Riccati equations tends to converge to the solution of the algebraic Riccati equations when going backward in time, which is not trivial, and requires several conditions to be satisfied.
 
  By considering the problem for a fixed finite time interval, for which we assume the temporal dynamics of the system will stabilize and solving with a pre-defined terminal condition, we end up with a Cauchy Initial Value Problem problem\footnote{See more on Cauchy Initial Value Problem at Definition~\ref{def:civp}.}, which has one unique solution under certain assumptions\footnote{This Theorem is given at Theorem~\ref{the:picard_lind}. For more information, see~\cite{existence}.}. Let us then discuss the procedure of solving the corresponding CDRE.
 
 
\section{CDRE Solution Properties}
\label{app:solving_the_CDRE}
Let us consider a single CDRE, as given at~\eqref{eqn:CDRE}. Since it is a matrix equation, in order to solve one CDRE is to actually solve a set of differential equations. Theorem~\ref{the:picard_lind} guarantees only local existence
and uniqueness for the solution of a given initial value problem, so we need further justification to conclude existence and uniqueness in a given closed time interval $\mathcal{I} \equiv [t_0, T_f]$. Luckily, we have the following already proven theorem:

\begin{theorem}[\textbf{CDRE Global Existence and Uniqueness}, Theorem 8 from~\cite{riccati_review}]
Let us consider the CDRE at~\eqref{eqn:CDRE} on a given closed interval $\mathcal{I}$. Then given any terminal condition $Q_f$, there exists a unique solution on the interval, denoted $P_{\mathcal{I},Q_f} (t)$.
\label{the:global_CDRE_existence}
\end{theorem}

\begin{remark}
As stated in~\cite{riccati_review}, what Theorem~\ref{the:global_CDRE_existence} actually does is guarantee existence of a Lipschitz constant $\varepsilon$ as required in Theorem~\ref{the:picard_lind}, on any closed interval $\mathcal{I}$, no matter how large. Thus we have global existence and uniqueness for $P(t)$, regardless of the given terminal condition and interval, which we call $P_{\mathcal{I},Q_f} (t)$.
\end{remark}

So this last result guarantees we have a unique solution to the CDRE. This is already a very exciting and perhaps surprising result, as the presumably more complex set of differential equations is guaranteed to yield a result, and not only that, but also a unique one, while neither promise is granted in the algebraic case.

Let us now discuss the properties of this unique solution. We mentioned we are interested only in the non-negative solutions that satisfy the Riccati equation. One of the most important properties of the solutions to the CDRE is a result of the following theorem:

\begin{theorem}[\textbf{CDRE Positive Semi-Definiteness}, Theorem 9 from~\cite{riccati_review}]
Let us consider the CDRE at~\eqref{eqn:CDRE} on a given closed interval $\mathcal{I}$ along with a terminal condition $Q_f \geq 0$. Let $P_{\mathcal{I},Q_f} (t)$ be the corresponding unique solution to the CARE. Then $P_{\mathcal{I},Q_f} (t) \geq 0$.
\end{theorem}

So not only do we have global existence and uniqueness, but given a positive semi-definite terminal condition, we are also guaranteed positive semi-definiteness of the solution. It's almost like a dream come true. The thing is, we cannot guarantee that this unique semi-definite solution is stabilizing without additional requirements, though let us recall that we are operating on a \textbf{finite} interval and thus the asymptotic stability of the solution does not matter to us, as we assume the system to never reach infinity.

To conclude, while solving the CDRE on a finite time interval, we are guaranteed a unique, non-negative solution. That is fine, but now we wish to see how does this solution correlate to the solution of the CARE we started out with, i.e. when we let $T_f \rightarrow \infty$.

\section{Asymptotic Behaviour Of The CDRE}
Let us now consider the asymptotic behaviour of the solutions to the CDRE, as in when letting $T_f \rightarrow \infty$\footnote{See page 52 of~\cite{riccati_review} under 'Asymptotic Behaviour Of The Solution' to be convinced that letting $T_f \rightarrow \infty$ is equivalent, by a change of variables, to letting $t_0 \rightarrow - \infty$ and then consider the interval $(-\infty, T_f] $ instead.}. One interesting thing to ask is whether the solutions can diverge in the limit. The following theorem provides us a cheerful answer:


\begin{theorem}[\textbf{Global CDRE Solution Finiteness}, Theorem 11 from~\cite{riccati_review}]
Let us consider the CDRE at~\eqref{eqn:CDRE} on either a given closed interval $\mathcal{I} \equiv [t_0, T_f]$ or a semi-open interval $\mathcal{I} \equiv [t_0, T_f)$ with $T_f \rightarrow \infty$. If the pair $(A,B)$ is stabilizable, then for any given terminal condition $Q_f \in \mathbb{R}^{n \times n}$, any solution to the CDRE is finite on the entire interval,  i.e. $\forall t \in \mathcal{I} \ \colon \ P_{\mathcal{I},Q_f} (t) < \infty$.
\end{theorem}

We are then guaranteed that, given the standard condition of the stabilizability of $(A,B)$, any limiting solution will not diverge, but do we still have uniqueness? To make a long story short (if one may, Theorems 12-16 of~\cite{riccati_review} are fascinating in demonstrating this result), let's jump to the main event of our conclusions:

\begin{theorem}[\textbf{CDRE Limiting Solution Existence, Uniqueness and Stability}, Theorem 17 from~\cite{riccati_review}]
Let us consider the CDRE at~\eqref{eqn:CDRE} on a given semi-open interval $\mathcal{I} \equiv [t_0, T_f)$ with $T_f \rightarrow \infty$. If the pair $(A, B)$ is stabilizable and the pair $(A, \sqrt{Q})$ is detectable then for any terminal condition $Q_f \geq 0$ we have:
\begin{equation}
    \lim_{T_f \rightarrow \infty} P_{\mathcal{I},Q_f} (t)= P_{\infty}
\end{equation}
where $P_{\infty} \in \mathbb{R}^{n \times n}$ is a constant, finite, non-negative matrix that is the unique stabilizing solution to the corresponding CARE at~\eqref{eqn:riccati} resulting from dropping the derivative term from the original CDRE.
\end{theorem}

This is very exciting indeed, as we have just discovered that when considering the CARE for large systems ($n \ll 100$), even with the pair $(A, B)$ stabilizable and the pair $(A, \sqrt{Q})$ detectable, we may face serious difficulties when trying to converge on that unique solution. So in that case, the corresponding CDRE comes to our aid, and with any terminal condition, given enough time to converge will yield the desired unique stabilizing solution.
 
 \begin{definition}
 Let us consider the CARE at~\eqref{eqn:riccati}. Assuming they exist, let $P_o \in \mathbb{R}^{n \times n}$ be the optimizing solution and $P_s \in \mathbb{R}^{n \times n}$ be the stabilizing solution to the CARE.
 \end{definition}
  With that, let us consider the following two tables from the last page of~\cite{riccati_review} that summarize the behaviour of the solutions of the CARE and the CDRE:
  
 \begin{center}
 \begin{table}[H]
  \label{tbl:are}
  \centering
  \includegraphics[width=0.65\linewidth]{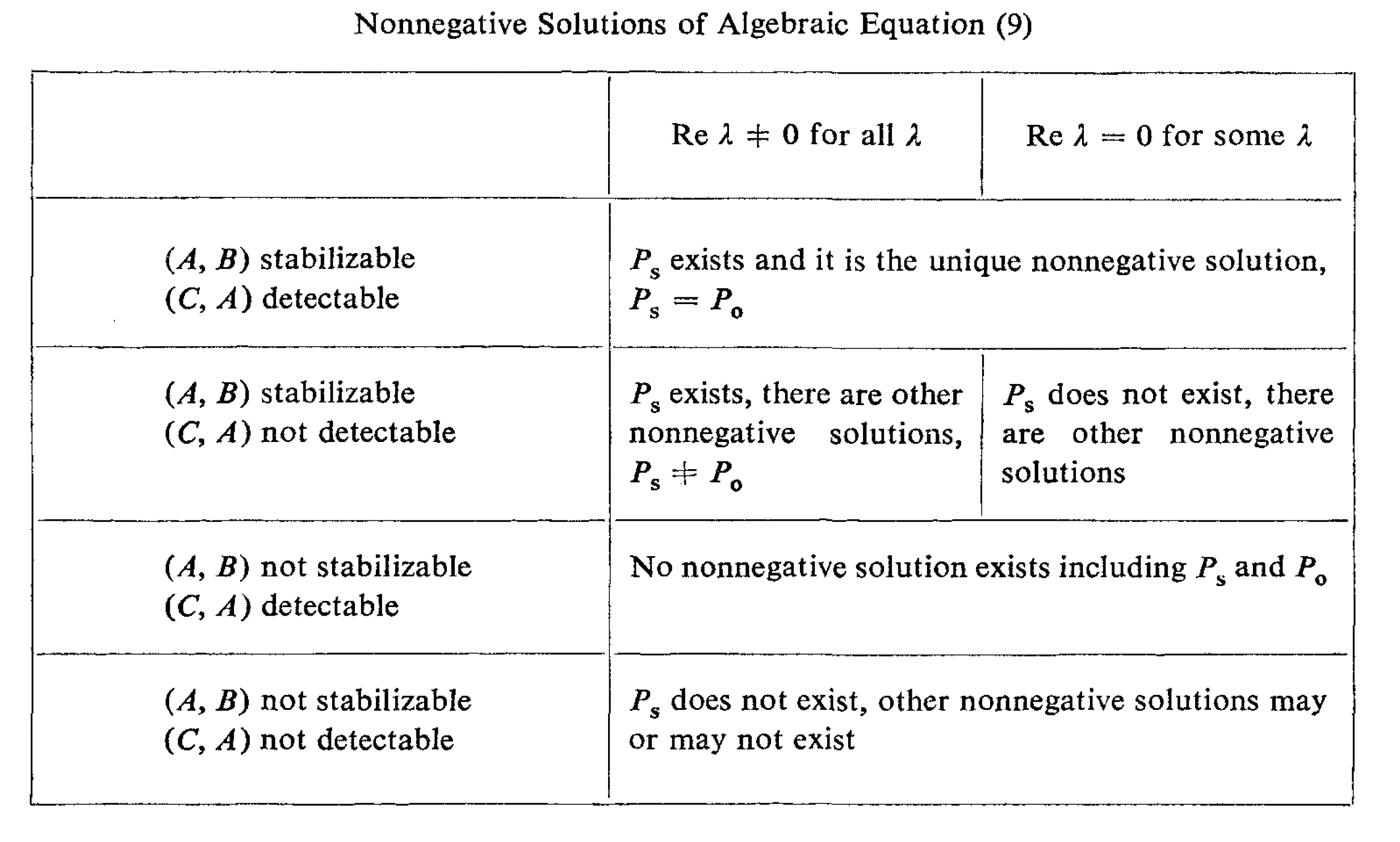}
  \caption{Properties of the CARE under different situations for $(A,B)$ and $(A,C)$}
\end{table}
\end{center}

\begin{center}
\begin{table}[H]
  \label{tbl:dre}
  \centering
  \includegraphics[width=0.65\linewidth]{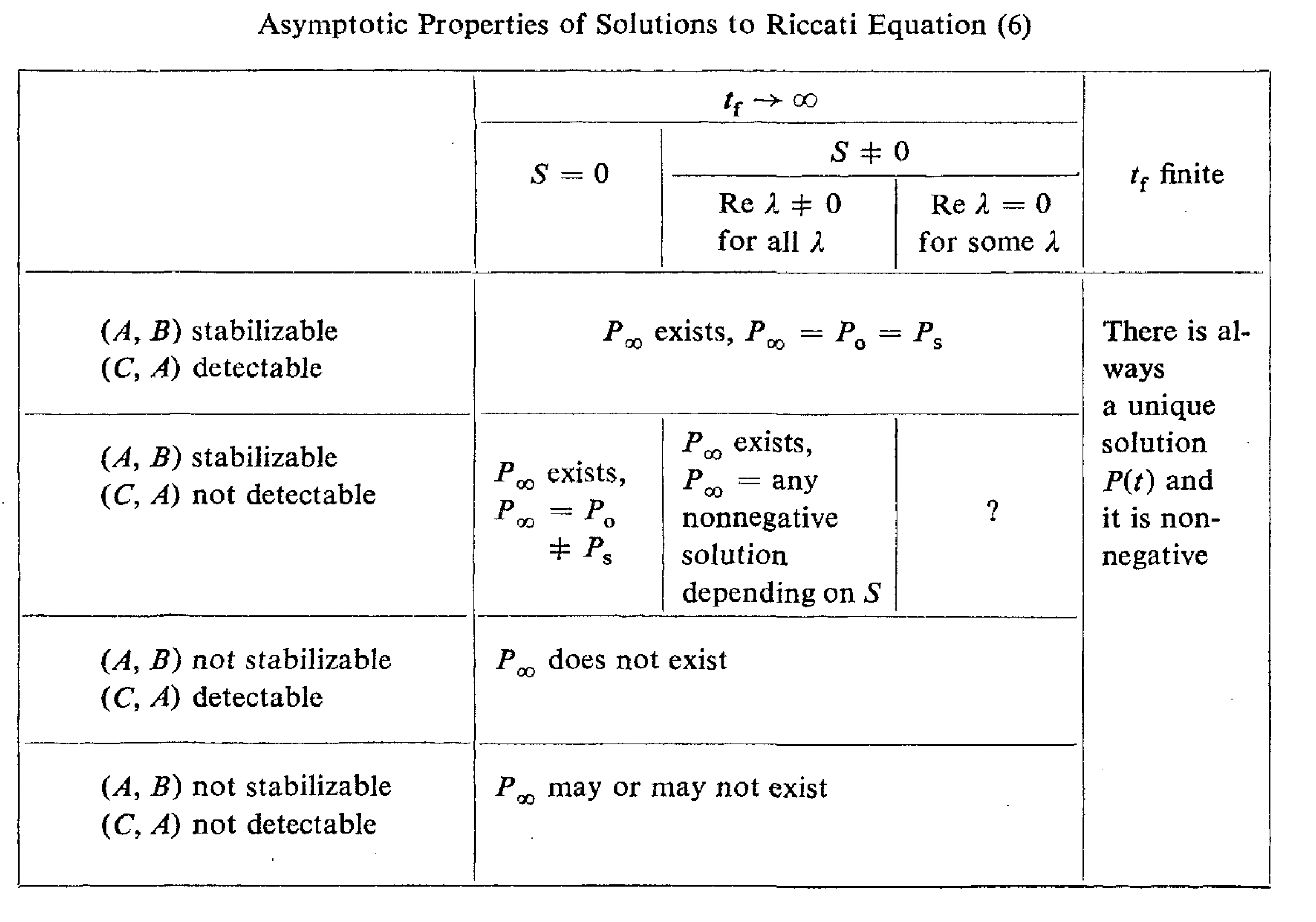}
  \caption{Properties of the CDRE under different situations for $(A,B)$ and $(A,C)$}
\end{table}
\end{center}
 \section{\texttt{PyDiffARE} Algorithm}

As we already laid out, in the case of infinite horizon, the number of solutions to the CARE can grow exponentially with the length of the state vector $n$, the set of GCAREs at \eqref{eqn:GCARE} can grow exponentially with $n$ multiplied by the number of objectives $N$, and in both cases there can be infinite and even uncountable many solutions. This can cause serious difficulties when trying to solve for the unique stabilizing solution, especially in large scale system and with multiple objectives to consider. 

We now suggest a novel approach to help with this matter; instead of solving the set of algebraic equations while hoping to arrive at the desired solution, we propose to transform it to a set of sets of differential equations, each set for each objective, assign appropriate terminal values and iteratively solve them backwards until convergence is reached. 

In this section we will provide the full algorithm, written in psuedocode and the python implementation appears in 

Let us consider the following:
\begin{enumerate}
    \item Our algorithm considers the following parameters:
\begin{itemize}
        \item $A \in \mathbb{R}^{n \times n}$ - the dynamics matrix;
        \item $B \in \mathbb{R}^{m \times n}$ - the input matrix;
        \item $Q \in \mathbb{R}^{n \times n}$ - the state weights, such that  $Q \geq 0$;
        \item $R \in \mathbb{R}^{m \times m}$ - the input weights, such that  $R > 0$.
    \item $T_f \in \mathcal{I}$ - the cut-off time;
    \item $L \in \mathbb{N}$ - the number of sampling points.
    \item $\varepsilon \in (0,1)$ - the convergence threshold;
    \item $\eta \in \mathbb{N}$ - the number of elements to consider for convergence;
    \item $\mathbf{x_0} \in \mathbb{R}^n$ - the initial state vector (optional)
    \item $\mathbf{x_T} \in \mathbb{R}^n$ - the terminal state vector (optional)
\end{itemize}
\item With these, let us define the \textbf{sampling time} $\delta \in \mathbb{R}^{+}$ and \textbf{sampling points} $T \in \mathbb{R}^{{+}^L}$ respectively as:
\begin{equation}
    \begin{aligned}
    \delta \coloneqq \frac{T_f}{L} \ ; \ T \coloneqq (j \delta)_{j=1}^L
    \end{aligned}
\end{equation}
\item Let \pyth{CDRE(A, B, Q, R, P_f, epsilon)} be a function such that:
\begin{itemize}
    \item \pyth{CDRE} receives the matrices $A, B, Q, R$, a terminal condition $P_f \in \mathbb{R}^{n \times n}$, such that $P_f \geq 0$, and a series of strictly decreasing time points $T \coloneqq (t_j)_{j=1}^L$ in the the interval $[t_0, T_f]$, i.e. they satisfy $\forall 1 \leq j \leq L -1  \ \colon \ t_0 \leq t_{j+1} < t_j \leq T_f$;
    \item \pyth{CDRE} returns the solution to the corresponding CDRE at~\ref{eqn:CDRE} at the time points of $T$, denoted $P_T \coloneqq \left(P(t)\right)_{t \in T}$ with a tolerance of $\varepsilon$;
    \item Let \pyth{simulate_x_T(T)} be a functions that:
    \begin{itemize}
        \item \pyth{simulate_x_T(T)} receives the backwards time $T$, reverses it and solves the state model to obtain $X_T \coloneqq \left(\mathbf{x}(t)\right)_{t \in T}$;
        \item \pyth{simulate_x_T(T)} returns $X_T[-1]$, which is the simulated last state variable, denoted $\mathbf{\tilde{x}_T}$.
    \end{itemize}
    In considering the function for the CDRE, there is a huge amount of known solvers, all trying to compete with regards to runtime. Here is an example for such an algorithm that solves CDREs in $O(n^3)$~\cite{BennerPeter2013RMfS} and even proclaims to do so in the future at $O(n^2)$.
\end{itemize}
\end{enumerate}

Using these, let us then define the following \pyth{PyDiffARE} algorithm:

\RestyleAlgo{ruled}
\SetKwComment{Comment}{/* }{ */}
\begin{algorithm}[H]
\caption{PyDiffARE}\label{alg:CARE_to_CDRE}
    \KwIn{$A, B, Q, R, T_f, L, 
\varepsilon, \eta, \mathbf{x_0} , \mathbf{x_T} \ \textrm{(optional, default } \mathbf{0_n})$}
    x\_converged $\gets$ False, P\_converged $\gets$ False\; 
    norms $\gets [ \ ]$\;
    $T_{f_i} \gets T_f, P_f \gets None$\;
    \While{x\_converged is False}{
        \While{P\_converged is False}{
        $\delta \gets T_{f} / L$\;
        $T$ $\gets [ \ j \delta \ ]_{j=L}^1 = [ \ L \delta, \ (L-1)\delta , \ (L-2)\delta, \ \cdots , \ \delta \ ]$\;  
        $P_T \gets $ CDRE $(A,B,Q,R, P_f, T)$\;
        $P_f \gets P_T[-1]$\;
        norms.insert($\|P_0\|$)\;
        \If{length of norms $> \eta$}{
                norms.pop(0)\;
                }
        \If{length of norms = $\eta$}{
        P\_converged $\gets$ True\;
        \For{$i \leftarrow 0$ \KwTo length of norms - 2 }{
        \If{$\| $norms$[i+1]$ $-$ norms$[i]$$\| \geq \varepsilon$}{converged $\gets$ False \;}
        }
        }
        $T_f \gets T_f - \delta$\;
        }
        $\mathbf{\tilde{x}_T} \gets$ simulate\_x\_T $(T)$\;
        \If{$\|\mathbf{\tilde{x}_T} - \mathbf{x_T} \| < \varepsilon$}{x\_converged $\gets$ True\;}
        $T_{f_i} \gets T_{f_i} + 1$\;
        $T_f \gets T_{f_i}$\;
        $\delta \gets T_{f} / L$\;
        $T$ $\gets [ \ j \delta \ ]_{j=L}^1 = [ \ L \delta, \ (L-1)\delta , \ (L-2)\delta, \ \cdots , \ \delta \ ]$\;
    }
    \KwRet{$P_f$}\;
    \label{alg:pydiffare}
\end{algorithm}

\newpage

\chapter{Case Studies}
\label{chap:case_studies}
\epigraph{\textit{Nothing is particularly hard if you divide it into small jobs.}}{\textbf{Henry Ford}}

To illustrate the methods we proposed, we implemented it on the two following, systems which we found to be well suited for showcasing the suggested D\&C approach and infinite horizon solution technique:

\begin{enumerate}
    \item An \textbf{inverted pendulum on a moving cart} - a well known, simple model, used as a '\textit{hello world}' tutorial-like introductory use case. This use case will be fully examined from top to bottom, according to the following steps:
    \begin{itemize}
        \item First, we will use full rigor to define the non-linear system model;
        \item Then we will perform linearization to define the LTI state space model;
        \item Afterwards the described decomposition technique will be implemented upon the system according to pre-defined virtual objectives and corresponding inputs, relevant to this use case; 
        \item Then the case of continuous infinite-horizon cost functions will be considered, and the performance of a standard LQR will be compared with that of a multiplayer game according to the defined virtual objectives;
    \end{itemize}
    \item A non-linear \textbf{quadrotor system} - an aerial vehicle with four engines that is widely used and studied. We will use this model to demonstrate how dynamically changing objectives can be taken into account using our method. This use case is the same one we used in the article we authored, and so it has a full description which we already covered, some of it in the article and some on previous works we mentioned, but will still display in this work in a concise manner.
\end{enumerate}
In this section we will lay out the specific models implemented for both systems and the results acquired.

\section{Inverted Pendulum on a Moving Cart}
An \textbf{inverted pendulum} is a well-known two-dimensional system in classical mechanics that has been solved by various means. It it comprised of a cart moving along a horizontal axis, attached to a rod pendulum, free to rotate about the point of attachment to the cart. The main objective associated with this system is to stabilize the pendulum perpendicular to the cart, so that even if the cart moves, the pendulum is still situated at a right angle with respect to it. 

The input to this system is usually a \textbf{linear force} along the movement axis of the cart, so the control is required to alter this force direction and magnitude to stabilize the pendulum.
To illustrate our method, we will also add an input at the form of a \textbf{pure torque} applied directly to the pendulum. In doing so, we will assign the task of stabilizing the pendulum to that torque. 

In order to consider a balance between multiple objectives, we will add an additional objective in the form of moving the cart to some designated point along its axis, and assign the linear force to doing so.

\subsection{System Modeling}
Let us denote the system as $\mathcal{S}_{IP}$ where $IP$ stands for inverted pendulum. Suppose the cart has a mass of $m_c \ [kg] \in \mathbb{R}^+$ and the pendulum attached to it has a mass of $m_p \ [kg] \in \mathbb{R}^+$. Also suppose the pendulum has uniform mass, a length of $L \ [m] \in \mathbb{R}^+$ and moment of inertia of $I \ [kg \cdot m^2] \in \mathbb{R}^+$.

\subsubsection{State Vector}
Let us define the state vector of the system. Let $t_0=0 \ [s]$ and let $x(t) \ [m] \in \mathcal{X}_x$ be the position of the cart, (where $\mathcal{X}_x \subseteq \mathbb{R}$ is the interval along the $x$-axis which the cart can travel by) defined and continuously differentiable for all time $t \ [s] \in \mathcal{I} = \mathbb{R}^{\geq t_0}= \mathbb{R}^+ \cup \{ 0 \}$.  With that, the time-derivative of $x(t)$, namely $\dot{x}(t) \ [\frac{m}{s}] \in \mathcal{X}_v$ is the velocity of the cart at time $t$ (where $\mathcal{X}_v \subseteq \mathbb{R}$ is the set of all the possible velocities of the cart).
Let $\theta(t) \ [rad] \in \mathcal{X}_\theta$ be the offset angle of the pendulum from the vertical axis, (where $\mathcal{X}_\theta = [0, 2 \pi]$ is all the possible values for the angle) and so $\dot{\theta}(t) \ [ \frac{rad}{s} ] \in \mathcal{X}_\omega$ is the angular velocity at time $t$ (where $\mathcal{X}_\omega$ is the set of all possible angular velocities). Let us denote the state vector of $\mathcal{S}_{IP}$ as $\mathbf{x}(t) \colon \mathbb{R}^+ \cup \{0\} \rightarrow \mathcal{X}$, where $\mathcal{X} \coloneqq \bigtimes_{q \in \{ x, \theta, v, \omega \}} \mathcal{X}_q$, with $\mathcal{X} \subseteq \mathbb{R}^n = \mathbb{R}^4$, as $n \equiv 4$ in this case. Let then define $\mathbf{x}(t)$ as such:

\begin{equation}
    \mathbf{x}(t) \coloneqq \begin{bmatrix}
        x(t) \\
        \theta(t) \\
        \dot{x}(t) \\
        \dot{\theta}(t)
    \end{bmatrix}.
\end{equation}

Let $\mathbf{x_0} \coloneqq \begin{bmatrix}
        x_0 &
        \theta_0 &
        \dot{x}_0 &
        \dot{\theta}_0
    \end{bmatrix}^T$ be the initial value for the state vector.

\subsubsection{System Input}
As for the system input, let us assume there are two distinct inputs applied to $\mathcal{S}_{IP}$:
\begin{itemize}
    \item A linear force $F(t) \ [N] \in \mathcal{U}_F$ applied to the cart along the direction of the $x$ coordinate (where $\mathcal{U}_F \subseteq \mathbb{R}$ is the set of all possible values for $F$);
    \item A pure torque $M(t) \ [N \cdot m] \in \mathcal{U}_M$ applied to the pendulum along the direction of the $\theta$ coordinate (where $\mathcal{U}_M \subseteq \mathbb{R}$ is the set of all possible values for $M$).
\end{itemize}

Let us define the input vector $\mathbf{u}(t) \in \bigtimes_{q \in \{F, M\}} \mathcal{U}_q$ of $\mathcal{S}_{IP}$ as:

\begin{equation}
    \mathbf{u}(t) \coloneqq \begin{bmatrix}
        F(t) \\
        M(t)
    \end{bmatrix}
\end{equation}

\subsubsection{System Diagram}
Let us consider the following diagram to visualize $\mathcal{S}_{IP}$:

\begin{figure}[H]

    \centering
    \begin{tikzpicture} [thick]
\newcommand{\ang}{30}
\draw [brown!80!red] (-2,0) -- (2,0);
\fill [pattern = crosshatch dots,
    pattern color = brown!80!red] (-2,0) rectangle (2,-.2);
\begin{scope} [draw = black,
    fill = blue!20, 
    dot/.style = {black, radius = .025}]
\filldraw [rotate around = {-\ang:(0,1.5)}] (.09,1.5) -- 
    node [very near end, right] {$m_p,L, I$}
    +(0,2) arc (0:180:.09) 
    coordinate [pos = .5] (T) -- (-.09,1.5);
\filldraw (-.65,.15) circle (.15);
\fill [dot] (-.65,.15) circle;
\filldraw (.65,.15) circle (.15);
\fill [dot] (.65,.15) circle;
\filldraw (-1,1.5) -- coordinate [pos = .5] (F)
    (-1,.3) -- node [above = .3cm] {$m_c$}
    (1,.3) -- (1,1.5) 
    coordinate (X) -- (.1,1.5)
    arc (0:180:.1) -- (-1.014,1.5);
\fill [dot] (0,1.52) circle;
\end{scope}
    \draw (T) -- (0,1.52) coordinate (P);
    \draw [-stealth] (P) + (0,-0.5) arc (-90:50:0.8) node [black, near end, right] {$\theta(t)$};
    \draw [stealth-] (P) + (-1.1,1.3) arc (90:90-\ang:3) node [black, near start, left, above] {$M(t)$} ;
    \draw [stealth-] (F) -- node [above] {$F(t)$} + (-1,0);
    \draw [dashed,->] (P) + (1,-0.02) -- node [near end, above right] {$x_c$} + (2.5,-0.02);
    \draw [-stealth] (X) + (0,-.5) -- node [near end, below right] {$x(t)$} + (1,-.5);
    \draw [dashed,->] (P) + (0,-0.5) |- node [near end, above right] {$y$} +(0,2.2);
\end{tikzpicture}
    \caption{A description of the system $\mathcal{S}_{IP}$.}
    \label{fig:S_IP}
\end{figure}
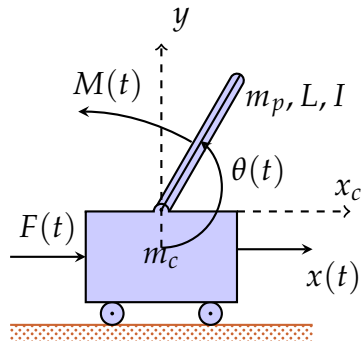

where the origin of the dynamic coordinate system $\begin{bmatrix}
        x_c(t) &
        y(t) &
        z(t)
    \end{bmatrix}^T \in \mathbb{R}^3$ is situated at the point of attachment of the cart and pendulum, with an offset of $x(t)$ from a fixed point $\mathbf{r}_0 \coloneqq \begin{bmatrix}
        x_0 &
        y_0 &
        z_0
    \end{bmatrix}^T \in \mathbb{R}^3$ situated at the cart's initial horizontal position, so that $x_c(t) = x_0 + x(t)$, all respect to the three dimensional origin $\mathbf{0_3} \in \mathbb{R}^3$.

\subsubsection{System Physical Model}

To derive the physical model of the pendulum, we will employ principles from analytical mechanics, specifically the Lagrangian mechanics formalism.~\cite{lagrangian} Let us define the Lagrangian of $\mathcal{S}_{IP}$:

\begin{equation}
    \mathcal{L}\big(\mathbf{x}(t)\big) = \mathcal{L}\Big(x(t), \theta (t), \dot{x}(t), \dot{\theta} (t) \Big) = T\big(\mathbf{x}(t)\big) - V\big(\mathbf{x}(t)\big)
    \label{eqn:L}
\end{equation}



where $\mathcal{L}$ is the Lagrangian and $T$ and $V$ are the kinetic and potential energies of $\mathcal{S}_{IP}$, respectively, all of which are functions from $\mathcal{X}$ to $\mathbb{R}$. In our case, $\mathcal{S}_{IP}$ can be considered a system of two bodies - a point mass $m_c$ which is the cart and a rigid body which is the pendulum, of mass $m_p$ and moment of inertia $I$ relative to its center of mass. Let us assume the coordinate system situated at the initial position of the cart $\mathbf{r}_0$ has unit vectors $\hat{\mathbf{i}} \coloneqq \begin{bmatrix}
        1 &
        0 &
        0
    \end{bmatrix}^T$ and $\hat{\mathbf{j}} \coloneqq \begin{bmatrix}
        0 &
        1 &
        0
    \end{bmatrix}^T$. Note that since that coordinate system is fixed, the derivative of these unit vectors with respect to time are zero. With respect to the fixed coordinate system at $\mathbf{r}_0$, let us consider the kinetic energy of $\mathcal{S}_{IP}$:

\begin{equation}
\begin{aligned}
    T\big(\mathbf{x}(t)\big)
    & = \frac{1}{2} m_c || \dot{\mathbf{r}}_c\big(\mathbf{x}(t)\big)||^2 + \frac{1}{2} m_p  ||\dot{\mathbf{r}}_p\big(\mathbf{x}(t)\big)||^2 \\
    & + \frac{1}{2} \boldsymbol{\omega}\big(\mathbf{x}(t)\big) \cdot \mathbf{H}\big(\mathbf{x}(t)\big)
\end{aligned}
\label{eqn:kinetic_energy}
\end{equation}

where $\mathbf{r}_c\big(\mathbf{x}(t)\big) \in \mathbb{R}^3$ is the positional vector of the cart, $\mathbf{r}_p\big(\mathbf{x}(t)\big) \in \mathbb{R}^3$ is the positional vector to the center of mass of the pendulum, $\boldsymbol{\omega}\big(\mathbf{x}(t)\big) \in \mathbb{R}^3$ is the angular velocity vector of the pendulum and $\mathbf{H}\big(\mathbf{x}(t)\big) \in \mathbb{R}^3$ is the angular momentum vector of the pendulum. Denoting $l \coloneqq \frac{L}{2}$, let us write expressions for the positional vectors:

\begin{equation}
\begin{aligned}
    \mathbf{r}_c\big(\mathbf{x}(t)\big) & = x(t) \ \hat{\mathbf{i}};\\
    \mathbf{r}_p\big(\mathbf{x}(t)\big) & = \Big( x(t) + l \sin \big( \theta (t) \big) \Big) \ \hat{\mathbf{i}} + l \cos \big( \theta (t) \big) \ \hat{\mathbf{j}}.
\end{aligned}
\label{eqn:positions}
\end{equation}

Differentiating, we have the linear velocity vectors, while using the fact that the unit vectors are time-independent:

\begin{equation}
\begin{aligned}
    \dot{\mathbf{r}}_c\big(\mathbf{x}(t)\big) & = \dot{x}(t) \ \hat{\mathbf{i}};\\
    \dot{\mathbf{r}}_p\big(\mathbf{x}(t)\big) & = \Big( \dot{x}(t) + l \cos \big( \theta (t) \big)\dot{\theta} (t) \Big) \ \hat{\mathbf{i}} - l \sin \big( \theta (t) \big)\dot{\theta} (t) \ \hat{\mathbf{j}}.
\end{aligned}
\end{equation}

Taking the the squared norms of these vectors:

\begin{equation}
\begin{aligned}
    ||\dot{\mathbf{r}}_c\big(\mathbf{x}(t)\big)||^2 & = \dot{x}^2(t) \ ;\\
    ||\dot{\mathbf{r}}_p\big(\mathbf{x}(t)\big)||^2 & = \Big( \dot{x}(t) + l \cos \big( \theta (t) \big)\dot{\theta} (t) \Big)^2 + \Big(-l \sin \big( \theta (t) \big)\dot{\theta}(t) \Big)^2\\
    & = \dot{x}^2(t) + 2 l \dot{x}(t)\cos \big( \theta (t) \big)\dot{\theta} (t) + l^2 \cos^2 \big( \theta (t) \big)\dot{\theta}^2 (t)\\
    & + l^2 \sin^2 \big( \theta (t) \big)\dot{\theta}^2 (t)\\
    & = \dot{x}^2(t) + 2 l \dot{x}(t)\cos \big( \theta (t) \big)\dot{\theta} (t) + l^2 \dot{\theta}^2 (t).
\end{aligned}
\label{eqn:velocities_squared_norms}
\end{equation}

Let us assume the $z$ axis perpendicular to the  plane has a unit vector $\hat{\mathbf{k}} \coloneqq \begin{bmatrix}
        0 &
        0 &
        1
    \end{bmatrix}^T$. Thus the angular velocity and momentum vectors can be written as:

\begin{equation}
\begin{aligned}
    \boldsymbol{\omega}\big(\mathbf{x}(t)\big) & = \dot{\theta} (t) \hat{\mathbf{k}};\\
    \mathbf{H}\big(\mathbf{x}(t)\big) & = I \dot{\theta} (t) \hat{\mathbf{k}}.
\end{aligned}
\end{equation}

So the dot product of these two is:

\begin{equation}
\begin{aligned}
    \boldsymbol{\omega}\big(\mathbf{x}(t)\big) \cdot \mathbf{H}\big(\mathbf{x}(t)\big) = I \dot{\theta}^2 (t).
\end{aligned}
\label{eqn:angular_dot}
\end{equation}

Plugging \eqref{eqn:velocities_squared_norms} and \eqref{eqn:angular_dot} to \eqref{eqn:kinetic_energy} we have:

\begin{equation}
\begin{aligned}
     T\big(\mathbf{x}(t)\big)   & = \frac{1}{2} m_c \dot{x}^2(t) \\
    & + \frac{1}{2} m_p \Big(\dot{x}^2(t)  + 2 l \dot{x}(t)\cos \big( \theta (t) \big)\dot{\theta} (t) + l^2 \dot{\theta}^2 (t)\Big) \\
    & + \frac{1}{2} I \dot{\theta}^2 (t).
    \label{eqn:T}
\end{aligned}
\end{equation}

Now let us derive the potential energy of $\mathcal{S}_{IP}$. To do so, we'll use as reference the horizontal line that goes through the point where the cart and pendulum are attached at. So the point of reference $\mathbf{r}_{ref}$ satisfies:
\begin{equation}
    \mathbf{r}_{ref} = \mathbf{r}_c.
\end{equation}
With respect to it, the only conservative force acting upon $\mathcal{S}_{IP}$ is the gravitational force acting on the pendulum, and thus:

\begin{equation}
    \mathbf{F}_{C} =  \mathbf{F}_{{mg}_p} = - m_p g \  \hat{\mathbf{j}}
\end{equation}

where $\mathbf{F}_{C}$ is the total conservative force acting upon the system and  $\mathbf{F}_{{mg}_p}$ is the gravitational force acting on the pendulum. The potential energy is by definition:

\begin{equation}
\begin{aligned}
    V\big(\mathbf{x}(t)\big) & =   \int_{\mathbf{r}_{ref} }^{\mathbf{r}_p} \mathbf{F}_{C} \cdot \mathrm{d}\mathbf{r} =  \int_{\mathbf{r}_{c}}^{\mathbf{r}_p} \left(- m_p g\right) \  \hat{\mathbf{j}} \cdot \left(\mathrm{d}x \hat{\mathbf{i}} + \mathrm{d}y \hat{\mathbf{j}} + \mathrm{d}z \hat{\mathbf{k}}\right) \\
    & =  \int_{\mathbf{r}_c \cdot \hat{\mathbf{i}}}^{\mathbf{r}_p \cdot \hat{\mathbf{i}}} \ 0\mathrm{d}x + \int_{\mathbf{r}_c \cdot \hat{\mathbf{j}}}^{\mathbf{r}_p \cdot \hat{\mathbf{j}}} \left(- m_p g\right) \mathrm{d}y  + \int_{\mathbf{r}_c \cdot \hat{\mathbf{k}}}^{\mathbf{r}_p \cdot \hat{\mathbf{k}}} 0 \mathrm{d}z\\
    &  = - m_p g\int_{0}^{l \cos \left( \theta (t) \right)} \ \mathrm{d}y = - m_pgl \cos \big( \theta (t) \big).
\end{aligned}
\label{eqn:V}
\end{equation}

Plugging \eqref{eqn:T} and \eqref{eqn:V} to \eqref{eqn:L} we have the following expression for the Lagrangian:

\begin{equation}
\begin{aligned}
    \mathcal{L}\big(\mathbf{x}(t)\big) & = \frac{1}{2} m_c \dot{x}^2(t) + \frac{1}{2} m_p \Big(\dot{x}^2(t)  + 2 l \dot{x}(t)\cos \big( \theta (t) \big)\dot{\theta} (t) + l^2 \dot{\theta}^2 (t)\Big) \\
    & + \frac{1}{2} I \dot{\theta}^2 (t) + m_pgl \cos \big( \theta (t) \big).
    \label{eqn:L_final}
\end{aligned}
\end{equation}

To consider the Euler-Lagrange equation, let us define our virtual coordinates as $x(t)$ and $\theta(t)$. For each of these, let us express its corresponding Euler-Lagrange equation:

\begin{equation}
    \begin{aligned}
        \frac{\mathrm{d}}{\mathrm{d}t}\left(\frac{\partial \mathcal{L}\big(\mathbf{x}(t)\big)}{\partial \dot{x}(t)}\right) - \frac{\partial \mathcal{L}\big(\mathbf{x}(t)\big)}{\partial x(t)} = F_{{n.c.}_x}(t); \\
        \frac{\mathrm{d}}{\mathrm{d}t}\left(\frac{\partial \mathcal{L}\big(\mathbf{x}(t)\big)}{\partial \dot{\theta}(t)}\right) - \frac{\partial \mathcal{L}\big(\mathbf{x}(t)\big)}{\partial \theta(t)} = F_{{n.c.}_\theta}(t);
    \end{aligned}
    \label{eqn:E-L}
\end{equation}

where $F_{{n.c.}_x}(t)$ and $F_{{n.c.}_\theta}(t)$ are the non conservative forces acting along the trajectories of the virtual coordinates $x(t)$ and $\theta(t)$, respectively. Plugging in the Lagrangian from \eqref{eqn:L_final}, let us now evaluate the derivatives of the Lagrangian with respect to each virtual coordinate and its time-derivative:

\begin{equation}
    \begin{aligned}
        \frac{\partial \mathcal{L}\big(\mathbf{x}(t)\big)}{\partial \dot{x}(t)} & =  \left(m_c + m_p\right)  \dot{x}(t) + m_p l \cos \big( \theta (t) \big)\dot{\theta} (t); \\
        \frac{\partial \mathcal{L}\big(\mathbf{x}(t)\big)}{\partial x(t)} & =  0; \\
        \frac{\partial \mathcal{L}\big(\mathbf{x}(t)\big)}{\partial \dot{\theta}(t)} & =  m_p \left(  l \dot{x}(t)\cos \left(\theta(t)\right) + l^2 \dot{\theta}(t)\right) + I\dot{\theta}(t);\\
        \frac{\partial \mathcal{L}\big(\mathbf{x}(t)\big)}{\partial \theta(t)} & = -m_p l \sin \left(\theta(t)\right) \left( g + \dot{x}(t) \dot{\theta}(t)\right).
    \end{aligned}
\end{equation}

Now let us calculate the corresponding time derivatives:

\begin{equation}
    \begin{aligned}
        \frac{\mathrm{d}}{\mathrm{d}t}\left(\frac{\partial \mathcal{L}\big(\mathbf{x}(t)\big)}{\partial \dot{x}(t)}\right) & =  \left(m_c + m_p\right)  \ddot{x}(t)  + m_p l\Big(  \cos \big( \theta (t) \big)\ddot{\theta} (t) - \sin \big( \theta (t) \big)\dot{\theta}^2 (t)\Big); \\
        \frac{\mathrm{d}}{\mathrm{d}t}\left(\frac{\partial \mathcal{L}\big(\mathbf{x}(t)\big)}{\partial \dot{\theta}(t)}\right) & =  \left( m_pl^2 + I \right) \ddot{\theta}(t) + m_pl \left( \ddot{x}(t)\cos \left(\theta(t)\right) - \dot{x}(t)\sin \left(\theta(t)\right)\dot{\theta}(t)\right).
    \end{aligned}
\end{equation}

Plugging all to \eqref{eqn:E-L} with $ F_{{n.c.}_x}(t) = F(t)$ and $F_{{n.c.}_\theta}(t) = M(t)$ we finally have the model for $\mathcal{S}_{IP}$:

\begin{equation}
    \begin{aligned}
        \left(m_c + m_p\right)  \ddot{x}(t) & + m_p l\Big(  \cos \big( \theta (t) \big)\ddot{\theta} (t) - \sin \big( \theta (t) \big)\dot{\theta}^2 (t)\Big) = F(t); \\
        \left( m_pl^2 + I \right) \ddot{\theta}(t) & + m_pl  \Big(\cos \left(\theta(t)\right)\ddot{x}(t) +  g  \sin \left(\theta(t)\right)\Big) = M(t).
    \end{aligned}
\end{equation}

Solving for $\ddot{x}(t)$ from the first equation we have:

\begin{equation}
    \begin{aligned}
        \ddot{x}(t) = \frac{1}{m_c + m_p} \left(F(t) + m_p l\Big(  \sin \big( \theta (t) \big)\dot{\theta}^2 (t) - \cos \big( \theta (t) \big)\ddot{\theta} (t) \Big)\right).
    \end{aligned}
    \label{eqn:x_ddot_first}
\end{equation}

Plugging this in the second equation:

\begin{equation}
    \begin{aligned}
        \left( m_pl^2 + I \right) \ddot{\theta}(t) & + m_pl  \Bigg[\frac{\cos \left(\theta(t)\right)}{m_c + m_p} \Big[F(t) + m_p l\Big(   \sin \big( \theta (t) \big)\dot{\theta}^2 (t) \\ & - \cos \big( \theta (t) \big)\ddot{\theta} (t)\Big)\Big] 
         +  g  \sin \left(\theta(t)\right)\Bigg] = M(t).
    \end{aligned}
\end{equation}

Solving for $\ddot{\theta}(t)$:

\begin{equation}
    \begin{aligned}
        \ddot{\theta}(t) = & \frac{1}{m_pl^2 + I - \frac{m^2_pl^2}{m_c + m_p}\cos^2 \big( \theta (t) \big)}  \Bigg[M(t) \\ & - m_pl  \bigg[\frac{\cos \left(\theta(t)\right)}{m_c + m_p} \Big[F(t) + m_p l   \sin \big( \theta (t) \big)\dot{\theta}^2 (t) \Big] 
         +  g  \sin \left(\theta(t)\right)\bigg] \Bigg]
    \end{aligned}
\end{equation}

Plugging in \eqref{eqn:x_ddot_first} we have:

\begin{equation}
    \begin{aligned}
        \ddot{x}(t) =& \frac{1}{m_c + m_p} \Bigg[F(t) + m_p l\Bigg(  \sin \big( \theta (t) \big)\dot{\theta}^2 (t) \\  -& \cos \big( \theta (t) \big)\frac{1}{m_pl^2 + I - \frac{m^2_pl^2}{m_c + m_p}\cos^2 \big( \theta (t) \big)}  \Bigg[M(t) \\  -& m_pl  \bigg[\frac{\cos \left(\theta(t)\right)}{m_c + m_p} \Big[F(t) + m_p l   \sin \big( \theta (t) \big)\dot{\theta}^2 (t) \Big] 
         +  g  \sin \left(\theta(t)\right)\bigg] \Bigg] \Bigg)\Bigg].
    \end{aligned}
\end{equation}

\subsubsection{Linearization}

We now would like to perform linearization in order to use our method. To do so, we'll need to choose an equilibrium point to linearize $\mathcal{S}_{IP}$ around. Considering we would like to stabalize the system around the point where $\theta(t) \approx \pi$, which is a non-stable equilibrium point, we'll linearize the system around it. Let us denote $\theta(t) = \pi + \delta(t)$. We have:

\begin{equation}
\begin{aligned}
    \sin \left(\theta(t)\right) = & \sin \left(\pi + \delta(t)\right) = \sin \left(\pi - (-\delta(t))\right) \\ = & \sin \left( (-\delta(t))\right) = - \sin \left( \delta(t)\right);\\
    \cos \left(\theta(t)\right) = &  \cos \left(\pi - (-\delta(t))\right) \\ = & -\cos \left( (-\delta(t))\right) = - \cos \left( \delta(t)\right).
    \label{eqn:theta_to_delta}
\end{aligned}
\end{equation}
for all $t \in \mathbb{R}^+$. Let us now consider a small angle approximation for $\delta(t)$, as in we assume all powers of $\delta (t)$ equal or greater than $2$ is are negligible, as in $\forall k \geq 2 \ ; \ \delta ^k(t) \approx 0$. Using the Maclaurin expansion for the trigonometric functions we have:

\begin{equation}
    \begin{aligned}
    \sin \left(\delta(t)\right) & = \sum^{\infty}_{n=0} \frac{(-1)^n}{(2n+1)!} \left(\delta(t)\right)^{2n+1} = \delta(t) - \frac{\delta^3(t)}{3!} + \frac{\delta^5(t)}{5!} - \cdots  \\ & \approx \delta(t) = \theta(t) - \pi;\\
        \cos \left(\delta(t)\right) &= \sum^{\infty}_{n=0} \frac{(-1)^n}{(2n)!} \left(\delta(t)\right)^{2n} = 1 - \frac{\delta^2(t)}{2!} + \frac{\delta^4(t)}{4!} - \cdots \approx 1;
    \end{aligned}
\end{equation}
for all $t \in \mathbb{R}^+$. Plugging this approximation in \eqref{eqn:theta_to_delta} we have:
\begin{equation}
\begin{aligned}
    \sin \left(\theta(t)\right) & \approx \pi - \theta(t);\\
    \cos \left(\theta(t)\right) & \approx - 1.
\end{aligned}
\end{equation}
Moreover, we assume products of $x(t)$, $\theta(t)$ and their time derivatives are relatively small and close to zero. Thus the model reduces to:

\begin{equation}
    \begin{aligned}
        \ddot{\theta}(t) = & \frac{m_c + m_p}{m_cm_pl^2 + I (m_c + m_p)}  \Bigg[M(t) + m_pl  \bigg[\frac{1}{m_c + m_p} F(t) 
         +  g  \theta(t)\bigg] \Bigg]\\
         = &  \frac{m_pl}{m_cm_pl^2 + I (m_c + m_p)}F(t) + \frac{m_c + m_p}{m_cm_pl^2 + I (m_c + m_p)} M(t)\\
         + & \frac{m_pgl(m_c + m_p)}{m_cm_pl^2 + I (m_c + m_p)}  \theta(t)\\
         \ddot{x}(t) =& \frac{1}{m_c + m_p} \Bigg[F(t) + m_p l\Bigg(  \frac{m_c + m_p}{m_cm_pl^2 + I (m_c + m_p)}  \Bigg[M(t) \\  +& m_pl  \bigg[\frac{1}{m_c + m_p} F(t)
         +  g  \theta(t)\bigg] \Bigg] \Bigg)\Bigg]\\
         = & \frac{1}{m_c + m_p}\left( 1 + \frac{m^2_pl^2}{m_cm_pl^2 + I (m_c + m_p)}\right)F(t) \\
         + & \frac{m_pl}{m_cm_pl^2 + I (m_c + m_p)} M(t) + \frac{m^2_pgl^2}{m_cm_pl^2 + I (m_c + m_p)}   \theta(t)\\
         = & \frac{m_pl^2 + I}{m_cm_pl^2 + I (m_c + m_p)}F(t) + \frac{m_pl}{m_cm_pl^2 + I (m_c + m_p)} M(t) \\
         + & \frac{m^2_pgl^2}{m_cm_pl^2 + I (m_c + m_p)}   \theta(t).
    \end{aligned}
    \label{eqn:linear_model}
\end{equation}

\subsubsection{State Space}
Let us define:

\begin{equation}
    \begin{aligned}
        D \coloneqq & m_cm_pl^2 + I (m_c + m_p);\\
        a_{21} \coloneqq & \frac{m^2_p g l^2}{D}; \\
        a_{31} \coloneqq & \frac{m_p g l (m_c + m_p)}{D}; \\
        b_{21} \coloneqq & \frac{m_pl^2 + I}{D};\\
        b_{31} \coloneqq&  \frac{m_p l}{D}; \\
        b_{22} \coloneqq& b_{31};\\
        b_{32} \coloneqq& \frac{m_c + m_p}{D}.
    \end{aligned}
\end{equation}

Plugging this into \eqref{eqn:linear_model} we have:

\begin{equation}
    \begin{aligned}
        \ddot{\theta}(t) = &  b_{31}F(t) + b_{32} M(t) + a_{31}  \theta(t)\\
         \ddot{x}(t) =& b_{21}F(t) + b_{22} M(t)
         + a_{21} \theta(t).
    \end{aligned}
\end{equation}

With this, we finally arrive at the following state space for $\mathcal{S}_{IP}$:

\begin{equation}
    \begin{bmatrix}
        \dot{x}(t) \\
        \dot{\theta}(t) \\
        \ddot{x}(t) \\
        \ddot{\theta}(t)
    \end{bmatrix} = \underbrace{\begin{bmatrix}
        0 & 0 & 1 & 0 \\
        0 & 0 & 0 & 1 \\
        0 & a_{21} & 0 & 0 \\
        0 & a_{31} & 0 & 0
    \end{bmatrix}}_{A} \begin{bmatrix}
        x(t) \\
        \theta(t) \\
        \dot{x}(t) \\
        \dot{\theta}(t)
    \end{bmatrix} + \underbrace{\begin{bmatrix}
        0 & 0\\
        0 & 0\\
        b_{21} & b_{22}\\
        b_{31} & b_{32}
    \end{bmatrix}}_{B} \begin{bmatrix}
        F(t) \\
        M(t) 
    \end{bmatrix} 
\end{equation}

which of course, we will denote as:

\begin{equation}
    \dot{\mathbf{x}}(t) = A \mathbf{x}(t) + B \mathbf{u}(t)
\end{equation}

\subsection{Infinite Horizon LQR Control}

Let us now consider a standard LQR controller for $\mathcal{S}_{IP}$, i.e. consider our approach for the case where $N \equiv 1$. Let us set $t_0 \equiv 0$, and let us define the following objective:

\begin{equation}
\begin{aligned}
    O \coloneqq & \left( Q, R\right),
\end{aligned}
\end{equation}

where:

\begin{itemize}
    \item Let us define:
    \begin{equation}
        \begin{aligned}
            Q_x \coloneqq & \begin{bmatrix}
        q_l & 0 & 0 & 0 \\
        0 & q_s & 0 & 0 \\
        0 & 0 & q_m & 0 \\
        0 & 0 & 0 & q_s
    \end{bmatrix};\\
    Q_\theta \coloneqq & \begin{bmatrix}
        q_s & 0 & 0 & 0 \\
        0 & q_l & 0 & 0 \\
        0 & 0 & q_s & 0 \\
        0 & 0 & 0 & q_m
        \end{bmatrix},
        \end{aligned}
    \end{equation}
    for some $q_s, q_m, q_l \in \mathbb{R}^+$ that satisfy $q_l \gg q_m \gg q_s$. Notice $Q_x$ affects $x(t)$ the most, $\dot{x}(t)$ a bit less, and almost does not affect $\theta (t)$ and $\dot{\theta} (t)$ at all, while $Q_\theta$ works vice versa - it affects $\theta (t) $ the most, $\dot{\theta}(t)$ a bit less, and almost does not affect $x (t)$ and $\dot{x} (t)$. Let us set $Q$ to average these two, i.e.:
    \begin{equation}
        Q \coloneqq \frac{Q_x + Q_\theta}{2} = \frac{1}{2} \begin{bmatrix}
        q_l + q_s & 0 & 0 & 0 \\
        0 & q_l + q_s & 0 & 0 \\
        0 & 0 & q_m + q_s & 0 \\
        0 & 0 & 0 & q_m + q_s
    \end{bmatrix}.
    \end{equation}
    Since $q_l \gg q_m \gg q_s$, we can approximate:
    \begin{equation}
        Q \approx \frac{1}{2} \begin{bmatrix}
        q_l & 0 & 0 & 0 \\
        0 & q_l & 0 & 0 \\
        0 & 0 & q_m & 0 \\
        0 & 0 & 0 & q_m
    \end{bmatrix}.
    \end{equation}
    \item Let us define:
    \begin{equation}
        \begin{aligned}
            R \coloneqq & \begin{bmatrix}
        r & 0 \\
        0 & r \\
    \end{bmatrix},
        \end{aligned}
    \end{equation}
    for some $r \in \mathbb{R}^+$.
\end{itemize}

Using the approximation for $Q$, the associated cost function will be:

\begin{equation}
\begin{aligned}
        J\Big(\mathbf{x}(t), \mathbf{v}(t), t_0\Big) = &
       \int_{t_0}^{\infty} \Big[\mathbf{x}(t)^TQ\mathbf{x}(t) + \mathbf{v}^T(t)R\mathbf{v}(t)\Big]\mathrm{d}t\\
       = & \int_{0}^{\infty} \Big[\mathbf{x}(t)^TQ\mathbf{x}(t) + \mathbf{u}^T(t)R\mathbf{u}(t)\Big]\mathrm{d}t \\
       = & \frac{1}{2}\bigint_{0}^{\infty} \begin{bmatrix}
        x(t) \\
        \theta(t) \\
        \dot{x}(t) \\
        \dot{\theta}(t)
    \end{bmatrix}^T \begin{bmatrix}
        q_l & 0 & 0 & 0 \\
        0 & q_l & 0 & 0 \\
        0 & 0 & q_m & 0 \\
        0 & 0 & 0 & q_m
    \end{bmatrix}\begin{bmatrix}
        x(t) \\
        \theta(t) \\
        \dot{x}(t) \\
        \dot{\theta}(t)
    \end{bmatrix}\mathrm{d}t\\
    & + \bigintsss_{0}^{\infty}\begin{bmatrix}
        F(t) \\
        M(t) 
    \end{bmatrix} ^T\begin{bmatrix}
        r & 0 \\
        0 & r \\
    \end{bmatrix}\begin{bmatrix}
        F(t) \\
        M(t) 
    \end{bmatrix} \mathrm{d}t\\
    = & \frac{1}{2}\int_{0}^{\infty} \left[q_l \left( x^2(t) + \theta^2(t)\right) + q_m \left( \dot{x}^2(t) + \dot{\theta}^2(t)\right) \right]\mathrm{d}t\\
    + & r \int_{0}^{\infty} \left[F^2(t) + M^2(t)\right]\mathrm{d}t.
\end{aligned}
\end{equation}

Now let us consider the optimal inputs values. Using the theory laid out in appendix \ref{subsection:cont_inf_lqr} we have that the optimal input adheres the expression in equation \eqref{eqn:optimal_proposed_value}: $\mathbf{v^*}(t) = - R^{-1}B^T P \mathbf{x^*}(t)$, where the positive-definite matrix $P \in \mathbb{R}^{n \times n}$ is obtained by solving the appropriate CARE as described in equation \eqref{eqn:riccati}. 



\subsection{System Decomposition}
So, let us consider two virtual objectives, one for the $x$ coordinate and the other for the $\theta$ coordinate, while taking changing values for the weights of one of them, say for the $x$ coordinate. So let us define the following virtual vector:

\begin{equation}
    \mathbf{v}(t) \coloneqq \begin{bmatrix}
        v_x(t)\\
        v_\theta(t)
    \end{bmatrix}
\end{equation}
with corresponding lengths $(m_x, m_\theta) \coloneqq (1, 1)$ such that:
\begin{equation}
\begin{aligned}
    v_x(t) \coloneqq b_{21}F(t) + b_{22} M(t)\\
    v_\theta(t) \coloneqq b_{31}F(t) + b_{32} M(t).
\end{aligned}
\end{equation}
Let us define:
\begin{equation}
    M_x \coloneqq \begin{bmatrix}
        b_{21} & b_{22}
    \end{bmatrix} \ ; \ M_\theta \coloneqq \begin{bmatrix}
        b_{31} & b_{32}
    \end{bmatrix}.
\end{equation}
With that, let the augmented division matrix of $\mathcal{S}_{IP}$ to the corresponding decomposed system $\mathcal{S}_{{IP}_{x, \theta}}$ be:
\begin{equation}
    M \coloneqq \begin{bmatrix}
        M_x \\
        M_\theta
    \end{bmatrix} = \begin{bmatrix}
        b_{21} & b_{22} \\
        b_{31} & b_{32}
    \end{bmatrix}
\end{equation}
Thus we have:
\begin{equation}
    \begin{bmatrix}
        v_x(t) \\
        v_\theta(t) 
    \end{bmatrix} = \begin{bmatrix}
          b_{21} & b_{22} \\
          b_{31} &  b_{32}
    \end{bmatrix} \begin{bmatrix}
        F(t) \\
        M(t)
    \end{bmatrix}.
\end{equation}
And more concisely:
\begin{equation}
    \mathbf{v}(t) = M \mathbf{u}(t).
\end{equation}
Inverting $M$ we have:
\begin{equation}
    \mathbf{u}(t) = M^{-1} \mathbf{v}(t).
\end{equation}
Where:
\begin{equation}
    M^{-1} = 
        \frac{1}{b_{21}b_{32}-b_{22}b_{31}}\begin{bmatrix} b_{32} & -b_{22}\\
        -b_{31} & b_{21}
    \end{bmatrix}
\end{equation}
Correspondingly, using the lengths $(m_x, m_\theta)$, let us denote:
\begin{equation}
    \Tilde{M}_x \coloneqq \frac{1}{b_{21}b_{32}-b_{22}b_{31}} \begin{bmatrix}
        b_{32}\\
        -b_{31}
    \end{bmatrix} \ ; \ \Tilde{M}_\theta \coloneqq \frac{1}{b_{21}b_{32}-b_{22}b_{31}} \begin{bmatrix}
        -b_{22}\\
        b_{21}
    \end{bmatrix}.
\end{equation}
So we have:
\begin{equation}
    M^{-1} = \begin{bmatrix}
        \Tilde{M}_x & \Tilde{M}_\theta
    \end{bmatrix}
\end{equation}
Plugging to the model at~\eqref{eqn:basic_sys} we have:
\begin{equation}
\begin{aligned}
    \dot{\mathbf{x}}(t) & = A\mathbf{x}(t) + B \mathbf{u}(t) = A\mathbf{x}(t) + B M^{-1} \mathbf{v}(t) \\
    & = A\mathbf{x}(t) + B\begin{bmatrix}
        \Tilde{M}_x & \Tilde{M}_\theta
    \end{bmatrix} \mathbf{v}(t)\\
    & = A\mathbf{x}(t) +\begin{bmatrix}
        B\Tilde{M}_x & B\Tilde{M}_\theta
    \end{bmatrix}\mathbf{v}(t)\\
    & = A\mathbf{x}(t) +
    \frac{1}{b_{21}b_{32}-b_{22}b_{31}} \begin{bmatrix}
        \begin{bmatrix}
        0 & 0\\
        0 & 0\\
        b_{21} & b_{22}\\
        b_{31} & b_{32}
    \end{bmatrix}   \begin{bmatrix}
        b_{32}\\
        -b_{31}
    \end{bmatrix} & \begin{bmatrix}
        0 & 0\\
        0 & 0\\
        b_{21} & b_{22}\\
        b_{31} & b_{32}
    \end{bmatrix} \begin{bmatrix}
        -b_{22}\\
        b_{21}
    \end{bmatrix}
    \end{bmatrix}\mathbf{v}(t)\\
    & = A\mathbf{x}(t) +  \frac{1}{b_{21}b_{32}-b_{22}b_{31}} \begin{bmatrix}
    0 & 0 \\
    0 & 0\\
    b_{21}b_{32}-b_{22}b_{31} & 0\\
    0 & b_{21}b_{32}-b_{22}b_{31}
    \end{bmatrix}\mathbf{v}(t)\\
    & = A\mathbf{x}(t) + \begin{bmatrix}
        0 & 0\\
        0 & 0\\
        1 & 0\\
        0 & 1
    \end{bmatrix}\begin{bmatrix}
        v_x(t)\\
        v_\theta(t)
    \end{bmatrix} = A\mathbf{x}(t) + \begin{bmatrix}
        0\\
        0\\
        1\\
        0
    \end{bmatrix} v_x(t) + \begin{bmatrix}
        0\\
        0\\
        0\\
        1
    \end{bmatrix} v_\theta(t)
    \end{aligned}
\end{equation}

Denoting:

\begin{equation}
    B_x \coloneqq \begin{bmatrix}
        0\\
        0\\
        1\\
        0
    \end{bmatrix} \ ; \ B_\theta \coloneqq \begin{bmatrix}
        0\\
        0\\
        0\\
        1
    \end{bmatrix},
\end{equation}

we have the following model for the decomposed system $\mathcal{S}_{{IP}_{x, \theta}}$:

\begin{equation}
    \dot{\mathbf{x}}(t) = A\mathbf{x}(t) + \sum_{i \in \{ x, \theta \}} B_i v_i(t)
\end{equation}

Correspondingly, let us define the following virtual objectives for $\mathcal{S}_{{IP}_{x, \theta}}$:

\begin{equation}
\begin{aligned}
    O_{x} \coloneqq & \left( Q_x, \left( R_{xx}, 0_{n\times n} \right), B_{x}\right);\\
    O_{\theta} \coloneqq & \left( Q_\theta, \left( 0_{n\times n}, R_{\theta \theta} \right), B_\theta \right).
\end{aligned}
\end{equation}

The associated virtual cost functions of $\mathcal{S}_{{IP}_{x, \theta}}$ will be:

\begin{equation}
\begin{aligned}
        J_{x}\Big(\mathbf{x}(t), v_{x}(t) \Big) = &
       \int_{0}^{\infty} \Big[\mathbf{x}(t)^TQ_x\mathbf{x}(t) + v^T_{x}(t)R_{xx}v_{x}(t) \Big]\mathrm{d}t;\\
       J_{\theta}\Big(\mathbf{x}(t), v_{\theta}(t)\Big) = &
       \int_{0}^{\infty} \Big[\mathbf{x}(t)^TQ_\theta\mathbf{x}(t)  + v_{\theta}^T(t)R_{\theta\theta}v_{\theta}(t)\Big]\mathrm{d}t.
\end{aligned}
\end{equation}

\subsection{Python Implementation Details}

We have written an appropriate subclass for $\mathcal{S}_{IP}$ and the detailed decomposed system $\mathcal{S}_{{IP}_{x, \theta}}$. The full code is detailed in the package repository~\cite{package}. To enable quick simulation, we implemented the aforementioned mathematical derivation hard-coded in the class, so that the constructor of the class is of the following form:

\begin{minted}[linenos,tabsize=2,breaklines]{python}
import numpy as np
from abc import ABC
from PyDiffGame.PyDiffGame import PyDiffGame
from PyDiffGame.ContinuousPyDiffGame import ContinuousPyDiffGame

class InvertedPendulum(PyDiffGame, ABC):
    def __init__(self,
                 m_c: float,
                 m_p: float,
                 p_L: float,
                 q_s: float = 1,
                 q_m: float = 100,
                 q_l: float = 10000,
                 r: float = 0.001,
                 x_0: np.array = None,
                 x_T: np.array = None,
                 T_f: float = None,
                 L: int = PyDiffGame._L_default,
                 multiplayer: bool = True,
                 regular_LQR: bool = False,
                 show_animation: bool = True
                 ):
\end{minted}

A few takeaways here to mention:

\begin{itemize}
    \item In simpler terms, to simulate the closed loop dynamics we just need to specify the parameters $m_c, m_p$ and $L$ (noted as \pyth{p_L} in the class constructor to not be confused with the number or data points $L$);
    \item We proposed default values for $q_s, q_m, q_l$ and $r$ as $1, 100, 1000$ and $0.001$, and number of data points which is the parent class default of $1000$;
    \item In order for us to simulate the state space model under the acquired dynamics, we would also define $\mathbf{x_0}$;
    \item If we do not define $\mathbf{x_T}$, we then would like to consider the regulation case, i.e., the case where the state variable $\mathbf{x}(t)$ is desired to converge to the zero vector;
    \item If we do define $\mathbf{x_T}$, then we would consider the situation of signal tracking, where we would like the state variable $\mathbf{x}(t)$ to converge to the given fixed vector $\mathbf{x_T}$;
    \item The \pyth{regular_LQR} flag allows to run a regular LQR, which is equivalent to a 1-person differential game;
    \item The \pyth{show_animation} flag allows to simulate a running simulation of the pendulum as a visual rod on a cart using Python's \pyth{matplotlib} package.
\end{itemize}
  
\subsection{Python Simulation }
To demonstrate the implementation of our D\&C method, we used several values for the aforementioned variables, and compared a regular LQR with a 2-person non-zero-sum game, from which we used the described theory to obtain a Nash Equilibrium. 

We computed the following metric both for $\mathcal{S}_{IP}$ and $\mathcal{S}_{{IP}_{x,\theta}}$:
    \underline{$T_c \in \mathcal{I}$}: the required time to attain the condition for convergence of the state variable $\mathbf{x}(t)$ as described in the \pyth{PyDiffARE} algorithm, formally described as the time at which the following conditions holds:
    \begin{equation}
        \| \mathbf{x}(t) - \mathbf{x_T} \| < \varepsilon
    \end{equation}

\subsubsection{First Signal Tracking Simulation}
Let us define for the first simulation:

\begin{equation}
    \mathbf{x_0} \coloneqq \begin{bmatrix}
        x_0 \\
        \theta_0 \\
        \dot{x}_0 \\
        \dot{\theta}_0
    \end{bmatrix} = \begin{bmatrix}
        20 \ [m] \\
        \pi / 3 \ [rad] \\
        0 \\
        0
    \end{bmatrix} \ ; \ \mathbf{x_T} \coloneqq  \begin{bmatrix}
        0 \\
        \pi \ [rad] \\
        0 \\
        0
    \end{bmatrix}
\end{equation}


\begin{table}[h!]
\begin{center}
  \begin{tabular}{ | c | c | c |}
  \hline
      \backslashbox{\textbf{$m_c\ [kg], m_p\ [kg], L\ [m]$}}{\textbf{$T_c \ [sec]$}} & \textbf{$\mathcal{S}_{IP}$} &  \textbf{$\mathcal{S}_{{IP}_{x,\theta}}$}  \\ 
      \hline
    \textbf{10, 1, 1} & 3 & 3\\ \hline
    \textbf{20, 5, 2} & 4 & 3\\ \hline
    \textbf{50, 8, 3} & 5 & 3\\ \hline
    \textbf{100, 10, 4} & 7 & 3 \\\hline
    \textbf{200, 15, 1} & 8 & 3\\
    \hline
  \end{tabular}
\end{center}
\caption{Convergence times of $\mathcal{S}_{IP}$ and $\mathcal{S}_{{IP}_{x,\theta}}$ for different $m_c, m_p, L$ values at the first simulation.}
\end{table}
\clearpage
  
\begin{center}
\begin{figure}[h!t!]
\centering
\includegraphics[width=0.82\columnwidth]{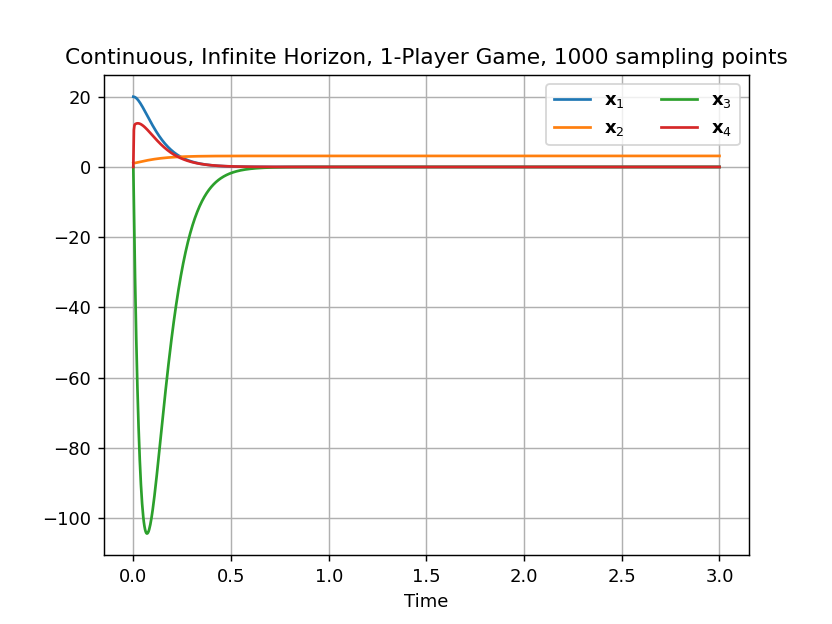}
\caption{LQR with $(m_c, m_p, L) \equiv (10, 1, 1)$}
\end{figure}
\end{center}

\begin{center}
\begin{figure}[h!t!]
\centering
\includegraphics[width=0.82\columnwidth]{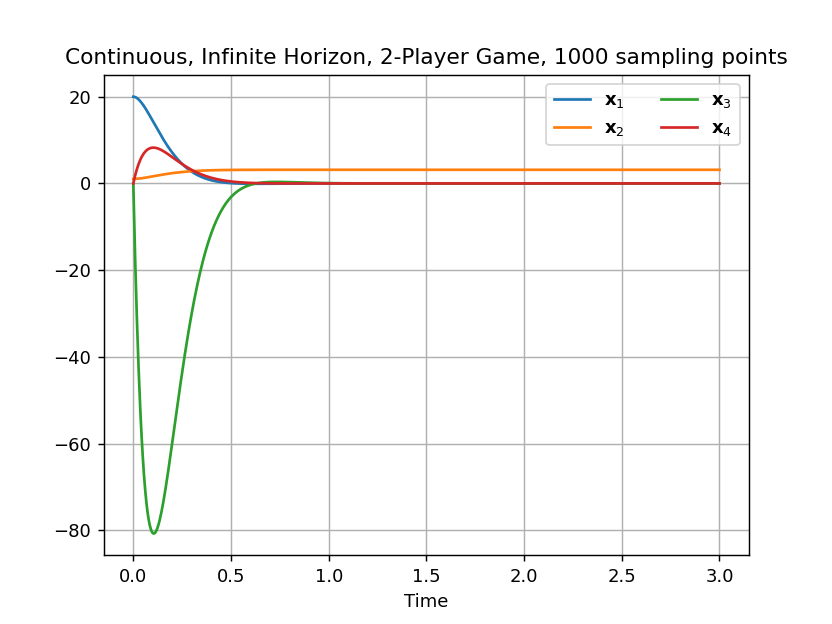}
\caption{Nash Equilibrium with $(m_c, m_p, L) \equiv (10, 1, 1)$}
\end{figure}
\end{center}


\begin{center}
\begin{figure}[h!t!]
\centering
\includegraphics[width=0.82\columnwidth]{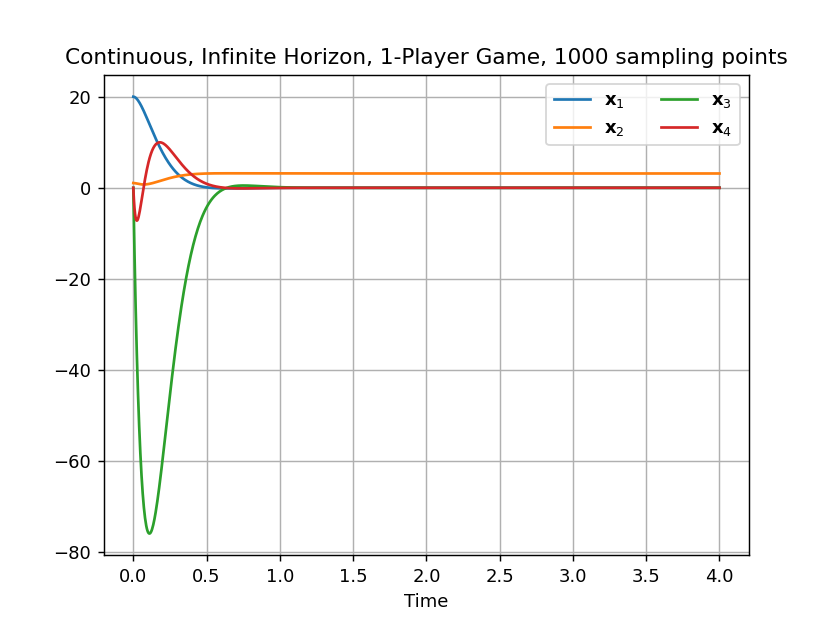}
\caption{LQR with $(m_c, m_p, L) \equiv (20, 5, 2)$}
\end{figure}
\begin{figure}[h!t!]
\centering
\includegraphics[width=0.82\columnwidth]{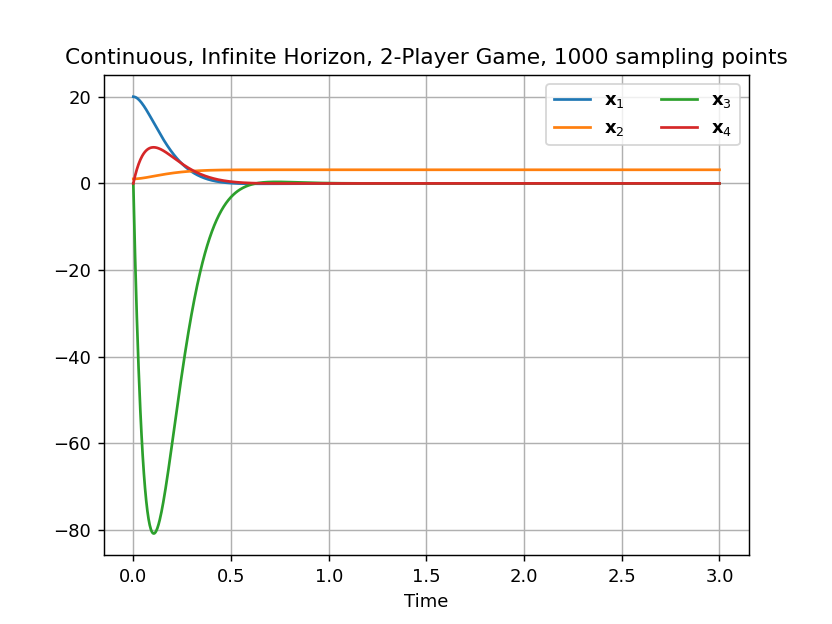}
\caption{Nash Equilibrium with $(m_c, m_p, L) \equiv (20, 5, 2)$}
\end{figure}
\end{center}


\begin{center}
\begin{figure}[h!t!]
\centering
\includegraphics[width=0.82\columnwidth]{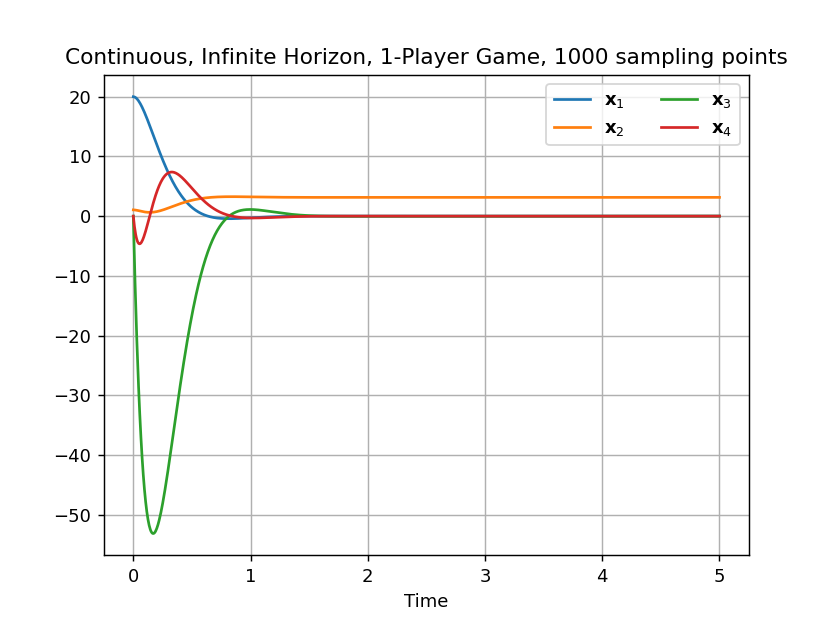}
\caption{LQR with $(m_c, m_p, L) \equiv (50, 8, 3)$}
\end{figure}
\begin{figure}[h!t!]
\centering
\includegraphics[width=0.82\columnwidth]{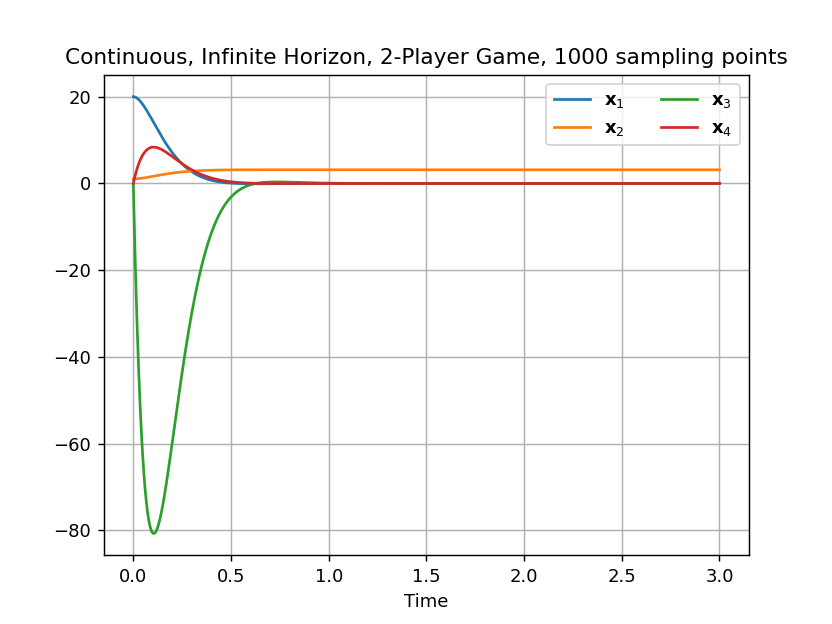}
\caption{Nash Equilibrium with $(m_c, m_p, L) \equiv (50, 8, 3)$}
\end{figure}
\end{center}


\begin{center}
\begin{figure}[h!t!]
\centering
\includegraphics[width=0.82\columnwidth]{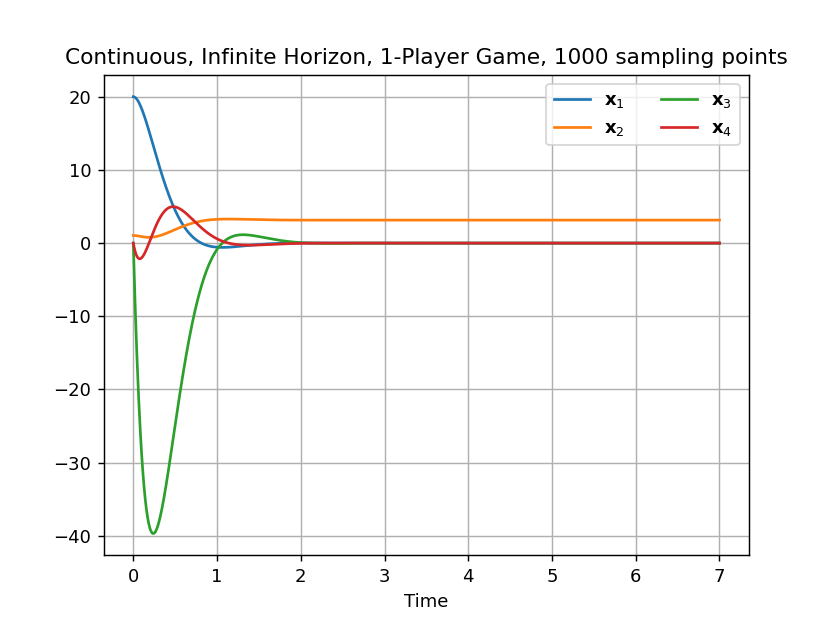}
\caption{LQR with $(m_c, m_p, L) \equiv (100, 10, 4)$}
\end{figure}
\begin{figure}[h!t!]
\centering
\includegraphics[width=0.82\columnwidth]{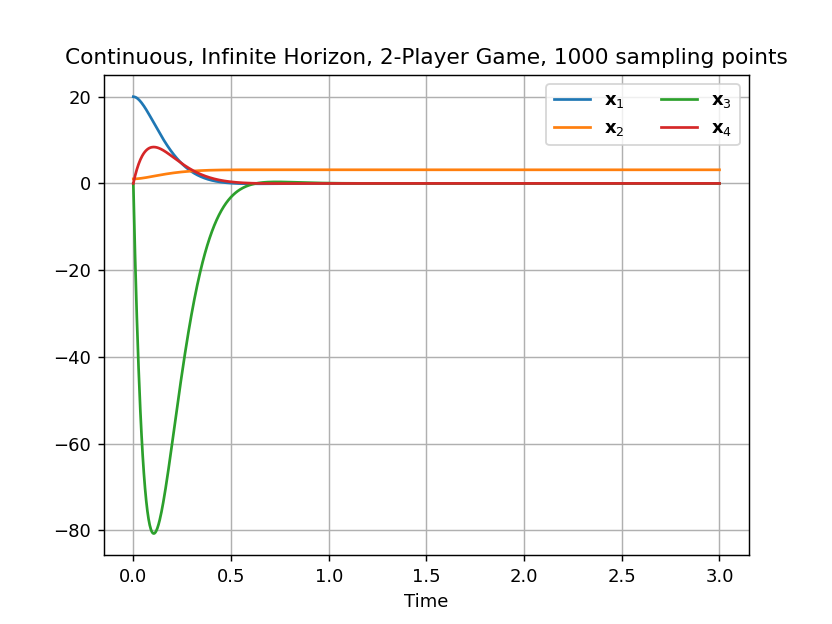}
\caption{Nash Equilibrium with $(m_c, m_p, L) \equiv (100, 10, 4)$}
\end{figure}
\end{center}

\begin{center}
\begin{figure}[h!t!]
\centering
\includegraphics[width=0.82\columnwidth]{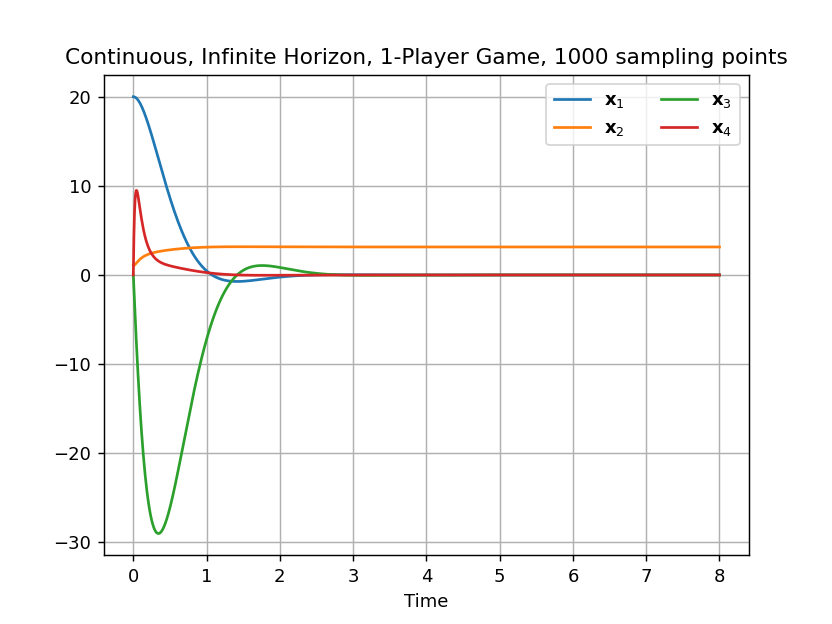}
\caption{LQR with $(m_c, m_p, L) \equiv (200, 15, 1)$}
\end{figure}
\begin{figure}[h!t!]
\centering
\includegraphics[width=0.82\columnwidth]{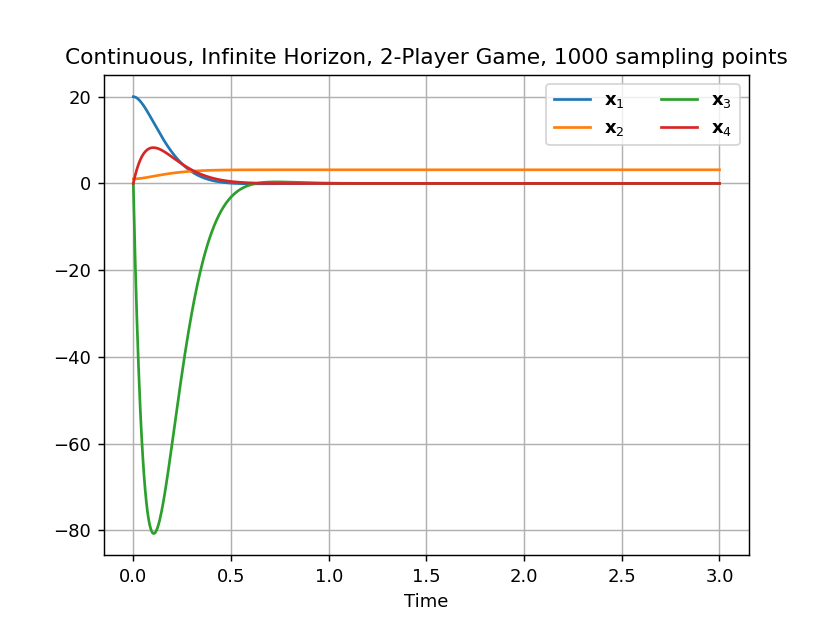}
\caption{Nash Equilibrium with $(m_c, m_p, L) \equiv (200, 15, 1)$}
\end{figure}
\end{center}

\subsubsection{Second Signal Tracking Simulation}

Let us define for the second simulation:

\begin{equation}
    \mathbf{x_0} \coloneqq \begin{bmatrix}
        x_0 \\
        \theta_0 \\
        \dot{x}_0 \\
        \dot{\theta}_0
    \end{bmatrix} = \begin{bmatrix}
        20 \ [m] \\
        \pi / 3 \ [rad] \\
        0 \\
        0
    \end{bmatrix} \ ; \ \mathbf{x_T} \coloneqq  \begin{bmatrix}
        -3 \ [m]\\
        \pi / 4 \ [rad] \\
        10 \ [m / s]\\
        5 \ [rad / s]
    \end{bmatrix}
\end{equation}

\begin{table}[h!]
\begin{center}
  \begin{tabular}{ | c | c | c |}
  \hline
      \backslashbox{\textbf{$m_c\ [kg], m_p\ [kg], L\ [m]$}}{\textbf{$T_c \ [sec]$}} & \textbf{$\mathcal{S}_{IP}$} &  \textbf{$\mathcal{S}_{{IP}_{x,\theta}}$}  \\ 
      \hline
    \textbf{10, 1, 1} & 12.5 & 16\\ \hline
    \textbf{20, 5, 2} & 17 & 16\\ \hline
    \textbf{50, 8, 3} & 25 & 16\\ \hline
    \textbf{100, 10, 4} & 33 & 16 \\\hline
    \textbf{200, 15, 1} & 45 & 16\\
    \hline
  \end{tabular}
\end{center}
\caption{Convergence Times of $\mathcal{S}_{IP}$ and $\mathcal{S}_{{IP}_{x,\theta}}$ for different $m_c, m_p, L$ values at the second simulation.}
\end{table}

  
\begin{center}
\begin{figure}[h!t!]
\centering
\includegraphics[width=0.85\columnwidth]{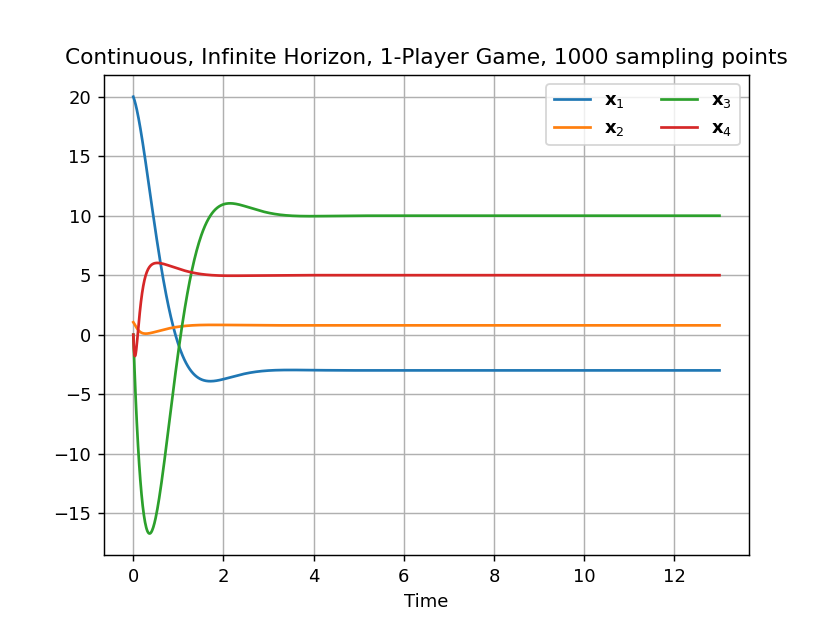}
\caption{LQR with $(m_c, m_p, L) \equiv (10, 1, 1)$}
\end{figure}
\end{center}

\begin{center}
\begin{figure}[h!t!]
\centering
\includegraphics[width=0.85\columnwidth]{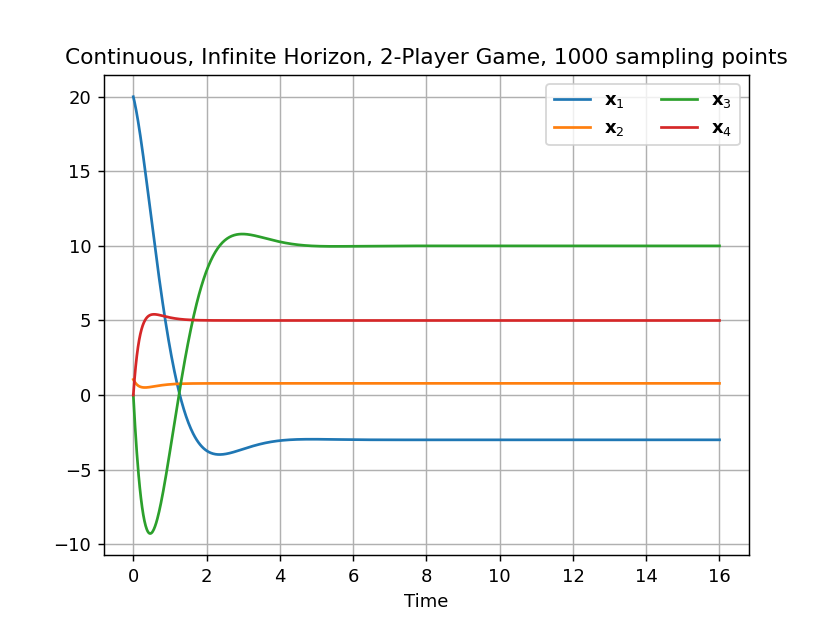}
\caption{Nash Equilibrium with $(m_c, m_p, L) \equiv (10, 1, 1)$}
\end{figure}
\end{center}


\begin{center}
\begin{figure}[h!t!]
\centering
\includegraphics[width=0.85\columnwidth]{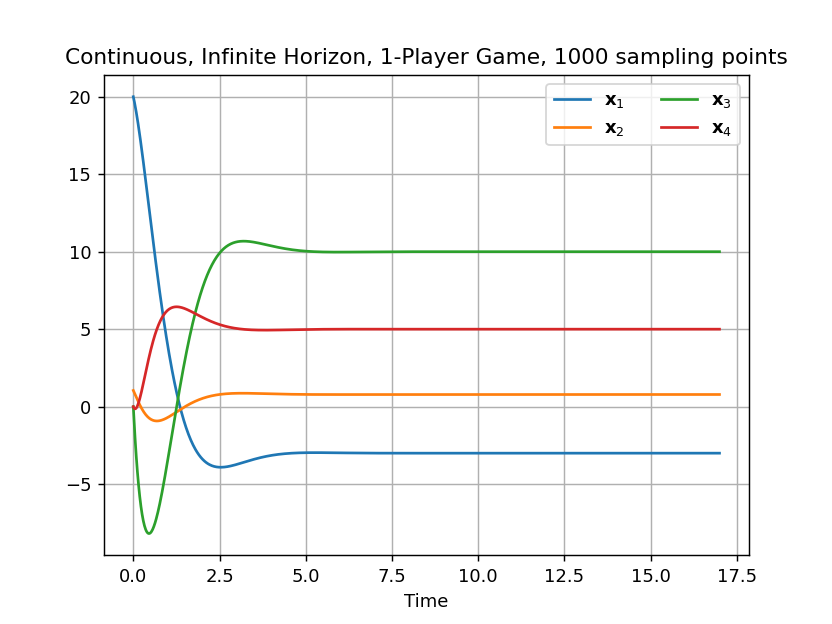}
\caption{LQR with $(m_c, m_p, L) \equiv (20, 5, 2)$}
\end{figure}
\end{center}

\begin{center}
\begin{figure}[h!t!]
\centering
\includegraphics[width=0.85\columnwidth]{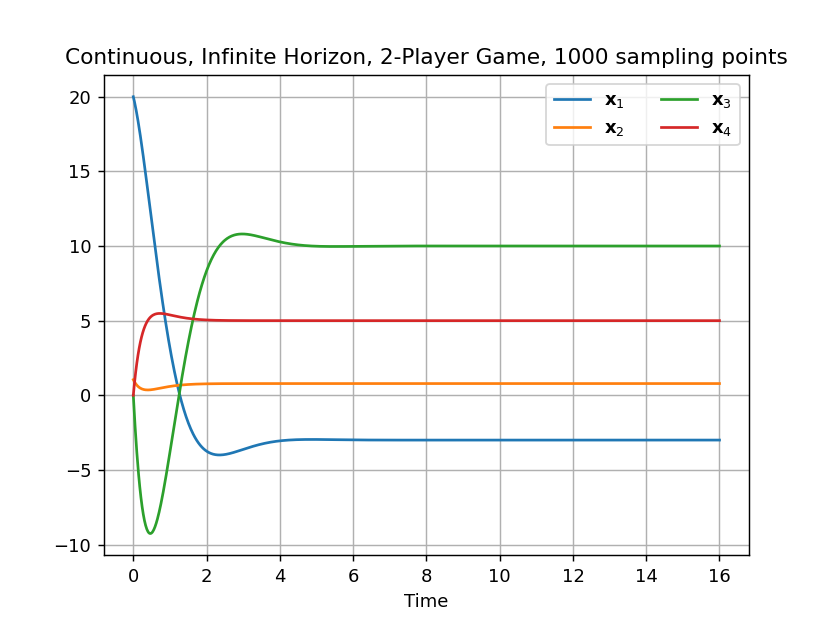}
\caption{Nash Equilibrium with $(m_c, m_p, L) \equiv (20, 5, 2)$}
\end{figure}
\end{center}


\begin{center}
\begin{figure}[h!t!]
\centering
\includegraphics[width=0.85\columnwidth]{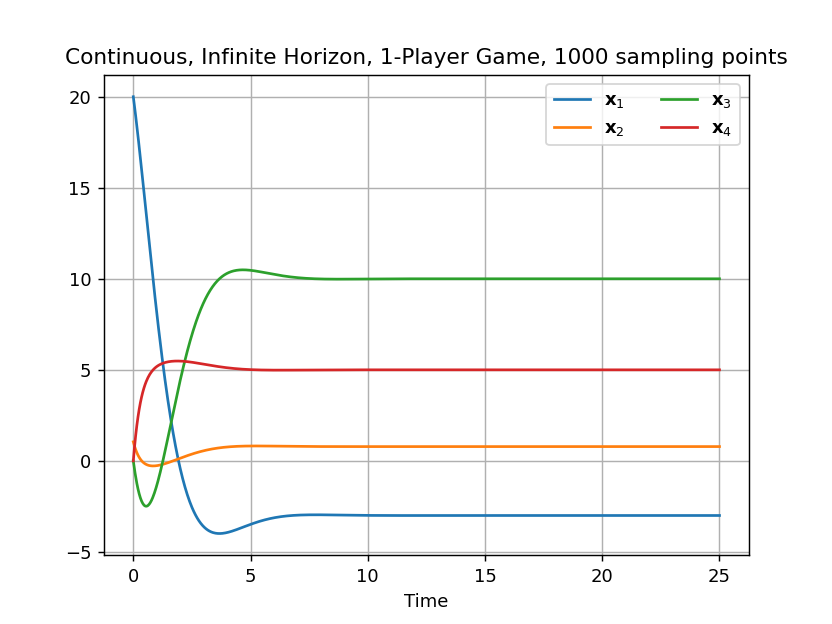}
\caption{LQR with $(m_c, m_p, L) \equiv (50, 8, 3)$}
\end{figure}
\end{center}

\begin{center}
\begin{figure}[h!t!]
\centering
\includegraphics[width=0.85\columnwidth]{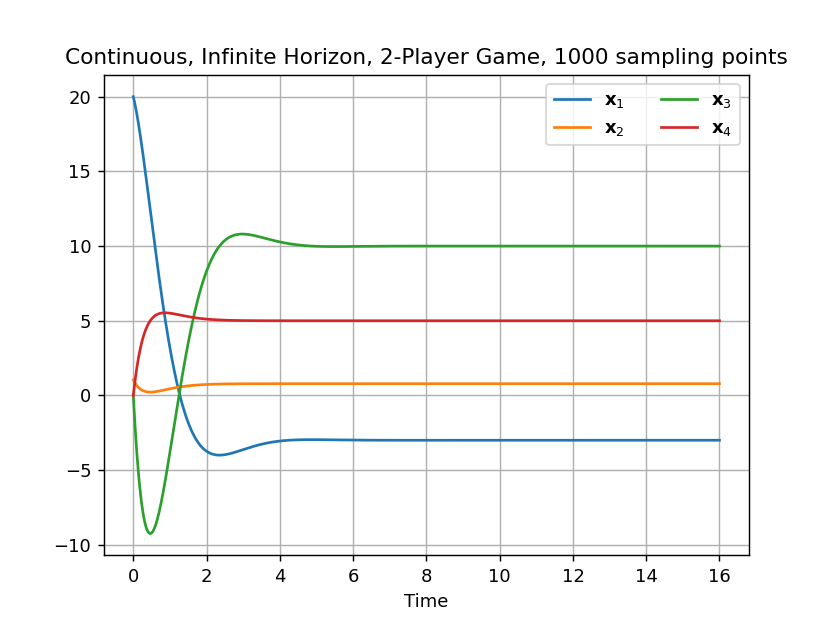}
\caption{Nash Equilibrium with $(m_c, m_p, L) \equiv (50, 8, 3)$}
\end{figure}
\end{center}


\begin{center}
\begin{figure}[h!t!]
\centering
\includegraphics[width=0.85\columnwidth]{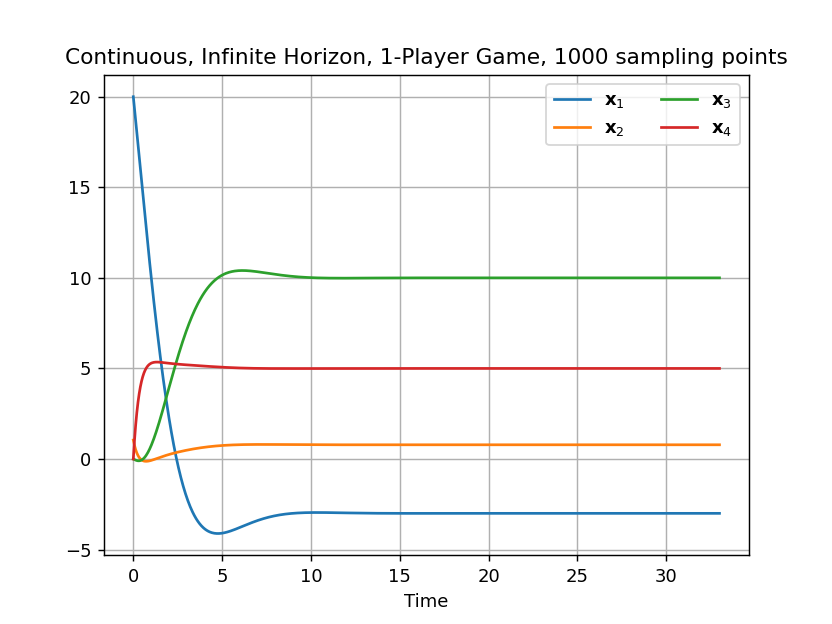}
\caption{LQR with $(m_c, m_p, L) \equiv (100, 10, 4)$}
\end{figure}
\end{center}

\begin{center}
\begin{figure}[h!t!]
\centering
\includegraphics[width=0.85\columnwidth]{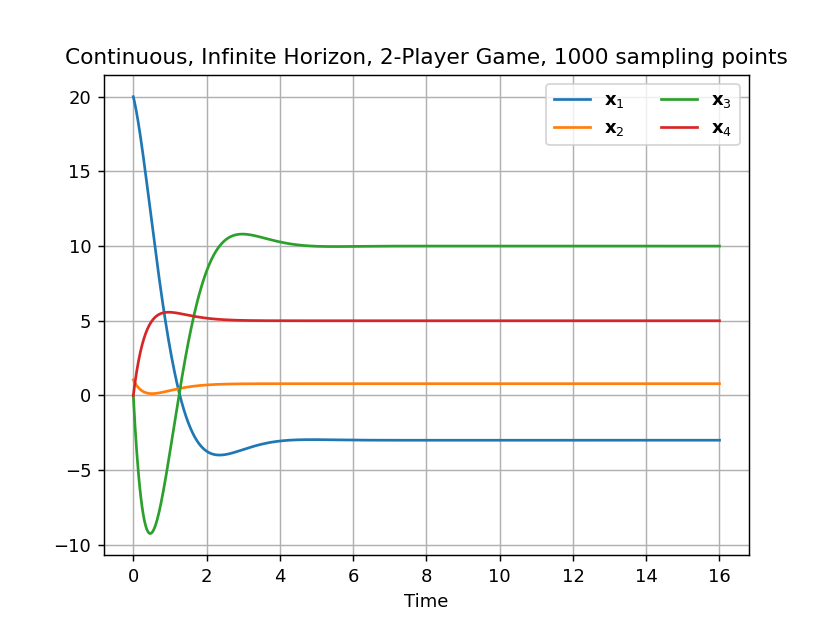}
\caption{Nash Equilibrium with $(m_c, m_p, L) \equiv (100, 10, 4)$}
\end{figure}
\end{center}


\begin{center}
\begin{figure}[h!t!]
\centering
\includegraphics[width=0.85\columnwidth]{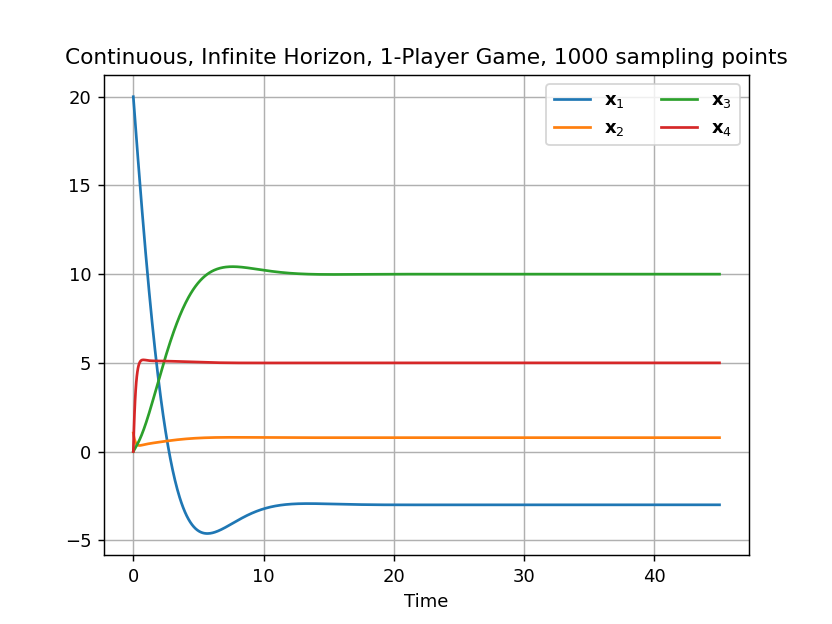}
\caption{LQR with $(m_c, m_p, L) \equiv (200, 15, 1)$}
\end{figure}
\begin{figure}[h!t!]
\centering
\includegraphics[width=0.85\columnwidth]{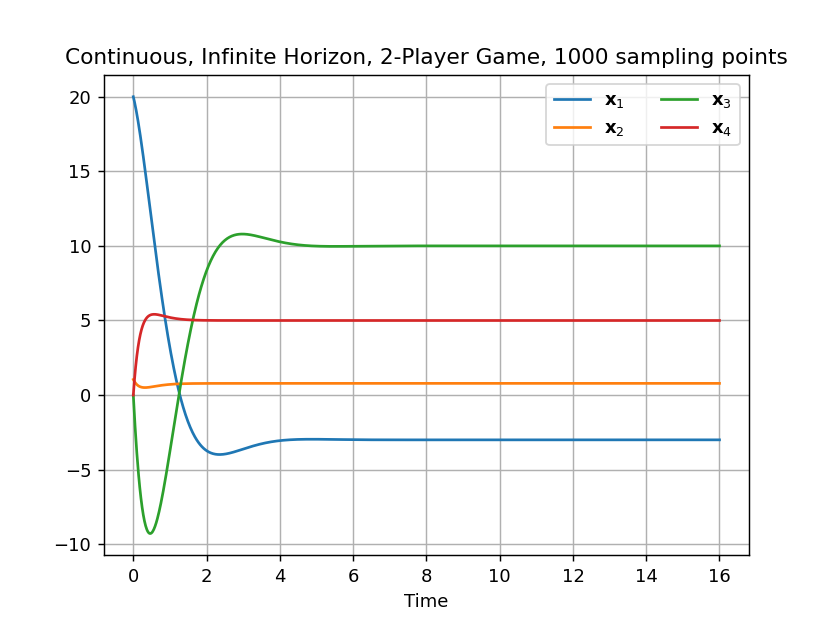}
\caption{Nash Equilibrium with $(m_c, m_p, L) \equiv (200, 15, 1)$}
\end{figure}
\end{center}

\section{Quadrotor}
\label{sect:quadrotor}
 A quadrotor has four rotors and six degrees of freedom, making it an underactuated system~\cite{underactuated}, meaning the number of its independent controlled inputs is less than the number of its configuration variables (spatial and angular variables). 
 
 The use of our suggested approach for the quadrotor allows us to apply control laws to selected subsets of the state variables separately and sequentially. This will allow us to achieve hierarchical control clusters, that operate at different time scales such that the lower ones help to control those higher. 
 
 In this section we will use this hierarchical control approach to showcase our methodology first in the context of a simpler low level control problem and then in the context of the more complex and dynamic high level control. 
\subsection{System Modeling}
Let us denote our system as $\mathcal{S}_Q$ where $Q$ is for quadrotor. Let us then define the state vector for $\mathcal{S}_Q$:
\begin{equation}
    \mathbf{x}(t) \coloneqq \begin{bmatrix}
    		\phi(t)\\
    		\dot{\phi}(t)\\
    		\theta (t)\\
    		\dot{\theta}(t)\\
    		\psi(t)\\
    		\dot{\psi}(t)\\
    		z(t)\\
    		\dot{z}(t)\\
    		x(t)\\
    		\dot{x}(t)\\
    		y(t)\\
    		\dot{y}(t)
    \end{bmatrix},
    \label{non-linear_state}
\end{equation}
where $x(t),y(t),z(t) \in \mathbb{R}$ are the positional variables of the system relative to a fixed three dimensional coordinate system in $\mathbb{R}^3$ and $\phi(t), \theta(t), \psi(t)$ represent the roll, pitch and yaw angles (respectively). 

So in this case we have that $n$ is $12$.

To specify the input vector of $\mathcal{S}_Q$, let us assume the following inputs:
\begin{itemize}
    \item Thrust, denoted $T(t) \in \mathbb{R}$ acting upon the quadrotor;
    \item Three torques along the axes $x,y,z$ as $\tau_1(t), \tau_2(t), \tau_3(t) \in \mathbb{R}$ respectively.
\end{itemize}
 Using that, let us define the input vector of $\mathcal{S}_Q$ as:

\begin{equation}\label{quad_input}
    \mathbf{u}(t) \coloneqq \begin{bmatrix}
    T(t)\\
    \tau_1(t)\\
    \tau_2(t)\\
    \tau_3(t)
    \end{bmatrix}.
\end{equation}
We will use a non-linear model for the quadrotor as suggested in~\cite{bouabdallah_thesis}. The model is described by equation~\eqref{eqn:dynamical_system} with the following expression for the non-linear function $f$ that will result in:
\begin{equation}\label{quad_model}
    \dot{\mathbf{x}}(t) = f\left(\mathbf{x}(t),\mathbf{u}(t))\right)\coloneqq\left[\begin{matrix}\begin{matrix}\begin{matrix}
    \dot{\phi}(t)\\
    a_1\dot{\theta}(t)\dot{\psi}(t)+a_2\dot{\theta}(t)\Omega_r(t)+b_1\tau_1(t)\\
    \dot{\theta}(t)\\
    \end{matrix}\\
    \dot{\phi}(t)
    a_3\dot{\psi}(t)- a_4\dot{\phi}(t)\Omega_r(t)+b_2\tau_2(t)\\
    \dot{\psi}(t)\\
    \end{matrix}
    \\
    \begin{matrix}
    a_5\dot{\theta}(t)\dot{\phi}(t) +b_3\tau_3(t)\\
    \dot{z}(t)\\
    g-\frac{1}{m}\cos\left({\phi(t)}\right)\cos\left({\theta (t)}\right)T(t)\\
    \end{matrix}\\
    \begin{matrix}\dot{x}\\
    \frac{1}{m}u_x(t) {T(t)} \\
    \begin{matrix}
    \dot{y}\\
    \frac{1}{m}u_y(t){T(t)} \\
    \end{matrix}\\
    \end{matrix}\\
    \end{matrix}\right],
\end{equation}
where:
\begin{itemize}
    \item $m$ is the overall mass of the quadrotor;
    \item $u_x(t) \coloneqq \cos\left({\phi(t)}\right)\sin\left({\theta(t)}\right)\cos\left({\psi(t)}\right)+\sin\left({\phi(t)}\right)\sin\left({\psi(t)}\right)$;\\
    $u_y(t) \coloneqq \cos\left({\phi(t)}\right)\sin\left({\theta(t)}\right)\sin\left({\psi(t)}\right)-\sin\left({\phi(t)}\right)\cos\left({\psi(t)}\right)$;
    \item $a_1 \coloneqq (I_{yy}-I_{zz})/I_{xx} \ ; \ a_2 \coloneqq J_r/I_{xx};\\
    a_3 \coloneqq  (I_{zz}-I_{xx})/I_{yy} \ ; \ a_4 \coloneqq J_r/I_{yy};\\
    a_5 \coloneqq  (I_{xx}-I_{yy})/I_{zz}$;
    \item $b_1 \coloneqq l/I_{xx} \ ; \ b_2 \coloneqq l/I_{yy} \ ; \ b_3 \coloneqq 1/I_{zz}$;
    \item $I_{xx}, I_{yy}, I_{zz}$ are the moments of inertia of the quadrotor relative to the corresponding axes;
    \item $J_r$ is the polar moment of inertia of the quadrotor;
    \item $l$ is the horizontal distance from one of the propellers center to the center of gravity of the quadrotor;
    \item $\Omega_r(t)$ is the overall residual propeller angular speed.
\end{itemize}
\subsection{Quadrotor Control Strategy}
We show the performance of the suggested methodology in the context of the quadrotor \textbf{hierarchical control strategy}~\cite{Hierarchical_control_1, Hierarchical_control_2} proposed in~\cite{hanoch}, that uses an \textbf{image-based visual servoing} to control micro aerial vehicles (MAVs) such as a quadrotor in indoor environments. The MAV in~\cite{hanoch} is equipped with an angular velocities sensors and a front facing camera that allows it to visualize the environment it moves in. The method proposed in~\cite{hanoch} shows a way to stabilize the MAV in the middle of a corridor using only the corridor lines and angular velocities measurements.

To simulate these circumstances, we will first assume the state of the system is fully known.
Then we assume the corresponding image that the quadrotor is witnessing is simulated according to the current state at hand. The block diagram, in Fig.~\ref{control_diagram}, provides a high-level representation of the simulated system.
\vspace{-5mm}
\begin{center}
\begin{figure}[h!t!]
\centering
\includegraphics[width=1\columnwidth]{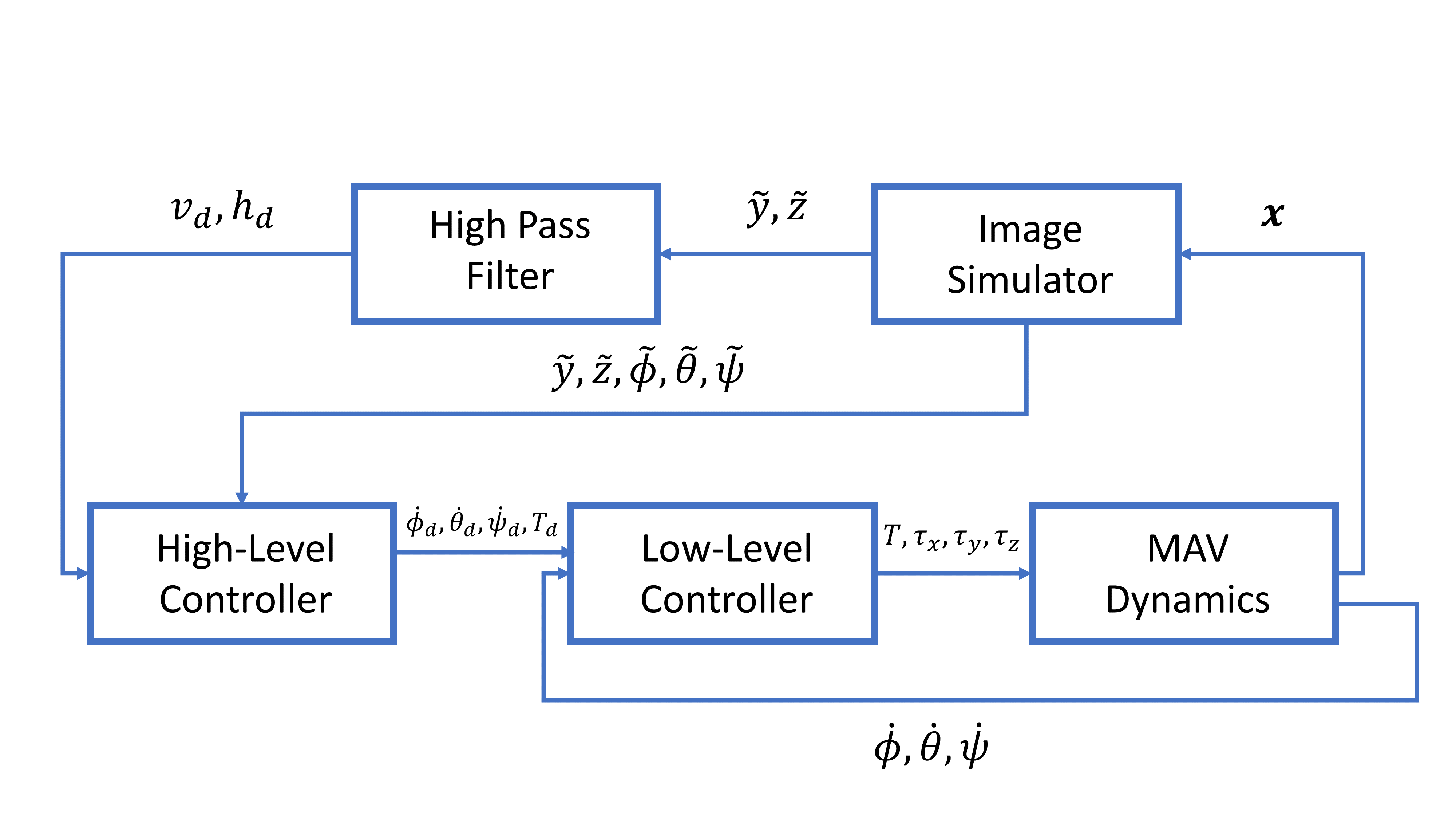}
\caption{Overview of the Control Diagram of $\mathcal{S}_Q$}
\label{control_diagram}
\end{figure}
\end{center}
Let us follow through the actual steps of this model:
\begin{enumerate}
    \item The MAV dynamics block is assigned with the model~\eqref{quad_model}. The full state of the system $\mathbf{x}(t)$ is obtained by solving the differential equation at~\eqref{quad_model};
    \item This state is fed into the image simulator that produces all the necessary variables for the simulation of the image the quadrotor observes;
    \item As described in~\cite{hanoch}, a High-Pass Filter (HPF)~\cite{HPF} is then applied to simulate differentiation of some of the variables for simplification;
    \item The parameters from the HPF and the image simulator are then passed to the high-level controller that produces the quadrotor's desired thrust and angular rates;
    \item Finally, the low-level controller receives the output from the high-level controller with the current angular velocities to produce control thrust and torques, $u$, to be applied to the quadrotor.
\end{enumerate}
    
\subsection{Image Output Parameters}
As described by Fig.~\ref{control_diagram}, the image generates approximations for the positions in the $y$ and $z$ axes and the angles $\phi(t)$, $\theta(t)$, $\psi(t)$ denoted as $\Tilde{y}(t)$, $\Tilde{z}(t)$, $\Tilde{\phi}(t)$, $\Tilde{\theta}(t)$, $\Tilde{\psi}(t)$ respectively. Then a High-Pass Filter is applied as a differentiator that acts upon $\Tilde{y}(t)$ and $\Tilde{z}(t)$ to generate their corresponding virtual derivatives $h_d(t)$ and $v_d(t)$ respectively. All of these are then passed into the high-level controller.

\subsection{Low-Level Control - Angular Rate Controller}
As illustrated by Fig.~\ref{control_diagram}, the low-level controller, denoted $\mathcal{S}_{LL}$, (which is the angular rate controller) inputs are the desired angular velocities and thrust, $\dot{\phi}_d(t), \dot{\theta}_d(t), \dot{\psi}_d(t), T_d(t)$. The expected outputs are the three torques along the axes, $\tau_1(t), \tau_2(t), \tau_3(t)$, and the desired thrust (pass-through).

\subsubsection{Low-Level Controller Model}
Let us define the following reduced state vector to define the lower-level system $\mathcal{S}_{LL}$:

\begin{equation}
    \mathbf{x}_{LL}(t) \coloneqq \begin{bmatrix}
    \dot{\phi}(t)\\
    \dot{\theta}(t)\\
    \dot{\psi}(t)
    \end{bmatrix}.
\end{equation}

So we have $n_{LL} \equiv 2$. We assume the following LTI system model for $\mathcal{S}_{LL}$:

\begin{equation}
    \dot{\mathbf{x}}_{LL}(t) = A_{LL}\mathbf{x}_{LL}(t) + B_{LL}\mathbf{u}_{LL}(t),
\end{equation}

where the non-linear system in~\eqref{quad_model} is dynamically linearized using an LPV modeling approach as in~\cite{bouabdallah_thesis}. The result is, as described in~\cite{bouabdallah_thesis}:

\begin{equation}
    A_{LL} \coloneqq \begin{bmatrix}
        0 & \frac{1}{2}a_1\dot{\psi}(t)  & \frac{1}{2}a_1\dot{\theta}(t) \\
        \frac{1}{2}a_3\dot{\psi}(t)  & 0 & \frac{1}{2}a_3\dot{\phi}(t) \\
        \frac{1}{2}a_5\dot{\theta}(t)  & \frac{1}{2}a_5\dot{\phi}(t) & 0 
    \end{bmatrix} \ ; \
     B_{LL} \coloneqq \begin{bmatrix}
        b_1 & 0  & 0 \\
        0  & b_2 & 0 \\
        0  & 0 & b_3 
    \end{bmatrix}.
\end{equation}

\subsubsection{Low-Level Controller Decomposition}
We assume three virtual control objectives, where $\mathbf{u}_{LL}(t)$ is defined as a linear combination of the scalar objective-matched control inputs, $u_1(t)$,$u_2(t)$,$u_3(t)$, as follows:

\begin{equation}
\mathbf{u}_{LL}(t) \coloneqq \begin{bmatrix}
\tau_1(t)\\
\tau_2(t)\\
\tau_3(t)
\end{bmatrix}=\begin{bmatrix}
u_1(t)\\
u_2(t)\\
u_3(t)
\end{bmatrix},
\end{equation}
so the augmented division matrix of $\mathcal{S}_{LL}$, namely $M_{LL}$, is simply the identity matrix:
\begin{equation}
\begin{aligned}
    M_{LL} \coloneqq \begin{bmatrix}
        1 & 0  & 0 \\
        0  & 1 & 0 \\
        0  & 0 & 1 \\
    \end{bmatrix} = I_{3}.
\end{aligned}
\end{equation}

Using $M_{LL}$ we calculate the objective-matched $(B_j)_{j=1}^3$:

\begin{equation}
\begin{aligned}
    B_1 \coloneqq \begin{bmatrix}
        b_1 \\
        0 \\
        0 
    \end{bmatrix}
    B_2 \coloneqq \begin{bmatrix}
        0 \\
        b_2 \\
        0 
    \end{bmatrix}
    B_3 \coloneqq \begin{bmatrix}
        0 \\
        0 \\
        b_3 
    \end{bmatrix}
\end{aligned}
\end{equation}

For each objective, the user-defined weight matrices $(Q_i)_{i=1}^3, (R_i)_{i=1}^3$ are:
\begin{equation}
\begin{aligned}
    Q_1\coloneqq & \operatorname{diag}[10^3,10^2,10^2]\\
    Q_2\coloneqq & \operatorname{diag}[10^2,10^3,10^2]\\
    Q_3\coloneqq & \operatorname{diag}[10^2,10^2,10^3]\\
    R_1\coloneqq & R_2\coloneqq R_3\coloneqq 0.1
\end{aligned}
\end{equation}
One of the motivations for the decomposition to several objectives is the ability to choose different weight matrices $Q_i$ and $R_i$ that represent different trade-offs by different objectives. 

To compute the control signals we used the method suggested in Chapter~\ref{chap:d&c} with the above objective matched parameters. Fig.~\ref{fig:Angular_rate_control} presents the results of the angular rate controller used to control $\dot{\theta}(t)$.  

\begin{figure}[ht]
    \centering
    \includegraphics[width=0.82\columnwidth]{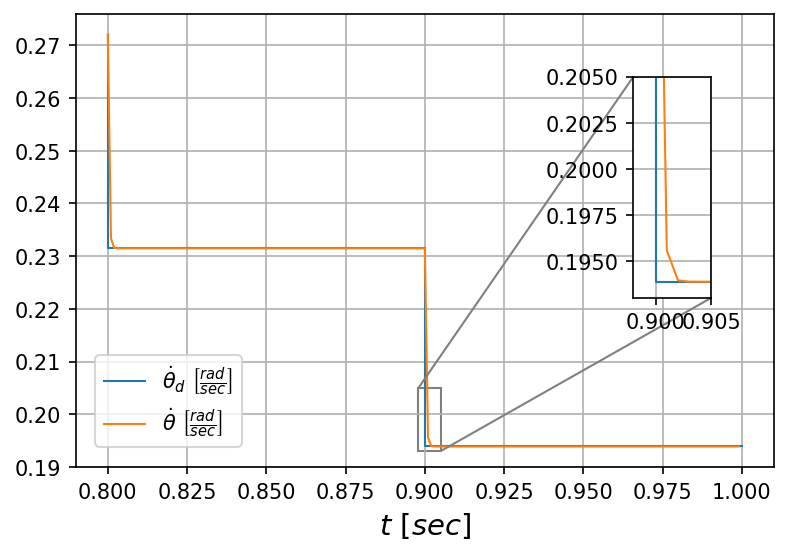}
    \caption{Performance of the angular rate controller of $\mathcal{S}_Q$, where $\dot{\theta}(t)$, the orange line, is tracking the desired angular rate $\dot{\theta}_d(t)$, the blue line.}
    \label{fig:Angular_rate_control}
\end{figure}

\subsection{High-Level Control}
Let $\mathcal{S}_{HL}$ be the high-level controller system, with a state variable $\mathbf{x}_{HL}(t)$ of the form:
\begin{equation}
    \mathbf{x}_{HL}(t) \coloneqq \begin{bmatrix}
    \phi (t)\\
    \theta (t)\\
    \psi(t)\\ 
    y(t)\\
    \dot{y}(t)\\
    z(t)\\
    \dot{z}(t)\\
    {\Theta}_1(t)\\
    {\Theta}_2(t)
    \end{bmatrix},
\end{equation}
where ${\Theta}_i(t) :=\int_{\tau=0}^{t}{\theta(\tau)} d\tau$, but ${\Theta}_1$ differs from ${\Theta}_2$ by the desired converging value, as will be elaborated on. We have $n_{HL} \equiv 9$.

\subsubsection{High-Level Controller Model}
The following model is obtained from a linearization of the process around $\bar{\mathbf{x}}_{HL}(t) \equiv \mathbf{0_9}$:  

\begin{equation}
    \dot{\mathbf{x}}_{HL}(t) = A_{HL}\mathbf{x}_{HL}(t) + B_{HL}\mathbf{u}_{HL}(t),
\end{equation}
where:

\begin{equation}
    A_{HL} \coloneqq \begin{bmatrix}
        0 & 0 & 0 &  0 & 0 & 0 &  0 & 0 & 0 \\
        0 & 0 & 0 &  0 & 0 & 0 &  0 & 0 & 0 \\
        0 & 0 & 0 &  0 & 0 & 0 &  0 & 0 & 0 \\
        0 & 0 & 0 &  0 & 1 & 0 &  0 & 0 & 0 \\
        g & 0 & 0 &  0 & 0 & 0 &  0 & 0 & 0 \\
        0 & 0 & 0 &  0 & 0 & 0 &  1 & 0 & 0 \\
        0 & 0 & 0 &  0 & 0 & 0 &  0 & 0 & 0 \\
        0 & 1 & 0 &  0 & 0 & 0 &  0 & 0 & 0 \\
        0 & 1 & 0 &  0 & 0 & 0 &  0 & 0 & 0
    \end{bmatrix} \ ; \
    B_{HL} \coloneqq \begin{bmatrix}
        1 & 0 & 0 & 0\\
        0 & 1 & 0 & 0\\
        0 & 0 & 1 & 0\\
        0 & 0 & 0 & 0\\
        0 & 0 & 0 & 0\\
        0 & 0 & 0 & 0\\
        0 & 0 & 0 & -\frac{1}{m}\\
        0 & 0 & 0 & 0\\
        0 & 0 & 0 & 0
    \end{bmatrix}.
\end{equation}

\subsubsection{Low-Level Controller Decomposition}

We define two control objectives:
\begin{enumerate}
    \item The first objective is similar to the one described in~\cite{hanoch}, stabilizing the quadrotor at the center of the corridor, facing forward;
    \item The second one is a dynamic objective on the forward velocity,  $\dot{x}$. The objective changes as a function of an approximation of the relative distance of the quadrotor to the wall $c(t):=\left|\frac{\Tilde{y}(t)}{a_y}\right|$, where $a_y$ is half of the corridor width. If $c(t)<0.5$, which means the quadrotor is closer to the center of the corridor than it is to the wall, the objective is to reach a given constant forward velocity and if $c(t) \geq 0.5$ then the objective is to reduce forward velocity to zero.
\end{enumerate}
 
We then define $\mathbf{u}_{HL}(t)$ as a linear combination of two objective-matched control inputs, of the form:
\begin{equation}
    \mathbf{u}_1(t) \coloneqq \begin{bmatrix}
    u_{1_1}(t)\\
    u_{1_2}(t)\\
    u_{1_3}(t)\\
    u_{1_4}(t)
    \end{bmatrix} \in \mathbb{R}^4 ; \ u_2(t) \in \mathbb{R}.
\end{equation}
Then we set:
\begin{equation}
\mathbf{u} \coloneqq \begin{bmatrix} 
\dot{\phi}_d(t)\\  
\dot{\theta}_d(t) \\ 
\dot{\psi}_d(t)\\
T_d (t)
\end{bmatrix}
=\begin{bmatrix} 
u_{1_1}(t)\\   
u_{1_2}(t) + u_2(t) \\  
u_{1_3}(t)\\ 
u_{1_4} (t)
\end{bmatrix},
\end{equation}

so the matrix $M_{HL}$ is:

\begin{equation}
\begin{aligned}
    M = \begin{bmatrix}
        1  & 0 & 0 & 0 & 0  \\
        0  & 1 & 0 & 0 & 1  \\
        0  & 0 & 1 & 0 & 0  \\
        0  & 0 & 0 & 1 & 0  \\
    \end{bmatrix}.
\end{aligned}
\end{equation}

Using $M$ we calculate the objective-matched $(B_j)_{j=1}^2$:

\begin{equation}
    B_1 \coloneqq B \ ; \  B_2 \coloneqq \begin{bmatrix}
    0\\
    1\\
    0\\
    0\\
    0\\
    0\\
    0\\
    0\\
    0
    \end{bmatrix}.
\end{equation}
The first objective weight matrices $Q_1$ and $R_1$ are constant, with the following form:
\begin{equation}
    \begin{aligned}
        Q_1& = \operatorname{diag}[10,10,10,0.001,0.001,0.1,0.0005,0,0]\\
    R_1& = \operatorname{diag}[10,10,10,0.01]
    \end{aligned}
\end{equation}

\subsubsection{Dynamically Changing Objective}

Now we introduce a dynamically changing objective to illustrate our method fully. The second objective state weight matrix $Q_2(t)$ dynamically changes as a function of $c(t)$:

\begin{itemize}
    \item If $c(t)<0.5$, the weight related to the 8'th state-parameter, ${\Theta}_1$, which is proportional to the forward velocity, gets a high value. ${\Theta}_1$ is initialized to ${\Theta}_1(0)=-\frac{3}{9.81}$;
    \item If $c(t) \geq 0.5$, where the objective is to reduce forward velocity to zero, the weight related to the 9'th state parameter, ${\Theta}_2$, which is also proportional to the forward velocity, gets a high value. ${\Theta}_2$ is initialized to ${\Theta}_2(0)=0$.
\end{itemize}

 and:

\begin{equation}
\begin{aligned}
    Q_2(t) &\coloneqq \begin{cases} \operatorname{diag}[0,0,0,0,0,0,0,1000,0] & ; c(t)<0.5\\
    \operatorname{diag}[0,0,0,0,0,0,0,0,wall\_prox]& ; c(t) \geq 0.5
    \end{cases}\\
    R_2 &\coloneqq 10.
\end{aligned}
\end{equation}
\subsection{Simulation Results}
Figures~\ref{fig:xdot_control} and~\ref{fig:yz_control} present the simulation results of the high level controller, applied as a part of the entire system (as described at Fig.~\ref{control_diagram}). The initial non-linear system state vector~\ref{non-linear_state} was chosen to be,
\begin{equation}
     \mathbf{x_0} \coloneqq [
    		0.1, \
    		0, \
    		0, \
    		0, \
    		0.1, \
    		0, \
    		-1, \
    		0, \
    		0, \
    		0, \ 
    		0.3, \
    		0
    ]^T
\end{equation}
where:
\begin{itemize}
    \item The values $z=-1.25[m]$ and $y=0[m]$ represent the center of the corridor;
    \item The ceiling height is $a_z=-2.5 [m]$;
    \item The corridor width is $a_y=0.55[m]$;
    \item The $z$ axis is pointing down.
\end{itemize}

Fig.~\ref{fig:xdot_control} shows the effects of the dynamic objective has on the forward velocity. Figures~\ref{fig:attitude_control} and~\ref{fig:yz_control} show the performance of both the low level and high level controllers in archiving the two control objectives, the quadrotor converge to the center of the corridor while tracking the forward velocity according to $c(t)$.
\begin{figure}[H]
    \centering
    \includegraphics[width=0.7\columnwidth]{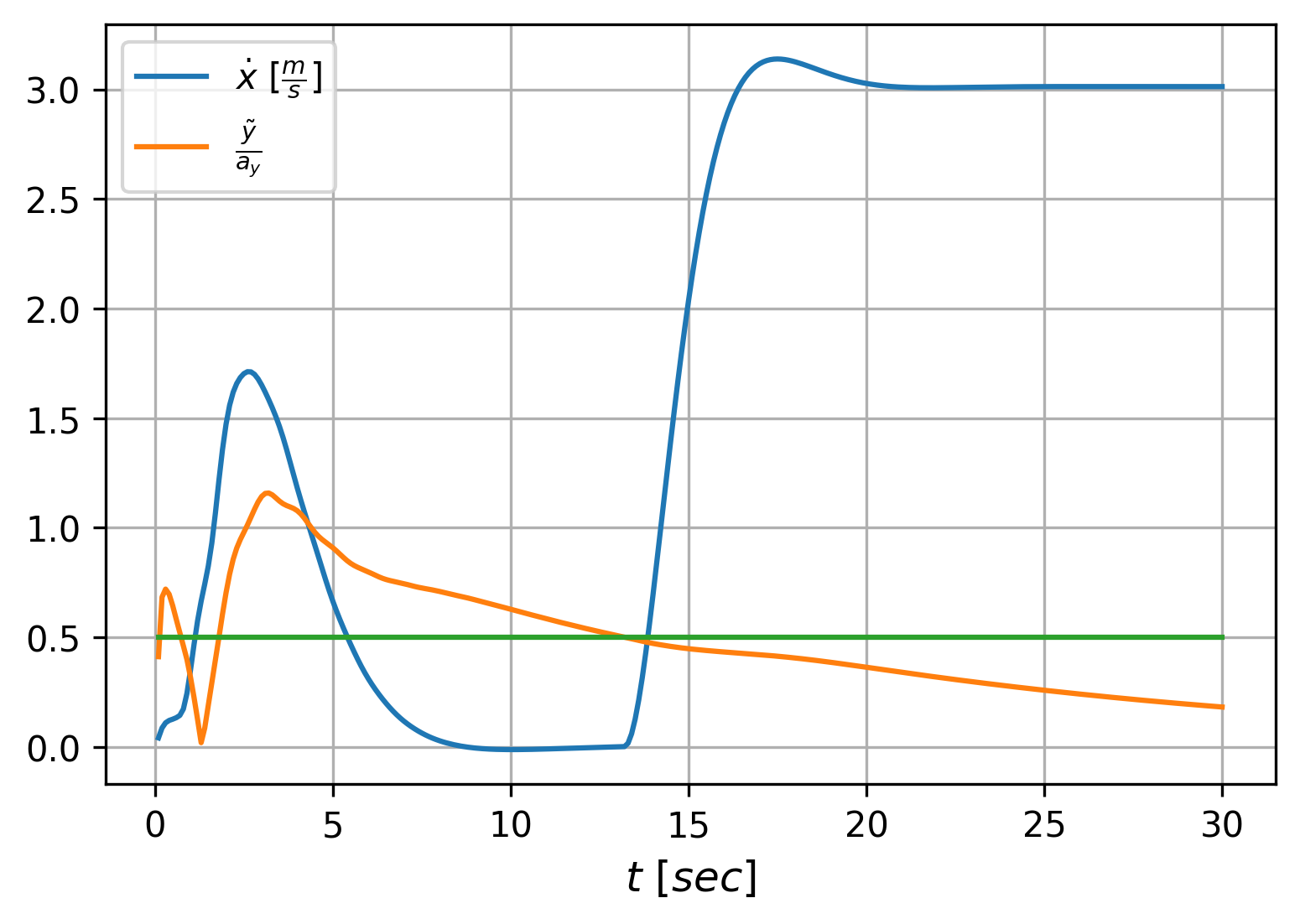}
    \caption{Performance of the high level controller of $\mathcal{S}_Q$; the affects of the dynamic objective has on the forward velocity.}
    \label{fig:xdot_control}
\end{figure}
\begin{figure}[ht]
    \centering
    \includegraphics[width=0.8\columnwidth]{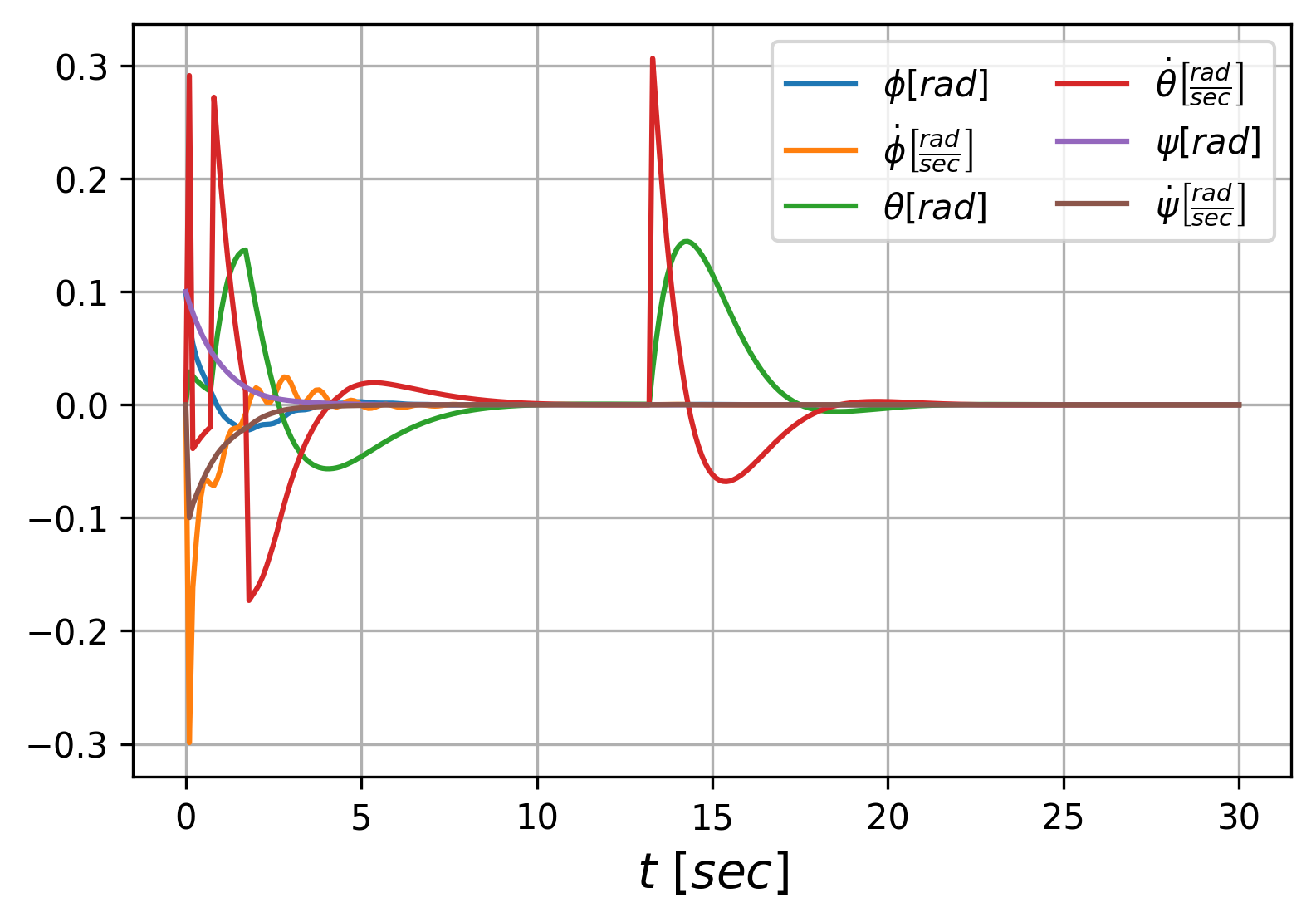}
    \caption{Performance of the high level controller of $\mathcal{S}_Q$; attitude regulation}
    \label{fig:attitude_control}
\end{figure}
\begin{figure}[ht]
    \centering
    \includegraphics[width=0.8\columnwidth]{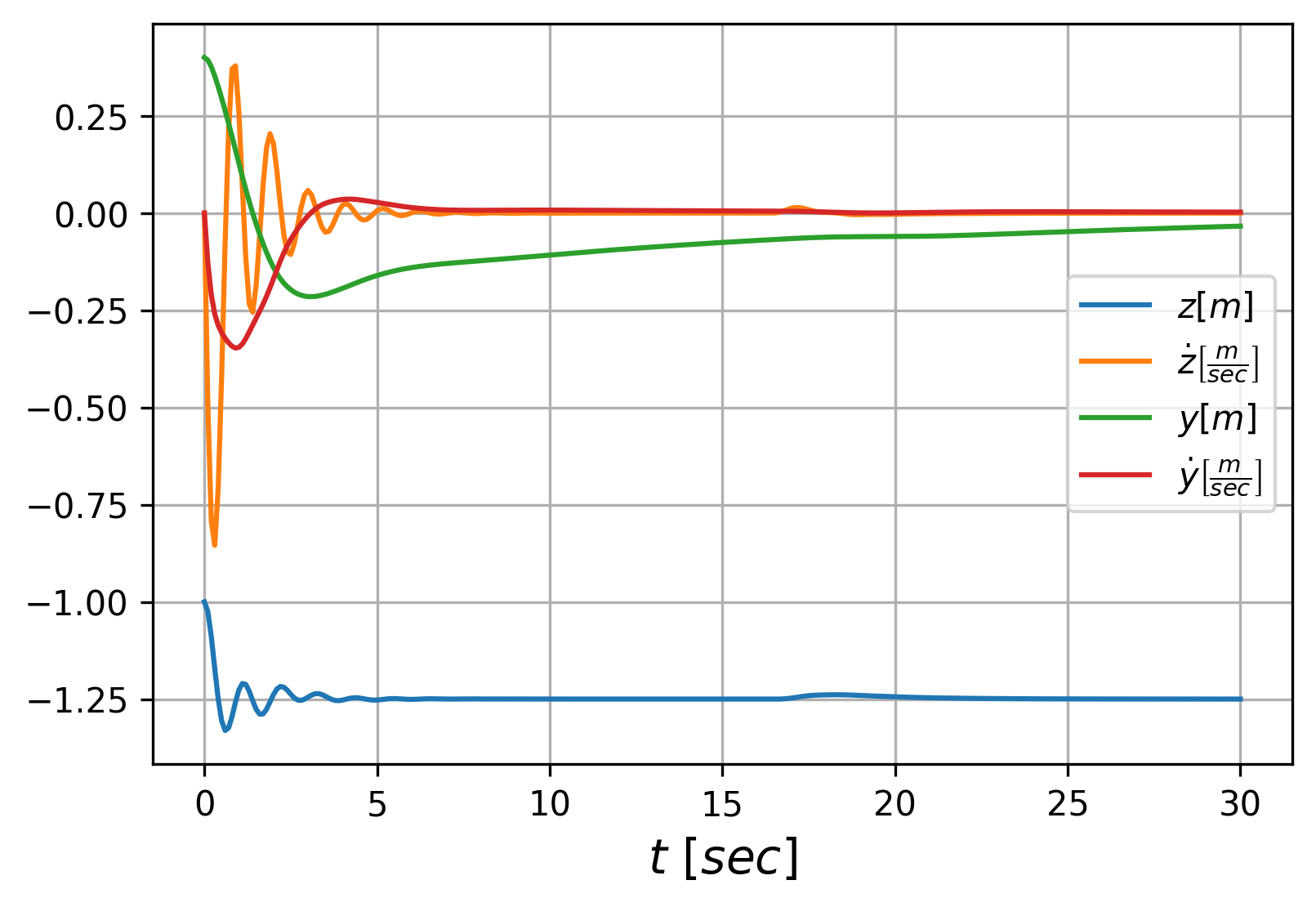}
    \caption{Performance of the high level controller of $\mathcal{S}_Q$; $y$ and $z$ regulation that represent convergence to the center of the corridor.}
    \label{fig:yz_control}
\end{figure}

\chapter{Discussion and Future Work}
\label{chap:conclusions}
\epigraph{\textit{In the 16th century, [Niccolò] Machiavelli - in an attempt to get back in the good graces of the powerful - wrote a slim volume called The Prince. In that book he showed the powers that be how to control the people. That book is a statement: separate and rule, divide and conquer. That's five hundred years ago and it still works, because we allow ourselves to be lead around with holes through our noses.}}{\textbf{Maya Angelou}}

In this work we explored different methodologies for dealing with the composition of complex controllers for single-agent multi-objective control systems. We reviewed previous work that was done in the field and then went on to extend some of the techniques to construct a novel methodology.

We have applied differential games to single-agent multi-objective control systems. While the concept of Divide and Conquer is well known in Computer Science and was used in control problems for many decades now, we have shown that the anecdotal paradigm is applicable to single-agent problems in control theory as well. 

We have laid out the mathematical foundations for the formulation of a non-zero-sum, non-cooperative differential game for players that correlate to virtual objectives, defined by decomposing a dynamical system as a control engineer sees fit, based on a preconceived variation of its actual inputs.

That differential game, under certain circumstances we explored, can yield a Nash Equilibrium. We saw in this case that the obtained Nash Equilibrium yields a dynamic balance that takes the virtual cost function into aggregated account. The balance that equilibrium provides was then used to construct corresponding control policies for each virtual actuator and then create a closed loop control policy framework under which the system can operate upon.

We developed an open-source Python package that implemented the described method, both for the finite horizon and infinite horizon case, and both for continuous and discrete time systems.

That same package was used to depict the state evolution of the systems we considered as case studies: First an inverted pendulum on a moving cart, with a force acting on the cart and pure torque acting on the pendulum. We carefully designed the weights to attribute for our desired virtual inputs, and use those same weights to simulate the system with an LQR controller. We set a convergence criterion for the state variable with respect to a predefined terminal desired state and saw across the board that for different parameter values we get faster convergence and better transient response using our carefully designed differential game.

Then we went on and illustrated with the quadrotor how dynamically changing objectives that can interfere with one another, such as keeping wall distance and moving forward, are handled using our methodology. We saw that a dynamically calculated Nash Equilibrium operates well in these scenarios, and can yield the desired system response with relative ease of design.

This work also introduced a novel algorithm for solving Algebraic Riccati Equations, which was given in psuedocode and implemented in the package we worked with. We laid the mathematical foundations for the convergence conditions of the solution of our algorithm to the unique optimizing, stabilizing solution of the optimal control problem.

In future work we intend on:

\begin{itemize}
    \item Authoring an article showcasing the novel ARE solution approach;
    \item Consider expanding the Python package to deal with more complex scenarios, such as linear time-varying systems;
    \item Consider the discrete time finite horizon Nash Equilibrium case.
\end{itemize}

\chapter{Appendices}
\label{chap:appendices}
\epigraph{\textit{But it is impossible to divide a cube into two cubes, or a fourth power into fourth powers, or generally any power beyond the square into like powers; of this I have found a remarkable demonstration. This margin is too narrow to contain it.}}{\textbf{Pierre de Fermat}}

In this section, we will present formal introductory as well as advanced concepts linear algebra and then in control theory, serving as an appendix to the theoretical preliminaries required to showcase our contribution. We will define several concepts and then deduce several notions we will consider as known theorems or lemmas, which we will not prove, and others as propositions and corollaries, which we will prove.


\section{Linear Algebra}
\begin{definition}
Let $k, h \in \mathbb{N}$. Then we denote:
\begin{itemize}
    \item $\mathbf{0}_k \in \mathbb{C}^k$ as a vector of $k$ zeros;
    \item $\mathbf{1}_k \in \mathbb{C}^k$ as a vector of $k$ ones;
    \item $0_{k,h} \in \mathbb{C}^{k \times h}$ as a matrix of all zeros;
    \item $1_{k,h} \in \mathbb{C}^{k \times h}$ as a matrix of all zeros;
    \item $I_k \in \mathbb{C}^{k \times h}$ as the identity matrix of order $k$.
\end{itemize}
\end{definition}
\begin{definition}
Let $r \in \mathbb{N}$ and let $\{ \mathbf{v}_i \}_{i=1}^r \subseteq \mathbb{C}^k$ be a set of $r$ vectors of length $k$. Then:
\begin{itemize}
    \item $\{ \mathbf{v}_i \}_{i=1}^r$ is \textbf{a linearly dependant set} if there exist a series of scalars $(\alpha_i)_{i=1}^r$, not all zero, such that:
\begin{equation}
    \sum_{i=1}^r \alpha_i \mathbf{v}_i = \mathbf{0}_k
\end{equation}
\item If $\{ \mathbf{v}_i \}_{i=1}^r$ is not linearly dependant then it is called a \textbf{linearly independant set}.
\end{itemize}
\end{definition}
\begin{definition}
Let $S = \{ \mathbf{v_i} \}_{i=1}^r \subseteq \mathbb{C}^k$ be a set of $r$ vectors of length $k$. The the \textbf{linear span} or \textbf{span} of $S$ is defined as the set of all finite linear combinations of vectors of S, i.e,:
\begin{equation}
    \spn (S ) \coloneqq {\displaystyle \left\{{\left.\sum _{i=1}^{k}\alpha _{i}\mathbf{v_i}\;\right|\;k\in \mathbb {N} ,\mathbf{v_i}\in S,\alpha _{i}\in \mathbb{C}}\right\}.}
\end{equation}
\end{definition}
\begin{definition}
Let $V \subseteq \mathbb{C}^k$ be a vector space of length $k$. A \textbf{basis} $B$ of $V$ is a linearly independent subset of $V$ that spans $V$, i.e. $B \subseteq V$ and:
\begin{equation}
    \spn (B ) = V.
\end{equation}
\end{definition}
\begin{definition}
Let $M \in \mathbb{C}^{k \times m}$ be a matrix. The \textbf{rank} of $M$, denoted as $\rank (M)$, is the maximal number of linearly independent columns of $M$.
\end{definition}
\begin{theorem}
Let $M \in \mathbb{C}^{k \times m}$ be a matrix. Then:
\begin{equation}
    \rank (M) \leq \min \{k, m \}
\end{equation}
\end{theorem}
\begin{definition}
Let $M \in \mathbb{C}^{k \times m}$ be a matrix. We say $M$ has \textbf{full rank} if:
\begin{equation}
    \rank (M) = \min \{k, m \}.
\end{equation}
\end{definition}
\begin{definition}
Let $M \in \mathbb{C}^{k \times m}$ be a matrix with elements $((m_{i,j})_{j=1}^m)_{i=1}^k$. The \textbf{conjugate transpose} of $M$ is denoted $M^\dagger \in \mathbb{C}^{m \times k}$ and is defined by considering the elements $m^\dagger_{i,j} \equiv (m_{j,i})^\dagger$.
\end{definition}
\begin{definition}
Let $M \in \mathbb{R}^{k \times m}$ be a real matrix with elements $((m_{i,j})_{j=1}^m)_{i=1}^k$. The \textbf{transpose} of $M$ is denoted $M^T \in \mathbb{R}^{m \times k}$ and is defined by considering the elements $m^T_{i,j} \equiv m_{j,i}$.
\end{definition}
\begin{definition}
Let $M \in \mathbb{C}^{k \times k}$ be a square matrix. Then:
\begin{enumerate}
    \item $M$ is \textbf{Hermitian} if $M = M^\dagger$;
    \item $M$ is \textbf{anti-Hermitian} if $M = -M^\dagger$.
\end{enumerate}
\end{definition}
\begin{definition}
Let $M \in \mathbb{R}^{k \times k}$ be a real square matrix. Then:
\begin{enumerate}
    \item $M$ is \textbf{symmetric} if $M = M^T$;
    \item $M$ is \textbf{antisymmetric} if $M = -M^T$.
\end{enumerate}
\end{definition}
\begin{definition}
Let $M \in \mathbb{C}^{k \times k}$ be a square matrix with elements $((m_{i,j})_{j=1}^k)_{i=1}^k$. Then:
\begin{itemize}
    \item The \textbf{main diagonal} of $M$ are the elements of the form $m_{i,i}$.
    \item The \textbf{superdiagonal} of $M$ are the elements of the form $m_{i,i+1}$.
\end{itemize}

\end{definition}
\begin{definition}
Let $M \in \mathbb{C}^{k \times k}$ be a square matrix with elements $((m_{i,j})_{j=1}^k)_{i=1}^k$. Then:
\begin{itemize}
    \item $M$ is \textbf{upper-triangular} if:
    \begin{equation}
        \forall 1 \leq j < i \leq k \ \colon \ m_{i,j} \equiv 0
    \end{equation}
    \item $M$ is \textbf{lower-triangular} if:
    \begin{equation}
        \forall 1 \leq i < j \leq k \ \colon \ m_{i,j} \equiv 0
    \end{equation}
    \item $M$ is \textbf{triangular} if it is upper-triangular or lower-triangular
\end{itemize}
\end{definition}
\begin{definition}
 Let $\mathbf{v} \coloneqq (v_i)_{i=1}^k \in \mathbb{C}^k$ be a vector of length $k$. Then the \textbf{matrix diagonal of $\mathbf{v}$}, denoted $\diag (\mathbf{v})$, is the matrix composed of the sorted elements of $\mathbf{v}$ on the main diagonal and zeros elsewhere, i.e.:
\begin{equation}
    \diag (\mathbf{v}) \coloneqq \begin{bmatrix}
    v_1 & 0 & \cdots & 0 & 0\\
    0 & v_2 & \cdots & 0 & 0\\
    \vdots & \vdots & & \vdots & \vdots\\
    0 & 0 &  \cdots & v_{k-1} & 0 \\
    0 & 0 & \cdots & 0 & v_k
    \end{bmatrix}.
\end{equation}
\end{definition}
\begin{definition}
Let $\mathbf{V} \coloneqq (V_i)_{i=1}^k \in \mathbb{C}^{k \times \sum_{i=1}^k k_i \times \sum_{i=1}^k k_i}$ be a series of $k$ square matrices where each one $V_i \in \mathbb{C}^{k_i \times k_i}$ is a square matrix with its corresponding shape $k_i$. Then a matrix $J$ is a \textbf{block matrix diagonal of $\mathbf{V}$}, denoted  $\diag (\mathbf{V})$, is the matrix composed of the matrices in $\mathbf{V}$ on the main diagonal and zeros elsewhere, i.e.:
\begin{equation}
    \diag (\mathbf{V}) \coloneqq \begin{bmatrix}
    V_1 & 0 & \cdots & 0 & 0\\
    0 & V_2 & \cdots & 0 & 0\\
    \vdots & \vdots & & \vdots & \vdots\\
    0 & 0 &  \cdots & V_{k-1} & 0 \\
    0 & 0 & \cdots & 0 & V_k
    \end{bmatrix}.
\end{equation}
\end{definition}
\subsection{Eigenvalues and Eigenvectors}
\begin{definition}
Let $M \in \mathbb{C}^{k \times k}$ be a square matrix. An \textbf{eigenvalue} $\lambda$ and an \textbf{eigenvector} $\mathbf{v}$ of the matrix $M$ are a complex scalar $\lambda \in \mathbb{C}$ and a vector $\mathbf{v} \in \mathbb{C}^{n} \setminus \{ \mathbf{0}_n \}$, respectively, satisfying:
\begin{equation}
      M \mathbf{v} = \lambda \mathbf{v}.
      \label{eqn:eig}
\end{equation}
\end{definition}
\begin{definition}
Let $M \in \mathbb{C}^{k \times k}$ be a square matrix. The \textbf{characteristic polynomial} $p_M(\lambda)$ of $M$ is defined by:
\begin{equation}
    p_M(\lambda) \coloneqq \det (\lambda I_k - M)
\end{equation}
where $\det$ is the \textbf{determinant} operator.
\end{definition}
\begin{theorem}
Let $M \in \mathbb{C}^{k \times k}$ be a square matrix and let $p_M(\lambda)$ be its characteristic polynomial. Then $p_M(\lambda)$ is a polynomial of degree $k$.
\label{the:char_1}
\end{theorem}
\begin{theorem}
Let $M \in \mathbb{C}^{k \times k}$ be a square matrix and $p_M(\lambda)$ be its characteristic polynomial. Then the roots of $p_M(\lambda)$ are precisely the eigenvalues of $M$.
\label{the:char_2}
\end{theorem}
\begin{corollary}
Let $M \in \mathbb{C}^{k \times k}$ be a square matrix. Then $M$ has exactly $k$, not necessarily unique, eigenvalues.
\label{cor:k_eig}
\end{corollary}
\begin{proof}
The claim follows directly from Theorems~\ref{the:char_1} and~\ref{the:char_2} with the Fundamental Theorem of Algebra.
\end{proof}
\begin{definition}
Let $M \in \mathbb{C}^{k \times k}$ be a square matrix and let $p_M(\lambda)$ be its characteristic polynomial. Let $\tilde{\lambda} \in \mathbb{C}$ be an eigenvalue of $M$, which by Theorem~\ref{the:char_2} is also a root of $p_M(\lambda)$. The \textbf{algebraic multiplicity of $\tilde{\lambda}$} is defined as the number of times it occurs as a root of $p_M(\lambda)$ and is denoted by $\mu_M (\lambda)$.
\end{definition}
\begin{definition}
Let $M \in \mathbb{C}^{k \times k}$ be a square matrix and let $p_M(\lambda)$ be its characteristic polynomial. The \textbf{spectrum} of $M$ is the multiset of its eigenvalues, i.e.:
\begin{equation}
    \sigma (M) \coloneqq \{ \lambda \in \mathbb{C} \ | \ \exists \mathbf{v} \in \mathbb{C}^k \setminus \{ \mathbf{0}_k \} \ \colon \ M \mathbf{v} = \lambda \mathbf{v} \},
\end{equation}
where each eigenvalue is listed according to its algebraic multiplicity $\mu_M (\lambda)$.
\end{definition}
\begin{corollary}
Let  $M \in \mathbb{C}^{k \times k}$ be a square matrix with corresponding spectrum $\sigma (M)$ and algebraic multiplicities $(\mu_M (\lambda))_{\lambda \in \sigma (M)}$. Then:
\begin{equation}
    \sum_{\lambda \in \sigma (M)} \mu_M (\lambda) = k.
\end{equation}
\label{cor:k_alg_sum}
\end{corollary}
\begin{proof}
The proof follows by the definition of algebraic multiplicity and Corollary~\ref{cor:k_eig}.
\end{proof}
\begin{theorem}
Let $M \in \mathbb{C}^{k \times k}$ be a square matrix. Then $M$ has $O(k)$ linearly independent eigenvectors.
\end{theorem}
\begin{theorem}
Let $M \in \mathbb{R}^{k \times k}$ be a real square matrix. Then a complex scalar $\lambda \in \mathbb{C} \setminus \mathbb{R}$ is an eigenvalue of $M$ if and only if $\lambda^\dagger$ is also an eigenvalue of $M$, where $\lambda^\dagger$ is the complex-conjugate of $\lambda$.
\end{theorem}
\subsection{Inverse Matrix}
\begin{definition}
Let $M \in \mathbb{C}^{k \times k}$ be a square matrix. Then $M$ is \textbf{invertible} if there exists a square matrix $X \in \mathbb{C}^{k \times k}$ such that:
\begin{equation}
    M X = X M = I_{k}.
\end{equation}
We then denote $M^{-1} \coloneqq X$ and refer to $M^{-1}$ as the \textbf{inverse} of $M$. A matrix which is not invertible is also called \textbf{singular}.
\end{definition}
\begin{theorem}
Let  $M \in \mathbb{C}^{k \times k}$ be a square matrix with corresponding spectrum $\sigma (M)$ and rank $\rank (M)$. The following statements are equivalent:
\begin{itemize}
    \item $M$ is invertible;
    \item $M$ has a full rank, i.e. $\rank (M) = k$;
    \item All of the eigenvalues of $M$ are non-zero, i.e. $0 \not\in \sigma (M)$.
\end{itemize}
\end{theorem}
\begin{definition}
Let $M \in \mathbb{R}^{k \times k}$ be a real square matrix. We say $M$ is \textbf{orthogonal} if:
\begin{equation}
    MM^T = I_{k},
\end{equation}
which is equivalent to:
\begin{equation}
    M^T = M^{-1}.
\end{equation}
Notice by this definition we have that any orthogonal matrix is also invertible.
\end{definition}
\begin{definition}
Let $M_1, M_2 \in \mathbb{C}^{k \times k}$ be two square matrices. We say $M_1$ and $M_2$ are \textbf{similar} if there exists an invertible matrix $D \in \mathbb{C}^{k \times k}$ such that:
\begin{equation}
    M_2 = D^{-1} M_1 D.
\end{equation}
\end{definition}
\begin{proposition}
Let $M_1, M_2 \in \mathbb{C}^{k \times k}$ be two similar matrices, so let $D \in \mathbb{C}^{k \times k}$ such that $M_2 = D^{-1} M_1 D$. Let $\lambda \in \mathbb{C}$ be an eigenvalue of $M_1$ associated with the eigenvector $\mathbf{v} \in \mathbb{C}^{k}$. Then $\lambda$ is also an eigenvalue of $M_2$ associated with the eigenvector $D^{-1}\mathbf{v}$.
\end{proposition}
\begin{proof}
Since $M_1$ and $M_2$ are similar:
\begin{equation}
    M_2 = D^{-1} M_1 D \Longleftrightarrow DM_2D^{-1} = M_1.
\end{equation}
For $M_1$, $\lambda$ and $\mathbf{v}$ we have by definition:
\begin{equation}
    M_1 \mathbf{v} = \lambda \mathbf{v},
\end{equation}
and thus:
\begin{equation}
    DM_2D^{-1} \mathbf{v} = \lambda \mathbf{v} \rightarrow M_2D^{-1} \mathbf{v} = \lambda D^{-1}\mathbf{v},
\end{equation}
thus $D^{-1} \mathbf{v}$ is an eigenvector of $M_2$ with eigenvalue $\lambda$.
\end{proof}
\subsection{Matrix Diagonalization}
\begin{definition}
Let $M \in \mathbb{C}^{k \times k}$ be a square matrix with corresponding spectrum $\sigma(M)$. Then:
\begin{itemize}
    \item $M$ is \textbf{diagonalizable} if it is similar to a diagonal matrix, i.e., there exists an invertible matrix $D \in \mathbb{C}^{k \times k}$ and a vector $\mathbf{v} \in \mathbb{C}^k$ such that:
\begin{equation}
    M = D^{-1} \Lambda D,
\end{equation}
and where $\Lambda \in \mathbb{C}^{k \times k}$ is defined by:
\begin{equation}
    \Lambda \coloneqq \diag \left(\mathbf{v}\right).
\end{equation}
In this case the matrix $D$ is said to be \textbf{$M$-diagonalizing}.
\item If $M$ is not diagonalizable, then it is called \textbf{defective}.
\end{itemize}
\end{definition}

\begin{theorem}
Let $M \in \mathbb{C}^{k \times k}$ be a square matrix with corresponding spectrum $\sigma (M)$. Then $M$ is diagonalizable, i.e. there exist an invertible matrix $D$ and a vector $\mathbf{v} \in \mathbb{C}^k$ such that $M = D^{-1}\diag \left( \mathbf{v} \right) D$ if and only if the following statements hold:
\begin{itemize}
    \item $M$ has exactly $k$ linearly independent eigenvectors;
    \item The eigenvectors of $M$ form a basis of $\mathbb{C}^k$;
    \item The vector $\mathbf{v}$ is composed of the eigenvalues of $M$, i.e., $\mathbf{v} \equiv (\lambda)_{\lambda \in \sigma (M)}$;
    \item The columns of the matrix $D$ are composed of the eigenvectors of $M$.
    
\end{itemize}

 The order of the eigenvalues and eigenvectors correspond, so $D$ and $\mathbf{v}$ are not necessarily unique.
\end{theorem}
\begin{theorem}[\textbf{Spectral Theorem}]
Let $M \in \mathbb{R}^{k \times k}$ be a real symmetric matrix. Then it can be diagonalized by a real orthogonal matrix.
\end{theorem}
\subsection{Real Matrix Functions}
\begin{definition}
Let $f \colon \mathbb{R} \rightarrow \mathbb{R}$ be a real function and $\mathcal{D} \subset \mathbb{R}$ an open interval. We say $f$ is \textbf{real analytic on $\mathcal{D}$} if it is given by a real power series which is convergent locally in $\mathcal{D}$, i.e. if for all $x_0 \in \mathcal{D}$ we can write:
\begin{equation}
    f(x) = \sum_{n=1}^\infty a_n (x - x_0)^n.
\end{equation}
\end{definition}
\begin{definition}
Let $M \in \mathbb{R}^{k \times k}$ be real a diagonalizable matrix with corresponding spectrum $\sigma(M)$, i.e., there exists an invertible matrix $D \in \mathbb{R}^{k \times k}$ such that: $M = D^{-1} \diag (( \lambda )_{\lambda \in \sigma(M)}) D$. Let $f \colon \mathbb{R} \rightarrow \mathbb{R}$ be a real analytic function on some open set $\mathcal{D} \subset \mathbb{R}$. Then we extend the definition of $f$ to act upon and return matrices, as in we define the \textbf{matrix function} $f \colon \mathbb{R}^{k \times k} \rightarrow \mathbb{R}^{k \times k}$ to be:
\begin{equation}
    f(M) = D^{-1} \diag \left(( f(\lambda) )_{\lambda \in \sigma(M)}\right) D,
\end{equation}.
\end{definition}
\begin{definition}
 Let $M \in \mathbb{R}^{k \times k}$ be a real diagonalizable matrix. As a special case of the former definition, we define the \textbf{matrix exponential}~\cite{lie_groups} of $M$ as follows:
\begin{equation}
    e^M \coloneqq \sum_{s=0}^\infty \frac{M^s}{s!}
\end{equation}
\end{definition}
\subsection{Jordan Normal Form}
\begin{definition}
Let $M \in \mathbb{C}^{k \times k}$ be a square matrix and let $\lambda \in \mathbb{C}$ be an eigenvalue of $M$. A \textbf{generalized eigenvector} of $M$ corresponding to the eigenvalue $\lambda$ is a vector $\mathbf{v} \in \mathbb{C}^k \setminus \{ \mathbf{0}_k \}$ satisfying:
\begin{equation}
    (M - \lambda I_k )^p \mathbf{v} = \mathbf{0}_k
\end{equation}
for some $p \in \mathbb{N}$. Notice that by that definition, eigenvectors are generalized eigenvectors with $p \equiv 1$.
\end{definition}

\begin{definition}
Let $M \in \mathbb{C}^{k \times k}$ be a square matrix and let $\lambda \in \mathbb{C}$ be an eigenvalue of $M$ with algebraic multiplicity $\mu_M (\lambda)$. A \textbf{$\lambda$-Jordan block of $M$}, denoted $J_{M , \lambda} \in \mathbb{C}^{\mu_M (\lambda) \times \mu_M (\lambda)}$, is a square upper-triangular matrix composed of zeroes everywhere except for the main diagonal, which is filled with the value of $\lambda$, and for the superdiagonal, which is composed of ones, i.e.:
\begin{equation}
J_{M , \lambda} \coloneqq \left.\left[ \vphantom{\begin{array}{c}1\\1\\1\\1\\1\end{array}}
                  \smash{\underbrace{
                      \begin{array}{ccccc}
                             \lambda&1&0&\cdots &0\\
                             0&\lambda&1&\cdots &0\\
                             \vdots&\vdots&\vdots&\ddots&\vdots\\
                             0&0&0&\cdots &1\\
                             0&0&0&\cdots &\lambda
                      \end{array}
                      }_{\mu_M (\lambda)\text{ columns}}}
              \right]\right\}
              \,_{\mu_M (\lambda)\text{ rows}}.
\end{equation}
\end{definition}
\vspace{10px}
\begin{definition}
Let $M \in \mathbb{C}^{k \times k}$ be a square matrix with corresponding spectrum $\sigma(M)$. With that, let us define $\mathbf{J}_M \in \mathbb{C}^{k \times \sum_{\lambda \in \sigma (M)} \mu_M (\lambda) \times \sum_{\lambda \in \sigma (M)} \mu_M (\lambda)}$ as such:
\begin{equation}
    \mathbf{J}_M \coloneqq (J_{M, \lambda})_{\lambda \in \sigma (M)}.
\end{equation}
Notice by Corollary~\ref{cor:k_alg_sum} we have $\sum_{\lambda \in \sigma (M)} \mu_M (\lambda) = k$ so $\mathbf{J}_M \in \mathbb{C}^{k \times k \times k}$. Then a matrix $J_M \in \mathbb{C}^{k \times k}$ is a \textbf{$M$-Jordan Matrix} if it satisfies:
\begin{equation}
    J_M \coloneqq \diag (\mathbf{J}_M) = \begin{bmatrix}
    J_{M, \lambda_1} & 0 & \cdots & 0 & 0\\
    0 & J_{M, \lambda_2} & \cdots & 0 & 0\\
    \vdots & \vdots & & \vdots & \vdots\\
    0 & 0 &  \cdots & J_{M, \lambda_{k-1}} & 0 \\
    0 & 0 & \cdots & 0 & J_{M, \lambda_{k}}
    \end{bmatrix},
\end{equation}
where $J_M$ is unique up to permutations of the $\lambda$-Jordan blocks.
\end{definition}
\begin{theorem}[\textbf{Jordan Normal Form}]
Let $M \in \mathbb{C}^{k \times k}$ be a square matrix with corresponding spectrum $\sigma(M)$. Then $M$ is similar a $M$-Jordan matrix $J_M$, i.e., there exist an invertible matrix $D \in \mathbb{C}^{k \times k}$ such that:
\begin{equation}
    D^{-1} M D = J_M = \diag \left((J_{M, \lambda})_{\lambda \in \sigma (M)}\right),
\end{equation}
where $D$ is composed of generalised eigenvectors of $M$. Notice $J_M$ and $D$ are unique up to permutations of the $\lambda$-Jordan blocks of $M$.
\end{theorem}
\begin{theorem}[\textbf{Real Jordan Normal Form}]
Let $M \in \mathbb{R}^{k \times k}$ be a real square matrix. Then $M$ is similar to a $M$-Jordan matrix $J_M$, i.e. $D^{-1} M D = J_M$, where $D$ is composed of generalised eigenvectors of $M$, where $D$ can be chosen to be real, i.e. $D \in \mathbb{R}^{k \times k}$.
\label{the:real_jordan}
\end{theorem}

\subsection{Definite Matrices}
\begin{definition}
Let $M \in \mathbb{C}^{k \times k}$ be a Hermitian matrix. Then: 
\begin{itemize}
    \item $M$ is \textbf{positive definite} if:
    \begin{equation}
        \forall \mathbf{w} \in \mathbb{C}^k \setminus \{\mathbf{0}_k\} \ \colon \ \mathbf{w}^\dagger M \mathbf{w} > 0,
    \end{equation}
    and we denote $M > 0$.
    \item $M$ is \textbf{positive semi-definite} if:
    \begin{equation}
        \forall \mathbf{w} \in \mathbb{C}^k \setminus \{\mathbf{0}_k\} \ \colon \ \mathbf{w}^\dagger M \mathbf{w}\geq 0,
    \end{equation}
    and we denote $M \geq 0$.
\end{itemize}
The same two definitions hold when $M$ is real and symmetric, but replacing the conjugate transpose with a regular transpose.
\end{definition}
\begin{theorem}
Let $M \in \mathbb{C}^{k \times k}$ be a Hermitian matrix. Then:
\begin{enumerate}
    \item $M$ is positive definite if and only if all of its eigenvalues are all real and positive;
    \item $M$ is positive semi- definite if and only if all of its eigenvalues are all real and non-negative.
\end{enumerate}
\end{theorem}
\subsection{Hamiltonian Matrices}
\label{subsect:Hamiltonian_Matrices}
\begin{definition}
Given a real symmetric matrix $M \in \mathbb{C}^{k \times k}$, an \textbf{invariant subspace} of $M$ is a subspace $\mathcal{W} \subseteq \mathbb{C}^{k}$ that is preserved by $M$, as in it satisfies:

\begin{equation}
    \forall \mathbf{w} \in \mathbb{C}^k \ ; \ \mathbf{w} \in \mathcal{W} \longrightarrow M\mathbf{w} \in \mathcal{W}.
    \label{eqn:inv_sub_cond}
\end{equation}

We then refer to $\mathcal{W}$ as \textbf{$M$-invariant}.
\end{definition}

\begin{definition}
For all $k \in \mathbb{N}$, let us define the \textbf{skew-symmetric matrix of order $k$} as a matrix $W_k \in \mathbb{R}^{2k \times 2k}$ with the following form:

\begin{equation}
    W_k \coloneqq \begin{bmatrix}
    0_{k \times k} & -I_{k}\\
    I_{k} &0_{k \times k}
    \end{bmatrix}.
\end{equation}
\end{definition}

\begin{definition}
For all $k \in \mathbb{N}$, let us consider $W_k$, the skew-symmetric matrix of order $k$, to define that a matrix $H \in \mathbb{R}^{2k \times 2k}$ is \textbf{Hamiltonian} if we have:

\begin{equation}
    W_k H^T W_k = H_k.
    \label{eqn:hamiltonian_matrix_condition}
\end{equation}
\end{definition}

\begin{proposition}
Let $k \in \mathbb{N}$. Then $W_k$, the skew-symmetric matrix of order $k$, is orthogonal.
\end{proposition}
\begin{proof}
\begin{equation}
    W_k W_k^T =  \begin{bmatrix}
    O_{k \times k} & -I_{k}\\
    I_{k} &O_{k \times k}
    \end{bmatrix} \begin{bmatrix}
    O_{k \times k} & I_{k}\\
    -I_{k} &O_{k \times k}
    \end{bmatrix} = \begin{bmatrix}
    I_{k} & 0_{k \times k}\\
    0_{k \times k} &I_{k}
    \end{bmatrix}  = I_{2k \times 2k}.
\end{equation}
\end{proof}
\begin{proposition}
Let $H \in \mathbb{R}^{2k \times 2k}$ be Hamiltonian. Then its spectrum $\sigma (H)$ is symmetric about the imaginary axis.
\label{prop:H_symmetry}
\end{proposition}
\begin{proof}
Notice $W_k$ is antisymmetric so $W_k = - W_k^T$. We also have from the lemma that $W_k^T = W_k^{-1}$. Thus $W_k = - W_k^{-1}$.  Let us replace $W_k$ from the left of the LHS of \eqref{eqn:hamiltonian_matrix_condition} with $-W_k^{-1}$, multiply both sides by $-1$ and transpose both sides to have:

\begin{equation}
    W_k^{-1} H W_k = -H^T.
\end{equation}

Thus we have that $H$ and $-H^T$ are similar matrices, and thus they have the same spectrum. So $\lambda \in \mathbb{C}$ is an eigenvalue of $H$ if and only if $-\lambda^{\dagger}$ is also, where $\lambda^{\dagger}$ is the complex conjugate of $\lambda$.
\end{proof}

\section{Continuous System Modeling}
\label{app:sys_modeling}
Let us set the proper definitions to describe how we define systems and the means to model them in this work.
\subsection{Basic Continuous System Definitions}
\begin{definition}
A \textbf{system} $\mathcal{S}$ is a set of objects in a geometrical space. A system can be any physical set of objects in a real-life scenario, with the condition that they behave in a deterministic manner.
\end{definition}
\begin{definition}
A \textbf{dynamical system} is a system $\mathcal{S}$ that is governed by a function that describes its time dependence in the geometrical space it is defined at.~\cite{Dynamical_Systems}
\end{definition}
\begin{definition}
Given a dynamical system $\mathcal{S}$, $t_0 \in \mathbb{R}^+ \cup \{ 0 \}$ is the \textbf{initial time} the evolution of $\mathcal{S}$ starts at. In many cases we set $t_0 \equiv 0$.
\end{definition}

\begin{definition}
Given a dynamical system $\mathcal{S}$, $T_f \in \mathbb{R}^+ \cup \{ \infty \}$ is the \textbf{horizon of evolution}, which is the terminal time $\mathcal{S}$ operates at. It can either be finite or infinite.
\end{definition}
\begin{definition}
Let $\mathcal{S}$ be a dynamical system with initial time $t_0$ and horizon $T_f$. For any $t \in [t_0, T_f)$, we define the \textbf{$t$-suffix time interval} $\mathcal{I}(t) \colon [t_0, T_f) \rightarrow \mathcal{P}(\mathbb{R})$\footnote{Given a set $S$, $\mathcal{P}(S)$ is the power set of $S$, which is the set of all subsets of $S$.} as:
      \begin{equation}
     \mathcal{I}(t) \coloneqq \begin{cases}
    [t, T_f] & ;T_f < \infty \\
    [t, T_f) & ;\textrm{else}.
\end{cases}
\end{equation}
\end{definition}
\begin{definition}
Given a dynamical system $\mathcal{S}$ with initial time $t_0$ and horizon $T_f$, the system's \textbf{time interval} $\mathcal{I}$ is defined as the $t_0$-suffix time interval, i.e.:
      \begin{equation}
     \mathcal{I} \coloneqq \mathcal{I}(t_0) = \begin{cases}
    [t_0, T_f] & ;T_f < \infty \\
    [t_0, T_f) & ;\textrm{else}.
\end{cases}
\label{eqn:I_interval}
\end{equation}
\end{definition}
\begin{definition}
  One way to model a dynamical physical system is by means of the \textbf{state-space representation}, a mathematical model of the given physical system as a set of inputs, outputs and state variables related by first-order differential equations for continuous time or difference equations for discrete time~\cite{A_First_Course_in_Differential_Equations}. Given a dynamical system $\mathcal{S}$ operating on the time interval $\mathcal{I}$, $\mathbf{x} \colon \mathcal{I} \rightarrow \mathcal{X}$ is the \textbf{state vector} of $\mathcal{S}$, where the vector subspace $\mathcal{X} \subseteq \mathbb{R}^n$ is defined for the length of $\mathbf{x}(t)$, namely $n \in \mathbb{N}$ and is composed of all possible state vectors $\mathcal{S}$ can assume. $\mathcal{X}$ is refereed to as the \textbf{reachable states set} of $\mathcal{S}$. As mentioned, it is a subspace, so the origin is among them ($\mathbf{0}_n \in \mathcal{X}$). $\mathbf{x}(t)$ is composed of all the time-dependant state variables $\big(x_i(t)\big)_{i=1}^n$ necessary to define the system at that point in time.
 \end{definition}
 \begin{definition}
 Given a dynamical system $\mathcal{S}$ operating on the time interval $\mathcal{I}$, $\mathbf{u} \colon \mathcal{I} \rightarrow \mathcal{U} $ is the \textbf{input vector} where the vector subspace $\mathcal{U} \subseteq \mathbb{R}^{m}$ is defined for the length of $\mathbf{u}(t)$, namely $m \in \mathbb{N}$ and is composed of all possible input vectors $\mathcal{S}$ can assume - which is equivalent to the condition that they are \textbf{admissible}, a definition dependant upon a given \textbf{cost function} $J(t)$, which we will describe in the optimal control chapter. $\mathcal{U}$ is refereed to as the \textbf{admissible states set} of $\mathcal{S}$, with regards to $J(t)$. $\mathbf{u}(t)$ is composed of all the time-dependant individual input components $\big(u_i(t)\big)_{i=1}^m$.
 \end{definition}
 \begin{definition}
 Given a dynamical system $\mathcal{S}$ operating on the time interval $\mathcal{I}$, with state and input vectors $\mathbf{x}(t) \in \mathcal{X}$ and $\mathbf{u}(t) \in \mathcal{U}$,  the \textbf{dynamic function} ${\displaystyle f \colon \mathcal{X} \times \mathcal{U} \rightarrow \mathbb{R}^n}$ of $\mathcal{S}$ is a continuous function for all $t \in \mathcal{I}$ and $\big(\mathbf{x}(t), \mathbf{u}(t) \big)  \in \mathcal{X} \times \mathcal{U}$.
 \end{definition}
 \begin{definition}
 Given a dynamical system $\mathcal{S}$ operating on the time interval $\mathcal{I}$ and with a dynamic function $f$, the \textbf{state-space model} of $\mathcal{S}$ is defined by means of the following first order differential equation~\cite{intelligent_control}:
 
 \begin{equation}
     \dot{\mathbf{x}}(t) = f\big(\mathbf{x}(t), \mathbf{u}(t)\big),
     \label{eqn:dynamical_system}
 \end{equation}
  \end{definition}
  \begin{remark}
  The dot notation in~\eqref{eqn:dynamical_system} indicates the first order time derivative:
 \begin{equation}
     \dot{\mathbf{x}}(t) \coloneqq \frac{\mathrm{d}}{\mathrm{d}t}\mathbf{x}(t).
 \end{equation}
  \end{remark}

 \begin{remark}
 Given an LTI system $\mathcal{S}$ and its state vector $\mathbf{x}(t)$, Equation~\eqref{eqn:dynamical_system} constitutes a timed differential equation with respect to the temporal function $\mathbf{x}(t)$. Thus, in that context, let us assume $\mathbf{x}(t)$ is differentiable for all $t \in \mathcal{I}$.
 \end{remark}
 \begin{definition}
   Given a system $\mathcal{S}$ with initial time $t_0$ and state vector $\mathbf{x}(t)$, the \textbf{initial condition} of $\mathcal{S}$ is defined as:
   \begin{equation}
       \mathbf{x_0} \coloneqq \mathbf{x}(t_0)
   \end{equation}
 \end{definition}
 \begin{definition}
 Given a dynamical system $\mathcal{S}$, its associated \textbf{Cauchy Initial Value Problem} is defined as its initial condition $\mathbf{x_0}$ and state space model as described in Equation~\eqref{eqn:dynamical_system}.
 \label{def:civp}
 \end{definition}
 \begin{definition}
 Let $\mathcal{S}$ be a dynamical system with initial condition $\mathbf{x_0}$ and let $r \in \mathbb{R}^+$. The \textbf{$r$-distanced states} of $\mathcal{S}$ are defined as:
 
 \begin{equation}
     \mathcal{X}_r \coloneqq \{ \mathbf{x}(t) \in \mathcal{X} \ | \ \| \mathbf{x}(t) - \mathbf{x_0} \| \leq r \}.
 \end{equation}
 \end{definition}
 \begin{definition}
  Let $\mathcal{S}$ be a dynamical system operating on the time interval $\mathcal{I}$ and with a dynamic function $f \colon \mathcal{X} \times \mathcal{U} \rightarrow \mathbb{R}^n$. We say $f$ is \textbf{locally Lipschitz-continuous} if:
  \begin{enumerate}
      \item There exists $r \in \mathbb{R}^+$ such that $\mathcal{X}_r$, the $r$-distanced states of $\mathcal{S}$, satisfy $\mathcal{X} \subseteq \mathcal{X}_r$.
      \item There exists $L \in \mathbb{R}^+$ such that:
      \begin{equation}
      \begin{aligned}
          \forall t \in \mathcal{I} \ ; \ & \forall \mathbf{u}(t) \in \mathcal{U} \ ; \ \forall \mathbf{x_1}(t), \mathbf{x_2}(t) \in \mathcal{X}_r \ ; \\
          & \| f\left(\mathbf{x_1}(t), \mathbf{u}(t)\right) - f\left(\mathbf{x_2}(t), \mathbf{u}(t)\right)\| \leq L \| \mathbf{x_1}(t) - \mathbf{x_2}(t) \|
      \end{aligned}
      \end{equation}
  \end{enumerate}
 \end{definition}
 \begin{theorem}[\textbf{Local Picard–Lindelöf Theorem for Dynamical Systems}~\cite{Methods_PDE}]
 \label{the:picard_lind}
Let $\mathcal{S}$ be a dynamical system with initial time $t_0$ and a dynamic function $f$. If $f$ is locally Lipschitz-continuous, then there exists $\varepsilon > 0$ such that the Cauchy Initial Value Problem associated with $\mathcal{S}$ has a unique solution in the time interval $[t_0 - \varepsilon, t_0 + \varepsilon]$.
\end{theorem} 
\begin{remark}
In this work we assume we are dealing with dynamical systems in which it is reasonable to assume Theorem~\ref{the:picard_lind} and then be guaranteed an existence and uniqueness of a solution to the problem.
\end{remark}

\begin{definition}
Given a dynamical system $\mathcal{S}$ operating on the time interval $\mathcal{I}$, $\mathbf{y} \colon \mathcal{I} \rightarrow \mathcal{Y}$ is the \textbf{output vector} where the vector subspace $\mathcal{Y} \subseteq \mathbb{R}^p$ is defined for some length $p \in \mathbb{N}$, is the set of all possible output vectors. $\mathbf{y}(t)$ is composed of all the time-dependant output variables $\big(y_i(t)\big)_{i=1}^p$.
\end{definition}
\begin{definition}
Given a dynamical system $\mathcal{S}$ operating on the time interval $\mathcal{I}$ with an output vector $\mathbf{y}(t)$, the \textbf{output of the system} is defined by the model:

\begin{equation}
    \mathbf{y}(t) = g\big( \mathbf{x}(t), \mathbf{u}(t) \big),
\end{equation}

where ${\displaystyle g:\mathcal{X} \times \mathcal{U}\rightarrow \mathcal{Y}}$ is a continuous function for all $t \in \mathcal{I}$ and $\big( \mathbf{x}(t), \mathbf{u}(t)\big) \in \mathcal{X} \times \mathcal{U}$.
The output vector of the system correlates to the variables the control engineer actually measures and observes, that are some function of the state vector.
\end{definition}
\subsection{Continuous System Stability}
\begin{definition}
Given a dynamical system $\mathcal{S}$, 
the \textbf{stability} of the system tries to reason about what happens with the system state after a long period of time, i.e. as $t \rightarrow \infty$. The state vector of the system $\mathcal{S}$ can converge to some specific constant value, or maybe converge to some time-periodical function, or otherwise diverge.~\cite{Terrell+2009}

\end{definition}

\begin{definition}
Given a dynamical system $\mathcal{S}$, the simplest kind of stable behavior is exhibited by convergence of the system to a set of \textbf{equilibrium points}, or \textbf{equilibria}, which is a vector subspace of state vectors  $\overline{\mathcal{X}} \subseteq \mathcal{X}$ satisfying:

\begin{equation}
    \overline{\mathcal{X}} \coloneqq \{ \overline{\mathbf{x}}(t) \in \mathcal{X}\ | \ \forall \tau \in \mathcal{I} \colon \ f\big( \overline{\mathbf{x}}(\tau)\big) = \mathbf{0}_n \}.
    \label{eqn:eqilibrium_points}
\end{equation}

As well as the fact that $\overline{\mathcal{X}}$ is a subspace, Equation~\ref{eqn:dynamical_system} and the definition of $\overline{\mathcal{X}}$ dictates that $\mathbf{0}_n \in \overline{\mathcal{X}}$, so $\overline{\mathcal{X}} \neq \emptyset$.
\end{definition}
\begin{remark}

  Given a dynamical system $\mathcal{S}$ and its equilibrium points $\overline{\mathcal{X}}$, let $\overline{\mathbf{x}}(t) \in \overline{\mathcal{X}}$. Since $f(\overline{\mathbf{x}}(t)) = \dot{\overline{\mathbf{x}}}(t) = \mathbf{0}_n$, we can always apply a coordinate transformation of the form $\overline{\mathbf{z}}(t) \coloneqq \mathbf{x}(t) - \overline{\mathbf{x}}(t)$
 such that:
 
 \begin{equation}
     \dot{\overline{\mathbf{z}}}(t) = \dot{\mathbf{x}}(t) - \dot{\overline{\mathbf{x}}}(t)= \dot{\mathbf{x}}(t) = f\big( \mathbf{x}(t) \big) = f\big( \overline{\mathbf{z}}(t) + \overline{\mathbf{x}}(t)\big).
     \label{eqn:change_coordinates}
 \end{equation}
 Defining $g \colon \mathcal{X} \rightarrow \mathbb{R}^n$ where $g\big( \overline{\mathbf{z}}(t) \big) \coloneqq f\big( \overline{\mathbf{z}}(t) + \overline{\mathbf{x}}(t)\big)$ and plugging into~\eqref{eqn:change_coordinates}
 we'll have:
 \begin{equation}
     \dot{\overline{\mathbf{z}}}(t) = g\big( \overline{\mathbf{z}}(t) \big),
 \end{equation}
 which describes a dynamical system, denoted $\overline{\mathcal{S}}$, with its state vector $\overline{\mathbf{z}}(t)$ shifted by $\overline{\mathbf{x}}(t)$ from $\mathbf{x}(t)$. Thus, in studying equilibrium points, it is sufficient to assume the equilibrium point occurs at nullity.
\end{remark}
\begin{definition}
Given a dynamical system $\mathcal{S}$ and its equilibrium points $\overline{\mathcal{X}}$, several definitions for the \textbf{stability of equilibrium points} exist.
Let us consider the nullity equilibrium point: $\overline{\mathbf{x}}(t) \equiv \mathbf{0}_n$. Then:

 \begin{itemize}
    \item $\overline{\mathbf{x}}(t)$ is said to be a \textbf{Lyapunov stable equilibrium point} if, for any $\varepsilon > 0 $, there exists $\delta >0$ such that:
    \begin{equation}
        \| \mathbf{x}(t_0) \| < \delta \Rightarrow \forall t \in \mathcal{I} \colon \ \| \mathbf{x}(t) \| < \varepsilon.
    \end{equation}
    \item $\overline{\mathbf{x}}(t)$ is said to be \textbf{asymptotically stable equilibrium point} if
    it is Lyapunov stable and there exists $\delta >0$ such that:
    \begin{equation}
        \| \mathbf{x}(t_0) \| < \delta \Rightarrow \lim_{t \rightarrow \infty}\| \mathbf{x}(t) \| = 0.
        \label{eqn:null_asymptotic_stability}
    \end{equation}
\end{itemize}

\end{definition}
\begin{definition}
 In this work we will define \textbf{system stability} as well. Given a dynamical system $\mathcal{S}$ and its equilibrium points $\overline{\mathcal{X}}$, we say it is Lyapunov stable if there exists an equilibrium $\overline{\mathbf{x}}(t) \in \overline{\mathcal{X}}$ of $\mathcal{S}$ that is Lyapunov stable, and likewise for asymptotic stability.
\end{definition}
\subsection{Continuous LTI Systems}
\label{app:LTI}
\begin{definition}
Given a dynamical system $\mathcal{S}$ with a dynamic function $f$, we say it is \textbf{linear-time-invariant} (LTI) if $f$ is of the form: $f\big(\mathbf{x}(t), \mathbf{u}(t)\big) \equiv A\mathbf{x}(t) + B\mathbf{u}(t)$, so Equation~\eqref{eqn:dynamical_system} becomes a first-order linear differential equation with time-invariant coefficients of the form:

\begin{equation}
    \dot{\mathbf{x}}(t) =  A\mathbf{x}(t) + B \mathbf{u}(t),
    \label{eqn:basic_sys}
\end{equation}

where:
\begin{itemize}
    \item $A \in \mathbb{R}^{n \times n}$ is a time-invariant matrix representing the \textbf{unforced dynamics} of $\mathcal{S}$;
    
    \item $B \in \mathbb{R}^{n \times m}$ is the time-invariant \textbf{input coefficients} matrix of $\mathcal{S}$.
\end{itemize}
\end{definition}
\begin{definition}

Given an LTI system $\mathcal{S}$ and its output function $g$, the \textbf{output of an LTI system} is defined by assigning a linear value for $g$ of the form:

\begin{equation}
    \mathbf{y}(t) = C \mathbf{x}(t) ,
\end{equation}

where $C \in \mathbb{R}^{n \times p}$ is the \textbf{output coefficients matrix}. 
\end{definition}

\begin{definition}
Given an LTI system $\mathcal{S}$, its corresponding \textbf{open-loop block diagram}~\cite{LTI} is visualized at figure~\ref{fig:S_diagram}:

\begin{figure}[H]
\centering
\begin{tikzpicture}
\node[draw,
    circle,
    minimum size=0.8cm
] (sum) at (0,0){};
 
\draw (sum.north east) -- (sum.south west)
    (sum.north west) -- (sum.south east);
 
\draw (sum.north east) -- (sum.south west)
(sum.north west) -- (sum.south east);
 
\node[left=-1pt] at (sum.center){\tiny $+$};
\node[below] at (sum.center){\tiny $+$};
 
\node [draw,
    minimum width=2cm,
    minimum height=1.2cm,
    right=1.5cm of sum
]  (integrator) {$\int \mathrm{d}t$};
 
\node [draw,
    minimum width=2cm, 
    minimum height=1.2cm,
    left=1cm of sum
] (B) {$B$};
 
\node [draw,
    minimum width=2cm, 
    minimum height=1.2cm, 
    below=1cm of integrator
]  (A) {$A$};

\node [draw,
    minimum width=2cm, 
    minimum height=1.2cm, 
    right=1.5cm of integrator
]  (C) {$C$};
 
\draw[-stealth] (sum.east) -- (integrator.west)
    node[midway,above]{$\dot{\mathbf{x}}(t)$};
 
 \draw[-stealth] (integrator.east) -- (C.west)
    node[midway,above](x){$\mathbf{x}(t)$};
    
\draw[-stealth] (C.east) -- ++ (2,0)
    node[midway,above](y){$\mathbf{y}(t)$};
 
 
\draw[-stealth] (x.south) |- (A.east);
 
\draw[-stealth] (A.west) -| (sum.south) 
    node[near end,left]{};
 
\draw[-stealth] ++(-4.5,0) 
    node[above]{$\mathbf{u}(t)$} -- (B.west);

\draw[-stealth] (B.east) -- (sum.west)
    node[midway,above]{};
 
\end{tikzpicture}
\caption{Open-Loop Block diagram of the system model of $\mathcal{S}$ described in~\eqref{eqn:basic_sys}.} \label{fig:S_diagram} 
\end{figure}
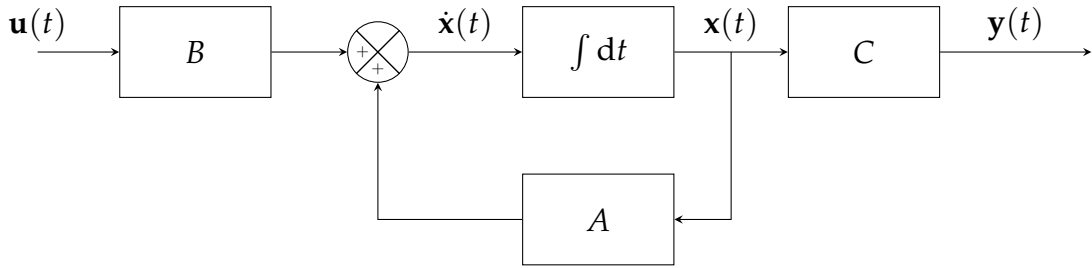
\end{definition}
This situation is referred to as '\textit{open-loop}' as there is no immediate connection between the state vector $\mathbf{x}(t)$ and the input $\mathbf{u}(t)$, i.e. the input of the system does not necessarily take the current system state into account.
\begin{remark}
Given an LTI system $\mathcal{S}$, from now on, when we say we are considering the system, we in effect assume we posses the exact constant values of $A$, $B$, $C$, $t_0, T_f$ (and $\mathcal{I}$ consequently) and $\mathbf{x_0}$, and leave the input vector $\mathbf{u}(t)$ not defined explicitly, as this term captures the notion that \textit{some} input of length $m$ can be fed into $\mathcal{S}$, and is case-specific based on the desired application. 

\end{remark}

\begin{definition}
Given an LTI system $\mathcal{S}$, the \textbf{forced state solution} of $\mathcal{S}$ is the temporal dynamics of the state vector $\mathbf{x}(t)$ while taking the effect of the input into account.
\end{definition}

\begin{remark}

Given an LTI system $\mathcal{S}$, the \textbf{scalar forced state solution} of $\mathcal{S}$ is obtained when considering a reduced version of Equation~\eqref{eqn:basic_sys} with $n=m=1$ for the system $\mathcal{S}$, adhering:

\begin{equation}
    \dot{x}(t)=ax(t)+bu(t).
    \label{eqn:scalar_lti}
\end{equation}

where $x(t), u(t), a, b \in \mathbb{R}$. This can be written as:

\begin{equation}
    \dot{x}(t) - ax(t) = bu(t).
    \label{eqn:scalar_basic_sys_1}
\end{equation}

Multiplying both sides of Equation~\eqref{eqn:scalar_basic_sys_1} by the integrating factor $e^{-at}$ makes the LHS a product derivative such that:

\begin{equation}
    \frac{\mathrm{d}}{\mathrm{d}t} \left( e^{-at} x(t) \right) = e^{-at} b u(t),
    \label{eqn:scalar_basic_sys_2}
\end{equation}

which can then be integrated directly. Taking the definite integral $\int_{t_0}^t \mathrm{d} \tau $ of both sides of Equation~\eqref{eqn:scalar_basic_sys_2} yields:

\begin{equation}
    \int_{t_0}^t\frac{\mathrm{d}}{\mathrm{d}\tau} \left( e^{-a\tau} x(\tau) \right) \mathrm{d}\tau = \int_{t_0}^te^{-a\tau} b u(\tau)\mathrm{d}\tau.
    \label{eqn:scalar_basic_sys_3}
\end{equation}

Since the integrand at the LHS is a derivative as mentioned, we have:

\begin{equation}
      e^{-at} x(t) -  e^{-a t_0} x(t_0)  = \int_{t_0}^te^{-a\tau} b u(\tau)\mathrm{d}\tau.
    \label{eqn:scalar_basic_sys_4}
\end{equation}

Moving sides and multiplying by $e^{at}$ we have:

\begin{equation}
      x(t) = e^{a(t- t_0)
      } x(t_0)  + \int_{t_0}^te^{a(t-\tau)} b u(\tau)\mathrm{d}\tau.
    \label{eqn:scalar_basic_sys_5}
\end{equation}
\end{remark}
\begin{definition}
Given an LTI system $\mathcal{S}$, in the case where $n,m>1$, the derivation of the state solution is done in a similar fashion as for the scalar case, by using $e^{-At}$ as the integrating factor, which is the  matrix exponential of $A$. Then the \textbf{higher order LTI forced state solution} of $\mathcal{S}$ at Equation~\eqref{eqn:basic_sys} is given by~\cite{LTI_book}:

\begin{equation}
\begin{aligned}
    \mathbf {x}(t) = & \underbrace{e^{A(t-t_0)}{\mathbf  {x_0}}}_{\mathbf{x_h}(t)}+\underbrace{\int _{{t_{0}}}^{t}e^{A(t-\tau)}{B}{\mathbf  {u}}(\tau )\mathrm{d}\tau}_{\mathbf{x_p}(t)} \\
    = & \mathbf{x_h}(t) + \mathbf{x_p}(t),
\end{aligned}
    \label{eqn:LTI_solution}
\end{equation}

which is the higher-order generalization of Equation~\eqref{eqn:scalar_basic_sys_5}. 
\end{definition}
\begin{remark}
Given an LTI system $\mathcal{S}$, the final expression for its state solution derived in Equation~\eqref{eqn:scalar_basic_sys_5} is comprised of two terms - the first one $\mathbf{x_h}(t)$ is where the term\textbf{ unforced dynamics} comes into place for $A$, as one can see that in the unforced case, $\dot{\mathbf{x}}(t) = A\mathbf{x}(t)$ the matrix $A$ (along with the initial condition $\mathbf{x_0}$) dictates the dynamic behaviour of the state vector, in what is called the \textbf{zero-input response}. The second term $\mathbf{x_p}(t)$ in Equation~\eqref{eqn:scalar_basic_sys_5} is the \textbf{convolution integral} of the input $\mathbf{u}(t)$ and it represents the particular solution to the input, with an initial condition of zero.
\end{remark}
\begin{definition}
Given an LTI system $\mathcal{S}$, the \textbf{poles} of $\mathcal{S}$ are defined as the spectrum of $A$.
\end{definition}
\begin{definition}
Given an LTI system $\mathcal{S}$ and its poles, let us define \textbf{LTI stability}. Plugging in Equation~\eqref{eqn:LTI_solution} to the condition at~\eqref{eqn:null_asymptotic_stability} we have that if $\mathcal{S}$ is asymptotically stable, then:

\begin{equation}
    \lim_{t \rightarrow \infty} \| \mathbf{x}(t) \| = \lim_{t \rightarrow \infty} \Bigg\| \mathbf e^{A(t-t_0)}{\mathbf  {x_0}}+\int _{{t_{0}}}^{t}e^{A(t-\tau)}{B}{\mathbf  {u}}(\tau )\mathrm{d}\tau \Bigg\| = 0.
    \label{eqn:unforced_LTI_asymtotic_stability}
\end{equation}

This must hold for all values of  for $\mathcal{S}$, so thus this corresponds to:

\begin{equation}
     \lim_{t \rightarrow \infty} \| e^{At} \| = 0,
\end{equation}

which is satisfied if all the poles of $\mathcal{S}$ are located at left-hand side of the complex plane, meaning:

\begin{equation}
     \forall \lambda \in \sigma ( A ) \colon \ \mathrm{Re}(\lambda) < 0.
\end{equation}

\end{definition}

\section{Continuous Controlled Systems}
\label{app:controlled_systems}

Now, after we set the proper definitions to consider physical systems, let us add definitions for setting controllers that act upon them.
Given a dynamical system $\mathcal{S}$, a \textbf{controller} is an element we add to the system which is designed specifically for some purpose. Mathematically, applying control is assigning a specific value for $\mathbf{u}(t)$ that is some carefully-designed temporal vector function, expected to take the necessary information into consideration, such that it outputs the desired value required by the demands set.
\subsection{Feedback Control}
Given a dynamical system $\mathcal{S}$, a \textbf{feedback} or \textbf{closed loop controller} is a controller which takes the current state of the system into consideration, i.e. it is a function of the state vector $\mathbf{x}(t)$. We will focus on such controllers in this work.
\begin{definition}
Given a dynamical system $\mathcal{S}$, a \textbf{linear feedback controller} is a feedback controller with the following form:

\begin{equation}
    \mathbf{u}(t) = - K(t) \mathbf{x}(t),
    \label{eqn:linear_controller}
\end{equation}

where $K \colon \mathcal{I} \rightarrow \mathbb{R}^{m \times n}$ is some temporal matrix function which maps a given state vector to the desired resulting controlled inputs. 

\end{definition}
\begin{remark}
The minus sign in~\eqref{eqn:linear_controller} is assigned to indicate this is a \textbf{negative} feedback controller, which tends to reduce the fluctuations of the resulting output. Negative feedback generally promotes stability, in contrast to \textbf{positive} feedback controller, that tends to lead to instability via exponential growth, oscillation or chaotic behavior.
\end{remark}
\begin{remark}

Given an LTI $\mathcal{S}$ system with a linear feedback controller $K$, plugging~\eqref{eqn:linear_controller} into~\eqref{eqn:basic_sys} we get the following controlled model:

\begin{equation}
    \dot{\mathbf{x}}(t) = \Big(A - BK(t)\Big) \mathbf{x}(t).
    \label{eqn:controlled_system}
\end{equation}

Equation~\eqref{eqn:controlled_system} illustrates the term \textbf{closed-loop control}, as the block diagram described in figure 
\ref{fig:S_diagram} can be redrawn to connect $\mathbf{x}(t)$ to $\mathbf{u}(t)$ as induced by Equation~\eqref{eqn:linear_controller} and thus close the loop between the plant output and input:

\begin{figure}[H]
\centering
\begin{tikzpicture}
\node[draw,
    circle,
    minimum size=0.8cm
] (sum) at (0,0){};
 
\draw (sum.north east) -- (sum.south west)
    (sum.north west) -- (sum.south east);
 
\draw (sum.north east) -- (sum.south west)
(sum.north west) -- (sum.south east);
 
\node[left=-1pt] at (sum.center){\tiny $-$};
\node[below] at (sum.center){\tiny $+$};
 
\node [draw,
    minimum width=2cm,
    minimum height=1.2cm,
    right=1.5cm of sum
]  (integrator) {$\int \mathrm{d}t$};
 
\node [draw,
    minimum width=2cm, 
    minimum height=1.2cm,
    left=1cm of sum
] (B) {$B$};
 
\node [draw,
    minimum width=2cm, 
    minimum height=1.2cm, 
    below=1cm of integrator
]  (A) {$A$};

\node [draw,
    minimum width=2cm, 
    minimum height=1.2cm, 
    right=1.5cm of integrator
]  (C) {$C$};

\node [draw,
    minimum width=2cm, 
    minimum height=1.2cm,
    below=1cm of A
] (K) {$K(t)$};
 
\draw[-stealth] (sum.east) -- (integrator.west)
    node[midway,above]{$\dot{\mathbf{x}}(t)$};
 
 \draw[-stealth] (integrator.east) -- (C.west)
    node[midway,above](x){$\mathbf{x}(t)$};
    
\draw[-stealth] (C.east) -- ++ (2,0)
    node[midway,above](y){$\mathbf{y}(t)$};
 
\draw[-stealth] (x.south) |- (A.east);

\draw[-stealth] (x.south) |- (K.east);
 
\draw[-stealth] (A.west) -| (sum.south) 
    node[near end,left]{};
 
 \draw[-stealth] (K.west) -| (B.south) 
    node[near end,left]{};

\draw[-stealth] (B.east) -- (sum.west)
    node[midway,above]{};
 
\end{tikzpicture}
\caption{Closed-loop block diagram describing the controlled system model of $\mathcal{S}$ given in~\eqref{eqn:controlled_system}.} \label{fig:S_controlled_diagram}
\end{figure}
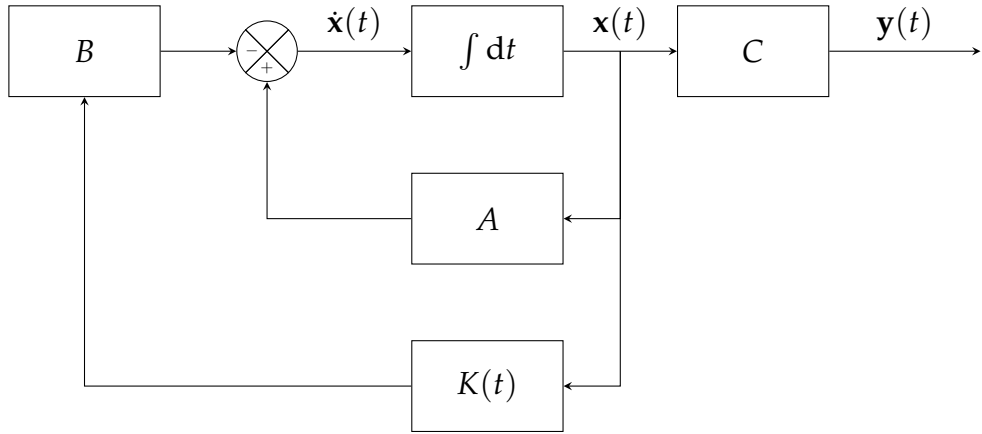
\end{remark}
\begin{definition}
Given an LTI $\mathcal{S}$ system with a linear feedback controller $K$,
the \textbf{controlled dynamics} of $\mathcal{S}$ is defined as:
\begin{equation}
    A_{cl}(t) \coloneqq A-BK(t).
\end{equation}
Thus, the controlled model in~\eqref{eqn:controlled_system} can be described as:

\begin{equation}
    \dot{\mathbf{x}}(t) = A_{cl}(t) \mathbf{x}(t),
    \label{eqn:concise_controlled_system}
\end{equation}

and the diagram in~\ref{fig:S_controlled_diagram} can be described more concisely as:

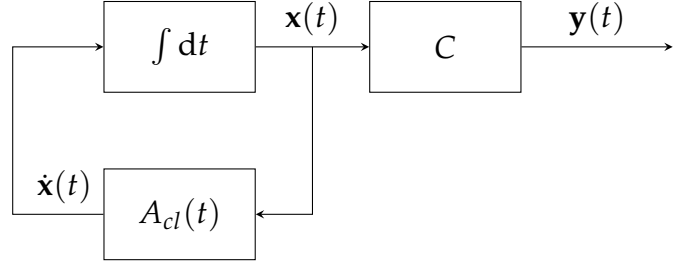
\begin{figure}[H]
\centering
\begin{tikzpicture}
 
\node [draw,
    minimum width=2cm,
    minimum height=1.2cm
]  (integrator) at (0,0){$\int \mathrm{d}t$};
 
\node [draw,
    minimum width=2cm, 
    minimum height=1.2cm, 
    below=1cm of integrator
]  (A) {$A_{cl}(t)$};

\node [draw,
    minimum width=2cm, 
    minimum height=1.2cm, 
    right=1.5cm of integrator
]  (C) {$C$};
 
 
 \draw[-stealth] (integrator.east) -- (C.west)
    node[midway,above](x){$\mathbf{x}(t)$};
    
\draw[-stealth] (C.east) -- ++ (2,0)
    node[midway,above](y){$\mathbf{y}(t)$};
 
\draw[-stealth] (x.south) |- (A.east);
 
\draw[-stealth] (A.west) node[above left] {$\dot{\mathbf{x}}(t)$} -- + (-12mm,0) |- (integrator.west);
 
\end{tikzpicture}
\caption{Concise closed-loop block diagram of the controlled model for $\mathcal{S}$ given in ~\eqref{eqn:controlled_system}.} \label{fig:S_concise_controlled_diagram}
\end{figure}
\end{definition}
\begin{definition}
 Given an LTI system $\mathcal{S}$ with its poles and a linear feedback controller $K$, let us consider the simple case in which $K$ is constant with respect to time. This also makes $A_{cl}$ constant and then Equation~\eqref{eqn:controlled_system} can be easily solved with:

\begin{equation}
    \mathbf{x}(t) = e^{A_{cl}t} \mathbf{x}_0.
    \label{eqn:x_solution_for_constant_controller}
\end{equation}

This is referred to as a \textbf{constant stabilizing controller}. Thus we get the following necessary condition for asymptotic stability of $\mathcal{S}$:

\begin{equation}
    \lim_{t \rightarrow \infty} \| e^{A_{cl}t} \| = 0.
    \label{eqn:LTI_asymtotic_stability}
\end{equation}
 
 This condition is satisfied if and only if all the poles of the controlled dynamics of $\mathcal{S}$ are in the left-hand side of the complex plane, meaning:

\begin{equation}
     \forall \lambda \in \sigma ( A_{cl} ) \colon \ \mathrm{Re}(\lambda) < 0.
     \label{eqn:cons_stab_cond}
\end{equation}

\end{definition}
\begin{remark}
Given an LTI system $\mathcal{S}$, if its unforced system design is not asymptotically stable, one can design a constant controller such that the closed-loop system will be stable, by manually setting the eigenvalues of $A_{cl}$.
\end{remark}



\subsection{Controllability}
The notion of \textbf{controllability} plays a crucial role in the analysis of a control system, as it affects the measure at which the system can be manipulated by the designing engineer.
\begin{definition}
Given a dynamical system $\mathcal{S}$ \textbf{controllability} relates to the ability to move the state vector $\mathbf{x}(t)$ of the given system $\mathcal{S}$ from any initial vector $\mathbf{x}(t_1) \in \mathcal{X}$ to any other final vector $\mathbf{x}(t_2) \in \mathcal{X}$ within some finite time window $[t_1, t_2] \subseteq \mathbb{R}$. The controllability of a system $\mathcal{S}$ is defined with respect to the pair $(A, B)$ and can be defined is several ways. 
\end{definition}
\begin{definition}
Given an LTI system $\mathcal{S}$, its associated \textbf{controllability matrix} $\mathcal{C} \in \mathbb{R}^{n \times n m}$ is defined as:
\begin{equation}
    \mathcal{C} \coloneqq \begin{bmatrix}
          B &  AB & A^2 B \ \cdots \ A^{n-1} B 
            \end{bmatrix}.
\end{equation}
\end{definition}
\begin{theorem}
Given an LTI system $\mathcal{S}$ with closed loop dynamics and its associated controllability matrix $\mathcal{C}$, the pair $(A, B)$ is controllable if the matrix $\mathcal{C} \in \mathbb{R}^{n \times n m}$ has full rank, i.e.:
\begin{equation}
    \rank \left( \mathcal{C} \right) = \min \{ n, nm \} = n.
\end{equation}
\end{theorem}



\begin{definition}
Given an LTI system $\mathcal{S}$ with a state vector $\mathbf{x}(t) = (x_i(t))_{i=1}^n$, let us assume it is not controllable. Thus there exist some $1 \leq r \leq n$ individual variables $X_{n.c.} = \{x_i(t)\}_{i=1}^r \subseteq \{x_i(t)\}_{i=1}^n$ that are not controllable. In this case, we say that $\mathcal{S}$ is \textbf{stabilizable} if all of its uncontrollable state variables $X_{n.c.}$ are stable, in a manner of Lyapunov or asymptotic stability, based on the context.
\end{definition}
\subsection{Observability}
\begin{definition}
Given a dynamical system $\mathcal{S}$, the property of \textbf{observability} correlates to a measure of the ability to infer the individual state variables of a given system, using information about its output. The term relates to the ability to estimate the state vector $\mathbf{x}(t)$ of a given system  using its output $\mathbf{y}(t)$. The observability of a system $\mathcal{S}$ is defined with respect to the pair $(A, C)$ and can be defined is several ways. 
\end{definition}
\begin{definition}
Given an LTI system $\mathcal{S}$, its associated \textbf{observability matrix} $\mathcal{O} \in \mathbb{R}^{np \times n}$ is defined as:
\begin{equation}
    \mathcal{O} \coloneqq \begin{bmatrix}
          C \\  CA \\ C A^2 \\ \vdots \\  C A^{n-1} 
            \end{bmatrix}.
\end{equation}
\end{definition}
\begin{theorem}
Given an LTI system $\mathcal{S}$, the pair $(A,C)$ is observable if the matrix $\mathcal{O} \in \mathbb{R}^{np \times n}$ full rank, i.e.:
\begin{equation}
    \rank \left( \mathcal{O} \right) = \min \{ np, n \} = n
\end{equation}
\end{theorem}

\begin{definition}
Given an LTI system $\mathcal{S}$ with a state vector $\mathbf{x}(t) = (x_i(t))_{i=1}^n$, let us assume it is not observable. Thus there exist some individual variables $X_{n.o.} = \{x_i(t)\}_{i=1}^r \subseteq \{x_i(t)\}_{i=1}^n$ where $1 \leq r \leq n$ that are not observable. In this case, we say that $\mathcal{S}$ is \textbf{detectable} if all of its unobservable state variables $X_{n.o.}$ are stable, in a manner of Lyapunov or asymptotic stability, based on the context.
\end{definition}
\begin{definition}
Given an LTI system $\mathcal{S}$, we say it is \textbf{zero-state observable} if setting its input to and output to zero, i.e. $\mathbf{u}(t) \equiv \mathbf{0}_m$ and  $\mathbf{y}(t) \equiv \mathbf{0}_p$, yields that state vector is also zero, i.e. $\mathbf{x}(t) = \mathbf{0}_n$.
\end{definition}

\section{Scalar HJB Equation}
\label{app:scalar_HJB}

Let us consider the scalar HJB solution, i.e. when $n=m=1$. This considers the system described in~\eqref{eqn:scalar_lti} which reduces Equation~\eqref{eqn:HJB} to the form:

\begin{equation}
\begin{aligned}
    -\frac{1}{4} \frac{b^2}{r} \partial_x J^{*^2}(t) + ax^*(t) \partial_x J{^{*}}(t) + qx^{*^2}(t) = 0,
       \label{eqn:scalar_HJB}
\end{aligned}
\end{equation}

where $x(t), J^*(t), a, b, q, r \in \mathbb{R}$ with $r > 0$ and $q \geq 0$. Thus this can be solved directly using the scalar quadratic equation to obtain:

\begin{equation}
\begin{aligned}
    \partial_x J_1^*(t) & = \frac{ax^*(t) + \sqrt{a^2 x^{*^2}(t)+ \frac{qb^2}{r}x^{*^2}(t)}}{\frac{b^2}{2r} } \\
    & = \frac{ax^*(t) + |x^*(t)|\sqrt{a^2 + \frac{qb^2}{r}}}{\frac{b^2}{2r} };\\
    \partial_x J_2^*(t) & = \frac{ax^*(t) - |x^*(t)|\sqrt{a^2 + \frac{qb^2}{r}}}{\frac{b^2}{2r}}.
    \label{eqn:scalar_cost_1}
\end{aligned}
\end{equation}

Notice since $r > 0$ and $q \geq 0$, the expression in the square root is non-negative and thus both $\partial_x J_1^*(t)$ and $\partial_x J_2^*(t)$ are real functions. We can plug these to the optimal input from~\eqref{eqn:optimal_u} to yield:

\begin{equation}
\begin{aligned}
    u_1^*(t) = & - \frac{b}{2r} \partial_x J_1^{*}(t) = - \frac{b}{2r} \frac{2r}{b^2} \left(ax^*(t) + |x^*(t)|\sqrt{a^2 + \frac{qb^2}{r}}\right)\\
    & - \frac{1}{b} \left(ax^*(t) + |x^*(t)|\sqrt{a^2 + \frac{qb^2}{r}}\right);\\
    u_2^*(t) = & - \frac{1}{b} \left(ax^*(t) - |x^*(t)|\sqrt{a^2 + \frac{qb^2}{r}}\right).
\end{aligned}
\end{equation}

Plugging these two possible inputs to the system model yields the following two possible optimal state dynamics:

\begin{equation}
\begin{aligned}
    \dot{x}^*(t) = & ax^*(t) + bu_1^*(t) = ax^*(t) - b\frac{1}{b} \left(ax^*(t) + |x^*(t)|\sqrt{a^2 + \frac{qb^2}{r}}\right) \\
     = & - |x^*(t)|\sqrt{a^2 + \frac{qb^2}{r}} < 0\\
    \dot{x}^*(t) = & ax^*(t) + bu_2^*(t)  = |x^*(t)|\sqrt{a^2 + \frac{qb^2}{r}} >0
\end{aligned}
\end{equation}

Thus we conclude that only the positive optimal solution $u_1^*(t)$ stabilizes the optimal state and thus $u_2^*(t)$ is discarded.
\begin{definition}
Let us consider an LTI system $\mathcal{S}$. We say that the state vector $\mathbf{x}(t)$ is \textbf{positive} if:
\begin{equation}
    \forall t \in\mathcal{I} \ ; \ x^*(t) > 0
\end{equation}
\end{definition}

Let us assume without loss of generality that the optimal state trajectory is positive. Moreover, let us denote:

\begin{equation*}
    \begin{aligned}
         \psi \coloneqq & \frac{2r\left(a + \sqrt{a^2 + \frac{qb^2}{r}}\right)}{b^2}
    \end{aligned}
\end{equation*}

Plugging this to~\eqref{eqn:scalar_cost_1} and writing $\partial_x J_1^*(t)$ explicitly:

\begin{equation}
\begin{aligned}
    \frac{\partial J_1^*(x^*(t), u^*(t), t)}{\partial x^*(t)} & = \psi x^*(t),
\end{aligned}
\end{equation}

which is a simple, separable PDE that can be integrated directly with respect to $x^*(t)$ to yield:

\begin{equation}
\begin{aligned}
   J_1^*(x^*(t), u^*(t), t) & = \frac{\psi}{2} x^{*^2}(t) + c(u^*(t), t),
\end{aligned}
\end{equation}

where $c$ is a function of $u^*(t)$ and $t$ only. So we have that the optimal cost is quadratic with respect to the state.


   
   

\section{Discrete Time LTI System Modeling}
\label{sect:lti_discrete_appendix}

 In this section we will present the fundamentals for discrete-time control theory. Discrete controllers can be designed using one of these two approaches:

\begin{itemize}
    \item \textbf{Direct design}:\\
    We describe the system as a discrete-time system which is equivalent to the continuous-time system at all the discrete measuring points and then design a discrete controller for the discrete system directly.
    \item \textbf{Indirect design}:\\
    We design a continuous controller for the original continuous system and then approximate the discrete controller by sampling its continuous counterpart at discrete measuring points. 
\end{itemize}

In this work we will focus on direct design of LTI systems and compare its performance with using the continuous model. 

\subsection{LTI Model Discretization}
\label{app:lti_discretization}
We would now want to discretize $\mathcal{S}$ to an equivalent discretized system denoted $\mathcal{S}_D$.

\subsubsection{Time Interval Discretization}
To perform discretization, we would first need to consider discrete intervals for which we would sample the system.
\begin{definition}
Let $L \in \mathbb{N}$ be the number of \textbf{sampling intervals} which we wish to sample $\mathcal{S}$ at.  We would consider values of at least $L \geq 10$. 
\end{definition}
\begin{definition}
In case the system evaluation horizon is finite ($T_f < \infty$) we would discretize the performance interval of the system $\mathcal{I} = [t_0, T_f]$ into $L$ equal intervals, with a length we will refer to as the \textbf{sampling period} $\delta \in (0, 1)$ which we define as:
    \begin{equation}
        \delta \coloneqq \frac{T_f - t_0}{L}.
    \end{equation}
    Notice since we require $\delta < 1$, it must hold that $L >T_f - t_0$.
\end{definition}
\begin{definition}
For each $0 \leq k \leq L-1$, let the $k$'th \textbf{initial and terminal sampling times} $t_{0_k}$ and $T_{f_k}$ be:
    \begin{equation}
        \begin{aligned}
             t_{0_k} & \coloneqq t_0 + k \delta;\\
             T_{f_k} & \coloneqq t_0 + (k+1) \delta.
        \end{aligned}
    \end{equation}
\end{definition}
\begin{definition}
With that, let the
    \textbf{$k$'th sampling interval} be defined as:
    
    \begin{equation}
        I_k \coloneqq \left[t_{0_k}, T_{f_k}\right] \subseteq \mathcal{I}.
    \end{equation}
\end{definition}
\begin{remark}
With that we have:
\begin{equation}
    \begin{aligned}
        \bigcup_{k=0}^{L-1} I_k & = \bigcup_{k=0}^{L-1} \left[t_{0_k}, T_{f_k}\right] = \left[t_{0_0}, T_{f_0}\right] \cup \left[t_{0_1}, T_{f_1}\right] \cup \cdots \cup \left[t_{0_{L-1}}, T_{f_{L-1}}\right] \\
        & = \left[t_0 , t_0 + \delta\right] \cup \left[t_0 + \delta, t_0 + 2\delta\right] \cup \cdots \cup \left[t_0 + (L-1)\delta, t_0 + L\delta \right] \\
        & = \left[t_0 , t_0 + L\delta\right] = \left[t_0 , t_0 + L \cdot \frac{T_f - t_0}{L}\right] = \left[t_0 , T_f\right] =  \mathcal{I},
    \end{aligned}
    \end{equation}
    i.e., the union of all the sampling intervals equals the original interval.
\end{remark}
\begin{definition}
For each such interval $I_k$, let us denote the \textbf{right-open interval} and \textbf{left-open interval} as $I_{k_r}$  and $I_{k_l}$ respectively, where:
    \begin{equation}
        \begin{aligned}
             I_{k_r} \coloneqq & I_{k} \setminus \{ T_{f_k} \} =  \left[t_{0_k}, T_{f_k}\right);\\
             I_{k_l} \coloneqq & I_{k} \setminus \{ t_{0_k} \} =  \left(t_{0_k}, T_{f_k}\right].
        \end{aligned}
    \end{equation}
\end{definition}
\begin{definition}
If $T_f \rightarrow \infty$ then we would require an additional design parameter $\tilde{T}_f \in \mathcal{I}$ we will refer as the \textbf{trimming time} to determine at which finite time to trim the system progression and then divide the interval $\tilde{\mathcal{I}} \coloneqq [t_0, \tilde{T}_f]$ as mentioned for the finite horizon case.
\end{definition}

\subsubsection{State Discretization by Sampling}

\begin{definition}
Let us define the sampled values of $\mathbf{x}(t)$, the state vector of $\mathcal{S}$. For each $0 \leq k \leq L-1$, let us denote the \textbf{$k$'th state sample} as $\mathbf{x}[k]$, with square brackets and define it to be the value of $\mathbf{x}(t)$ at the beginning of the interval $I_k$. In this case we define $\mathbf{x}: \bigcup_{i=0}^{L-1}\{i\} \rightarrow \mathcal{X}_D$, where $\mathcal{X}_D \subseteq \mathcal{X}$ is the set of all sampled state measurements, and:
\begin{equation}
    \mathbf{x}[k] \coloneqq \mathbf{x}(t_{0_k}).
    \label{eqn:state_sampling}
\end{equation}
\end{definition}

\begin{remark}
Notice this definition yields $\mathbf{x}[0]=\mathbf{x}(t_{0_0}) = \mathbf{x}(t_0 + 0 \cdot \delta) = \mathbf{x}(t_0)= \mathbf{x_0}$.
\end{remark}

\subsubsection{Input Discretization by Sampling and Zero-Order-Hold}
  Let us define the sampled values for $\mathbf{u}(t)$, the input vector of $\mathcal{S}$. The go-to technique for input sampling is the \textbf{Zero-Order-Hold} (ZOH) technique, i.e., holding the control signal $\mathbf{u}(t)$ constant at each interval, by means of the following steps:
  \begin{definition}
  For each $0 \leq k \leq L-1$, let us denote the \textbf{$k$'th sampled input measurement} as $\mathbf{u}[k]$ with square brackets, and define it by the value of $\mathbf{u}(t)$ at the beginning of the interval $I_k$. In this case $\mathbf{u}: \bigcup_{i=0}^{L-1}\{i\} \rightarrow \mathcal{U}_D$, where $\mathcal{U}_D \subseteq \mathcal{U}$ is the set of all input measurements, and:
  
  \begin{equation} 
  \mathbf{u}[k] \coloneqq \mathbf{u}(t_{0_k}) .
  \label{eqn:input_sampling}
  \end{equation}
  \end{definition}
  \begin{definition}
  Given a set of input samples $ \left(\mathbf{u}[k]\right)_{k=0}^{L-1}$, we define a new input signal by holding each sample across each right-open interval $I_{k_r}$. In that we define the \textbf{zero-order-hold input signal} $\mathbf{u}_{ZOH}: \mathcal{I} \rightarrow \mathcal{U}_D$, in the following form:

  \begin{equation}
     \mathbf{u}_{ZOH}(\tau) \coloneqq \begin{cases}
    \mathbf{u}[0] & ;\tau \in I_{0_r} \\
    \mathbf{u}[1] & ;\tau \in I_{1_r} \\
    & \vdots\\
    \mathbf{u}[L-1] & ;\tau \in I_{{L-1}_r}
\end{cases}
\label{eqn:zoh}
\end{equation}
  \end{definition}
\begin{remark}
As mentioned, note that we switch $\mathbf{u}_{ZOH}(\tau)$ at each sampling point $(k\delta)_{k=0}^{L-1}$ so it is right-continuous.
\end{remark}

\subsubsection{Discrete LTI System Model}

\begin{proposition}
Given an LTI system $\mathcal{S}$ adhering the model at~\eqref{eqn:basic_sys} with an initial condition $\mathbf{x_0}$, let $L \in \mathbb{N}$ and let $ \left(\mathbf{x}[k]\right)_{k=0}^{L-1}$ and $ \left(\mathbf{u}[k]\right)_{k=0}^{L-1}$ be $L$ corresponding state and input samples measured from $\mathcal{S}$. Then for all $0 \leq k < L-1$, by applying a zero-order-hold on the input samples, an equivalent discretized system $\mathcal{S}_D$ can be expressed by the following recurrence relation:
\begin{equation}
\begin{aligned}
     \mathbf{x}[k+1] = \tilde{A}\mathbf{x}[k] + \tilde{B} \mathbf{u}[k],
\end{aligned}
\end{equation}
with $\mathbf{x}[0] \coloneqq \mathbf{x_0}$, for some matrices $\tilde{A} \in \mathbb{R}^{n \times n}, \tilde{B} \in \mathbb{R}^{n \times m}$.
\end{proposition}
\begin{proof}
Let us define $L$ continuous LTI systems $\left(\mathcal{S}_k\right)_{k=0}^{L-1}$, operating at each interval $I_k$. Note that for all $0 \leq k \leq L-1$ we have $T_{f_k} - t_{0_k} = \delta$. 

\begin{itemize}
\item We define each system $\mathcal{S}_k$ such that it adheres the LTI model at Equation~\eqref{eqn:basic_sys} for all $\tau_k \in I_k$ for which it is defined, but with the assigned input $\mathbf{u}(\tau_k) \equiv \mathbf{u}_{ZOH}(\tau_k)$, when $\mathbf{u}_{ZOH}(\tau_k)$ is defined using the input samples $ \left(\mathbf{u}[k]\right)_{k=0}^{L-1}$ as demonstrated in Equation~\eqref{eqn:zoh}.
    \item For each system $\mathcal{S}_k$ and for any $\tau_k \in I_k$, using the expression for the solution of an LTI system given in~\eqref{eqn:LTI_solution}, we can express the value $\mathbf{x}\left(\tau_k\right)$. Since $\tau_k \in I_k$ we can write: $\tau_k = t_{0_k} + \eta_k$ where $\eta_k \in [0, \delta]$. With that, we have: 
    \begin{equation}
    \begin{aligned}
        \mathbf{x}\left(\tau_k\right) & = \mathbf{x}\left(t_{0_k} + \eta_k\right) \\
        & = e^{A\left(t_{0_k} + \eta_k-t_{0_k}\right)}{\mathbf{x}(t_{0_k})}+\int _{t_{0_k}}^{t_{0_k} + \eta_k}e^{A\left(t_{0_k} + \eta_k-\zeta_k\right)}{B}{\mathbf  {u}}_{ZOH}(\zeta_k)\mathrm{d}\zeta_k\\
        & = e^{\eta_k A}{\mathbf{x}(t_{0_k})} +\int _{t_{0_k}}^{t_{0_k} + \eta_k}e^{A\left(t_{0_k} + \eta_k-\zeta_k\right)}{B}{\mathbf  {u}}_{ZOH}(\zeta_k)\mathrm{d}\zeta_k.
    \end{aligned}
    \end{equation}
    Using the ZOH condition at~\eqref{eqn:zoh} we have that each term $\mathbf{u}_{ZOH}(\zeta_k)$ is constant in the integral bounds and is equal to $\mathbf{u}[k]$. Thus we have:
    \begin{equation}
    \begin{aligned}
        \mathbf{x}(t_{0_k} + \eta_k) & = e^{\eta_k A}{\mathbf{x}(t_{0_k})} +\int _{t_{0_k}}^{t_{0_k} + \eta_k}e^{A\left(t_{0_k} + \eta_k-\zeta_k\right)}{B}\mathrm{d}\zeta_k \ \mathbf{u}[k].
    \end{aligned}
    \end{equation}
    \item For the integral, let us employ $L$ changes of variables of the form: $t_k \coloneqq \tau_k - \zeta_k = t_{0_k} +\eta_k-\zeta_k$. With that and the fact that $B$ is also constant since each system $\mathcal{S}_k$ is LTI, we have:
    \begin{equation}
    \begin{aligned}
        \mathbf{x}(t_{0_k} + \eta_k) = e^{\eta_k A}{\mathbf{x}(t_{0_k})}+\int _{0}^{\eta_k}e^{At_k}\mathrm{d}t_k \ {B}{\mathbf  {u}}[k].
    \end{aligned}
    \end{equation}
    
    Let us denote:

\begin{equation}
    \begin{aligned}
\tilde{A}(\eta) \coloneqq & \ e^{\eta A};\\
\tilde{B}(\eta) \coloneqq & \int_{0}^\eta \tilde{A}(t) \mathrm{d}t \ B = \int_{0}^\eta e^{t A} \mathrm{d}t \ B.
\label{eqn:A_B_tilde}
\end{aligned}
\end{equation}
\begin{definition}

Using the state samples  $ \left(\mathbf{x}[k]\right)_{k=0}^{L-1}$ by their definition at~\eqref{eqn:state_sampling} we define the \textbf{Continuous-Discrete Perturbation Equation} (CDPE) :
\begin{equation}
    \begin{aligned}
        \mathbf{x}(t_{0_k} + \eta_k) = \tilde{A}(\eta_k){\mathbf{x}[k]}+\tilde{B}(\eta_k){\mathbf  {u}}[k].
        \label{eqn:CDPE}
    \end{aligned}
    \end{equation}
\end{definition}

Setting $\eta_k \equiv \delta$, let us further denote:

\begin{equation}
    \begin{aligned}
\tilde{A} \coloneqq & \ e^{\delta A};\\
\tilde{B} \coloneqq & \int_{0}^\delta e^{t A} \mathrm{d}t \ B.
\end{aligned}
\end{equation}

Plugging this to the CDPE at~\eqref{eqn:CDPE} we have that in each interval end we have:

\begin{equation}
    \begin{aligned}
        \mathbf{x}(t_{0_k} + \delta) = \tilde{A}{\mathbf{x}[k]}+\tilde{B}{\mathbf  {u}}[k].
        \label{eqn:discrete_LTI_model_1}
    \end{aligned}
    \end{equation}

Defining $\mathbf{x}[k+1] \coloneqq \mathbf{x}\left(T_{f_k}\right)$ and since $t_{0_k} + \delta = T_{f_k}$ we have that the terminal values $\left(\mathbf{x}\left(T_{f_k}\right)\right)_{k=0}^{L-1}$ plugged into~\eqref{eqn:discrete_LTI_model_1} define a recurrence relation of the form:

\begin{equation}
\begin{aligned}
     \mathbf{x}[k+1] = \tilde{A}\mathbf{x}[k] + \tilde{B} \mathbf{u}[k].
    \label{eqn:discrete_LTI_model_2}
\end{aligned}
\end{equation}

\item The discrete equivalent of $\mathcal{S}$, denoted $\mathcal{S}_D$ is then defined by simulating each system  $\mathcal{S}_k$ according to the recurrence relation for its corresponding left-open interval interval $I_{k_l}$ such that $\mathcal{S}_D$ adheres Equation~\eqref{eqn:discrete_LTI_model_2} with the initial condition $\mathbf{x}[0]=\mathbf{x_0}$.
\end{itemize}
\end{proof}

\subsection{LTI Model Performance Index Discretization}
\label{app:performance_discretization}

As we mentioned, in performing direct design, we convert a continuous system into a discrete one. In that procedure the performance index is also to be discretized. There are several ways to go about this, so we will go about the method suggested at~\cite{Digital_Control_Dynamic_Systems}. Let us consider the infinite and finite horizon cases separately:

\subsubsection{Discretizing LTI Infinite Horizon Performance}

\begin{proposition}
Given an LTI system $\mathcal{S}$ adhering the model at~\eqref{eqn:basic_sys}, with an LQR cost function $J(t)$, let us consider the optimal control problem given in~\eqref{eqn:optimal_control_problem_formulation} with $T_f \rightarrow \infty$. Let $L \in \mathbb{N}$ and let $ \left(\mathbf{x}[k]\right)_{k=0}^{L-1}$ and $ \left(\mathbf{u}[k]\right)_{k=0}^{L-1}$ be corresponding $L$ measured state and input samples of $\mathcal{S}$. Then for all $0 \leq k \leq L-1$, by applying zero-order-hold on the input samples and first-order-hold on the state samples, a discretized cost function $J_D$ can be approximated by:
\begin{equation}
    \begin{aligned}
        J_D\Big(\mathbf{x}[k], \mathbf{u}[k], k\Big) 
        \equiv \sum_{t=k}^{L-1} \Big[  {\mathbf{x}[t]}^T Q \mathbf{x}[t] +  \mathbf{u}^T[t] R \mathbf{u}[t] \Big].
    \end{aligned}
    \end{equation}
\end{proposition}
\begin{proof}

    In the LTI infinite horizon case Equation~\eqref{eqn:general_cost} reduces to Equation~\ref{eqn:infinte_value_integral}. Let $\tilde{T}_f \in \mathcal{I}$ be the trimming time. We would like to consider the cost in the trimmed interval $\tilde{\mathcal{I}} = [t_0, \tilde{T}_f] $, while applying ${\mathbf  {u}}_{ZOH}(t)$ as the input. Let us denote this cost as $\tilde{J}\big( \mathbf{x}(t), \mathbf{u}(t), \tilde{T}_f\big)$ and define it by:
    \begin{equation}
    \begin{aligned}
        \tilde{J}\big( \mathbf{x}(t), \mathbf{u}(t), \tilde{T}_f\big) & \coloneqq J\big( \mathbf{x}(t), \mathbf{u}(t), t_0 \big) - J\big( \mathbf{x}(t), \mathbf{u}(t), \tilde{T}_f \big)\\
        & = \int_{t_0}^{\infty} \Big[ \mathbf{x}^T(\tau)Q\mathbf{x}(\tau) + \mathbf{u}_{ZOH}^T(\tau)R\mathbf{u}_{ZOH}(\tau) \Big]\mathrm{d}\tau \\
        - & \int_{\tilde{T}_f}^{\infty} \Big[ \mathbf{x}^T(\tau)Q\mathbf{x}(\tau) + \mathbf{u}_{ZOH}^T(\tau)R\mathbf{u}_{ZOH}(\tau) \Big] 
        \mathrm{d}\tau \\
        = & \int_{t_0}^{\tilde{T}_f} \Big[ \mathbf{x}^T(\tau)Q\mathbf{x}(\tau) + \mathbf{u}_{ZOH}^T(\tau)R\mathbf{u}_{ZOH}(\tau) \Big] \mathrm{d}\tau\\
        = & \int_{\tilde{\mathcal{I}}} \Big[ \mathbf{x}^T(\tau)Q\mathbf{x}(\tau) + \mathbf{u}_{ZOH}^T(\tau)R\mathbf{u}_{ZOH}(\tau) \Big] \mathrm{d}\tau.
        \label{eqn:J_tilde_1}
    \end{aligned}
    \end{equation}
    With the number of intervals $L \in \mathbb{N}$, let the sampling period be $\delta \equiv \frac{\tilde{T}_f - t_0}{L}$, defining a partitioning of the trimmed interval into $L$ equal intervals $\tilde{I}_k$ such that  $\bigcup_{k=0}^{L-1}\tilde{I}_k = \tilde{\mathcal{I}}$. Now we can rewrite~\eqref{eqn:J_tilde_1} as:
    \begin{equation}
    \begin{aligned}
        \tilde{J}\big( \mathbf{x}(t), \mathbf{u}(t), \tilde{T}_f\big) & = \int_{\tilde{\mathcal{I}}} \Big[ \mathbf{x}^T(\tau)Q\mathbf{x}(\tau) + \mathbf{u}_{ZOH}^T(\tau)R\mathbf{u}_{ZOH}(\tau) \Big] \mathrm{d}\tau \\
        & = \int_{\tilde{I}_0} \Big[ \mathbf{x}^T(\tau_0)Q\mathbf{x}(\tau_0) + \mathbf{u}^T_{ZOH}(\tau_0)R\mathbf{u}_{ZOH}(\tau_0) \Big] \mathrm{d}\tau_0  \\
        + & \int_{\tilde{I}_1} \Big[ \mathbf{x}^T(\tau_2)Q\mathbf{x}(\tau_1) + \mathbf{u}^T_{ZOH}(\tau_1)R\mathbf{u}_{ZOH}(\tau_1) \Big] \mathrm{d}\tau_1 +  \cdots 
        \\
          + & \int_{\tilde{I}_{L-1}} \Big[ \mathbf{x}^T(\tau_{L-1})Q\mathbf{x}(\tau_{L-1}) + \mathbf{u}^T_{ZOH}(\tau_{L-1})R\mathbf{u}_{ZOH}(\tau_{L-1}) \Big] \mathrm{d}\tau_{L-1}
        \\
         = &\sum_{k=0}^{L-1} \int_{\tilde{I}_k} \Big[ \mathbf{x}^T(\tau_k)Q\mathbf{x}(\tau_k) + \mathbf{u}_{ZOH}^T(\tau_k)R\mathbf{u}_{ZOH}(\tau_k) \Big] \mathrm{d}\tau_k\\
         = & \sum_{k=0}^{L-1} \int_{t_{0_k}}^{T_{f_k}} \Big[ \mathbf{x}^T(\tau_k)Q\mathbf{x}(\tau_k) + \mathbf{u}^T_{ZOH}(\tau_k)R\mathbf{u}_{ZOH}(\tau_k) \Big] \mathrm{d}\tau_k .
    \end{aligned}
    \end{equation}
   Let us now use $L$ changes of variables $\eta_k \coloneqq \tau_k - t_{0_k}$. Notice since $\tau_k \in [t_{{0_k}}, T_{{f_k}}] = [t_{{0_k}}, t_{{0_k}} + \delta]$, we have: $\eta_k \in [0, \delta]$. Plugging in, we have:
   \begin{equation}
    \begin{aligned}
        \tilde{J}\big( \mathbf{x}(t), \mathbf{u}(t), \tilde{T}_f\big) 
        & = \sum_{k=0}^{L-1} \int_{0}^{\delta} \Big[ \mathbf{x}^T(\eta_k + t_{0_k})Q\mathbf{x}(\eta_k + t_{0_k}) \\
        & + \mathbf{u}_{ZOH}^T(\eta_k + t_{0_k})R\mathbf{u}_{ZOH}(\eta_k + t_{0_k}) \Big] \mathrm{d}\eta_k .
    \end{aligned}
    \end{equation}
    Since $\eta_k \in [0, \delta]$ we can use the ZOH condition at~\eqref{eqn:zoh} and the CDPE equation at~\eqref{eqn:CDPE} to give:
    \begin{equation}
    \begin{aligned}
         \tilde{J}\big( \mathbf{x}(t), \mathbf{u}(t), \tilde{T}_f\big) & \\
          = \sum_{k=0}^{L-1} \int_{0}^{\delta} \Big[& \left(\tilde{A}(\eta_k){\mathbf{x}[k]}+\tilde{B}(\eta_k){\mathbf  {u}}[k]\right)^TQ\left(\tilde{A}(\eta_k){\mathbf{x}[k]}+\tilde{B}(\eta_k){\mathbf  {u}}[k]\right)
        \\
        & + \mathbf{u}^T[k]R\mathbf{u}[k] \Big] \mathrm{d}\eta_k  \\
          =\sum_{k=0}^{L-1} \int_{0}^{\delta} \Big[& \left({\mathbf{x}^T[k]}\tilde{A}^T(\eta_k)+{\mathbf  {u}}^T[k]\tilde{B}^T(\eta_k)\right)Q\left(\tilde{A}(\eta_k){\mathbf{x}[k]}+\tilde{B}(\eta_k){\mathbf  {u}}[k]\right)
        \\
        & + \mathbf{u}^T[k]R\mathbf{u}[k] \Big] \mathrm{d}\eta_k \\
        = \sum_{k=0}^{L-1} \int_{0}^{\delta} \Big[& {\mathbf{x}^T[k]}\tilde{A}^T(\eta_k)Q\tilde{A}(\eta_k)\mathbf{x}[k]+{\mathbf{x}[k]}^T\tilde{A}^T(\eta_k)Q\tilde{B}(\eta_k){\mathbf  {u}}[k] \\
         + {\mathbf  {u}}^T[k]&\tilde{B}^T(\eta_k)Q\tilde{A}(\eta_k){\mathbf{x}[k]}
         + \mathbf{u}^T[k]\left(\tilde{B}^T(\eta)Q\tilde{B}(\eta_k) + R \right)\mathbf{u}[k] \Big] \mathrm{d}\eta_k \\
         = \sum_{k=0}^{L-1} \Big[ {\mathbf{x}^T[k]}&\int_{0}^{\delta}\tilde{A}^T(\eta_k)Q\tilde{A}(\eta_k)\mathrm{d}\eta_k \  \mathbf{x}[k]+{\mathbf{x}[k]}^T\int_{0}^{\delta}\tilde{A}^T(\eta_k)Q\tilde{B}(\eta_k)\mathrm{d}\eta_k \ {\mathbf  {u}}[k] \\
         + {\mathbf  {u}}^T[k]\int_{0}^{\delta}\tilde{B}^T&(\eta_k)Q\tilde{A}(\eta_k)\mathrm{d}\eta_k \ {\mathbf{x}[k]}
         + \mathbf{u}^T[k]\int_{0}^{\delta}\left(\tilde{B}^T(\eta_k)Q\tilde{B}(\eta_k) + R \right)\mathrm{d}\eta_k \ \mathbf{u}[k] \Big] \\
         = \sum_{k=0}^{L-1} \Big[ {\mathbf{x}^T[k]}&\int_{0}^{\delta}\tilde{A}^T(\eta_k)Q\tilde{A}(\eta_k)\mathrm{d}\eta_k \  \mathbf{x}[k]+2{\mathbf{x}^T[k]}\int_{0}^{\delta}\tilde{A}^T(\eta_k)Q\tilde{B}(\eta_k)\mathrm{d}\eta_k \ {\mathbf  {u}}[k] \\
         & + \mathbf{u}^T[k]\int_{0}^{\delta}\left(\tilde{B}^T(\eta_k)Q\tilde{B}(\eta_k) + R \right)\mathrm{d}\eta_k \ \mathbf{u}[k] \Big].
    \label{eqn:J_tilde_2}
    \end{aligned}
    \end{equation}
    
    Let us denote:
    
    \begin{equation}
    \begin{aligned}
        Q_D \coloneqq & \int_{0}^{\delta}\tilde{A}^T(\eta)Q\tilde{A}(\eta)\mathrm{d}\eta; \\
        N_D \coloneqq & \int_{0}^{\delta}\tilde{A}^T(\eta)Q\tilde{B}(\eta)\mathrm{d}\eta; \\
        R_D \coloneqq & \int_{0}^{\delta}\left(\tilde{B}^T(\eta)Q\tilde{B}(\eta) + R \right)\mathrm{d}\eta.
    \label{eqn:Q_N_R}
    \end{aligned}
    \end{equation}
    Plugging in~\eqref{eqn:J_tilde_2}, we have:
    \begin{equation}
    \begin{aligned}
         \tilde{J}\big( \mathbf{x}(t), \mathbf{u}(t), \tilde{T}_f\big) = \sum_{k=0}^{L-1} \Big[ {\mathbf{x}[k]}^T Q_D  \mathbf{x}[k]+2{\mathbf{x}[k]}^TN_D {\mathbf  {u}}[k] + \mathbf{u}^T[k] R_D \mathbf{u}[k] \Big].
    \end{aligned}
    \end{equation}
    So we transformed (a version of) the continuous cost function to account for the discrete state and input terms. We can then generalize this by considering this sum not necessarily from 0, but for any $0 \leq k \leq L -1$  onwards, to define the discretized cost function:
    \begin{equation}
    \begin{aligned}
        J_D\Big(\mathbf{x}[k], \mathbf{u}[k], k\Big) \equiv \sum_{t=k}^{L-1} \Big[ {\mathbf{x}[t]}^T Q_D  \mathbf{x}[t]+2{\mathbf{x}[t]}^TN_D {\mathbf  {u}}[t] + \mathbf{u}^T[t] R_D \mathbf{u}[t] \Big].
        \label{eqn:discretized_cost}
    \end{aligned}
    \end{equation}
    Let us denote $J_D[k]$ for conciseness.
    \subsubsection{State-Input Coupling}
    With the result we got at~\eqref{eqn:discretized_cost}, we see that a novel complexity has risen in the form of a cost of the coupling between the state and input vectors. There are numerous ways to handle this, though as since we have limited the scope of this work to not consider state-input coupling, we will strive to circumvent this matter by use of an appropriate approximation, covered by the following Lemma:
    \begin{lemma}
    Given an LTI system $\mathcal{S}$ adhering the model at~\eqref{eqn:basic_sys}, with an LQR cost function $J(t)$, let us consider the optimal control problem given in~\eqref{eqn:optimal_control_problem_formulation} with $T_f \rightarrow \infty$. Given $L \in \mathbb{N}$, the sampling period $\delta \in (0,1)$ and corresponding $L$ measured state and input samples
$ \left(\mathbf{x}[k]\right)_{k=0}^{L-1}$ and $ \left(\mathbf{u}[k]\right)_{k=0}^{L-1}$ of $\mathcal{S}$ that yield the discretized cost function $J_D[k]$ given at~\eqref{eqn:discretized_cost}. Then for all $0 \leq k \leq L-1$, by applying first-order-hold on the state samples, $J_D[k]$ can be approximated by:
\begin{equation}
    \begin{aligned}
        J_D\Big(\mathbf{x}[k], \mathbf{u}[k], k\Big) 
        \approx \delta \sum_{t=k}^{L-1} \Big[  {\mathbf{x}[t]}^T Q \mathbf{x}[t] +  \mathbf{u}^T[t] R \mathbf{u}[t] \Big],
    \end{aligned}
    \end{equation}
    i.e., the input-state cross-term can be neglected and when multiplying by $\delta$.
    \end{lemma}
    \begin{proof}
    The matrix $\tilde{A}(\eta)$ is defined at~\eqref{eqn:A_B_tilde} using the matrix exponential of $A$, times some constant $\eta \in (0, \delta]$. This can be written by terms of its power series representation:
    \begin{equation}
        \tilde{A}(\eta) = e^{\eta A} = \sum_{p = 0}^\infty \frac{1}{p!} \left(\eta A\right)^p = \frac{1}{0!} \left(\eta A\right)^0 + \frac{1}{1!} \left(\eta A\right)^1 + \frac{1}{2!} \left(\eta A\right)^2 + \cdots,
    \end{equation}
    where $A^0 = I_{n}$.
    This is where the approximation comes in: we now assume we take a considerably small sampling period, i.e., we assume $\delta \ll 1$. This means $\eta \ll 1$. Taking a linear approximation, we then consider $\eta ^ 2 \approx 0$. Thus we have:
    \begin{equation}
        \tilde{A}(\eta)  = I_{n} + \eta A + \frac{1}{2} \eta^2 A^2 + \cdots  = I_{n} + \eta A + O(\eta ^ 2) \approx  I_{n} + \eta A,
    \end{equation}
    which is equivalent to assuming a \textbf{first-order-hold} (FOH) for the state $\mathbf{x}(t)$, i.e. in discretizing the performance index, we assume the sampling period is so small that the change in state in between each interval is linear. Plugging this to the definition of $\tilde{B}(\eta)$ at~\eqref{eqn:A_B_tilde} we have:
    \begin{equation}
    \begin{aligned}
    \tilde{B}(\eta) = & \int_{0}^\eta e^{t A} \mathrm{d}t \approx \int_{0}^\eta \left( I_{n}  + t A \right) \mathrm{d}t \ B = \int_{0}^\eta  \mathrm{d}t B + \int_{0}^\eta t \mathrm{d}t AB \\
    = & \eta B + \frac{\eta ^2 }{2} AB  = \eta B + O(\eta ^2) \approx \eta B ;
    \end{aligned}
    \end{equation}
    since $\eta \in (0, \delta]$.
    Using the same approximation while plugging  $\tilde{A}(\eta)$ and $\tilde{B}(\eta)$ to~\eqref{eqn:Q_N_R} we have:
    \begin{equation}
    \begin{aligned}
        Q_D = & \int_{0}^{\delta}\tilde{A}^T(\eta)Q\tilde{A}(\eta)\mathrm{d}\eta \approx \int_{0}^{\delta}\left( I_{n}  + \eta A^T \right)Q\left( I_{n}  + \eta A \right)\mathrm{d}\eta \\
        = &  \int_{0}^{\delta}  \mathrm{d}\eta \ Q +  \int_{0}^{\delta} \eta \mathrm{d}\eta \ QA + \int_{0}^{\delta} \eta \mathrm{d}\eta \ A^T Q + \int_{0}^{\delta} \eta^2 \mathrm{d}\eta \ A^T Q A \approx \delta Q; \\
        N_D \coloneqq & \int_{0}^{\delta}\tilde{A}^T(\eta)Q\tilde{B}(\eta)\mathrm{d}\eta \approx \int_{0}^{\delta} \left( I_{n}  + \eta \tilde{A}^T \right)  \eta  \mathrm{d}\eta \ QB \approx 0_{n \times m}; \\
        R_D \coloneqq & \int_{0}^{\delta}\left(\tilde{B}^T(\eta)Q\tilde{B}(\eta) + R \right)\mathrm{d}\eta \approx \int_{0}^{\delta}\left(\eta ^ 2 B^T Q B + R \right)\mathrm{d}\eta \approx \delta R.
        \label{eqn:approx_Q_N_D}
    \end{aligned}
    \end{equation}
    So we see that applying FOH for the state causes the state-input coupling to vanish.
    Plugging the values from~\eqref{eqn:approx_Q_N_D} to the discretized cost at~\eqref{eqn:discretized_cost}, we have:
    
    \begin{equation}
    \begin{aligned}
        J_D\Big(\mathbf{x}[k], \mathbf{u}[k], k\Big) & \approx \sum_{t=k}^{L-1} \Big[ \delta {\mathbf{x}[t]}^T Q \mathbf{x}[t] + \delta \mathbf{u}^T[t] R \mathbf{u}[t] \Big]\\
        & = \delta \sum_{t=k}^{L-1} \Big[  {\mathbf{x}[t]}^T Q \mathbf{x}[t] +  \mathbf{u}^T[t] R \mathbf{u}[t] \Big].
    \end{aligned}
    \end{equation}
\end{proof}    
    \subsubsection{Approximated LTI LQR Discrete Cost}
     We got that the discretized cost is just the sum of the original costs evaluated at the measurement points, times the sampling period length. Since we are in the infinite horizon case, we do not an additional term for the final state, so this means we are just multiplying the discrete cost values by the same fixed constant. Since we are trying to find the minimum of this function, multiplying it by a factor will just cause the minimum to scale by that same factor, so it is redundant. Thus we define:
    \begin{equation}
    \begin{aligned}
        J_D\Big(\mathbf{x}[k], \mathbf{u}[k], k\Big) 
        \equiv \sum_{t=k}^{L-1} \Big[  {\mathbf{x}[t]}^T Q \mathbf{x}[t] +  \mathbf{u}^T[t] R \mathbf{u}[t] \Big].
    \end{aligned}
    \end{equation}
\end{proof}
\subsubsection{Discretizing LTI LQR Finite Horizon Performance}
\begin{proposition}
Given an LTI system $\mathcal{S}$ adhering the model at~\eqref{eqn:basic_sys}, with an LQR cost function $J(t)$, let us consider the optimal control problem given in~\eqref{eqn:optimal_control_problem_formulation} with $T_f < \infty$. Let $L \in \mathbb{N}$ and let $ \left(\mathbf{x}[k]\right)_{k=0}^{L-1}$ and $ \left(\mathbf{u}[k]\right)_{k=0}^{L-1}$ be corresponding $L$ measured state and input samples of $\mathcal{S}$. Then for all $0 \leq k \leq L-1$, by approximating the input by zero-order-hold and state as first-order-hold, a discretized cost function $J_D$ can be approximated by:
\begin{equation}
    \begin{aligned}
        J_D\Big(\mathbf{x}[k], \mathbf{u}[k], k\Big) 
        \equiv \delta \sum_{t=k}^{L-1} \Big[  {\mathbf{x}[t]}^T Q \mathbf{x}[t] +  \mathbf{u}^T[t] R \mathbf{u}[t] \Big] + \mathbf{x}^T(T_f)Q_f\mathbf{x}(T_f).
    \end{aligned}
    \end{equation}
\end{proposition}
\begin{proof}

In the LTI finite horizon case Equation~\eqref{eqn:general_cost} reduces to Equation~\ref{eqn:finite_value_integral}. We evaluate the integral in the already finite interval $\mathcal{I}  = [t_0, T_f]$ and given $L \in \mathbb{N}$, define $\delta \equiv \frac{T_f - t_0}{L}$. Notice now we have an additional term for the cost of the final state, in the form of $\mathbf{x}^T(T_f)Q_f\mathbf{x}(T_f)$. Thus in this case we define:

\begin{equation}
    \begin{aligned}
        \tilde{J}\big( \mathbf{x}(t), \mathbf{u}(t), T_f\big) \coloneqq & J\big( \mathbf{x}(t), \mathbf{u}(t), t_0 \big) - \mathbf{x}^T(T_f)Q_f\mathbf{x}(T_f)\\
        = & \int_{t_0}^{T_f} \Big[ \mathbf{x}^T(\tau)Q\mathbf{x}(\tau) + \mathbf{u}^T(\tau)R\mathbf{u}_{ZOH}(\tau) \Big] \mathrm{d}\tau\\
        = & \int_{\mathcal{I}} \Big[ \mathbf{x}^T(\tau)Q\mathbf{x}(\tau) + \mathbf{u}^T(\tau)R\mathbf{u}_{ZOH}(\tau) \Big] \mathrm{d}\tau.
      \end{aligned}
\end{equation}

We now have the same expression as was at~\eqref{eqn:J_tilde_1}, but for the interval $\mathcal{I}$ instead of $\tilde{\mathcal{I}}$. We can apply the same reasoning as we did for the infinite horizon case to obtain:

 \begin{equation}
    \begin{aligned}
        \tilde{J}\big( \mathbf{x}(t), \mathbf{u}(t), T_f\big) & \approx  \delta \sum_{k=0}^{L-1} \Big[  {\mathbf{x}[k]}^T Q \mathbf{x}[k] +  \mathbf{u}^T[k] R \mathbf{u}[k] \Big].
    \end{aligned}
    \end{equation}

This time, the multiplication by $\delta$ cannot be neglected since we have an additional term in the original cost function the integration did not apply on. So in this case we define:

\begin{equation}
    \begin{aligned}
        J_D\Big(\mathbf{x}[k], \mathbf{u}[k], k\Big) \equiv \delta \sum_{t=k}^{L-1} \Big[  {\mathbf{x}[t]}^T Q \mathbf{x}[t] +  \mathbf{u}^T[t] R \mathbf{u}[t] \Big] + \mathbf{x}^T(T_f)Q_f\mathbf{x}(T_f)
    \end{aligned}
    \end{equation}
    
    \end{proof}
    
\section{\textit{PyDiffGame} Package}

\subsection{Main \texttt{PyDiffGame} Class}
\label{pyth:main}

The following is the constructor to the main \pyth{PyDiffGame} class:

\begin{minted}[linenos,tabsize=2,breaklines]{python}
from __future__ import annotations

import time
import numpy as np
from numpy.linalg import eigvals, norm, LinAlgError
import matplotlib.pyplot as plt
from scipy.linalg import solve_continuous_are, solve_discrete_are
import warnings
from typing import Callable, Union

from abc import ABC, abstractmethod

class PyDiffGame(ABC):
    """
    Differential game abstract base class
    """

    # class fields
    __T_f_default: int = 2
    _L_default: int = 1000
    _epsilon_default: float = 10 ** (-7)
    _eta_default: int = 5

    def __init__(self,
                 A: np.array,
                 B: Union[list[np.array], np.array],
                 Q: Union[list[np.array], np.array],
                 R: Union[list[np.array], np.array],
                 x_0: np.array = None,
                 x_T: np.array = None,
                 T_f: float = None,
                 P_f: list[np.array] = None,
                 show_legend: bool = True,
                 epsilon: float = _epsilon_default,
                 L: int = _L_default,
                 eta: int = _eta_default,
                 force_finite_horizon: bool = False,
                 debug: bool = False
                 ):
\end{minted}

And these are the noteable class methods:

\begin{minted}[linenos,tabsize=2,breaklines]{python}
    def __verify_input(self):
        """
        Input checking method

        Raises
        ------
        Case-specific errors
        """

    def _converge_DREs_to_AREs(self):
        """
        Solves the game as backwards convergence of the differential
        finite-horizon game for repeated consecutive steps until the matrix norm converges
        """

        

    def _post_convergence(method: Callable) -> Callable:
        """
        A decorator static-method to apply on methods that need only be called after convergence
        """


    @_post_convergence
    def _plot(self, t: np.array, mat: np.array, is_P: bool, title: str = None):
        """
        Displays plots for the state variables with respect to time and the convergence of the values of P

        Parameters
        ----------
        t: numpy array of len(data_points)
            The time axis information to plot
        mat: numpy array
            The y-axis information to plot
        is_P: boolean
            Indicates whether to accommodate plotting for all the values in P_i or just for x
        title: str, optional
            The plot title to display
        """


    def __plot_variables(self, mat: np.array):
        """
        Displays plots for the state variables with respect to time

        Parameters
        ----------

        mat: numpy array
            The y-axis information to plot
        """


    def _plot_state_space(self):
        """
        Plots the state vector variables wth respect to time
        """


    def _plot_Y(self, C: np.array):
        """
        Plots the output vector variables wth respect to time

        Parameters
        ----------

        C: numpy array
            The output coefficients with respect to state
        """

    def __get_are_P_f(self) -> list[np.array]:
        """
        Solves the uncoupled set of algebraic Riccati equations to use as initial guesses for fsolve and odeint

        Returns
        ----------
        P_f: numpy array of numpy 2-d arrays, of len(N), of shape(n, n), solution of scipy's solve_are
            Final condition for the Riccati equation matrix
        """


    @abstractmethod
    def _update_K_from_last_state(self, *args):
        """
        Updates the controllers K after forward propagation of the state through time
        """

        pass

    @abstractmethod
    def _get_K_i(self, *args) -> np.array:
        """
        Returns the i'th element of the controllers K
        """

        pass

    @abstractmethod
    def _get_P_f_i(self, i: int) -> np.array:
        """
        Returns the i'th element of the final condition matrices P_f

        Parameters
        ----------
        i: int
            The required index
        """

        pass

    def _update_A_cl_from_last_state(self, k: int = None):
        """
        Updates the closed-loop dynamics with the updated controllers based on the relation:
        A_cl = A - sum_{i=1}^N B_i K_i

        Parameters
        ----------
        k: int, optional
            The current k'th sample index, in case the controller is time-dependant.
            In this case the update rule is:
            A_cl[k] = A - sum_{i=1}^N B_i K_i[k]
        """


    @abstractmethod
    def _update_Ps_from_last_state(self, *args):
        """
        Updates the matrices {P_i}_{i=1}^N after forward propagation of the state through time
        """

        pass

    @abstractmethod
    def is_A_cl_stable(self, *args) -> bool:
        """
        Tests Lyapunov stability of the closed loop

        Returns
        ----------
        is_stable: boolean
            Indicates whether the system has Lyapunov stability or not
        """

        pass

    @abstractmethod
    def _solve_finite_horizon(self):
        """
        Solves for the finite horizon case
        """

        pass

    @_post_convergence
    def _plot_finite_horizon_convergence(self):
        """
        Plots the convergence of the values for the matrices P_i
        """
        
    def _solve_and_plot_finite_horizon(self):
        """
        Solves for the finite horizon case and plots the convergence of the values for the matrices P_i
        """

    @abstractmethod
    def _solve_infinite_horizon(self):
        """
        Solves for the infinite horizon case using the finite horizon case
        """

        pass

    @abstractmethod
    @_post_convergence
    def _solve_state_space(self):
        """
        Propagates the game through time and solves for it
        """

        pass

    def solve_game_and_simulate_state_space(self):
        """
        Propagates the game through time, solves for it and plots the state with respect to time
        """


    @_post_convergence
    def get_costs(self) -> np.array:
        """
        Calculates the cost function value using the formula:
        J_i = int_{t=0}^T_f [ x(t)^T ( Q_i + sum_{j=1}^N  K_j(t)^T R_{ij} K_j(t) ) x(t) ] dt
        """

\end{minted}

\subsection{\texttt{ContinuousPyDiffGame} Class}

\begin{minted}[linenos,tabsize=2,breaklines]{python}
import numpy as np
from scipy.integrate import odeint
from numpy.linalg import eigvals, inv
from typing import Union

from PyDiffGame.PyDiffGame import PyDiffGame


class ContinuousPyDiffGame(PyDiffGame):
    """
    Continuous differential game base class


    Considers the system:
    dx(t)/dt = A x(t) + sum_{j=1}^N B_j v_j(t)
    """

    def __init__(self,
                 A: np.array,
                 B: Union[list[np.array], np.array],
                 Q: Union[list[np.array], np.array],
                 R: Union[list[np.array], np.array],
                 x_0: np.array = None,
                 x_T: np.array = None,
                 T_f: float = None,
                 P_f: list[np.array] = None,
                 show_legend: bool = True,
                 epsilon: float = PyDiffGame._epsilon_default,
                 L: int = PyDiffGame._L_default,
                 eta: int = PyDiffGame._eta_default,
                 force_finite_horizon: bool = False,
                 debug: bool = False
                 ):


        def __solve_N_coupled_diff_riccati(_: float, P_t: np.array) -> np.array:
            """
            odeint Coupled Differential Matrix Riccati Equations Solver function

            Parameters
            ----------
            _: float
                Integration point in time
            P_t: numpy array of numpy 2-d arrays, of len(N), of shape(n, n)
                Current integrated solution

            Returns
            ----------
            dP_tdt: numpy array of numpy 2-d arrays, of len(N), of shape(n, n)
                Current calculated value for the time-derivative of the matrices P_t
            """


    def _update_K_from_last_state(self, t: int):
        """
        After the matrices P_i are obtained, this method updates the controllers at time t by the rule:
        K_i(t) = R^{-1}_ii B^T_i P_i(t)

        Parameters
        ----------
        t: int
            Current point in time
        """

    def _update_Ps_from_last_state(self):
        """
        With P_f as the terminal condition, evaluates the matrices P_i by backwards-solving this set of
        coupled time-continuous Riccati differential equations:

        dP_idt = - A^T P_i(t) - P_i(t) A - Q_i + P_i(t) sum_{j=1}^N B_j R^{-1}_{jj} B^T_j P_j(t) +
                    [sum_{j=1, j!=i}^N P_j(t) B_j R^{-1}_{jj} B^T_j] P_i(t)
        """

    def _solve_finite_horizon(self):
        """
        Considers the following set of finite-horizon cost functions:
        J_i = x(T_f)^T F_i(T_f) x(T_f) + int_{t=0}^T_f [x(t)^T Q_i x(t) + sum_{j=1}^N u_j(t)^T R_{ij} u_j(t)] dt

        In the continuous-time case, the matrices P_i can be solved for, unrelated to the controllers K_i.
        Then when simulating the state space progression, the controllers K_i can be computed as functions of P_i
        for each time interval
        """


    def _solve_infinite_horizon(self):
        """
        Considers the following finite-horizon cost functions:
        J_i = int_{t=0}^infty [x(t)^T Q_i x(t) + sum_{j=1}^N u_j(t)^T R_{ij} u_j(t)] dt
        """

    def is_A_cl_stable(self) -> bool:
        """
        Tests Lyapunov stability of the closed loop:
        A continuous dynamic system governed by the matrix A_cl has Lyapunov stability iff:
        Re(eig) <= 0 forall eigenvalues of A_cl and there is at most one real eigenvalue at the origin
        """


    @PyDiffGame._post_convergence
    def _solve_state_space(self):
        """
        Propagates the game forward through time and solves for it by solving the continuous differential equation:
        dx(t)/dt = A_cl(t) x(t)
        In each step, the controllers K_i are calculated with the values of P_i evaluated beforehand.
        """
        
        def state_diff_eqn(x_t: np.array, _: float) -> np.array:
            """
            Scipy's odeint State Variables Solver function

            Parameters
            ----------
            x_t: numpy 1-d array of shape(n)
                Current integrated state variables
            _: float
                Current time

            Returns
            ----------
            dx_t_dt: numpy 1-d array of shape(n)
                Current calculated value for the time-derivative of the state variable vector x_t
            """
\end{minted}

\subsection{\texttt{DiscretePyDiffGame} Class}

\begin{minted}[linenos,tabsize=2,breaklines]{python}
import numpy as np
from scipy.optimize import fsolve
import quadpy
from numpy.linalg import eigvals, inv
from typing import Union

from PyDiffGame.PyDiffGame import PyDiffGame


class DiscretePyDiffGame(PyDiffGame):
    """
    Discrete differential game base class


    Considers the system:
    x[k+1] = A_tilda x[k] + sum_{j=1}^N B_j_tilda v_j[k]

    where A_tilda and each B_j_tilda are in discrete form, meaning they correlate to a discrete system,
    which, at the sampling points, is assumed to be equal to some equivalent continuous system of the form:
    dx(t)/dt = A x(t) + sum_{j=1}^N B_j v_j(t)

    Parameters
    ----------
    is_input_discrete: boolean, optional, default = False
        Indicates whether the input matrices A, B, Q, R are in discrete form or whether they need to be discretized
    """

    def __init__(self,
                 A: np.array,
                 Q: list[np.array],
                 B: list[np.array],
                 R: list[np.array],
                 is_input_discrete: bool = False,
                 x_0: np.array = None,
                 x_T: np.array = None,
                 T_f: Union[float, int] = None,
                 P_f: list[np.array] = None,
                 show_legend: bool = True,
                 epsilon: float = PyDiffGame._epsilon_default,
                 L: int = PyDiffGame._L_default,
                 eta: int = PyDiffGame._eta_default,
                 force_finite_horizon: bool = False,
                 debug: bool = False
                 ):

        def __solve_for_K_k(K_k_previous: np.array, k_1: int) -> np.array:
            """
            fsolve Controllers Solver function

            Parameters
            ----------
            K_k_previous: numpy array
                The previous estimation for the controllers at the k'th sample point
            k_1: int
                The k+1 index

            Returns
            ----------
            dP_tdt: list of numpy 2-d arrays, of len(N), of shape(n, n)
                Current calculated value for the time-derivative of the matrices P_t
            """


    def __discretize_game(self):
        """
        The input to the discrete game can either be given in continuous or discrete form.
        If it is not in discrete form, the model gets discretized using this method in the following manner:
        A = exp(delta_T A)
        B_i = int_{t=0}^delta_T exp(tA) dt B_i

        The performance index gets discretized in the following manner:
        Q_i = Q_i * delta_T
        R_i = Q_i / delta_T
        """


    def __simpson_integrate_Q(self):
        """
        Use Simpson's approximation for the definite integral of the state cost function term:
        x(t) ^T Q_f_i x(t) + int_{t=0}^T_f x(t) ^T Q_i x(t) ~ sum_{k=0)^K  x(t) ^T Q_i_k x(t)

        where:
        Q_i_0 = 1 / 3 * Q_i
        Q_i_1, ..., Q_i_K-1 = 4 / 3 * Q_i
        Q_i_2, ..., Q_i_K-2 = 2 / 3 * Q_i
        """
        

    def __initialize_finite_horizon(self):
        """
        Initializes the calculated parameters for the finite horizon case by the following steps:
            - The matrices P_i and controllers K_i are initialized randomly for all sample points
            - The closed-loop dynamics matrix A_cl is first set to zero for all sample points
            - The terminal conditions P_f_i are set to Q_i
            - The resulting closed-loop dynamics matrix A_cl for the last sampling point is updated
            - The state is initialized with its initial value, if given
         """


    def __get_K_i_shape(self, i: int) -> range:
        """
        Returns the i'th controller shape indices

        Parameters
        ----------
        i: int
            The desired controller index
        """

    def _update_K_from_last_state(self, k_1: int):
        """
        After initializing, in order to solve for the controllers K_i and matrices P_i,
        we start by solving for the controllers by the update rule:
        K_i[k] = [ R_ii + B_i^T P_i[k+1] B_i ] ^ {-1} B_i^T P_i[k+1] [ A - sum_{j=1, j!=i}^N B_j K_j[k] ]

        Parameters
        ----------
        k_1: int
            The current k+1 sample index
        """


    def _update_Ps_from_last_state(self, k_1: int):
        """
        Updates the matrices P_i with
        A_cl[k] = A - sum_{i=1}^N B_i K_i[k]

        Parameters
        ----------
        k_1: int
            The current k'th sample index
        """

    def is_A_cl_stable(self, k: int) -> bool:
        """
        Tests Lyapunov stability of the closed loop:
        A discrete dynamic system governed by the matrix A_cl has Lyapunov stability iff:
        |eig| <= 1 forall eigenvalues of A_cl
        """


    def _solve_finite_horizon(self):
        """
        Solves the system with the finite-horizon cost functions:
        J_i = sum_{k=1}^T_f [ x[k]^T Q_i x[k] + sum_{j=1}^N u_j[k]^T R_{ij} u_j[k] ]

        In the discrete-time case, the matrices P_i have to be solved simultaneously with the controllers K_i
        """


    def _solve_infinite_horizon(self):
        """
        Solves the system with the infinite-horizon cost functions:
        J_i = sum_{k=1}^infty [ x[k]^T Q_i x[k] + sum_{j=1}^N u_j[k]^T R_{ij} u_j[k] ]
        """


    @PyDiffGame._post_convergence
    def _solve_state_space(self):
        """
        Propagates the game through time and solves for it by solving the discrete difference equation:
        x[k+1] = A_cl[k] x[k]
        In each step, the controllers K_i are calculated with the values of P_i evaluated beforehand.
        """

\end{minted}

\subsection{\texttt{PyDiffARE} Python Implementation}

\begin{minted}[linenos,tabsize=2,breaklines]{python}
import numpy as np

def _converge_DREs_to_AREs(self):
    """
    Solves the game as backwards convergence of the differential
    finite-horizon game for repeated consecutive steps until the matrix norm converges
    """

    x_converged = False
    P_converged = False
    last_norms = []
    curr_iteration_T_f = self._T_f
    x_T = self._x_T if self._x_T is not None else np.zeros_like(self._x_0)

    while not x_converged:

        while not P_converged:
            self._backward_time = np.linspace(start=self._T_f,
                                              stop=self._T_f - self._delta,
                                              num=self._L)
            self._update_Ps_from_last_state()
            self._P_f = self._P[-1]
            last_norms += [norm(self._P_f)]

            if len(last_norms) > self.__eta:
                last_norms.pop(0)

            if len(last_norms) == self.__eta:
                P_converged = all([abs(norm_i - norm_i1) < self.__epsilon for norm_i, norm_i1
                                   in zip(last_norms, last_norms[1:])])
            self._T_f -= self._delta

        self._T_f = curr_iteration_T_f

        if self._x_0 is not None:
            curr_x_T_f_norm = self.simulate_x_T_f()
            x_converged = norm(curr_x_T_f_norm - x_T) < self.__epsilon
            curr_iteration_T_f += 1
            self._T_f = curr_iteration_T_f
            self._forward_time = np.linspace(start=0, stop=self._T_f, num=self._L)
        else:
            x_converged = True
\end{minted}






\cleardoublepage
\phantomsection
\addnumberlesstotoc{chapter}{Bibliography}
\bibliography{references.bib}
\bibliographystyle{template/myieeetr.bst}
\clearpage

\pagenumbering{gobble} 
\begin{titlepage}
    \begin{center}
        \vspace*{1cm}
        
        \includegraphics[width=0.1\textwidth]{logos/bgu.png}\\
        \selectlanguage{hebrew}
        אוניברסיטת בן-גוריון בנגב\\
        הפקולטה למדעי הטבע\\
        המחלקה למדעי המחשב
        
        \vspace{2cm}
        
        {\Large שימוש במשחקים דיפרנציאליים להרכבת הטיפול במשימות בקרה מתחרות}
    
        \vspace{1.5cm}

        יהושע שי קריחלי
        
        \vspace{1cm}
        
        חיבור לשם קבלת התואר ''מגיסטר'' בפקולטה למדעי הטבע
        
        \vspace{1cm}
        
        בהנחיית פרופ' גרא וייס, המחלקה למדעי המחשב\\
        וד"ר שי ארוגטי, המחלקה להנדסת מכונות
        
        \vfill
        חשוון תשפ''ג %
        \selectlanguage{english}
    \end{center}
\end{titlepage}
\clearpage
\selectlanguage{hebrew}
\begin{center}
    \Large
    שימוש במשחקים דיפרנציאליים להרכבת הטיפול
במשימות בקרה מתחרות

    \vspace{0.4cm}
    \large
    יהושע שי קריחלי
       
    \vspace{0.4cm}
    \large
    עבודת גמר לתואר מוסמך למדעי הטבע
      
    \vspace{0.2cm}
    \large
    אוניברסיטת בן-גוריון בנגב
    
    \vspace{0.2cm}
    \large
    
    \selectlanguage{english}
    2022
    
    \selectlanguage{hebrew}
    \vspace{0.6cm}
    \Large
    \textbf{תקציר}
\end{center}

אנו מציגים מתודולוגיית בקרת מערכות חדשנית בשיטת הפרד ומשול
ליישום של משחקים דיפרנציאליים במערכות דינאמיות בעלות סוכן יחיד ומספר רב של משימות. הגישה מבוססת על שיוך כל משימת בקרה לקלט וירטואלי שמתפקד כשחקן במשחק דיפרנציאלי לא שיתופי, בעל מספר צעדים סופי או אינסופי. המשחק מורכב מקבוצה מתואמת של שחקנים נציגים, כאשר כל אחד מנסה להשיג את האסטרטגיה הטובה ביותר עבור המטרה שלו, תוך שהוא מכיר את המדיניות האופטימלית שנבחרה על ידי שאר השחקנים. הגישה המוצגת הינה גמישה בכך שהיא משייכת פונקציית עלות וירטואלית לכל שחקן ובכך למעשה מספקת לכל שחקן פרמטרי משקול משלו עבור הקלט שלו, מצב המערכת כולה וכל שאר הקלטים הווירטואליים. על ידי הבטחת אי-שוויון נאש למשחק הזה, אנחנו מקבלים בקר משולב המספק איזון יציב בין המטרות ומאפשר למהנדס הבקרה לכוון את הפרמטרים מחדש בצורה נגישה לאורך מחזור תהליך התכנון. אנו מספקים פיתוח מתמטי מלא של השיטה המוצעת, הן עבור מערכות רציפות והן עבור מערכות בדידות בזמן, במטרה לשימוש במשימות סוכן יחיד בקנה מידה גדול, שבהן לעתים קרובות ריבוי משימות בקרה משולבות יכול לגרום לקונפליקטים ובכך לגרום לשקלול המערכת הכוללת מראש למאתגר ביותר. כדי להדגים את השימוש בגישה המוצעת, אנו מציגים חבילת פייתון הזמינה כקוד פתוח שבה ישנו שימוש באלגוריתם שפיתחנו לפתרון משוואות ריקטי אלגבריות אשר מופיעות בפיתוח המשחק הדיפרנציאלי לאורך אופק אינסופי. אנו בוחנים שתי מערכות בקרה מוכרות לגיבוש משחקים דיפרנציאליים מתאימים והבטחת פתרון; מטוטלת הפוכה על עגלה נעה ורחפן ארבע-מנועי עם בקרה היררכית, כמערכת לא ליניארית. אנו משווים את הביצועים המתקבלים עם אלה של טכניקת הבקרה האופטימלית RQL על פני מספר מטריקות בקרה ומציגים תוצאות עדיפות.
\selectlanguage{english}
\clearpage
\newpage

a

\newpage
a

\newpage
a

\newpage
a

\newpage

\end{document}